\documentclass[acmsmall,nonacm]{acmart}

\usepackage[vlined,boxed,linesnumbered]{algorithm2e}
\usepackage{amsmath}
\usepackage{amsfonts}	
\usepackage{epsfig}
\usepackage{graphics}
\usepackage{color}
\usepackage{comment}
\usepackage{paralist}
\usepackage{mathtools}
\usepackage{multirow}
\usepackage{subfigure}
\usepackage{eepic}
\usepackage{stmaryrd}
\usepackage{textcomp}
\usepackage{balance}
\usepackage[inline]{enumitem}
\usepackage{ulem}
\usepackage{array}
\usepackage{arydshln}
\usepackage{bigdelim}
\usepackage{makecell}
\usepackage{tikz}
\usetikzlibrary{positioning}

\AtBeginDocument{%
  }

\setcopyright{acmlicensed}
\copyrightyear{2024}
\acmYear{2024}
\acmDOI{XXXXXXX.XXXXXXX}

\acmConference[Conference acronym 'XX]{Make sure to enter the correct
  conference title from your rights confirmation emai}{June 03--05,
  2024}{Woodstock, NY}
\acmISBN{978-1-4503-XXXX-X/18/06}

\begin{document}

\newcommand{\OMIT}[1]{}

\DeclarePairedDelimiter\ceil{\lceil}{\rceil}
\DeclarePairedDelimiter\floor{\lfloor}{\rfloor}

\allowdisplaybreaks

\newcommand{\ghw}{\mathsf{GHW}}
\newcommand{\sjf}{\mathsf{SJF}}
\newcommand{\spantl}{\mathsf{SpanTL}}
\newcommand{\spanl}{\mathsf{SpanL}}
\newcommand{\logcfl}{\mathsf{LOGCFL}}
\newcommand{\nlogspace}{\mathsf{NL}}
\newcommand{\hard}{\text{-hard}}
\newcommand{\complete}{\text{-complete}}
\newcommand{\trees}[2]{\mathsf{trees}_{#1}[#2]}

\newcommand{\oprh}[3]{\mathsf{ORep}_{#3}(#1,#2)}
\newcommand{\probhom}[2]{\probrep{#1}{#2}}
\newcommand{\crsh}[3]{\mathsf{CRS}_{#3}(#1,#2)}

\newcommand{\op}{\mathit{op}}
\newcommand{\PS}{\mathcal{P}}
\newcommand{\viol}[2]{\mathsf{V}(#1,#2)}
\newcommand{\rs}[2]{\mathsf{RS}(#1,#2)}
\newcommand{\rsone}[2]{\mathsf{RS}^1(#1,#2)}
\newcommand{\crs}[2]{\mathsf{CRS}(#1,#2)}
\newcommand{\crss}[3]{\mathsf{CRS}_{#3}(#1,#2)}
\newcommand{\cancrs}[2]{\mathsf{CanCRS}(#1,#2)}
\newcommand{\cancrss}[3]{\mathsf{CanCRS}_{#3}(#1,#2)}
\newcommand{\opr}[2]{\mathsf{ORep}(#1,#2)}
\newcommand{\copr}[2]{\mathsf{CORep}(#1,#2)}
\newcommand{\ops}[3]{\mathsf{Ops}_{#1}(#2,#3)}
\newcommand{\opsone}[3]{\mathsf{Ops}^1_{#1}(#2,#3)}
\newcommand{\abs}[1]{\mathsf{abs}_{>0}(#1)}
\renewcommand{\abs}[1]{\mathsf{RL}(#1)}
\newcommand{\insP}{\ins{P}}
\newcommand\sem[1]{{[\![ #1 ]\!]}}
\newcommand{\probrep}[2]{\mathsf{P}_{#1}(#2)}
\newcommand{\oca}[2]{\mathsf{OCA}_{#2}(#1)}
\newcommand{\ocqa}{\mathsf{OCQA}}
\newcommand{\rrelfreq}[1]{\mathsf{RRFreq}(#1)}
\newcommand{\rrelfreqone}[1]{\mathsf{RRFreq}^1(#1)}
\newcommand{\orfreq}[2]{\mathsf{rrfreq}_{#1}(#2)}
\newcommand{\orfreqone}[2]{\mathsf{rrfreq}^1_{#1}(#2)}
\newcommand{\srelfreq}[1]{\mathsf{SRFreq}(#1)}
\newcommand{\srelfreqone}[1]{\mathsf{SRFreq}^1(#1)}
\newcommand{\srfreq}[2]{\mathsf{srfreq}_{#1}(#2)}
\newcommand{\srfreqone}[2]{\mathsf{srfreq}^1_{#1}(#2)}
\newcommand{\ur}{\mathsf{ur}}
\newcommand{\us}{\mathsf{us}}
\newcommand{\uo}{\mathsf{uo}}

\newcommand{\IS}{\mathsf{IS}}
\newcommand{\ISZ}{\mathsf{IS}_{\neq \emptyset}}
\newcommand{\ISC}{\mathsf{IS}^{\mathsf{con}}}
\newcommand{\CC}{\mathsf{CC}}
\newcommand{\cg}[2]{\mathsf{CG}(#1,#2)}

\newcommand{\crsone}[2]{\mathsf{CRS}^1(#1,#2)}
\newcommand{\coprone}[2]{\mathsf{CORep}^1(#1,#2)}

\newcommand{\mi}[1]{\mathit{#1}}
\newcommand{\ins}[1]{\mathbf{#1}}
\newcommand{\adom}[1]{\mathsf{dom}(#1)}
\renewcommand{\paragraph}[1]{\textbf{#1}}
\newcommand{\ra}{\rightarrow}
\newcommand{\fr}[1]{\mathsf{fr}(#1)}
\newcommand{\dep}{\Sigma}
\newcommand{\sch}[1]{\mathsf{sch}(#1)}
\newcommand{\ar}[1]{\mathsf{ar}(#1)}
\newcommand{\body}[1]{\mathsf{body}(#1)}
\newcommand{\head}[1]{\mathsf{head}(#1)}
\newcommand{\guard}[1]{\mathsf{guard}(#1)}
\newcommand{\class}[1]{\mathbb{#1}}
\newcommand{\pos}[2]{\mathsf{pos}(#1,#2)}
\newcommand{\app}[2]{\langle #1,#2 \rangle}
\newcommand{\crel}[1]{\prec_{#1}}

\newcommand{\ccrel}[1]{\prec_{#1}^+}

\newcommand{\tcrel}[1]{\prec_{#1}^{\star}}
\newcommand{\rctaa}{\class{CT}_{\forall \forall}^{\mathsf{res}}}
\newcommand{\rctaapr}{\mathsf{CT}_{\forall \forall}^{\mathsf{res}}}
\newcommand{\rctae}{\class{CT}_{\forall \exists}^{\mathsf{res}}}
\newcommand{\rctaepr}{\mathsf{CT}_{\forall \exists}^{\mathsf{res}}}
\newcommand{\base}[1]{\mathsf{base}(#1)}
\newcommand{\eqt}[1]{\mathsf{eqtype}(#1)}
\newcommand{\result}[1]{\mathsf{result}(#1)}
\newcommand{\chase}[2]{\mathsf{ochase}(#1,#2)}
\newcommand{\pred}[1]{\mathsf{pr}(#1)}
\newcommand{\origin}[1]{\mathsf{org}(#1)}
\newcommand{\eq}[1]{\mathsf{eq}(#1)}
\newcommand{\dept}[1]{\mathsf{depth}(#1)}

\newcommand{\comp}[2]{\mathsf{comp}_{#2}(#1)}

\newcommand{\rep}[2]{\mathsf{rep}_{#2}(#1)}
\newcommand{\repp}[2]{\mathsf{rep}_{#2}\left(#1\right)}
\newcommand{\rfreq}[2]{\mathsf{rfreq}_{#2}(#1)}
\newcommand{\homs}[3]{\mathsf{hom}_{#2,#3}(#1)}
\newcommand{\prob}[1]{\mathsf{#1}}
\newcommand{\key}[1]{\mathsf{key}(#1)}
\newcommand{\keyval}[2]{\mathsf{key}_{#1}(#2)}
\newcommand{\block}[2]{\mathsf{block}_{#2}(#1)}
\newcommand{\sblock}[2]{\mathsf{sblock}_{#2}(#1)}

\newcommand{\rt}[1]{\mathsf{root}(#1)}
\newcommand{\child}[1]{\mathsf{child}(#1)}

\newcommand{\var}[1]{\mathsf{var}(#1)}
\newcommand{\const}[1]{\mathsf{const}(#1)}
\newcommand{\pvar}[2]{\mathsf{pvar}_{#2}(#1)}

\newcommand{\att}[1]{\mathsf{att}(#1)}
\newcommand{\card}[1]{\sharp #1}

\newcommand{\pr}{\mathsf{Pr}}
\newcommand{\prsp}{\mathsf{PS}}

\newcommand{\sign}[1]{\mathsf{sign}(#1)}
\newcommand{\litval}[2]{\mathsf{lval}_{#2}(#1)}
\newcommand{\angletup}[1]{\langle #1 \rangle}

\def\qed{\hfill{\qedboxempty}      
  \ifdim\lastskip<\medskipamount \removelastskip\penalty55\medskip\fi}

\def\qedboxempty{\vbox{\hrule\hbox{\vrule\kern3pt
                 \vbox{\kern3pt\kern3pt}\kern3pt\vrule}\hrule}}

\def\qedfull{\hfill{\qedboxfull}   
  \ifdim\lastskip<\medskipamount \removelastskip\penalty55\medskip\fi}

\def\qedboxfull{\vrule height 4pt width 4pt depth 0pt}

\newcommand{\markfull}{\qedboxfull}
\newcommand{\markempty}{\qed}

\newcommand{\atrees}[2]{\mathsf{trees}_{#2}(#1)}
\newcommand{\leaves}[1]{\mathsf{leaves}(#1)}
\newcommand{\children}[1]{\mathsf{children}(#1)}

\newcommand{\accept}{\text{\rm accept}}
\newcommand{\init}{\text{\rm init}}
\newcommand{\reject}{\text{\rm reject}}
\newcommand{\spanm}{\text{\rm span}}
\newcommand{\node}{\text{\rm node}}
\newcommand{\assign}{\text{\rm assignment}}
\newcommand{\confs}{\text{\rm configurations}}
\newcommand{\labeling}{L}
\newcommand{\final}{\mathsf{end}}
\newcommand{\process}{\mathsf{Process}}
\newcommand{\pathv}{\mathsf{path}}
\newcommand{\states}{\mathsf{states}}
\newcommand{\curr}{\mathsf{current}}
\newcommand{\per}{\mathsf{per}}

\newtheorem{claim}[theorem]{Claim}
\newtheorem{fact}[theorem]{Fact}
\newtheorem{observation}{Observation}
\newtheorem{remark}{Remark}
\newtheorem{apptheorem}{Theorem}[section]
\newtheorem{appcorollary}[apptheorem]{Corollary}
\newtheorem{appproposition}[apptheorem]{Proposition}
\newtheorem{applemma}[apptheorem]{Lemma}
\newtheorem{appclaim}[apptheorem]{Claim}
\newtheorem{appfact}[apptheorem]{Fact}

\newtheorem{manualtheoreminner}{Theorem}
\newenvironment{manualtheorem}[1]{%
	\renewcommand\themanualtheoreminner{#1}%
	\manualtheoreminner
}{\endmanualtheoreminner}

\newtheorem{manualpropositioninner}{Proposition}
\newenvironment{manualproposition}[1]{%
	\renewcommand\themanualpropositioninner{#1}%
	\manualpropositioninner
}{\endmanualpropositioninner}

\newtheorem{manuallemmainner}{Lemma}
\newenvironment{manuallemma}[1]{%
	\renewcommand\themanuallemmainner{#1}%
	\manuallemmainner
}{\endmanuallemmainner}

\title{Combined Approximations for Uniform Operational Consistent Query Answering}

\author{Marco Calautti}
\affiliation{%
	\institution{University of Milan}
	\country{Italy}
}
\email{marco.calautti@unimi.it}

\author{Ester Livshits}
\affiliation{%
	\institution{University of Edinburgh}
	\country{UK}
}
\email{ester.livshits@ed.ac.uk}

\author{Andreas Pieris}
\affiliation{%
	\institution{University of Edinburgh}
	\country{UK}
}
\affiliation{%
	\institution{University of Cyprus}
	\country{Cyprus}
}
\email{apieris@inf.ed.ac.uk}

\author{Markus Schneider}
\affiliation{%
	\institution{University of Edinburgh}
	\country{UK}
}
\email{m.schneider@ed.ac.uk}

\renewcommand{\shortauthors}{Marco Calautti, Ester Livshits, Andreas Pieris, \& Markus Schneider}

\begin{abstract}
	Operational consistent query answering (CQA) is a recent framework for CQA based on revised definitions of repairs, which are built by applying a sequence of operations (e.g., fact deletions) starting from an inconsistent database until we reach a database that is consistent w.r.t.~the given set of constraints.
	It has been recently shown that there is an efficient approximation for computing the percentage of repairs that entail a given  query when we focus on primary keys, conjunctive queries, and assuming the query is fixed (i.e.,~in data complexity). However, it has been left open whether such an approximation exists when the query is part of the input (i.e.,~in combined complexity). 
	We show that this is the case when we focus on self-join-free conjunctive queries of bounded generelized hypertreewidth. We also show that it is unlikely that efficient approximation schemes exist once we give up one of the adopted syntactic restrictions, i.e., self-join-freeness or bounding the generelized hypertreewidth.
	Towards the desired approximation, we introduce a counting complexity class, called $\mathsf{SpanTL}$, show that each problem in it admits an efficient approximation scheme by using a recent approximability result about tree automata, and then place the problem of interest in $\mathsf{SpanTL}$. 
\end{abstract}

\setcopyright{acmlicensed}
\acmJournal{PACMMOD}
\acmYear{2024} \acmVolume{2} \acmNumber{2 (PODS)} \acmArticle{99} \acmMonth{5}\acmDOI{10.1145/3651600}

\ccsdesc[500]{Theory of computation~Incomplete, inconsistent, and uncertain databases}
\ccsdesc[500]{Information systems~Inconsistent data}

\keywords{inconsistency, consistent query answering, operational semantics, primary keys, conjunctive queries, approximation schemes}

\received{December 2023}
\received[revised]{February 2024}
\received[accepted]{March 2024}

\maketitle

\section{Introduction}\label{sec:introduction}

Operational consistent query answering (CQA) is a recent framework, introduced in 2018~\cite{CaLP18}, that allows us to compute conceptually meaningful answers to queries posed over inconsistent databases, that is, databases that do not conform to their specifications given in the form of constraints, and it is based on revised definitions of repairs and consistent answers.
In particular, the main idea underlying this new framework is to replace the declarative approach to repairs, introduced in the late 1990s by Arenas, Bertossi, and Chomicki~\cite{ArBC99}, and adopted by numerous works (see, e.g.,~\cite{CaCP19,CaCP21,CLPS21,FuFM05,FuMi07,GePW15,KoSu14,KoWi15,KoWi21,MaWi13}), with an {\em operational} one that explains the process of constructing a repair. In other words, we can iteratively apply operations (e.g., fact deletions), starting from an inconsistent database, until we reach a database that is consistent w.r.t.~the given set of constraints. This gives us the flexibility of choosing the probability with which we apply an operation, which in turn allows us to calculate the probability of an operational repair, and thus, the probability with which an answer is entailed. 
Probabilities can be naturally assigned to operations in many scenarios leading to inconsistencies. This is illustrated using the following example from~\cite{CaLP18}.

\begin{example}\label{exa:operational}
	Consider a data integration scenario that results in a database containing the facts ${\rm Emp}(1, {\rm Alice})$ and ${\rm Emp}(1,{\rm Tom})$ that violate the constraint that the first attribute of the relation name ${\rm Emp}$ (the id) is a key.
	Suppose we have a level of
	trust in each of the sources; say we believe that each is 50\% reliable. With probability $0.5 \cdot 0.5 = 0.25$ we do not trust either tuple and apply the operation that removes both facts. With probability $(1-0.25)/2=0.375$ we remove either ${\rm Emp}(1,{\rm Alice})$ or ${\rm Emp}(1,{\rm Tom})$.
	Therefore, we have three repairs, each with its probability: the empty repair with probability $0.25$, and the repairs consisting of ${\rm Emp}(1,{\rm Alice})$ or ${\rm Emp}(1,{\rm Tom})$ with probability $0.375$. 
	%
	Note that the standard CQA approach from~\cite{ArBC99} only allows the removal of one of the two facts (with equal probability $0.5$). It somehow assumes that we trust at least one of the sources, even though they are conflicting. \hfill\markfull
	%
\end{example}

As recently discussed in~\cite{CLPS22}, a natural way of choosing the probabilities assigned to operations is to follow the uniform probability distribution over a reasonable space. The obvious candidate for such a space is the set of operational repairs. 
This led to the so-called {\em uniform} operational CQA. Let us note that in the case of uniform operational repairs, the probability with which an answer is entailed is essentially the percentage of operational repairs that entail the answer in question. We know from~\cite{CLPS22} that, although computing the above percentage is $\sharp ${\rm P}-hard, it can be efficiently approximated via a fully polynomial-time randomized approximation scheme (FPRAS) when we focus on primary keys, conjunctive queries, and the query is fixed (i.e.,~in data complexity). However, it has been left open whether an efficient approximation exists when the query is part of the input (i.e.,~in combined complexity). Closing this problem is the main goal of this work.

\subsection{Our Contributions}

We show that the problem of computing the percentage of operational repairs that entail an answer to the given query admits an FPRAS in combined complexity when we focus on self-join-free conjunctive queries of bounded generalized hypertreewidth.
We further show that it is unlikely to have an efficient approximation scheme once we give up one of the adopted syntactic restrictions, i.e., self-join-freeness or bounding the generelized hypertreewidth.
Note that an analogous approximability result has been recently established for the problem of computing the probability of an answer to a conjunctive query over a tuple-independent probabilistic database~\cite{BrMe23}.
Towards our approximability result, our main technical task is to show that computing the numerator of the ratio in question admits an FPRAS since we know from~\cite{CLPS22} that the denominator can be computed in polynomial time. To this end, we introduce a novel counting complexity class, called $\mathsf{SpanTL}$, establish that each problem in $\mathsf{SpanTL}$ admits an FPRAS, and then place the problem of interest in $\mathsf{SpanTL}$.
The complexity class $\spantl$, apart from being interesting in its own right, will help us to obtain the desired approximation scheme for the problem of interest in a more modular and comprehensible way.

The class $\mathsf{SpanTL}$ relies on alternating Turing machines with output. Although the notion of output is clear for (non-)deterministic Turing machines, it is rather uncommon to talk about the output of an alternating Turing machine. To the best of our knowledge, alternating Turing machines with output have not been considered before. We define an output of an alternating Turing machine $M$ on input $w$ as a node-labeled rooted tree, whose nodes are labeled with (finite) strings over the alphabet of the machine, obtained from a computation of $M$ on $w$.
We then define $\mathsf{SpanTL}$ as the class that collects all the functions that compute the number of distinct accepted outputs of an alternating Turing machine with output with some resource usage restrictions.
For establishing that each problem in $\mathsf{SpanTL}$ admits an FPRAS, we exploit a recent approximability result in the context of automata theory, that is, the problem $\sharp\mathsf{NFTA}$ of counting the number of trees of a certain size accepted by a non-deterministic finite tree automaton admits an FPRAS~\cite{ACJR21}. In fact, we reduce 
each function in $\mathsf{SpanTL}$ to $\sharp\mathsf{NFTA}$.
Finally, for placing our problems in $\mathsf{SpanTL}$, we build upon an idea from~\cite{BrMe23} on how the generalized hypertree decomposition of the query can be used to encode a database as a node-labeled rooted tree. Note, however, that transferring the idea from~\cite{BrMe23} to our setting requires a non-trivial adaptation since we need to encode operational repairs.

At this point, we would like to stress that $\mathsf{SpanTL}$ generalizes the known complexity class $\mathsf{SpanL}$, which collects all the functions that compute the number of distinct accepted outputs of a non-deterministic logspace Turing machine with output~\cite{AlJe93}. We know that each problem in $\mathsf{SpanL}$ admits an FPRAS~\cite{ACJR21-NFA}, but
our problem of interest is unlikely to be in $\mathsf{SpanL}$ (unless $\mathsf{NL} = \mathsf{LOGCFL}$). This is because its decision version, which asks whether the number that we want to compute is greater than zero, is $\mathsf{LOGCFL}\hard$ (this is implicit in~\cite{GoLS02}), whereas the decision version of each counting problem in $\mathsf{SpanL}$ is in $\mathsf{NL}$.


\OMIT{
The obvious question is how the complexity of exact and approximate operational CQA is affected if we assign probabilities to operations according to the above refined ways. In particular, we would like to understand whether uniform operational CQA allows us to go beyond the relatively simple case of primary keys.


Consistent query answering (CQA) is an elegant framework, introduced in the late 1990s by Arenas, Bertossi, and Chomicki~\cite{ArBC99}, that allows us to compute conceptually meaningful answers to queries posed over inconsistent databases, that is, databases that do not conform to their specifications given in the form of integrity constraints.
The key elements underlying CQA are (i) the notion of {\em (database) repair} of an inconsistent database $D$, that is, a consistent database whose difference with $D$ is somehow minimal, and (ii) the notion of query answering based on {\em consistent answers}, that is, answers that are entailed by every repair.


\vspace{10mm}
Since deciding whether a candidate answer is a consistent answer is most commonly intractable in data complexity (in fact, even for primary keys and conjunctive queries, the problem is coNP-hard~\cite{ChMa05}), there was a great effort on drawing the tractability boundary for CQA; see, e.g.,~\cite{FuFM05,FuMi07,GePW15,KoSu14,KoWi15,KoWi21}.
Much of this effort led to interesting dichotomy results that precisely characterize when CQA is tractable/intractable in data complexity.
However, the tractable fragments do not cover many relevant scenarios that go beyond primary keys.

As extensively argued in~\cite{CaLP18}, the goal of a practically applicable CQA approach should be efficient approximate query answering with explicit error guarantees rather than exact query answering. In the realm of the CQA approach described above, one could try to devise efficient probabilistic algorithms with bounded one- or two-sided error. However, it is unlikely that such algorithms exist since, even for very simple scenarios (e.g., primary keys and conjunctive queries), placing the problem in tractable randomized complexity classes such as RP or BPP would imply that the polynomial hierarchy collapses~\cite{KaLi80}.
Another promising idea is to replace the rather strict notion of consistent answers with the more refined notion of relative frequency, that is, the percentage of repairs that entail an answer, and then try to approximate it via a fully polynomial-time randomized approximation scheme (FPRAS); computing it exactly is, unsurprisingly, $\sharp ${\rm P}-hard~\cite{MaWi13}. Indeed, for primary keys and conjunctive queries, one can approximate the relative frequency via an FPRAS; this is implicit in~\cite{DaSu07}, and it has been made explicit in~\cite{CaCP19}. Moreover, a recent experimental evaluation revealed that approximate CQA in the presence of primary keys and conjunctive queries is not unrealistic in practice~\cite{CaCP21}.
However, it seems that the simple case of primary keys is the limit of this approach. We have strong indications that in the case of FDs the problem of computing the relative frequency does not admit an FPRAS, while in the case of keys it is a highly non-trivial problem~\cite{CLPS21}.


The above limitations of the classical CQA approach led the authors of~\cite{CaLP18} to propose a new framework for CQA, based on revised definitions of repairs and consistent answers, which opens up the possibility of efficient approximations with error guarantees. The main idea underlying this new framework is to replace the declarative approach to repairs with an {\em operational} one that explains the process of constructing a repair.
In other words, we can iteratively apply operations (e.g., fact deletions), starting from an inconsistent database, until we reach a database that is consistent w.r.t.~the given set of constraints. This gives us the flexibility of choosing the probability with which we apply an operation, which in turn allows us to calculate the probability of an operational repair, and thus, the probability with which an answer is entailed. 
Probabilities can be naturally assigned to operations in many scenarios leading to inconsistencies. This is illustrated using an example from~\cite{CaLP18}:

\begin{example}\label{exa:operational}
	Consider a data integration scenario that results in a database containing the facts ${\rm Emp}(1, {\rm Alice})$ and ${\rm Emp}(1,{\rm Tom})$ that violate the constraint that the first attribute of the relation name ${\rm Emp}$ (the id) is a key.
	Suppose we have a level of
	trust in each of the sources; say we believe that each is 50\% reliable. With probability $0.5 \cdot 0.5 = 0.25$ we do not trust either tuple and apply the operation that removes both facts. With probability $(1-0.25)/2=0.375$ we remove either ${\rm Emp}(1,{\rm Alice})$ or ${\rm Emp}(1,{\rm Tom})$.
	Therefore, we have three repairs, each with its probability: the empty repair with probability $0.25$, and the repairs consisting of ${\rm Emp}(1,{\rm Alice})$ or ${\rm Emp}(1,{\rm Tom})$ with probability $0.375$. 
	%
	Note that the standard CQA approach only allows the removal of one of the two facts (with equal probability $0.5$). It somehow assumes that we trust at least one of the sources, even though they are in conflict. \hfill\markfull
	%
\end{example}

The preliminary data complexity analysis of operational CQA performed in~\cite{CaLP18} revealed that computing the probability of a candidate answer is $\sharp ${\rm P}-hard and inapproximable, even for primary keys and conjunctive queries.
%
%
%
However, these negative results should not be seen as the end of the story, but rather as the beginning since operational CQA gives us the flexibility to choose the probabilities assigned to operations.
Indeed, the main question left open by~\cite{CaLP18} is the following: how can we choose the probabilities assigned to operations in a conceptually meaningful way and, at the same time, the existence of an FPRAS is guaranteed?

A natural way of choosing those probabilities is to follow the uniform probability distribution over a reasonable space. The obvious candidates for such a space are (i) the set of operational repairs, (ii) the set of sequences of operations that lead to a repair (note that multiple such sequences can lead to the same repair), and (iii) the set of available operations at a certain step of the repairing process. This leads to the so-called {\em uniform operational CQA}; further justification of the above choices is given in Section~\ref{sec:uniform}, where the formal definitions underlying uniform operational CQA are given.
The obvious question is how the complexity of exact and approximate operational CQA is affected if we assign probabilities to operations according to the above refined ways. In particular, we would like to understand whether uniform operational CQA allows us to go beyond the relatively simple case of primary keys.

Our goal is to perform a complexity analysis of uniform operational CQA, and provide answers to the above central questions.
Our main findings can be summarized as follows:

\begin{enumerate}
	\item Exact uniform operational CQA remains $\sharp ${\rm P}-hard, even in the case of primary keys and conjunctive queries.
	
	\item Uniform operational CQA admits an FPRAS if we focus on primary keys and conjunctive queries.
	
\end{enumerate}

\OMIT{
A database is inconsistent if it does not conform to its specifications given in the form of integrity constraints. There is a consensus that inconsistency is a real-life phenomenon that arises due to many reasons such as integration of conflicting sources. With the aim of obtaining conceptually meaningful answers to queries posed over inconsistent databases, Arenas, Bertossi, and Chomicki introduced in the late 1990s the notion of Consistent Query Answering (CQA)~\cite{ArBC99}. The key elements underlying CQA are (i) the notion of {\em (database) repair} of an inconsistent database $D$, that is, a consistent database whose difference with $D$ is somehow minimal, and (ii) the notion of query answering based on {\em certain answers}, that is, answers that are entailed by every repair. A simple example, taken from~\cite{CaCP19}, that illustrates the above notions follows:

\begin{example}\label{exa:cqa}
	Consider the relational schema consisting of a single relation name $\text{\rm Employee}(\text{\rm id}, \text{\rm name}, \text{\rm dept})$ that comes with the constraint that the attribute {\text{\rm id}} functionally determines \text{\rm name} and \text{\rm dept}.
	Consider also the database $D$ consisting of the tuples:
	(1, \text{\rm Bob}, \text{\rm HR}), (1, \text{\rm Bob}, \text{\rm IT}), (2, \text{\rm Alice}, \text{\rm IT}), (2, \text{\rm Tim}, \text{\rm IT}).
	It is easy to see that $D$ is inconsistent since we are uncertain about Bob's department, and the name of the employee with id $2$. To devise a repair, we need to keep one tuple from each conflicting pair, which leads to a maximal subset of $D$ that is consistent.
	Observe now that the query that asks whether employees $1$ and $2$ work in the same department is true only in two out of four repairs, and thus, not entailed. \hfill\markfull
\end{example}

\noindent
\paragraph{Counting Repairs Entailing a Query.}
A key task in this context is to count the number of repairs of an inconsistent database $D$ w.r.t.~a set $\dep$ of constraints that entail a given query $Q$; for clarity, we base our discussion on Boolean queries.
Depending on the shape of the constraints and the query, the data complexity of the above problem can be tractable, i.e., in \text{\rm FP} (the counting analogue of \textsc{PTime}), or intractable, i.e., $\sharp \text{\rm P}$-complete (with $\sharp \text{\rm P}$ being the counting analogue of \text{\rm NP}).
In other words, given a set $\dep$ of constraints and a query $Q$, the problem $\sharp \prob{Repairs}(\dep,Q)$ that takes as input a database $D$, and asks for the number of repairs of $D$ w.r.t.~$\dep$ that entail $Q$, can be tractable or intractable depending on the shape of $\dep$ and $Q$. This leads to the natural question whether we can establish a complete classification, i.e., for every $\dep$ and $Q$, classify $\sharp \prob{Repairs}(\dep,Q)$ as tractable or intractable by simply inspecting $\dep$ and $Q$.

This is a highly non-trivial question for which Maslowski and Wijsen gave an affirmative answer providing that we concentrate on primary keys, i.e., at most one key constraint per relation name, and self-join-free conjunctive queries (SJFCQs), i.e., CQs that cannot mention a relation name more than once~\cite{MaWi13}. 
More precisely, they have established the following dichotomy result: given a set $\dep$ of primary keys, and an SJFCQ $Q$, $\sharp \prob{Repairs}(\dep,Q)$ is either in \text{\rm FP} or $\sharp$\text{\rm P}-complete, and we can determine in polynomial time, by simply analyzing $\dep$ and $Q$, which complexity statement holds.
An analogous dichotomy for arbitrary CQs with self-joins was established by the same authors in~\cite{MaWi14} under the assumption that the primary keys are simple, i.e., they consist of a single attribute. The question whether such a dichotomy result exists for arbitrary primary keys and CQs with self-joins remains a challenging open problem.


Although the picture is rather well-understood for primary keys and SJFCQs, once we go beyond primary keys we know very little concerning the existence of a complete data complexity classification as the one described above.
In particular, the dichotomy result by Maslowski and Wijsen does not apply when we consider the more general class of functional dependencies (FDs), i.e., constraints of the form $R : X \ra Y$, where $X,Y$ are subsets of the set of attributes of $R$, stating that the attributes of $X$ functionally determine the attributes of $Y$.
This brings us to the following question:

\smallskip

\noindent {\em \textbf{Question 1:} Can we lift the dichotomy result for primary keys and SJFCQs to the more general case of functional dependencies?}

\smallskip

The closest known result to the complexity classification asked by Question 1 is for the problem $\sharp \prob{Repairs}(\dep)$, where $\dep$ is a set of FDs, that takes as input a database $D$, and asks for the number of repairs of $D$ w.r.t.~$\dep$ (without considering a query). In particular, we know from~\cite{LiKW21} that whenever $\dep$ has a so-called left-hand side (LHS, for short) chain (up to equivalence), $\sharp \prob{Repairs}(\dep)$ is in \text{\rm FP}; otherwise, it is $\sharp\text{\rm P}$-complete. We also know that checking whether $\dep$ has an LHS chain (up to equivalence) is feasible in polynomial time. Let us recall that a set $\dep$ of FDs has a LHS chain if, for every two FDs $R : X_1 \ra Y_1$ and $R : X_2 \ra Y_2$ of $\dep$, $X_1 \subseteq X_2$ or $X_2 \subseteq X_1$.


\noindent
\paragraph{Approximate Counting.} Another key task in the context of database repairing is to classify $\sharp \prob{Repairs}(\dep,Q)$, for a set of constraints $\dep$ and a query $Q$, as approximable, that is, the target value can be efficiently approximated with error guarantees via a fully polynomial-time randomized approximation scheme (FPRAS), or as inapproximable. Of course, whenever $\sharp \prob{Repairs}(\dep,Q)$ is tractable, then it is trivially approximable. Thus, the interesting task is to classify the intractable cases as approximable or inapproximable.

For a set $\dep$ of primary keys, and a CQ $Q$ (even with self-joins), $\sharp \prob{Repairs}(\dep,Q)$ is always approximable; this is implicit in~\cite{DaSu07}, and it has been made explicit in~\cite{CaCP19}. However, for FDs this is not the case. Depending on the syntactic shape of $\dep$ and $Q$, $\sharp \prob{Repairs}(\dep,Q)$ can be approximable or not; these are actually results of the present work. This leads to the following question:




\noindent
\paragraph{Summary of Contributions.} 
Concerning Question (1), we lift the dichotomy of~\cite{MaWi13} for primary keys and SJFCQs to the general case of FDs (Theorem~\ref{the:fds-dichotomy}). To this end, we build on the dichotomy for the problem of counting repairs (without a query) from~\cite{LiKW21}, which allows us to concentrate on FDs with an LHS chain (up to equivalence) since for all the other cases we can inherit the $\sharp \text{P}$-hardness.
Therefore, our main technical task was actually to lift the dichotomy for primary keys and SJFCQs from~\cite{MaWi13} to the case of FDs with an LHS chain (up to equivalence). Although the proof of this result borrows several ideas from the proof of~\cite{MaWi13}, the task of lifting the result to FDs with an LHS chain (up to equivalence) was a non-trivial one. This is due to the significantly more complex structure of database repairs under FDs with an LHS chain compared to those under primary keys; further details are given in Section~\ref{sec:lhs-chain-fds}.

\OMIT{
Concerning Question (2), although we do not establish a complete classification, we provide results that, apart from being interesting in their own right, are crucial steps towards a complete classification (Theorem~\ref{the:apx-main-result}). After discussing the difficulty underlying a proper 
dichotomy (it will resolve the challenging open problem of whether counting maximal matchings in a bipartite graph is approximable), we show that, for every set $\dep$ of FDs with an LHS chain (up to equivalence) and a CQ $Q$ (even with self-joins), $\sharp \prob{Repairs}(\dep,Q)$ admits an FPRAS. On the other hand, we show that there is a very simple set $\dep$ of FDs such that, for every SJFCQ $Q$, $\sharp \prob{Repairs}(\dep,Q)$ does not admit an FPRAS (under a standard complexity assumption). 
}
}
}

\section{Preliminaries}\label{sec:preliminaries}
%

%
We assume the disjoint countably infinite sets $\ins{C}$ and $\ins{V}$ of {\em constants} and {\em variables}, respectively. For $n > 0$, let $[n]$ be the set $\{1,\ldots,n\}$.

\medskip

\noindent\paragraph{Relational Databases.}
A {\em (relational) schema} $\ins{S}$ is a finite set of relation names with associated arity; we write $R/n$ to denote that $R$ has arity $n > 0$.
A {\em fact} over $\ins{S}$ is an expression of the form $R(c_1,\ldots,c_n)$, where $R/n \in \ins{S}$ and $c_i \in \ins{C}$ for each $i \in [n]$. A {\em database} $D$ over $\ins{S}$ is a finite set of facts over $\ins{S}$. The {\em active domain} of $D$, denoted $\adom{D}$, is the set of constants occurring in $D$.

\medskip

\noindent
\paragraph{Key Constraints.}
A {\em key constraint} (or simply {\em key}) $\kappa$ over a schema $\ins{S}$ is an expression $\key{R} = A$, where $R/n \in \ins{S}$ and $A \subseteq [n]$. Such an expression is called an $R$-key.
Given an $n$-tuple of constants $\bar t = (c_1,\ldots,c_n)$ and a set $A = \{i_1,\ldots,i_m\} \subseteq [n]$, for some $m \in [n]$, we write $\bar t[A]$ for the tuple $(c_{i_1},\ldots,c_{i_m})$, i.e., the projection of $\bar t$ over the positions in $A$.
A database $D$ satisfies $\kappa$, denoted $D \models \kappa$, if, for every two facts $R(\bar t),R(\bar s) \in D$, $\bar t[A] = \bar s[A]$ implies $\bar t = \bar s$. We say that $D$ is {\em consistent} w.r.t.~a set $\dep$ of keys, written $D \models \dep$, if $D \models \kappa$ for each $\kappa \in \dep$; otherwise, it is {\em inconsistent} w.r.t.~$\dep$.
In this work, we focus on sets of {\em primary keys}, i.e., sets of keys that, for each predicate $R$ of the underlying schema, have at most one $R$-key.

\medskip

\noindent
\paragraph{Conjunctive Queries.}
A {\em (relational) atom} $\alpha$ over a schema $\ins{S}$ is an expression $R(t_1,\ldots,t_n)$, where $R/n \in \ins{S}$ and $t_i \in \ins{C} \cup \ins{V}$ for each $i \in [n]$.
A {\em conjunctive query} (CQ) $Q$ over $\ins{S}$ is an expression of the form $\textrm{Ans}(\bar x)\ \text{:-}\ R_1(\bar y_1), \ldots, R_n(\bar y_n)$, where $R_i(\bar y_i)$, for $ i \in [n]$, is an atom over $\ins{S}$, $\bar x$ are the {\em answer variables} of $Q$, and each variable in $\bar x$ is mentioned in $\bar y_i$ for some $i \in [n]$. We may write $Q(\bar x)$ to indicate that $\bar x$ are the answer variables of $Q$. When $\bar x$ is empty, $Q$ is called {\em Boolean}. Moreover, when $Q$ mentions every relation name of $\ins{S}$ at most once it is called {\em self-join-free} and we write $\sjf$ for the class of self-join-free CQs.
The semantics of CQs is given via homomorphisms. Let $\var{Q}$ and $\const{Q}$ be the set of variables and constants in $Q$, respectively. A {\em homomorphism} from a CQ $Q$ of the form $\textrm{Ans}(\bar x)\ \text{:-}\ R_1(\bar y_1), \ldots, R_n(\bar y_n)$ to a database $D$ is a function $h : \var{Q} \cup \const{Q} \ra \adom{D}$, which is the identity over $\ins{C}$, such that $R_i(h(\bar y_i)) \in D$ for each $i \in [n]$.
A tuple $\bar c \in \adom{D}^{|\bar x|}$ is an {\em answer to $Q$ over $D$} if there is a homomorphism $h$ from $Q$ to $D$ with $h(\bar x) = \bar c$. Let $Q(D)$ be the answers to $Q$ over $D$. For Boolean CQs, we write $D \models Q$, and say that $D$ {\em entails} $Q$, if $() \in Q(D)$.

A central class of CQs, which is crucial for our work, is that of {\em bounded generalized hypertreewidth}. The definition of generalized hypertreewidth relies on the notion of tree decomposition.
A \emph{tree decomposition} of a CQ $Q$ of the form $\textrm{Ans}(\bar x)\ \text{:-}\ R_1(\bar y_1), \ldots, R_n(\bar y_n)$ is a pair $(T, \chi)$, where $T = (V,E)$ is a tree and $\chi$ is a labeling function $V \ra 2^{\var{Q} \setminus \bar x}$, i.e., it assigns a subset of $\var{Q} \setminus \bar x$ to each node of $T$, such that (1) for each $i \in [n]$, there exists $v \in V$ such that $\chi(v)$ contains all the variables in $\bar y_i \setminus \bar x$, and (2) for each $t \in \var{Q} \setminus \bar x$, the set $\{v \in V \mid t \in \chi(v)\}$ induces a connected subtree of $T$; the second condition is generally known as the {\em connectedness condition}.
A {\em generalized hypertree decomposition} of $Q$ is a triple $(T,\chi,\lambda)$, where $(T,\chi)$ is a tree decomposition of $Q$ and, assuming $T = (V,E)$, $\lambda$ is a labeling function that assigns a subset of $\{R_1(\bar y_1), \ldots, R_n(\bar y_n)\}$ to each node of $T$ such that, for each $v \in V$, the terms in $\chi(v)$ are ``covered'' by the atoms in $\lambda(v)$, i.e., $\chi(v) \subseteq \bigcup_{R(t_1,\ldots,t_n) \in \lambda(v)} \{t_1,\ldots,t_n\}$.
The {\em width} of $(T,\chi,\lambda)$ is the number $\max_{v \in V} \{|\lambda(v)|\}$, i.e., the maximal size of a set of the form $\lambda(v)$ over all nodes $v$ of $T$. The {\em generalized hypertreewidth} of $Q$ is the minimum width over all its generalized hypertree decompositions.
For $k > 0$, we denote by $\ghw_k$ the class of CQs of generalized hypertreewidth $k$. Recall that $\ghw_1$ coincides with the class of {\em acyclic} CQs.
\section{Uniform Operational CQA}\label{sec:operational-cqa}

We recall the basic notions underlying the operational approach to consistent query answering, focusing on keys, as defined in~\cite{CaLP18}. We then use those notions to recall the recent uniform operational approach to consistent query answering considered in~\cite{CLPS22}.

\medskip

\noindent\paragraph{Operations.}
The notion of operation is the building block of the operational approach, which essentially removes some facts from the database in order to resolve a conflict. As usual, we write $\PS(S)$ for the powerset of a set $S$.

\begin{definition}[\textbf{Operation}]\label{def:operation}
	For a database $D$, a {\em $D$-operation} is a function $\op : \PS(D) \ra
	\PS(D)$ such that, for some non-empty set $F \subseteq D$, for every $D' \in \PS(D)$, $\op(D') = D' \setminus F$. We write $-F$ to refer to this operation. We say that $-F$ is {\em $(D',\dep)$-justified}, where $D' \subseteq D$ and $\dep$ is a set of keys, if $F \subseteq \{f,g\} \subseteq D'$ and $\{f,g\} \not\models \dep$.\hfill\markfull
\end{definition}

The operations $-F$ depend on the database $D$ as they are defined over $D$. Since $D$ will be clear from the context, we may refer to them simply as operations, omitting $D$. 
%
The main idea of the operational approach to CQA is to iteratively apply justified operations, starting from an inconsistent database $D$, until we reach a database $D' \subseteq D$ that is consistent w.r.t. the given set $\dep$ of keys.
Note that justified operations do not try to minimize the number of atoms that need to be removed. As argued in~\cite{CaLP18}, a set of facts that collectively contributes to a key violation should be considered as a justified operation during the iterative repairing process since we do not know a priori which atoms should be deleted.

\OMIT{
However, as discussed in~\cite{CaLP18}, we need to ensure that at each step of this repairing process, at least one violation is resolved. To this end, we need to keep track of all the reasons that cause the inconsistency of $D$ w.r.t.~$\dep$. This brings us to the notion of key violation.

\begin{definition}[\textbf{Key Violation}]\label{def:violation}
	For a database $D$ over a schema $\ins{S}$, a {\em $D$-violation} of a key $\kappa = \key{R} = A$ over $\ins{S}$ is a set $\{f,g\} \subseteq D$ of facts such that $\{f,g\} \not\models \kappa$.	
	We denote the set of $D$-violations of $\kappa$ by $\viol{D}{\kappa}$. Furthermore, for a set $\dep$ of keys, we denote by $\viol{D}{\dep}$ the set $\{(\kappa,v) \mid \kappa \in \dep \textrm{~~and~~} v \in \viol{D}{\kappa}\}$. \hfill\markfull
\end{definition}

Thus, a pair $(\kappa,\{f,g\}) \in \viol{D}{\dep}$ means that one of the reasons why the database $D$ is inconsistent w.r.t.~$\dep$ is because it violates $\kappa$ due to the facts $f$ and $g$.
As discussed in~\cite{CaLP18}, apart from forcing an operation to be fixing, i.e., to fix at least one violation, we also need to force an operation to remove a set of facts only if it contributes as a whole to a violation. Such operations are called justified.

\begin{definition}[\textbf{Justified Operation}]\label{def:justified}
	Let $D$ be a database over a schema $\ins{S}$ and $\dep$ a set of keys over $\ins{S}$. For a database $D' \subseteq D$, a
	$D$-operation $-F$ is called {\em $(D',\dep)$-justified}
	if there exists $(\kappa,\{f,g\}) \in \viol{D'}{\dep}$ such that $F \subseteq \{f,g\}$. \hfill\markfull
\end{definition}

Note that justified operations do not try to minimize the number of atoms that need to be removed. As argued in~\cite{CaLP18}, a set of facts that collectively contributes to a violation should be considered as a justified operation during the iterative repairing process since we do not know a priori which atoms should be deleted, and therefore, we need to explore all the possible scenarios.
}

\medskip

\noindent\paragraph{Repairing Sequences and Operational Repairs.}
The main idea of the operational approach discussed above is formalized via the notion of repairing sequence.
Consider a database $D$ and a set $\dep$ of keys. Given a sequence $s = (\mi{op}_i)_{1 \leq i \leq n}$ of $D$-operations, we define $D_{0}^{s} = D$ and $D_{i}^{s} = \mi{op}_i(D_{i-1}^{s})$ for $i \in [n]$, that is, $D_{i}^{s}$ is obtained by applying to $D$ the first $i$ operations of $s$. 

\begin{definition}[\textbf{Repairing Sequence}]\label{def:rep-sequence}
	Consider a database $D$ and a set $\dep$ of keys. A sequence of $D$-operations $s = (\op_i)_{1 \leq i \leq n}$ is called {\em $(D,\dep)$-repairing} if, for every $i \in [n]$, $\op_i$ is $(D_{i-1}^{s},\dep)$-justified. Let $\rs{D}{\dep}$ be the set of all $(D,\dep)$-repairing sequences. \hfill\markfull
\end{definition}

It is easy to verify that the length of a $(D,\dep)$-repairing sequence is linear in the size of $D$ and that the set $\rs{D}{\dep}$ is finite. For a $(D,\dep)$-repairing sequence $s = (\op_i)_{1 \leq i \leq n}$, we define its {\em result} as the database $s(D) = D_{n}^{s}$ and call it {\em complete} if $s(D) \models \dep$.
Let $\crs{D}{\dep}$ be the set of all complete $(D,\dep)$-repairing sequences.
We can now introduce the central notion of operational repair.

\begin{definition}[\textbf{Operational Repair}]\label{def:operational-repair}
	Let $D$ be a database over a schema $\ins{S}$ and $\dep$ a set of keys over $\ins{S}$. An {\em operational repair} of $D$ w.r.t.~$\dep$ is a database $D'$ such that $D' = s(D)$ for some $s \in \crs{D}{\dep}$. Let $\opr{D}{\dep}$ be the set of all operational repairs of $D$ w.r.t.~$\dep$. \hfill\markfull
\end{definition}


\OMIT{
\medskip

\noindent\paragraph{Operational Repairs.} 
An {\em operational repair} of a database $D$ w.r.t.~a set $\dep$ of keys is a database $D'$ such that $D' = s(D)$ for some $s \in \crs{D}{\dep}$. Let $\opr{D}{\dep}$ be the set of all operational repairs of $D$ w.r.t.~$\dep$.

Although every database of $\copr{D}{\dep}$ corresponds to a conceptually meaningful way of repairing the database $D$, we would like to have a mechanism that allows us to choose which candidate repairs should be considered for query answering purposes, and assign likelihoods to those repairs.

The fact that we can operationally repair an inconsistent database via repairing sequences gives us the flexibility of choosing which operations (that is, fact deletions) are more likely than others, which in turn allows us to talk about the probability of a repair, and thus, the probability with which an answer is entailed.
The idea of assigning likelihoods to operations extending sequences can be described as follows: for all possible extensions $s\cdot \op_1,\ldots, s\cdot \op_k$ of a repairing sequence $s$, we assign
probabilities $p_1,\ldots, p_k$ to them so they add up to $1$. This is done by exploiting a {\em tree-shaped Markov chain} that arranges its states (i.e., repairing sequences) in a rooted tree, where (i) the empty sequence of operations, which is by definition repairing, is the root, (ii) the children of each state are its possible extensions, and (iii) the set of states corresponding to complete sequences coincide with the set of leaves.
We write $\varepsilon$ for the empty sequence of operations.
We further write $\ops{s}{D}{\dep}$ for the set of $(D,\dep)$-repairing sequences $\{s' \in \rs{D}{\dep} \mid s' = s \cdot \op \text{ for some } D\text{-operation } \op\}$.

\begin{definition}[\textbf{Repairing Markov Chain}]\label{def:repaiting-mc}
	For a database $D$ and a set $\dep$ of keys, a {\em $(D,\dep)$-repairing Markov chain} is an edge-labeled rooted tree $T = (V,E,\ins{P})$, where $V = \rs{D}{\dep}$, $E \subseteq V \times V$, and $\ins{P} : E \ra [0,1]$, such that:
	\begin{enumerate}
		\item the root is the empty sequence $\varepsilon$,
		\item for a non-leaf node $s \in V$, $\{s' \mid (s,s') \in E\} = \ops{s}{D}{\dep}$,
		\item for a non-leaf node $s \in V$, $\sum_{t \in \{s' \mid (s,s') \in E\}} \ins{P}(s,t) = 1$, and
		\item $\{s \in V \mid s \text{ is a leaf}\} = \crs{D}{\dep}$.
	\end{enumerate}
	A {\em repairing Markov chain generator} w.r.t.~$\dep$ is a function $M_\dep$ assigning to every database $D$ a $(D,\dep)$-repairing Markov chain. \hfill\markfull
\end{definition}

For a database $D$ and a set $\dep$ of keys, we assume that a $(D,\dep)$-repairing Markov chain $(V,E,\insP)$ is compactly represented as a function $f : \rs{D}{\dep} \times \rs{D}{\dep} \ra [0,1] \cup \{\bot\}$ such that, for every pair $(s,s') \in \rs{D}{\dep} \times \rs{D}{\dep}$, $f(s,s') = \insP(s,s')$ if $(s,s') \in E$; otherwise, $f(s,s') = \bot$.
\OMIT{
We give a simple example, taken from~\cite{CLPS22}, that illustrates the notion of repairing Markov chain:

\begin{figure}[t]
	\centering
	\includegraphics[width=.45\textwidth]{example-mc.pdf}
	\caption{Repairing Markov Chain}
	\label{fig:markov-chain}
\end{figure}

\begin{example}\label{exa:repairing-mc}
	Consider the database $D=\{f_1,f_2,f_3\}$ over the schema $\ins{S} = \{R/3\}$, where $f_1 = R(a_1,b_1,c_1)$, $f_2 = R(a_1,b_2,c_2)$ and $f_3 = R(a_2,b_1,c_2)$. Consider also the set $\dep = \{\phi_1,\phi_2\}$ of FDs over $\ins{S}$, where $\phi_1 = R: A \ra B$ and $\phi_2 = R: C \ra B$, assuming that $(A,B,C)$ is the tuple of attributes of $R$. 
	It is easy to see that $D \not\models \dep$. In particular, we have that $\viol{D}{\dep} = \{(\phi_1,\{f_1,f_2\}), (\phi_2,\{f_2,f_3\})\}$.
	It is easy to verify that for the edge-labeled rooted tree $T = (V,E,\ins{P})$ in Figure~\ref{fig:markov-chain}, $V = \rs{D}{\dep}$, for a non-leaf node $s$ the set of its children is $\ops{s}{D}{\dep}$, and the set of leaves coincides with $\crs{D}{\dep}$. Hence, providing that $p_1+p_2+p_3+p_4+p_5=1$, $p_6+p_7+p_8=1$ and $p_9+p_{10}+p_{11}=1$, $T$ is a $(D,\dep)$-repairing Markov chain. \hfill\markfull
\end{example}
}
The purpose of a repairing Markov chain generator is to provide a way for defining a family of repairing Markov chains independently of the database. One can design a repairing Markov chain generator $M_\dep$ once and for a database $D$, the desired $(D,\dep)$-repairing Markov chain is simply $M_\dep(D)$.

We now recall the notion of operational repair: they are candidate operational repairs obtained via repairing sequences that are {\em reachable} leaves of a repairing Markov chain, i.e., leaves with non-zero probability. The probability of a leaf is coming from the so-called leaf distribution of a repairing Markov chain. Formally, given a database $D$ and a set $\dep$ of keys, the {\em leaf distribution} of a $(D,\dep)$-repairing Markov chain $T = (V,E,\ins{P})$ is a function $\pi$ that assigns to each leaf $s$ of $T$ a number from $[0,1]$ as follows: assuming that $(s_0,s_1)$, $(s_1,s_2)$, $\ldots$, $(s_{n-1},s_n)$, where $n \geq 0$, $\varepsilon = s_0$ and $s = s_n$, is the unique path in $T$ from $\varepsilon$ to $s$, $\pi(s) = \ins{P}(s_0,s_1) \cdot \ins{P}(s_1,s_2) \cdot \cdots \cdot \ins{P}(s_{n-1},s_n)$.
The set of {\em reachable leaves} of $T$, denoted $\abs{T}$, is the set of leaves of $T$ that have non-zero probability according to the leaf distribution of $T$.

\begin{definition}[\textbf{Operational Repair}]\label{def:operational-repair}
	Given a database $D$, a set $\dep$ of keys, and a repairing Markov chain generator $M_\dep$ w.r.t.~$\dep$, an {\em (operational) repair} of $D$ w.r.t.~$M_\dep$ is a database $D' \in \copr{D}{\dep}$ such that $D' = s(D)$ for some $s \in \abs{M_\dep(D)}$. 	Let $\opr{D}{M_\dep}$ be the set of all operational repairs of $D$ w.r.t.~$M_\dep$. \hfill\markfull
\end{definition}

An operational repair may be obtainable via multiple repairing sequences that are reachable leaves of the underlying
repairing Markov chain. The probability of a repair $D'$ is calculated by summing up the probabilities of all reachable leaves $s$ so that $D' = s(D)$.

\begin{definition}[\textbf{Operational Semantics}]\label{def:operational-semantics}
	Given a database $D$, a set $\dep$ of keys, and a repairing Markov chain generator $M_\dep$ w.r.t.~$\dep$, the probability of an operational repair $D'$ of $D$ w.r.t.~$M_\dep$ is
	\[
	\probrep{D,M_\dep}{D'}\ =\ \sum\limits_{s \in \abs{M_\dep(D)} \text{ and } D'=s(D)} \pi(s),
	\]
	where $\pi$ is the leaf distribution of $M_\dep(D)$.
	The {\em operational semantics} of $D$ w.r.t.~$M_\dep$ is defined as the set of repair-probability pairs $\sem{D}_{M_\dep} =
	\left\{\left(D',\probrep{D,M_\dep}{D'}\right) \mid D' \in \opr{D}{M_\dep}\right\}$.
	\hfill\markfull
\end{definition}
}

\medskip

\noindent\paragraph{Uniform Operational CQA.} We can now 
recall the uniform operational approach to consistent query answering introduced in~\cite{CLPS22} and define the main problem of interest.
The goal is to provide a way to measure how certain we are that a candidate answer is indeed a consistent answer. A natural idea is to compute the percentage of operational repairs that entail a candidate answer, which measures how often a candidate answer is an answer to the query if we evaluate it over all operational repairs. Here, we consider all the operational repairs to be equally important; hence the term ``uniform''. This leads to the notion of repair relative frequency. For a database $D$, a set $\dep$ of keys, a CQ $Q(\bar x)$, and a tuple $\bar c \in \adom{D}^{|\bar x|}$, the {\em repair relative frequency} of $\bar c$ w.r.t.~$D$, $\dep$, and $Q$ is
\begin{eqnarray*}
	\mathsf{RF}(D,\dep,Q,\bar c) &=& \frac{|\{D' \in \opr{D}{\dep} \mid \bar c \in Q(D')\}|}{|\opr{D}{\dep}|}.
\end{eqnarray*}


\subsection{Problem of Interest}

The problem of interest in the context of uniform operational CQA, focusing on {\em primary keys}, is defined as follows: for 
a class $\mathsf{Q}$ of CQs (e.g., $\sjf$ or  $\ghw_k$ for $k > 0$),

\medskip

\begin{center}
	\fbox{\begin{tabular}{ll}
			{\small PROBLEM} : & $\ocqa^\ur[\mathsf{Q}]$
			\\
			{\small INPUT} : & A database $D$, a set $\dep$ of primary keys, a query $Q(\bar x)$ from $\mathsf{Q}$,\\
			& a tuple $\bar c \in \adom{D}^{|\bar x|}$.
			\\
			{\small OUTPUT} : &  $\mathsf{RF}(D,\dep,Q,\bar c)$.
	\end{tabular}}
\end{center}

\medskip

\noindent The superscript $\ur$ stands for ``uniform repairs''. We can show that the above problem is hard, even for self-join-free CQs of bounded generalized hypertreewidth. 
Let $\sjf \cap \ghw_k$ for $k > 0$ being the class of self-join-free CQs of generalized hypertreewidth $k$. Then:

\def\theocqaexact{
	For every $k > 0$, $\ocqa^\ur[\sjf \cap \ghw_k]$ is $\sharp ${\rm P}-hard.
}

\begin{theorem}\label{the:ocqa-exact}
\theocqaexact
\end{theorem}

With the above intractability result in place, the question is whether the problem of interest is approximable, i.e.,~whether the target ratio can be approximated via a {\em fully polynomial-time randomized approximation scheme} (FPRAS).
An FPRAS for $\ocqa^\ur[\mathsf{Q}]$ is a randomized algorithm $\mathsf{A}$ that takes as input a database $D$, a set $\dep$ of primary keys, a query $Q(\bar x)$ from $\mathsf{Q}$, a tuple $\bar c \in \adom{D}^{|\bar x|}$, $\epsilon > 0$ and $0 < \delta < 1$, runs in polynomial time in $||D||$, $||\dep||$, $||Q||$, $||\bar c||$,\footnote{As usual, $||o||$ denotes the size of the encoding of a syntactic object $o$.} $1/\epsilon$ and $\log(1/\delta)$, and produces a random variable $\mathsf{A}(D,\dep,Q,\bar c,\epsilon,\delta)$ such that
$
\text{\rm Pr}(|\mathsf{A}(D,\dep,Q,\bar c,\epsilon,\delta) - \mathsf{RF}(D,\dep,Q,\bar c)|\ \leq\ \epsilon \cdot \mathsf{RF}(D,\dep,Q,\bar c))\ \geq\
1-\delta.
$
It turns out that the answer to the above question is negative, even if we focus on self-join-free CQs or CQs of bounded generalized hypertreewidth, under the widely accepted complexity assumption that ${\rm RP} \neq {\rm NP}$. Recall that RP is the complexity class of problems that are efficiently solvable via a randomized algorithm with a bounded one-sided error~\cite{ArBa09}.

\def\theocqaapx{
	Unless  ${\rm RP} = {\rm NP}$,
\begin{enumerate}
	\item There is no FPRAS for $\ocqa^\ur[\sjf]$.
	
	\item For every $k > 0$, there is no FPRAS for $\ocqa^\ur[\ghw_k]$.
\end{enumerate}
}

\begin{theorem}\label{the:ocqa-apx}
\theocqaapx
\end{theorem}

After a quick inspection of the proof of the above result, one can verify that the CQ used in the proof of item (1) is of unbounded generalized hypertreewidth, whereas the CQ used in the proof of item (2) is with self-joins. Hence, the key question that comes up is whether we can provide an FPRAS if we focus on self-join-free CQs of bounded generalized hypertreewidth.
The main result of this work provides an affirmative answer to this question:

\def\themainfpras{
	For every $k > 0$, $\ocqa^\ur[\sjf \cap \ghw_k]$ admits an FPRAS.
}

\begin{theorem}\label{the:main-fpras}
\themainfpras
\end{theorem}

\subsection{Plan of Attack}

For a database $D$ and a set $\dep$ of primary keys, $|\opr{D}{\dep}|$, that is, the denominator of the ratio in question, can be computed in polynomial time~\cite{CLPS22}.
Therefore, to establish Theorem~\ref{the:main-fpras}, it suffices to show that the numerator of the ratio can be efficiently approximated. In fact, it suffices to show that the following auxiliary counting problem admits an FPRAS; the notion of FPRAS for this problem is defined as expected. 
For $k>0$,

\medskip

\begin{center}
	\fbox{\begin{tabular}{ll}
			{\small PROBLEM} : & $\sharp\mathsf{Repairs}[k]$
			\\
			{\small INPUT} : & A database $D$, a set $\dep$ of primary keys, a query $Q(\bar x)$ from $\sjf$,\\
			& a generalized hypertree decomposition of $Q$ of width $k$, a tuple $\bar c \in \adom{D}^{|\bar x|}$.
			\\
			{\small OUTPUT} : &  $|\{D' \in \opr{D}{\dep} \mid \bar c \in Q(D')\}|$.
	\end{tabular}}
\end{center}

\medskip

Note that the above problem takes as input, together with the CQ, a generalized hypertree decomposition of it, whereas the main problem of interest $\ocqa^\ur[\sjf \cap \ghw_k]$ takes as input only the CQ. Despite this mismatch, the existence of an FPRAS for $\sharp\mathsf{Repairs}[k]$, for every $k > 0$, implies Theorem~\ref{the:main-fpras}. This is due to the well-known fact that, given a CQ $Q$ from $\ghw_k$, although it is hard to compute a generalized hypertree decomposition of $Q$ of width $k$, we can compute in polynomial time a generalized hypertree decomposition of $Q$ of width $\ell$, where $k \leq \ell \leq 3k + 1$~\cite{AdGG07,GoLS02}.
%
Hence, our task is to show that $\sharp\mathsf{Repairs}[k]$, for $k > 0$, admits an FPRAS. 
%
\OMIT{
To this end, we are going to introduce the novel counting complexity class $\spantl$, 
show that our counting problem in $\spantl$ admits an FPRAS, and place $\sharp\mathsf{Repairs}[k]$, for every $k>0$, in $\spantl$.
}
%

\OMIT{
 More precisely, for a database $D$, a set $\dep$ of keys, a query $Q(\bar x)$, and a tuple $\bar c \in \adom{D}^{|\bar x|}$, the {\em repair relative frequency} of $\bar c$ being an answer to $Q$ over some operational repair of $D$ is defined as
\begin{eqnarray*}
	\probrep{}{D,M_\dep,Q,\bar c} &=& \sum\limits_{(D',p) \in \sem{D}_{M_\dep} \textrm{~and~} \bar c \in
		Q(D')} p.
\end{eqnarray*} 
We can now talk about operational consistent answers. In particular, the set of {\em operational consistent answers} to $Q$ over $D$ according to $M_\dep$ is defined as the set $\big\{\big(\bar c, \probrep{}{D,M_\dep,Q,\bar c}\big) \mid \bar c \in \adom{D}^{|\bar x|}\big\}$.

The problem of interest in this context is defined as follows:

\medskip

\begin{center}
	\fbox{\begin{tabular}{ll}
			{\small PROBLEM} : & $\ocqa$
			\\
			{\small INPUT} : & A database $D$, a set $\dep$ of primary keys,\\
			& a repairing Markov chain generator $M_\dep$ w.r.t.~$\dep$,\\
			& a query $Q(\bar x)$, and a tuple $\bar c \in \adom{D}^{|\bar x|}$.
			\\
			{\small OUTPUT} : &  $\probrep{M_\dep,Q}{D,\bar c}$.
	\end{tabular}}
\end{center}
}
\section{A Novel Complexity Class}\label{sec:spantl}

\def\pronotinspanl{
	Unless $\nlogspace = \logcfl$, for each $k > 0$, $\sharp\mathsf{Repairs}[k] \not\in \spanl$.
}

Towards our main task, we would like to place $\sharp\mathsf{Repairs}[k]$, for each $k > 0$, in a counting complexity class $\mathsf{C}$ such that every problem in $\mathsf{C}$ admits an FPRAS. Such a complexity class is essentially a collection of functions of the form $f : \Lambda^* \ra \mathbb{N}$, for some alphabet $\Lambda$; note that whenever we use $\mathbb{N}$ as the codomain of such a function, we assume that $0 \in \mathbb{N}$. An FPRAS for a function $f : \Lambda^* \ra \mathbb{N}$ is a randomized algorithm $\mathsf{A}$ that takes as input $w \in \Lambda^*$, $\epsilon > 0$ and $0 < \delta < 1$, runs in polynomial time in $|w|$, $1/\epsilon$ and $\log(1/\delta)$, and produces a random variable $\mathsf{A}(w,\epsilon,\delta)$ such that $\text{\rm Pr}\left(|\mathsf{A}(w,\epsilon,\delta) - f(w)|\ \leq\ \epsilon \cdot f(w)\right)\ \geq\ 1-\delta.$
\OMIT{
{\color{red} In order to show that $\sharp\mathsf{Repairs}[k]$, for $k > 0$, admits an FPRAS, we would like to show that $\sharp\mathsf{Repairs}[k]$ belongs to a complexity class $\mathsf{C}$ such that every function in $\mathsf{C}$ admits an FPRAS. The notion of FPRAS for an arbitrary function $f : \Lambda^* \ra \mathbb{N}$, for some alphabet $\Lambda$, is defined in the obvious way, i.e., an FPRAS for $f$ is a randomized algorithm $\mathsf{A}$ that takes as input $w \in \Lambda^*$, $\epsilon > 0$ and $0 < \delta < 1$, runs in polynomial time in $|w|$, $1/\epsilon$ and $\log(1/\delta)$, and produces a random variable $\mathsf{A}(w,\epsilon,\delta)$ such that
$$\text{\rm Pr}\left(|\mathsf{A}(w,\epsilon,\delta) - f(w)|\ \leq\ \epsilon \cdot f(w)\right)\ \geq\ 1-\delta.$$
}}
The obvious candidate for the class $\mathsf{C}$ is $\spanl$, which, to the best our knowledge, is the largest complexity class of functions admitting an FPRAS~\cite{ArCJR19}. Let us recall the basics about $\spanl$.

\medskip
\noindent \paragraph{The Complexity Class $\spanl$.} 
A function $f : \Lambda^* \ra \mathbb{N}$, for some alphabet $\Lambda$, is in $\spanl$ if there is a non-deterministic logspace Turing machine $M$ with output, with $\Lambda$ being the input alphabet of $M$, such that, for every $w \in \Lambda^*$, $f(w)$ coincides with the number of distinct accepted outputs of $M$ when executed with $w$ as input. It has been recently shown that:

\begin{theorem}[\cite{ArCJR19}]\label{thm:spanl-fpras}
	Every function in $\spanl$ admits an FPRAS.
\end{theorem}

A well-known problem that belongs to $\spanl$ is $\sharp \mathsf{PosDNF}$, which takes as input a Boolean formula $\varphi$ in DNF, with only positive literals, and asks for the number $\sharp \varphi$ of truth assignments that satisfy $\varphi$. We can easily place $\sharp \mathsf{PosDNF}$ in $\spanl$ as follows.
Assume that $\varphi = C_1 \vee \cdots \vee C_m$ with Boolean variables $x_1,\ldots,x_n$. We first guess a disjunct $C_j$ for some $j \in [m]$. Then, for each Boolean variable $x_i$, in some fixed order, if $x_i \in C_j$, we output the symbol $1$; otherwise, we guess a value $\beta \in \{0,1\}$ and output $\beta$. The procedure accepts once all the variables have been considered. 
Roughly, to satisfy $\varphi$, it suffices to assign $1$ to all variables occurring in a disjunct of $\varphi$. Hence, the above procedure first guesses a disjunct $C_j$ to satisfy, and then non-deterministically outputs (the encoding of a) truth assignment that is guaranteed to satisfy $C_j$. 
One can easily verify that the above procedure requires only logarithmic space, and the number of its \emph{distinct} accepted outputs coincides with $\sharp \varphi$. Hence, $\sharp \mathsf{PosDNF} \in \spanl$, and thus, by Theorem~\ref{thm:spanl-fpras}, admits an FPRAS.

The question that comes up is whether $\sharp\mathsf{Repairs}[k]$, for $k > 0$, belongs to $\spanl$. We can show that this cannot be the case, under the widely accepted complexity-theoretic assumption that $\nlogspace \neq \logcfl$; recall that $\logcfl$ is the class of languages (i.e., decision problems) that are reducible in logspace to a context-free language.


\begin{proposition}\label{pro:notinspanl}
	\pronotinspanl
\end{proposition}

The above result is shown by exploiting the fact that, for each function $f : \Lambda^* \rightarrow \mathbb{N} \in \spanl$, its {\em decision version}, that is,
the decision problem that takes as input a string $w \in \Lambda^*$, and asks whether $f(w) > 0$, is in $\nlogspace$, whereas the decision version of $\sharp \mathsf{Repairs}[k]$, for each $k>0$, is $\logcfl\hard$.

\smallskip
\noindent \paragraph{Towards a New Complexity Class.} The above proposition leads to the following question: can we define a new complexity class $\mathsf{C}$ that includes $\spanl$ such that, every function in $\mathsf{C}$ admits an FPRAS, and, for every $k > 0$, $\sharp\mathsf{Repairs}[k] \in \mathsf{C}$.
We believe that such a complexity class, apart from allowing us to solve our main problem, is interesting in its own right.
We provide an affirmative answer to the above question by introducing the counting complexity class $\spantl$ ($\mathsf{T}$ stands for tree and $\mathsf{L}$ for logspace). Its definition relies on alternating Turing machines with output. 
Note that, although the notion of output is clear for (non-)deterministic Turing machines, it is rather uncommon to talk about the output of an alternating Turing machine. To the best of our knowledge, alternating Turing machines with output have not been considered before.

Before proceeding with the formal definition of $\spantl$, let us provide some intuition underlying this new complexity class by means of an analogy with how $\spanl$ is defined. 
Recall that a function $f$ belongs to $\spanl$ due to the existence of a non-deterministic logspace Turing machine $M$ with output such that the number of distinct outputs accepted by $M$ on input $w$ coincides with $f(w)$.
Our observation is that one can see each output string $\alpha_1 \alpha_2 \cdots \alpha_n$ of $M$ on input $w$ as a node-labeled path, i.e., a very simple node-labeled tree $T$ of the form $v_1 \rightarrow \cdots \rightarrow v_n$, where $v_i$ is labeled with $\alpha_i$, for $i \in [n]$, and $n$ is at most polynomial w.r.t.\ $|w|$. 
The tree $T$ is essentially extracted from a computation $\Gamma = C_1 \ra \cdots \ra C_m$ of $M$ on $w$ that outputs $\alpha_1,\ldots,\alpha_n$, where $m \geq n$.\footnote{As usual, a computation of $M$ on $w$ is a branch of the computation tree of $M$ on $w$.} We first contract $\Gamma$ into $\Gamma' = C_{i_1} \ra \cdots \ra C_{i_n}$, where $C_{i_1},\ldots,C_{i_n}$ are the configurations that write $\alpha_1,\ldots,\alpha_n$ on the output tape. Then, $T$ is induced from $\Gamma'$ by associating to each $C_{i_j}$ the node $v_j$ with label $\alpha_j$.

\OMIT{
that outputs the string $\alpha_1 \cdots \alpha_n$ is essentially a sequence of $m \geq n$ steps $s_1,\ldots,s_m$, where at each step the machine moves from one configuration to exactly one configuration, and, assuming that $s_{i_1},\ldots,s_{i_n}$, where $1 \leq i_1 < \cdots < i_n \leq m$, are the steps that write $\alpha_1,\ldots,\alpha_n$ on the output tape, the step $s_{i_1}$ essentially creates the root $v_1$ of $T$ with label $\alpha_1$, whereas the step $s_{i_j}$, for $j \in \{2,\ldots,n\}$, creates the node $v_j$ with label $\alpha_j$ whose parent is the node $v_{j-1}$ created by step $s_{i_{j-1}}$.
}

\OMIT{
Before proceeding with the formal definition of $\spantl$, let us provide some intuition underlying this new complexity class by means of an analogy with how $\spanl$ is defined. 
Recall that a function $f$ belongs to $\spanl$ due to the existence of a non-deterministic logspace Turing machine $M$ with output such that the number of distinct outputs accepted by $M$ on input $w$ coincides with $f(w)$.
Our observation is that one can see each output string $\alpha_1 \alpha_2 \cdots \alpha_n$ of $M$ on input $w$ as a node-labeled path, i.e., a very simple node-labeled tree $T$ of the form $v_1 \rightarrow v_2 \rightarrow \cdots \rightarrow v_n$, where $v_i$ is labeled with $\alpha_i$, for $i \in [n]$, and $n$ is at most polynomial w.r.t.\ $|w|$. 
Thus, a computation of $M$ on $w$ that outputs the string $\alpha_1 \cdots \alpha_n$ is essentially a sequence of $m \geq n$ steps $s_1,\ldots,s_m$, where at each step the machine moves from one configuration to exactly one configuration, and, assuming that $s_{i_1},\ldots,s_{i_n}$, where $1 \leq i_1 < \cdots < i_n \leq m$, are the steps that write $\alpha_1,\ldots,\alpha_n$ on the output tape, the step $s_{i_1}$ essentially creates the root $v_1$ of $T$ with label $\alpha_1$, whereas the step $s_{i_j}$, for $j \in \{2,\ldots,n\}$, creates the node $v_j$ with label $\alpha_j$ whose parent is the node $v_{j-1}$ created by step $s_{i_{j-1}}$.
}


With the above observation in mind, it should not be difficult to see how the notion of output can be extended to alternating Turing machines.
In particular, recall that alternating Turing machines (without output) generalize non-deterministic machines (without output) in the sense that each computation of an alternating machine $M$ on input $w$ is now a \emph{tree} of configurations, where all paths from the root to a leaf of the computation are understood to be executed in parallel, and the computation is accepted if each path from the root to a leaf is accepted. Therefore, an output of an alternating Turing machine $M$ on input $w$ is the node-labeled tree $T$ induced by the contraction of a computation of $M$ on $w$ as illustrated below:

\medskip
\centerline{
	\includegraphics[width=.35\textwidth]{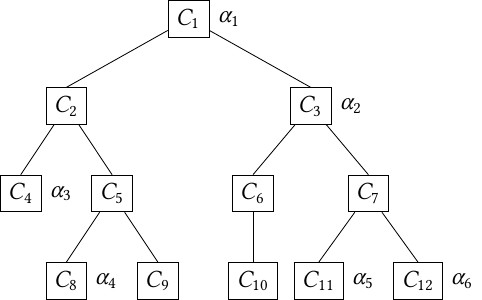}
	\includegraphics[width=.25\textwidth]{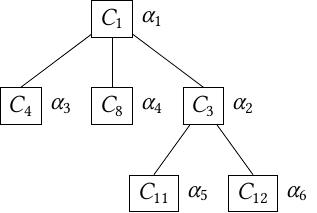}
	\includegraphics[width=.23\textwidth]{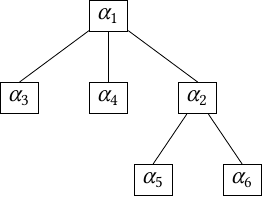}
}
\medskip

\noindent In particular, the leftmost tree is a computation of an alternating Turing Machine with output, where the configurations that output a symbol $\alpha_i$ are flagged with $\alpha_i$, the corresponding contraction is the middle tree, whereas the output is the right tree. 
We proceed, in the next section, to formalize the above intuition.



\OMIT{
we define an \emph{alternating Turing machine with output} as an alternating machine that is able at some step of some path of its computation, to output the label of a fresh node $v$ of a node-labeled tree $T$, and the parent of $v$ in $T$ is implicitly the node whose label has been previously outputted along the same path of the computation.
}


\subsection{Alternating Turing Machines with Output}

We consider alternating Turing machines with a read-only input tape, a read-write working tape, and a write-only one-way ``labeling'' tape. Furthermore, some of the states of the machine are classified as labeling states.
The labeling tape and states, which are rather non-standard in the definition of a Turing machine, should be understood as auxiliary constructs that allow us to build the labels of the nodes of an output. 
Let us clarify that we label nodes with strings instead of individual symbols in order to produce outputs that are easier to interpret. Note, however, that this choice does not lead to a more powerful computational model. In particular, we use a standard alternating Turing machine that, in addition, whenever it enters a labeling state, it, intuitively speaking, produces a fresh node $v$ of the output tree $T$ that is labeled with the string stored on the labeling tape, and then erases the content of the labeling tape, when entering a new configuration.
The parent of $v$ in $T$, as discussed above, is the last node produced along the same path of the computation. Finally, in order for each node of $T$ to have a well-defined parent, the initial state of the machine is labeling.
%

\begin{definition}[\textbf{Model}]\label{def:atm}
	An \emph{alternating Turing machine with output} (ATO) $M$ is a tuple 
	\[(S,\Lambda,s_\text{\rm init},s_\text{\rm accept},s_\text{\rm reject},S_\exists,S_\forall,S_{L},\delta),
	\] 
	where
	\begin{itemize}
		\item[-] $S$ is the finite set of states of $M$,
		\item[-] $\Lambda$ is a finite set of symbols, the alphabet of $M$, including the symbols $\bot$ (blank symbol) and $\triangleright$ (left marker),
		\item[-] $s_\text{\rm init} \in S$ is the initial state of $M$,
		\item[-] $s_\text{accept},s_\text{reject} \in S$ are the accepting and rejecting states of $M$, respectively,
		\item[-] $S_\exists,S_\forall$ are the existential and universal states of $M$, respectively, and they form a partition of $S \setminus \{s_\text{\rm accept},s_\text{\rm reject}\}$,
		\item[-] $S_L \subseteq S$ are the labeling states of $M$ including $s_\text{\rm init}$,
		\item[-] $\delta : (S \setminus \{s_\text{\rm accept},s_\text{\rm reject}\}) \times \Lambda \times \Lambda \rightarrow \PS(S \times \{-1,0,+1\} \times \{-1,0,+1\} \times \Lambda \times \Lambda^*)$ is the transition function of $M$. \hfill\markfull
	\end{itemize}
\end{definition}

The above definition follows the definition of a standard alternating Turing machine with an input and a working tape. The only difference is that some of the states can also be ``labeling'', and the transition function, besides of being able to read (resp., read/write) a symbol from the input (resp., from/to the working) tape, can also write a string from $\Lambda^*$ to the labeling tape.\footnote{ As it is standard with other definitions of Turing machines with output, an ATO can write the empty string on the labeling tape; this is to allow the machine to not write anything in its output when it is not required.}
We now define the semantics of an ATO. The definition is similar to the one for standard alternating Turing machine; we only need to clarify what happens when an ATO visits a state that is also labeling. In what follows, fix an ATO $M = (S,\Lambda,s_\text{\rm init},s_\text{\rm accept},s_\text{\rm reject},S_\exists,S_\forall,S_{L},\delta)$.

\medskip
\noindent \textbf{Configurations.}
A {\em configuration} of the ATO $M$ is a tuple $C = (s,x,y,z,h_x,h_y)$, where $s \in S$, $x,y,z$ are strings from $\Lambda^*$, and $h_x,h_y$ are positive integers.
Intuitively, if $M$ is in configuration $C$, then the input (resp., working, labeling) tape contains the infinite string $x \bot \bot \cdots$ (resp., $y \bot \bot \cdots$, $z \bot \bot \cdots$) and the cursor of the input (resp., working) tape points to the cell $h_x$, (resp., $h_y$). As usual, we use the left marker, which means that $x$ and $y$ are always starting with $\triangleright$. Moreover, $\delta$ is restricted in such a way that $\triangleright$ occurs exactly once in $x$ and $y$ and always as the first symbol.

\medskip
\noindent \textbf{Configuration Transition.}
We now define how $M$ transitions from one configuration to multiple configurations. Roughly, if $M$ is in some configuration $C$, then it simultaneously enters each configuration obtained from $C$ by applying one of the compatible transitions in $\delta$, similarly to a standard alternating Turing machine. The only addition is that when $M$ is in a configuration whose state is labeling, it first erases the content of the labeling tape before applying the transition function. This reflects the fact that the current content of the labeling tape is used to produce the label of a fresh node of the output, and the content of the labeling tape must be erased to allow the machine to construct a new label for a subsequent node. Formally,
assume that $M$ is in configuration $C = (s,x,y,z,h_x,h_y)$, where $x[h_x] = \alpha$ and $y[h_y] = \beta$, and assume that $\delta(s,\alpha,\beta) = \{(s_i,d_i,d_i',\beta_i,z_i)\}_{i \in [n]}$ for some $n > 0$. We write $C \ra_M \{C_1,\ldots,C_n\}$ to say that $M$ simultaneously enters the configurations $C_i = (s_i,x,y'_i,z_i',h_x',h_{y'_i})$ for $i \in [n]$, where $y'_i$ is obtained from $y$ by replacing $y[h_y] = \beta$ with $\beta_i$, $z_i' = z_i$ (resp., $z_i' = z \cdot z_i$) if $s \in S_L$ (resp., $s \not\in S_L$), $h_x' = h_x + d_i$, and $h_{y'_i} = h_y + d_i'$. 
Note that the cursors of the input and the working tapes cannot move left of the symbol $\triangleright$, which is achieved by restricting the transition function.
 
\medskip
\noindent \textbf{Computation.}
We are now ready to define the notion of computation of $M$ on input $w$. This is defined in exactly the same way as the computation of a standard alternating Turing machine on input $w$, with the difference that we now use the relation $\rightarrow_M$ defined above to determine the configurations that $M$ enters.
For a string  $w = w_1 \cdots w_n \in (\Lambda \setminus \{\bot,\triangleright\})^*$, the {\em initial configuration} of $M$ on input $w$ is $(s_\text{\rm init},\triangleright w, \triangleright, \epsilon,1,1)$. We call a configuration of $M$ {\em accepting} (resp., {\em rejecting}) if its state is $s_\text{\rm accept}$ (resp., $s_\text{\rm reject}$). We further call $C$ {\em existential} (resp., {\em universal}, {\em labeling}) if its state belongs to $S_{\exists}$ (resp., $S_\forall$, $S_L$).
We proceed to define the notion of computation of an ATO.

\begin{definition}[\textbf{Computation}]\label{def:computation}
	A {\em computation} of $M$ on input $w \in (\Lambda \setminus \{\bot,\triangleright\})^*$ is a finite node-labeled rooted tree $T = (V,E,\lambda)$, where $\lambda$ assigns to the nodes of $T$ configurations of $M$, s.t.:
	\begin{enumerate}
		\item[-] if $v \in V$ is the root node of $T$, then $\lambda(v)$ is the initial configuration of $M$ on input $w$,
		
		\item[-] if $v \in V$ is a leaf node of $T$, then $\lambda(v)$ is either an accepting or a rejecting configuration of $M$,
		
		\item[-] if $v \in V$ is a non-leaf node of $T$ with $\lambda(v) \ra_M \{C_1,\ldots,C_n\}$ for some $n > 0$, then: $\lambda(v)$ is existential implies $v$ has exactly one child $u \in V$ with $\lambda(u) = C_i$ for some $i \in [n]$, and $\lambda(v)$ is universal implies $v$ has $n$ children $u_1,\ldots,u_n \in V$ with $\lambda(u_i) = C_i$ for each $i \in [n]$.
	\end{enumerate}
	We say that $T$ is {\em accepting} if all its leaves are labeled with accepting configurations of $M$; otherwise, it is {\em rejecting}. \hfill\markfull 
\end{definition}

\medskip
\noindent \textbf{Output.}
We now define the notion of output of $M$ on input $w$, which is a node-labeled tree whose nodes are constructed when $M$ enters a labeling configuration. To this end, we need some auxiliary notions.
A path $v_1,\ldots,v_n$, for $n \geq 0$, in a computation $T = (V,E,\lambda)$ of $M$ is called {\em labeled-free} if $\lambda(v_i)$ is not a labeling configuration for each $i \in [n]$.
Moreover, given two nodes $u$ and $v$ of $T$, we say that $u$ reaches $v$ in $T$ via a labeled-free path if there is a path $u,v_1,\ldots,v_n,v$, where $n \geq 0$, in $T$ such that $v_1,\ldots,v_n$ is labeled-free.
The {\em output} of a computation $T = (V,E,\lambda)$ of $M$ on input $w \in (\Lambda \setminus \{\bot,\triangleright\})^*$ is the node-labeled rooted tree $O = (V',E',\lambda')$, where $\lambda' : V' \ra \Lambda^*$, such that:
\begin{itemize}
	\item[-] $V' = \{v \in V \mid \lambda(v) \text{ is a labeling configuration of } M\}$,
	\item[-] $E' = \{(u,v) \mid u \text{ reaches } v \text{ in } T \text{ via a } \text {labeled-free path}\}$,
	
	\item[-] for every $v' \in V'$, $\lambda'(v') = z$ assuming that $\lambda(v')$ is of the form $(\cdot,\cdot,\cdot,z,\cdot,\cdot)$.
\end{itemize}
We are now ready to define the output of $M$ on a certain input.

\begin{definition}[\textbf{Output}]\label{def:output}
	A node-labeled rooted tree $O$ is an {\em output} of $M$ on input $w$ if $O$ is the output of some computation $T$ of $M$ on input $w$. We further call $O$ a {\em valid output} if $T$ is accepting. \hfill\markfull
\end{definition}

\OMIT{
\begin{definition}[\textbf{Output}]\label{def:output}
	%
	The {\em output} of a computation $T = (V,E,\lambda)$ of $M$ on input $w \in (\Lambda \setminus \{\bot,\triangleright\})^*$ is the node-labeled rooted tree $T' = (V',E^\star,\lambda')$, where $\lambda' : V' \ra \Lambda^*$, such that:
	\begin{itemize}
		\item[-] if $v' \in V'$ is the root of $T'$, then $\lambda'(v') = \epsilon$, and
		
		\item[-] if $v' \in V'$, then $v'$  has $n$ children $u_1,\ldots,u_n$, for some $n \geq 0$, with $\lambda'(u_i) = z_i$ for each $i \in [n]$, where $v_1,\ldots,v_n$ are all the nodes of $T$ labeled with a configuration of the form $(q_i,\cdot,\cdot,z_i,\cdot,\cdot,\cdot)$, where $q_i \in Q_L$, that are reachable from the root of $T$ via a (possibly empty) label-free path.
	\end{itemize}
	We say that $T'$ is a {\em valid output} of $M$ on input $w$ if $T$ is an accepting configuration of $M$ on input $w$. \hfill\markfull
\end{definition}
}

\OMIT{
\subsection{Alternating Turing Machines with Output}

We consider standard alternating Turing machines with a read-only input tape, a read-write working tape, and a write-only output tape. The formal definition follows.

\begin{definition}[\textbf{Alternating Turing Machine}]\label{def:atm}
An \emph{alternating Turing machine with output} (ATO) $M$ is a tuple 
\[
(Q,\dep,q_\text{\rm init},q_\text{\rm accept},q_\text{\rm reject},Q_\exists,Q_\forall,Q_{L},\delta),
\] 
where
	\begin{itemize}
		\item[-] $Q$ is the finite set of states of $M$,
		\item[-] $\Lambda$ is a finite set of symbols, the {\em alphabet} of $M$, including the symbols $\bot$ (blank symbol) and $\triangleright$ (left marker),
		\item[-] $q_\text{\rm init} \in Q$ is the initial state of $M$,
		\item[-] $q_\text{accept},q_\text{reject} \in Q$ are the accepting and rejecting states of $M$, respectively,
		\item[-] $Q_\exists,Q_\forall$ are the existential and universal states of $M$, respectively, and they form a partition of $Q \setminus \{q_\text{\rm accept},q_\text{\rm reject}\}$,
		\item[-] $Q_L$ are the labelling states of $M$ including $q_\text{\rm init}$,
		\item[-] $\delta : (Q \setminus \{q_\text{\rm accept},q_\text{\rm reject}\}) \times \Lambda \times \Lambda \rightarrow \PS(Q \times \{-1,0,+1\} \times \{-1,0,+1\} \times \Lambda \times \Lambda)$ is the transition function of $M$. \hfill\markfull
	\end{itemize}
\end{definition}

Fix an ATO $M = (Q,\dep,q_\text{\rm init},q_\text{\rm accept},q_\text{\rm reject},Q_\exists,Q_\forall,Q_{L},\delta)$.
A {\em configuration} of $M$ is a tuple $C = (q,x,y,z,h_x,h_y,h_z)$, where $q \in Q$, $x,y,z$ are strings of $\Lambda^*$, and $h_x,h_y,h_z$ are positive integers. If $M$ is in configuration $C$, then the input (resp., working, output) tape contains the infinite string $x \bot \bot \cdots$ (resp., $y \bot \bot \cdots$, $z \bot \bot \cdots$) and the cursor points to the cell $h_x$ (resp., $h_y$, $h_z$). As usual, we use the left marker, which means that $x$ and $z$ are always starting with $\triangleright$. Moreover, the transition function $\delta$ is restricted in such a way that $\triangleright$ occurs exactly once in $x$ and $z$ and always as the first symbol.

Assume now that $M$ is in configuration $C = (q,x,y,z,h_x,h_y,h_z)$, where $x[h_x] = \alpha$ and $y[h_y] = \beta$, and assume that $\delta(q,\alpha,\beta) = \{(q_i,d_i,d_i',\beta_i,\gamma_i)\}_{i \in [n]}$ for some $n > 0$. Then, in one step, $M$ enters the configurations $C_i = (q_i,x,y_i,z_i,h_x',h_{y_i},h_{z_i})$ for $i \in [n]$, where $y_i$ is obtained from $y$ by replacing $y[h_y] = \beta$ with $\beta_i$, $z_i = \gamma_i$ (resp., $z_i = z\gamma_i$) if $q \in Q_L$ (resp., $q \not\in Q_L$), $h_x' = h_x + d_i$, $h_{y_i} = h_y + d_i'$, and $h_{z_i} = 1$ (resp., $h_{z_i} = h_z + 1$) if $q \in Q_L$ (resp., $q \not\in Q_L$). In this case we write $C \ra_M \{C_1,\ldots,C_n\}$. Note that the cursors of the input and the working tapes cannot move left of the symbol $\triangleright$. This is achieved by restricting the transition function.

The ATO $M$ receives an input $w = w_1 \cdots w_n \in (\Lambda \setminus \{\bot,\triangleright\})^*$. The {\em initial configuration} of $M$ on input $w$ is $(q_\text{\rm init},\triangleright w, \triangleright, \epsilon,1,1,1)$. We call a configuration of $M$ {\em accepting} (resp., {\em rejecting}) if its state is $q_\text{\rm accept}$ (resp., $q_\text{\rm reject}$). We further call $C$ {\em existential} (resp., {\em universal}, {\em labeling}) if its state belongs to $Q_{\exists}$ (resp., $Q_\forall$, $Q_L$).
We proceed to define the notion of computation of an ATO.

\begin{definition}[\textbf{Computation}]\label{def:computation}
A {\em computation} of $M$ on input $w \in (\Lambda \setminus \{\bot,\triangleright\})^*$ is a node-labeled rooted tree $T = (V,E,\lambda)$, where $\lambda$ assigns to the nodes of $T$ configurations of $M$, such that:
\begin{enumerate}
	\item[-] if $v \in V$ is the root node of $T$, then $\lambda(v)$ is the initial configuration of $M$ on input $w$,
	
	\item[-] if $v \in V$ is a leaf node of $T$, then $\lambda(v)$ is either an accepting or a rejecting configuration of $M$, and
	
\item[-] if $v \in V$ is an internal node of $T$ with $\lambda(v) \ra_M \{C_1,\ldots,C_n\}$ for $n > 0$, then $\lambda(v) \in Q_\exists$ implies $v$ has exactly one child $u \in V$ with $\lambda(u) = C_i$ for some $i \in [n]$, and $\lambda(v) \in Q_\forall$ implies $v$ has $n$ children $u_1,\ldots,u_n \in V$ with $\lambda(u_i) = C_i$ for $i \in [n]$.
\end{enumerate}
We say that $T$ is {\em accepting} if all its leaves are labeled with accepting configurations of $M$; otherwise, it is {\em rejecting}. \hfill\markfull 
\end{definition}

It remains to define the notion of output of $M$. A path $v_1,\ldots,v_n$, for some $n > 0$, in a computation $T = (V,E,\lambda)$ of $M$ is called {\em labeled-free} if $\lambda(v_i)$ is not a labeling configuration for each $i \in [n]$.

\begin{definition}[\textbf{Output}]\label{def:output}
%
The {\em output} of a computation $T = (V,E,\lambda)$ of $M$ on input $w \in (\Lambda \setminus \{\bot,\triangleright\})^*$ is the node-labeled rooted tree $T' = (V',E',\lambda')$, where $\lambda' : V' \ra \Lambda^*$, such that:
\begin{itemize}
	\item[-] if $v' \in V'$ is the root of $T'$, then $\lambda'(v') = \epsilon$, and
	
	\item[-] if $v' \in V'$ is a node other that the root of $T'$, then $v'$ has $n$ children $u_1,\ldots,u_n$, for some $n \geq 0$, with $\lambda'(u_i) = z_i$ for each $i \in [n]$, where $v_1,\ldots,v_n$ are all the nodes of $T$ labeled with a configuration of the form $(q_i,\cdot,\cdot,z_i,\cdot,\cdot,\cdot)$, where $q_i \in Q_L$, that are reachable from the root of $T$ via a (possibly empty) label-free path.
\end{itemize}
We say that $T'$ is a {\em valid output} of $M$ on input $w$ if $T$ is an accepting configuration of $M$ on input $w$. \hfill\markfull
\end{definition}
}

\subsection{The Complexity Class $\mathsf{SpanTL}$}

We now proceed to define the complexity class $\mathsf{SpanTL}$ and establish that each of its problems admits an efficient approximation scheme. To this end, we need to introduce the notion of well-behaved ATO, that is, an ATO with some resource usage restrictions.
In particular, we require (i) each computation to be of polynomial size, (ii) the working and labeling tapes to use logarithmic space, and (iii) each labeled-free path to have a bounded number of universal configurations.
Note that the first two conditions are direct extensions of the conditions that define non-deterministic logspace Turing machines, i.e., the size of each computation is bounded by a polynomial, and the space usage is bounded by a logarithm. The last item applies only to alternating Turing machines, as for non-deterministic Turing machines with output it holds trivially, assuming that we interpret a non-deterministic Turing machine with output as a special case of an ATO without universal states. We now formally define the above resource usage restrictions.

\begin{definition}[\textbf{Well-behaved ATO}]\label{def:well-behaved}
Consider an ATO $M = (S,\Lambda,s_\text{\rm init},s_\text{\rm accept},s_\text{\rm reject},S_\exists,S_\forall,S_{L},\delta)$. We call $M$ {\em well-behaved} if there is a polynomial function $\text{\rm pol} : \mathbb{N} \ra \mathbb{N}$ and an integer $k \geq 0$ such that, for every string $w = w_1 \cdots w_n \in (\Lambda \setminus \{\bot,\triangleright\})^*$ and computation $T = (V,E,\lambda)$ of $M$ on input $w$, the following hold:
\begin{itemize}
	\item[-] $|V| \in O(\text{\rm pol}(n))$,
	\item[-] for each $v \in V$ with $\lambda(v)$ being of the form $(\cdot,\cdot,y,z,\cdot,\cdot)$, $|y|,|z| \in O(\log(n))$, and
	\item[-] for each labeled-free path $v_1,\ldots,v_m$ of $T$, where $m \geq 0$, $|\{v_i \mid \lambda(v_i) \text{ is universal}\}| \leq k$. \hfill\markfull
\end{itemize}
\end{definition}

Let $\mathsf{span}_M : (\Lambda \setminus \{\bot,\triangleright\})^* \rightarrow \mathbb{N}$ be the function that assigns to each $w \in (\Lambda \setminus \{\bot,\triangleright\})^*$ the number of distinct valid outputs of $M$ on input $w$, i.e., $\mathsf{span}_M(w) = |\{T \mid T \text{ is a valid output of } M \text{ on } w\}|$. 
%
We finally define the complexity class
\[
\mathsf{SpanTL}\ =\ \left\{\mathsf{span}_M \mid M \text{ is a well-behaved ATO}\right\}.
\]

\def\prospanlinspantl{
	$\spanl \subseteq \spantl$. Furthermore, unless $\nlogspace = \logcfl$, $\spanl \subset \spantl$.
}

	We are going to show in the next section that $\sharp\mathsf{Repairs}[k]$ is in $\spantl$, for each $k>0$. Interestingly, we can show that there are other natural database counting problems, which are not in $\spanl$ unless $\nlogspace = \logcfl$, that belong to $\spantl$. In particular, we can show that the problem $\sharp \mathsf{GHWCQ}[k]$, which takes as input a database $D$, a CQ $Q$, and a generalized hypertree decomposition $H$ of $Q$ of width $k$,\footnote{Here, we assume that $H = (T,\chi,\lambda)$ considers also the output variables of $Q(\bar x)$. In other words, the tree decomposition $(T,\chi)$ considers all the variables of $Q$ and not only the variables $\var{Q} \setminus \bar x$. Note that this assumption is also needed in~\cite{ACJR21} for providing an FPRAS for the problem $\sharp \mathsf{GHWCQ}[k]$.} and asks for the number of answers to $Q$ over $D$, belongs to $\spantl$.
	Moreover, the uniform reliability problem $\sharp \mathsf{UR}[k]$, which takes as input a database $D$, a CQ $Q$ from $\sjf$, a generalized hypertree decomposition of $Q$ of width $k$, and a tuple $\bar c$, and asks for the number of subsets $D'$ of $D$ such that $\bar c \in Q(D')$, also belongs to $\spantl$.
	Finally, the problem $\sharp\mathsf{SRepairs}[k]$, defined as $\sharp\mathsf{Repairs}[k]$ with the difference that we focus on the classical subset repairs from~\cite{ArBC99}, is also in $\mathsf{SpanTL}$.	
	It is known that, for $k > 0$, $\sharp \mathsf{GHWCQ}[k]$ and $\sharp\mathsf{UR}[k]$ admit an FPRAS, which was shown via reductions to the problem of counting the number of trees of a certain size accepted by a non-deterministic finite tree automaton~\cite{ACJR21,BrMe23}.
	Our novel complexity class allows us to reprove the approximability of $\sharp \mathsf{GHWCQ}[k]$ and $\sharp\mathsf{UR}[k]$ in an easier way.
	The fact that $\sharp\mathsf{SRepairs}[k]$ admits an FPRAS is a result of this work.

	It is straightforward to see that $\spantl$ generalizes $\spanl$. Furthermore, since the decision version of every function in $\spanl$ is in $\nlogspace$, whereas the decision version of $\sharp \mathsf{GHWCQ}[k]$, $\sharp \mathsf{UR}[k]$, and $\sharp\mathsf{SRepairs}[k]$, for each $k>0$, is $\logcfl\hard$~\cite{GoLS02}, we get that:
	
	\begin{proposition}\label{pro:spanl-in-spantl}
		\prospanlinspantl
	\end{proposition}


A property of $\mathsf{SpanTL}$, which will be useful for our later development, is the fact that it is {\em closed under logspace reductions}. This means that, given two functions $f : \Lambda_1^* \rightarrow \mathbb{N}$ and $g : \Lambda_2^* \rightarrow \mathbb{N}$, for some alphabets $\Lambda_1$ and $\Lambda_2$, if there is a logspace computable function $h : \Lambda_1^* \rightarrow \Lambda_2^*$ with $f(w) = g(h(w))$ for all $w \in \Lambda_1^*$, and $g$ belongs to $\spantl$, then $f$ also belongs to $\spantl$.
	
\def\prologspaceclosure{
	$\spantl$ is closed under logspace reductions.
}

\begin{proposition}\label{pro:logspace-closure}
\prologspaceclosure
\end{proposition}

Another key property of $\mathsf{SpanTL}$, which is actually the crucial one that we need for establishing Theorem~\ref{the:main-fpras}, is that each of its problems admits an FPRAS. 


\def\thespantlfpras{
	Every function in $\mathsf{SpanTL}$ admits an FPRAS.
}

\begin{theorem}\label{the:spantl-fpras}
\thespantlfpras
\end{theorem}

To establish the above result we exploit a recent approximability result in the context of automata theory, that is, the problem of counting the number of trees of a certain size accepted by a non-deterministic finite tree automaton admits an FPRAS~\cite{ACJR21}.
Hence, for showing Theorem~\ref{the:spantl-fpras}, we reduce in polynomial time each function in $\mathsf{SpanTL}$ to the above counting problem for tree automata.

\OMIT{
In what follows, we recall the basics about non-deterministic finite tree automata and the associated counting problem, and then discuss the key ideas underlying our reduction.

\medskip

\noindent \paragraph{Ordered Trees and Tree Automata.} For an integer $k \geq 1$, a {\em finite ordered $k$-tree} (or simple {\em $k$-tree}) is a prefix-closed non-empty finite subset $T$ of $[k]^*$, that is, if $w \cdot i \in T$ with $w \in [k]^*$ and $i \in [k]$, then $w \cdot j \in T$ for every $w \in [i]$.
The root of $T$ is the empty string, and every maximal element of $T$ (under prefix ordered) is a leaf. For every $u,v \in T$, we say that $u$ is a child of $v$, or $v$ is a parent of $u$, if $u = v \cdot i$ for some $i \in [k]$. The size of $T$ is $|T|$. Given a finite alphabet $\Lambda$, let $\trees{k}{\Lambda}$ be the set of all $k$-trees in which each node is labeled with a symbol from $\Lambda$. By abuse of notation, for $T \in \trees{k}{\Lambda}$ and $u \in T$, we write $T(u)$ for the label of $u$ in $T$.

A {\em (top-down) non-deterministic finite tree automation} (NFTA) over $\trees{k}{\Lambda}$ is a tuple $A = (S,\Lambda,s_\text{\rm init},\delta)$, where $S$ is the finite set of states of $A$, $\Lambda$ is a finite set of symbols (the alphabet of $A$), $s_\text{\rm init} \in S$ is the initial state of $A$, and $\delta \subseteq S \times \Lambda \times \left(\bigcup_{i = 0}^{k} S^k\right)$ is the transition relation of $A$.
A {\em run} of $A$ over a tree $T \in \trees{k}{\Lambda}$ is a function $\rho : T \ra S$ such that, for every $u \in T$, if $u \cdot 1,\ldots,u \cdot n$ are the children of $u$ in $T$, then $(\rho(u),T(u),(\rho(u \cdot 1),\ldots,\rho(u \cdot n))) \in \delta$. In particular, if $u$ is a leaf, then $(\rho(u),T(u),()) \in \delta$. We say that $A$ {\em accepts} $T$ if there is a run $\rho$ of $A$ over $T$ with $\rho(\epsilon) = s_\text{\rm init}$, i.e., $\rho$ assigns to the root the initial state. We write $L(A) \subseteq \trees{k}{\Lambda}$ for the set of all trees accepted by $A$, i.e., the language of $A$. We further write $L_n(A)$ for the set of trees $\{T \in L(A) \mid |T| = n\}$, i.e., the set of trees of size $n$ accepted by $A$. The relevant counting problem for NFTA follows:

\medskip

\begin{center}
	\fbox{\begin{tabular}{ll}
			{\small PROBLEM} : & $\sharp\mathsf{NFTA}$
			\\
			{\small INPUT} : & An NFTA $A$  and a string $0^n$ for some $n \geq 0$.
			\\
			{\small OUTPUT} : &  $\left|\bigcup_{i = 0}^n L_i(A)\right|$.
	\end{tabular}}
\end{center}

\medskip

The notion of FPRAS for $\sharp\mathsf{NFTA}$ is defined in the obvious way. 
We know from~\cite{ACJR21} that $\sharp\mathsf{NFTA}_=$, defined as $\sharp\mathsf{NFTA}$ with the difference that it asks for $|L_n(A)|$, i.e., the number trees of size $n$ accepted by $A$, admits an FPRAS. By using this result, we can easily show that:

\def\thenfta{
	$\sharp\mathsf{NFTA}$ admits an FPRAS.
}

\begin{theorem}\label{the:nfta}
	\thenfta
\end{theorem}

\noindent\paragraph{The Reduction.} Recall that our goal is reduce in polynomial time every function of $\mathsf{SpanTL}$ to $\sharp\mathsf{NFTA}$, which, together with Theorem~\ref{the:nfta}, will immediately imply Theorem~\ref{the:spantl-fpras}. In particular, we need to establish the following technical result:

\def\proreductiontonfta{
Fix a function $f : \Lambda^* \ra \mathbb{N}$ of $\mathsf{SpanTL}$. For every $w \in \Lambda^*$, we can construct in polynomial time in $|w|$ an NFTA $A$ and a string $0^n$, for some $n \geq 0$, such that $f(w) = \left|\bigcup_{i = 0}^n L_i(A)\right|$.
}

\begin{proposition}\label{pro:reduction}
	\proreductiontonfta
\end{proposition}

We discuss the main ideas underlying the proof of Proposition~\ref{pro:reduction}. Since $f : \Lambda^* \ra \mathbb{N}$ belongs to $\mathsf{SpanTL}$, there exists a well-behaved ATO $M = (S,\Lambda',s_\text{\rm init},s_\text{\rm accept},s_\text{\rm reject},S_{\exists},S_{\forall},S_{L},\delta)$, where $\Lambda' = \Lambda \cup \{\bot,\triangleright\}$ and $\bot,\triangleright \not\in \Lambda$, such that $f$ is the function $\mathsf{span}_M$. 
The goal is, for an arbitrary string $w \in \Lambda^*$, to construct in polynomial time in $|w|$ an NFTA $A = (S^A,\Lambda^A,s_{\text{\rm init}}^{A}\delta^A)$ and a string $0^n$, for some $n \geq 0$, such that $\mathsf{span}_M(w) = |\bigcup_{i = 0}^{n} L_i(A)|$. This is done in two steps:
\begin{enumerate}
	\item We first construct an NFTA $A$ in polynomial time in $|w|$ such that $\mathsf{span}_M (w) = |L(A)|$, i.e., $A$ accepts $\mathsf{span}_M(w)$ trees.
	
	\item We define a polynomial function $\mathsf{pol} : \mathbb{N} \ra \mathbb{N}$ such that $L(A) = \bigcup_{i = 0}^{\mathsf{pol}(|w|)} L_i(A)$. 
\end{enumerate}
After completing the above two steps, it is clear that Proposition~\ref{pro:reduction} follows with $n = \mathsf{pol}(|w|)$. Let us now discuss the above two steps.

\medskip

\noindent\paragraph{\underline{Step 1: The NFTA}}

\smallskip

\noindent For the construction of the desired NFTA, we first need to introduce the auxiliary notion of the computation directed acyclic graph (DAG) of the ATO $M$ on input $w \in \Lambda^*$, which compactly represents all the computations of $M$ on $w$. Recall that a DAG $G$ is rooted if it has exactly one node, the root, with no incoming edges. We also say that a node of $G$ is a leaf if it has no outgoing edges.

\begin{definition}[\textbf{Computation DAG}]\label{def:computation_dag}
	The {\em computation DAG} of $M$ on input $w \in \Lambda^*$ is the DAG $G = (\mathcal{C},\mathcal{E})$, where $\mathcal{C}$ is the set of configurations of $M$ on $w$, defined as follows:
	\begin{enumerate}
		\item[-] if $C \in \mathcal{C}$ is the root node of $G$, then $C$ is the initial configuration of $M$ on input $w$,
		
		\item[-] if $C \in \mathcal{C}$ is a leaf node of $G$, then $C$ is either an accepting or a rejecting configuration of $M$, and
		
		\item[-] if $C \in \mathcal{C}$ is a non-leaf node of $G$ with $C \ra_M \{C_1,\ldots,C_n\}$ for $n > 0$, then (i) for each $i \in [n]$, $(C,C_i) \in \mathcal{E}$, and (ii) for every configuration $C' \not\in \{C_1,\ldots,C_n\}$ of $M$ on $w$, $(C,C') \not\in \mathcal{E}$. \hfill\markfull 
	\end{enumerate}
\end{definition}

As said above, the computation DAG $G$ of $M$ on input $w$ compactly represents all the computations of $M$ on $w$. In particular, a computation $(V,E,\lambda)$ of $M$ on $w$ can be constructed from $G$ by traversing $G$ from the root to the leaves and (i) for every universal configuration $C$ with outgoing edges $(C,C_1),\dots,(C,C_n)$, add a node $v$ with $\lambda(v)=C$ and children $u_1,\dots,u_n$ with $\lambda(u_i)=C_i$ for all $i\in[n]$, and (ii) for every existential configuration $C$ with outgoing edges $(C,C_1),\dots,(C,C_n)$, add a node $v$ with $\lambda(v)=C$ and a single child $u$ with $\lambda(u)=C_i$ for some $i\in [n]$.

\OMIT{
proof of Proposition~\ref{pro:reduction}. Since, by hypothesis, $f : \Lambda^* \ra \mathbb{N}$ belongs to $\mathsf{SpanTL}$, there exists a well-behaved ATO $M = (S,\Lambda',s_\text{\rm init},s_\text{\rm accept},s_\text{\rm reject},S_{\exists},S_{\forall},S_{L},\delta)$, where $\Lambda' = \Lambda \cup \{\bot,\triangleright\}$ and $\bot,\triangleright \not\in \Lambda$, such that $f$ is the function $\mathsf{span}_M$. 
The goal is, for an arbitrary string $w \in \Lambda^*$, to construct in polynomial time in $|w|$ an NFTA $A = (S^A,\Lambda^A,s_{\text{\rm init}}^{A}\delta^A)$ and a string $0^n$ from some $n \geq 0$ such that $\mathsf{span}_M(w) = |\bigcup_{i = 0}^{n} L_i(A)|$. This is done in two steps:
\begin{enumerate}
	\item We first construct an NFTA $A$ in polynomial time in $|w|$ such that $\mathsf{span}_M (w) = |L(A)|$, i.e., $A$ accepts $\mathsf{span}_M(w)$ trees.
	
	\item We define a polynomial function $\mathsf{pol} : \mathbb{N} \ra \mathbb{N}$ such that $L(A) = |\bigcup_{i = 0}^{\mathsf{pol}(|w|)} L_i(A)|$. 
\end{enumerate}
After completing the above two steps, it is clear that Proposition~\ref{pro:reduction} follows with $n = \mathsf{pol}(|w|)$. We proceed to discuss those two steps.

\medskip

\noindent\paragraph{\underline{Step 1: The NFTA}}

\smallskip

%
}

The construction of the NFTA $A$ is performed by $\mathsf{BuildNFTA}$, depicted in Algorithm~\ref{alg:dagtonfta}, which takes as input a string $w \in \Lambda^*$. It first constructs the computation DAG $G = (\mathcal{C},\mathcal{E})$ of $M$ on $w$, which will guide the construction of $A$. It then initializes the sets $S^A,\Lambda^A,\delta^A$, as well as the auxiliary set $Q$, which will collect pairs of the form $(C,U)$, where $C$ is a configuration of $\mathcal{C}$ and $U$ a set of tuples of states of $A$, to empty.
Then, it calls the recursive procedure $\mathsf{Process}$, depicted in Algorithm~\ref{alg:process}, which constructs the set of states $S^A$, the alphabet $\Lambda^A$, and the transition relation $\delta^A$, while traversing the computation DAG $G$ from the root to the leaves. Here, $S^A$, $\Lambda^A$, $\delta^A$, $G$, and $Q$ should be seen as global structures that can be used and updated inside the procedure $\mathsf{Process}$. Eventually, $\mathsf{Process}(\rt{G})$ returns a state $s \in S^A$, which acts as the initial state of $A$, and $\mathsf{BuildNFTA}(w)$ returns the NFTA $(S^A,\Lambda^A,s,\delta^A)$.

Concerning the procedure $\mathsf{Process}$, when we process a labeling configuration $C \in \mathcal{C}$, we add to $S^A$ a state $s_C$ representing $C$. 
Then, for every computation $T=(V,E,\lambda)$ of $M$ on $w$, if $V$ has a node $v$ with $u_1,\dots,u_n$ being the nodes reachable from $v$ via a labeled-free path, and $\lambda(u_i)$ is a labeling configuration for every $i\in[n]$, we add the transition $(s_C,z,(s_{C_1},\dots,s_{C_n}))$ to $\delta^A$, where $\lambda(v)=C$, $C$ is of the form $(\cdot,\cdot,\cdot,z,\cdot,\cdot,\cdot)$, and $\lambda(u_i)=C_i$ for every $i\in[n]$, for some arbitrary order $C_1,\dots,C_n$ over those configurations.
Since, by definition, the output of $T$ has a node corresponding to $v$ with children corresponding to $u_1,\dots,u_n$, using these transitions we ensure that there is a one-to-one correspondence between the outputs of $M$ on $w$ and the trees accepted by $A$.

Now, when processing non-labeling configurations, we accumulate all the information needed to add all these transitions to $\delta^A$. In particular, the procedure $\process$ always returns a set of tuples of states of $S^A$. For labeling configurations $C$, it returns a single tuple $(s_C)$, but for non-labeling configurations $C$, the returned set depends on the outgoing edges $(C,C_1),\dots,(C,C_n)$ of $C$. In particular, if $C$ is an existential configuration, it simply takes the union $P$ of the sets $P_i$ returned by $\process(C_i)$, for $i\in[n]$, because in every computation of $M$ on $w$ we choose only one child of each node associated with an existential configuration; each $P_i$ represents one such choice. If, on the other hand, $C$ is a universal configuration, then $P$ is the set $\bigotimes_{i \in [n]} P_i$ obtained by first computing the cartesian product $P' = \bigtimes_{i \in [n]} P_i$ and then merging each tuple of $P'$ into a single tuple of states of $S^A$. For example, with $P_1 = \{(),(s_1,s_2),(s_3)\}$ and $P_2 = \{(s_5),(s_6,s_7)\}$, $P_1 \otimes P_2$ is the set $\{(s_5),(s_6,s_7),(s_1,s_2,s_5),(s_1,s_2,s_6,s_7),(s_3,s_5),(s_3,s_6,s_7)\}$. We define $\bigotimes_{i \in [n]} P_i=\emptyset$ if $P_i=\emptyset$ for some $i\in[n]$.
We use the $\otimes$ operator since in every computation of $M$ on $w$ we choose all the children of each node associated with a universal configuration. When we reach a labeling configuration $C \in \mathcal{C}$, we add transitions to $\delta^A$ based on this accumulated information. In particular, we add a transition from $s_C$ to every tuple $(s_1,\dots,s_\ell)$ of states in $P$.

Concerning the running time of $\mathsf{BuildNFTA}(w)$, we first observe that the size (number of nodes) of the computation DAG $G$ of $M$ on input $w$ is polynomial in $|w|$. This holds since for each configuration $(\cdot,\cdot,y,z,\cdot,\cdot,\cdot)$ of $M$ on $w$, we have that $|y|,|z|\in O(\log(|w|))$. Moreover, we can construct $G$ in polynomial time in $|w|$ by first adding a node for the initial configuration of $M$ on $w$, and then following the transition function to add the remaining configurations of $M$ on $w$ and the
outgoing edges from each configuration.
Now, in the procedure $\process$ we use the auxiliary set $Q$ to ensure that we process each node of $G$ only once; thus, the number of calls to the $\process$ procedure is polynomial in the size of $|w|$. Moreover, in $\process(C)$, where $C$ is a non-labeling universal configuration that has $n$ outgoing edges $(C,C_1),\dots,(C,C_n)$, the size of the set $P$ is $|P_1| \times \dots \times |P_n|$, where $P_i=\process(C_i)$ for $i\in[n]$. When we process a non-labeling existential configuration $C$ that has $n$ outgoing edges $(C,C_1),\dots,(C,C_n)$, the size of the set $P$ is $|P_1|+\dots+ |P_n|$. We also have that $|P|=1$ for every labeling configuration $C$. Hence, in principle, many universal states along a labeled-free path could cause an exponential blow-up of the size of $P$. However, since $M$ is a well-behaved ATO, there exists $k\ge 0$ such that every labeled-free path of every computation $(V,E,\lambda)$ of $M$ on $w$ has at most $k$ nodes $v$ for which $\lambda(v)$ is a universal configuration. It is rather straightforward to see that every labeled-free path of $G$ also enjoys this property since this path occurs in some computation of $M$ on $w$. Hence, the size of the set $P$ is bounded by a polynomial in the size of $|w|$. From the above discussion, we get that $\mathsf{BuildNFTA}(w)$ runs in polynomial time in $|w|$, and the next lemma can be shown:

\def\lemmabuildnfta{
For a string $w\in \Lambda^*$, $\mathsf{BuildNFTA}(w)$ runs in polynomial time in $|w|$ and returns an NFTA $A$ such that $\spanm_M(w)=|L(A)|$.
}

\begin{lemma}\label{lem:buildnfta}
\lemmabuildnfta
\end{lemma}


\begin{algorithm}[t]
	\KwIn{A string $w \in \Lambda^*$}
	\KwOut{An NFTA $A$}
	\vspace{2mm}
	
	{Construct the computation DAG $G = (\mathcal{C},\mathcal{E})$ of $M$ on $w$;}
 {\\ $S^A := \emptyset$; $\Lambda^A := \emptyset$; $\delta^A := \emptyset$; $Q :=\emptyset$;}
	{\\ $P :=\mathsf{Process}(\rt{G})$;}
	{\\ Assuming $P = \{(s)\}$, $A := \left(S^A,\Lambda^A, s, \delta^A\right)$;}
	{\\\Return{$A$;}}
	\caption{The algorithm $\mathsf{BuildNFTA}$}\label{alg:dagtonfta}
\end{algorithm}

\begin{algorithm}[t]
	\KwIn{A configuration $C=(s,x,y,z,h_x,h_y)$ of $\mathcal{C}$}
	\vspace{2mm}
	\If{$(C,U) \in Q$}
	{\Return $U$;}
	\If{$C$ is a leaf of $G$}
	{\If{$s\in S_\labeling$}
		{{$S^A := S^A \cup \{s_C\}$;\\}
			{$\Lambda^A := \Lambda \cup \{z\}$;\\}
            \If{$s=s_\accept$}
            {$\delta^A := \delta^A \cup \{(s_C,z,())\}$;}
			{$Q :=Q \cup \{(C,\{(s_C)\})\}$;}
		}
		\lElseIf{$s=s_\accept$}{$Q := Q \cup \{(C,\{()\})\}$}
            \lElse{$Q : = Q \cup \{(C,\emptyset)\}$}}
	\Else{
		{let $C_1,\dots,C_n$ be an arbitrary order over the nodes of $G$ with an incoming edge from $C$;\\}
		\lForEach{$i \in [n]$}{
			{$P_i := \process(C_i)$}}
		\lIf{$s\in S_\exists$}
		{$P := \bigcup_{i \in [n]} P_i$}
		\lElse{    
			$P := \bigotimes_{i \in [n]} P_i$
			%
		}
		\If{$s\in S_\labeling$}
		{{$S^A := S^A \cup \{s_C\}$;\\} 
			{$\Lambda^A := \Lambda^A \cup \{z\}$;\\}
			\ForEach{$(s_1,\dots,s_\ell)\in P$}{$\delta^A := \delta^A \cup \{(s_C,z,(s_1,\dots,s_\ell))\}$;}
			{$Q := Q \cup \{(C,\{(s_C)\})\}$;}}
		\lElse{$Q := Q \cup \{(C,P)\}$}}
	{\Return $U$ with $(C,U) \in Q$;}
	\caption{The recursive procedure $\mathsf{Process}$}\label{alg:process}
\end{algorithm}

\medskip

\noindent\paragraph{\underline{Step 2: The Polynomial Function}}

\smallskip

\noindent 
%
Since $M$ is a well-behaved ATO, there exists a polynomial function $\mathsf{pol}:\mathbb{N}\rightarrow \mathbb{N}$ such that the size of every computation of $M$ on $w$ is bounded by $\mathsf{pol}(|w|)$. Clearly, $\mathsf{pol}(|w|)$ is also a bound on the size of the valid outputs of $M$ on $w$. From the proof of Lemma~\ref{lem:buildnfta}, we get that every tree accepted by $A$ has the same structure as some valid output of $M$ on $w$. Hence, we have that $\mathsf{pol}(|w|)$ is also a bound on the size of the trees accepted by $A$, and the next lemma follows:

\def\lempolynomial{
	It holds that $L(A) = \bigcup_{i = 0}^{\mathsf{pol}(|w|)} L_i(A)$.
}

\begin{lemma}\label{lem:polynomial}
\lempolynomial
\end{lemma}

Proposition~\ref{pro:reduction} readily follows from Lemmas~\ref{lem:buildnfta} and~\ref{lem:polynomial}.
}
\newcommand{\lIfElse}[3]{\lIf{#1}{#2 \textbf{else}~#3}}
\section{Combined Approximation}\label{sec:algorithms}


As discussed in Section~\ref{sec:operational-cqa}, for establishing Theorem~\ref{the:main-fpras}, it suffices to show that, for every $k>0$, $\sharp\mathsf{Repairs}[k]$ admits an FPRAS. To this end, by Theorem~\ref{the:spantl-fpras}, it suffices to show that

\begin{theorem}\label{the:in-spantl}
	For every $k>0$, $\sharp\mathsf{Repairs}[k]$ is in $\spantl$.
\end{theorem}

To prove the above result, for $k>0$, we need to devise a procedure, which can be implemented as a well-behaved ATO $M_R^k$,  that take as input a database $D$, a set $\dep$ of primary keys, a CQ $Q(\bar x)$ from $\sjf$, a generalized hypertree decomposition $H$ of $Q$ of width $k>0$, and a tuple $\bar c \in \adom{D}^{|\bar x|}$, such that $\mathsf{span}_{M_{R}^{k}}(D,\dep,Q,H,\bar c)$ is precisely $|\{D' \in \opr{D}{\dep} \mid \bar c \in Q(D')\}|.$
%
For the sake of simplicity, we are going to assume, w.l.o.g., that the input to this procedure is in a certain normal form. Let us first introduce this normal form.

%

\medskip
\noindent \paragraph{Normal Form.} 
Consider a CQ $Q(\bar x)$ from $\sjf$ and a generalized hypertree decomposition $H=(T, \chi, \lambda)$ of $Q$. We say that $H$ is $\ell$-uniform, for $\ell > 0$, if every non-leaf vertex of $T$ has exactly $\ell$ children. For any two vertices $v_1,v_2$ of $T$, we write $v_1 \prec_T v_2$ if $\dept{v_1} < \dept{v_2}$, or $\dept{v_1} = \dept{v_2}$ and $v_1$ precedes $v_2$ lexicographically. It is clear that $\prec_T$ is a total ordering over the vertices of $T$.
Consider an atom $R(\bar y)$ of $Q$. We say that a vertex $v$ of $T$ is a \emph{covering vertex for} $R(\bar y)$ (in $H$) if $\bar y \subseteq \chi(v)$ and $R(\bar y) \in \lambda(v)$~\cite{BrMe23}. 
We say that $H$ is \emph{complete} if each atom of $Q$ has at least one covering vertex in $H$.
%
%
%
The \emph{$\prec_T$-minimal covering vertex} for an atom $R(\bar y)$ of $Q$ (in $H$) is the covering vertex $v$ for $R(\bar y)$ for which there is no other covering vertex $v'$ for $R(\bar y)$ with $v' \prec_T v$; if $H$ is complete, every atom of $Q$ has a $\prec_T$-minimal covering vertex. We say that $H$ is \emph{strongly complete} if it is complete, and every vertex of $T$ is the $\prec_T$-minimal covering vertex of some atom of $Q$.

Now, given a database $D$, a CQ $Q$ from $\sjf$, and a generalized hypertree decomposition $H$ of $Q$, we say that the triple $(D,Q,H)$ is in \emph{normal form} if \emph{(i)} every relation name in $D$ also occurs in $Q$, and \emph{(ii)} $H$ is strongly complete and 2-uniform.
We can show that, given a database $D$, a set $\dep$ of primary keys, a CQ $Q(\bar x)$ from $\sjf$, a generalized hypertree decomposition $H$ of $Q$ of width $k>0$, and a tuple $\bar c \in \adom{D}^{|\bar x|}$, we can convert in logarithmic space the triple $(D,Q,H)$ into a triple $(\hat{D},\hat{Q},\hat{H})$, where $\hat{H}$ is a generalized hypertree decomposition of $\hat{Q}$ of width $k+1$, that is in normal form while
$|\{D' \in \opr{D}{\dep} \mid \bar c \in Q(D')\}| = |\{D' \in \opr{\hat{D}}{\dep} \mid \bar c \in \hat{Q}(D')\}|$. 
%
This, together with the fact that $\spantl$ is closed under logspace reductions (Proposition~\ref{pro:logspace-closure}), allows us to assume, w.l.o.g., that the input $D,\dep, Q, H, \bar c$ to the procedure that we are going to devise enjoys the property that $(D,Q,H)$ is in normal form.

\OMIT{
\smallskip
In what follows, consider a query  $Q(\bar x)$ from $\sjf$, and a generalized hypertree decomposition $H=(T, \chi, \lambda)$ of $Q$. We say that $H$ is $k$-uniform, for some integer $k > 0$, if every non-leaf vertex of $T$ has exactly $k$ children. For any two vertices $v_1,v_2$ of $T$, we write $v_1 \prec_T v_2$ if $\dept{v_1} < \dept{v_2}$, or, in case $\dept{v_1} = \dept{v_2}$, if $v_1$ precedes $v_2$ lexicographically. It is easy to verify that $\prec_T$ is a total ordering of the vertices of $T$.
Consider an atom $R(\bar y)$ occuring in the body of $Q$. We say that a vertex $v$ of $T$ is a \emph{covering vertex for} $R(\bar y)$ (in $H$) if $\bar y \subseteq \chi(v)$ and $R(\bar y) \in \lambda(v)$. 
We say that $H$ is \emph{complete} if each atom in the body of $Q$ has at least one covering vertex in $H$.
%
%
%
The \emph{$\prec_T$-minimal covering vertex} for some atom $R(\bar y)$ (in $H$) is the covering vertex $v$ for $R(\bar y)$ for which there is no other covering vertex $v'$ for $R(\bar y)$ with $v' \prec_T v$; note that every atom has a $\prec_T$-minimal covering vertex, if $H$ is complete. We say that $H$ is \emph{strongly complete} if $H$ is complete, and every vertex of $T$ is the $\prec_T$-minimal covering vertex of some atom of $Q$.

\medskip
\noindent
\textbf{Normal form.} For a database $D$, a query $Q$ from $\sjf$, and a generalized hypertree decomposition $H$ of $Q$, we say that $(D,Q,H)$ is in \emph{normal form} if \emph{1)} every relation in $D$ also occurs in $Q$, and \emph{2)} $H$ is strongly complete and 2-uniform.


\def\pronormalform{
Consider a database $D$, a set of primary keys $\dep$, a query $Q(\bar x)$ in $\sjf$, a generalized hypertree decomposition $H$ of $Q$ of width $k$, and $\bar c \in \adom{D}^{|\bar x|}$. Then, there exists a database $\hat{D}$, a query $\hat{Q}(\bar x)$ from $\sjf$, and a generalized hypertree decomposition $\hat{H}$ of $\hat{Q}$ of width $k+1$ such that:
\begin{itemize}
	\item $(\hat{D},\hat{Q},\hat{H})$ is in normal form,
	\item $|\{D' \in \opr{D}{\dep} \mid \bar c \in Q(D')\}| = |\{D' \in \opr{\hat{D}}{\dep} \mid \bar c \in \hat{Q}(D')\}|$,
	\item $|\{s \in \crs{D}{\dep} \mid \bar c \in Q(s(D))\}| = |\{s \in \crs{\hat{D}}{\dep} \mid \bar c \in \hat{Q}(s(\hat{D}))\}|$, and
	\item $(\hat{D},\hat{Q},\hat{H})$ can be computed in logarithmic space w.r.t.\ $D,\dep,Q,H,\bar c$.
\end{itemize}
}

\begin{proposition}\label{pro:normal-form}
\pronormalform
\end{proposition}
With the above proposition in place, together with the fact that $\spantl$ is closed under logspace reductions, to prove that $\sharp\mathsf{Repairs}[k]$ and $\sharp\mathsf{Sequences}[k]$ are in $\spantl$, for each $k > 0$, it is enough to focus on databases $D$, queries $Q$, and generalized hypertree decompositions $H$ such that $(D,Q,H)$ is in normal form.
}


\medskip
\noindent
\underline{\textbf{The Proof of Theorem~\ref{the:in-spantl}}}
\smallskip

\noindent We show Theorem~\ref{the:in-spantl} via the procedure $\mathsf{Rep}[k]$, depicted in Algorithm~\ref{alg:repairs}, for which we can show the following technical lemma:

\def\lemrepairsato{
	For every $k > 0$, the following hold:
	\begin{enumerate}
		\item $\mathsf{Rep}[k]$ can be implemented as a well-behaved ATO $M_R^k$.
		\item For a database $D$, a set $\dep$ of primary keys, a CQ $Q(\bar x)$ from $\sjf$, a generalized hypertree decomposition $H$ of $Q$ of width $k$, and a tuple $\bar c \in \adom{D}^{|\bar x|}$, where $(D,Q,H)$ is in normal form, $\mathsf{span}_{M_R^k}(D,\dep,Q,H,\bar c) = |\{D' \in \opr{D}{\dep} \mid \bar c \in Q(D')\}|$.
	\end{enumerate}
}
\begin{lemma}\label{lem:repairs-ato}
	\lemrepairsato
\end{lemma}

Theorem~\ref{the:in-spantl} follows from Lemma~\ref{lem:repairs-ato}. The rest of this section is devoted to discussing the procedure $\mathsf{Rep}[k]$ and why Lemma~\ref{lem:repairs-ato} holds. We first give some auxiliary notions and a discussion on how the valid outputs of $\mathsf{Rep}[k]$ (which are labeled trees) look like, which will help the reader to understand how $\mathsf{Rep}[k]$ works.

%

\medskip
\noindent \paragraph{Auxiliary Notions.}
Consider a database $D$, a set $\dep$ of primary keys, and a fact $\alpha = R(c_1,\ldots,c_n) \in D$. The \emph{key value} of $\alpha$ w.r.t. $\dep$ is
\begin{eqnarray*}
	\keyval{\dep}{\alpha}\
	= \left\{
	\begin{array}{ll}
		\langle R, \langle c_{i_1},\ldots,c_{i_m} \rangle \rangle & \text{if }\key{R} = \{i_1,\ldots,i_m\} \in \dep,\\
		&\\
		\langle R, \langle c_1,\ldots,c_n \rangle \rangle & \text{otherwise.}
	\end{array} \right.
\end{eqnarray*}
Moreover, we let $\block{\dep}{\alpha,D} = \{\beta \in D \mid \keyval{\dep}{\beta} = \keyval{\dep}{\alpha}\}$, 
and $\block{\dep}{R,D} = \{ \block{\dep}{R(\bar c),D} \mid R(\bar c) \in D\}$, for a relation name $R$.
%


\begin{example}\label{ex:repairs-algo}
	Consider now the following Boolean CQ $Q$
	\[
	\text{Ans}() \text{ :- } P(x,y),S(y,z),T(z,x),U(y,w),
	\]
	which has generalized hypertreewidth $2$. A generalized hypertree decomposition $H=(T,\chi,\lambda)$ of $Q$ of width $2$ is
		
	\medskip
	\centerline{\includegraphics[width=.37\textwidth]{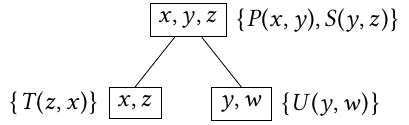}}
	\medskip
	
	\noindent where each box denotes a vertex $v$ of $T$, and the content of the box is 
	$\chi(v)$ while on its side is
	$\lambda(v)$.
	Consider now the set of primary keys $\dep = \{\key{R} = \{1\} \mid R \in \{P,S,T,U\}\}$ and the database
	\begin{multline*}
		D\ =\ \{ P(a_1,b), P(a_1,c), P(a_2,b), P(a_2,c), P(a_2,d),  S(c,d), S(c,e), \\
		T(d,a_1), U(c,f), U(c,g), U(h,i), U(h,j), U(h,k)\}.
	\end{multline*}
	%
	Each operational repair $D'$ in $\opr{D}{\dep}$ is such that, for each relation $R$ in $D$, and for each block $B \in \block{\dep}{R,D}$, either $B = \{\beta\}$ and $\beta \in D'$, or \emph{at most} one fact of $B$ occurs in $D'$; recall that justified operations can remove pairs of facts, and thus a repairing sequence can leave a block completely empty.
	For example, the database
	$$ D' = \{P(a_1,c), S(c,d), T(d,a_1), U(c,f), U(h,i)\}$$
	is an operational repair of $\opr{D}{\dep}$ where we keep the atoms occurring in $D'$ from their respective block in $D$, while for all the other blocks we do not keep any atom.
	Assuming a lexicographical order among the key values of all facts of $D$, we can unambiguously encode the database $D'$ as the following labeled tree $\mathcal{T}$:
	
	\vspace*{2mm}
	\centerline{\includegraphics[width=0.56\textwidth]{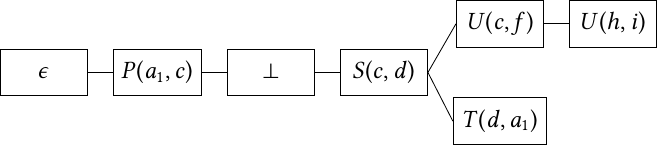}}
	\vspace*{2mm}
	
	\noindent
	The root of $\mathcal{T}$ is labeled with the empty string $\epsilon$, and the label of each other node denotes the choice made on a block $B$ of $\block{\dep}{R,D}$, for some relation name $R$, in order to construct $D'$, i.e., either we keep a certain fact from $B$, or we keep none ($\bot$).
	Crucially, the shape of $\mathcal{T}$ is determined by $H$. In particular, since the root $v$ of $H$ is such that $\lambda(v)$ mentions the relation names $P$ and $S$, then, after its root, $\mathcal{T}$ proceeds with a path where nodes correspond to choices relative to the blocks of $P$ and $S$, and the choices are ordered according to the order over key values discussed before. Then, since $v$ has two children $u_1,u_2$ such that $\lambda(u_1)$ and $\lambda(u_2)$ mention $T$ and $U$, respectively, $\mathcal{T}$ continues in two parallel branches, where in the first branch it lists the choices for the blocks of $T$, while in the second branch it lists the choices for the blocks of $U$.\hfill\markfull
\end{example}

We can now discuss the procedure $\mathsf{Rep}[k]$. In what follows, a \emph{tuple mapping} is an expression of the form $\bar x \mapsto \bar c$, where $\bar x$ is a tuple of terms and $\bar c$ a tuple of constants with $|\bar x| = |\bar c|$, and we say that a set $A$ of tuple mappings is \emph{coherent} if, for any $\bar x \mapsto \bar t \in A$, $\bar x[i] \in \ins{C}$ implies $\bar x[i] = \bar t[i]$, for all $i \in [|\bar x|]$, and for any $\bar x \mapsto \bar t, \bar y \mapsto \bar u \in A$, $\bar x[i] = \bar y[j]$ implies $\bar t[i] = \bar u[j]$, for all $i \in [|\bar x|]$, $j \in [|\bar y|]$.

\begin{algorithm}[t]
	\SetInd{0.7em}{0.7em}
	\DontPrintSemicolon
	\SetArgSty{textnormal}
	\KwIn{A database $D$, a set $\dep$ of primary keys, a CQ $Q(\bar x)$ from $\sjf$, a generalized hypertree decomposition $H=(T,\chi,\lambda)$ of $Q$ of width $k$, and $\bar c \in \adom{D}^{|\bar x|}$, where $(D,Q,H)$ is in normal form}
	\vspace{2mm}
	
	{$v := \mathsf{root}(T)$; $A := \emptyset$;\\}
	\vspace{1mm}
	
	{Assuming $\lambda(v) = \{ R_{i_1}(\bar y_{i_1}), \ldots, R_{i_\ell}(\bar y_{i_\ell})\}$, guess a set $A' =$ $\{ \bar y_{i_1} \mapsto \bar c_1,\ldots,\bar y_{i_\ell} \mapsto \bar c_\ell\}$, with $R_{i_j}(\bar c_j) \in D$ for $j \in [\ell]$, and verify $A \cup A' \cup \{\bar x \mapsto \bar c\}$ is coherent; if not, \textbf{reject};\\}\label{line:begin-repairs}
	
	\vspace{1mm}
	\For{$j=1, \ldots, \ell$}{
		\If{$v$ is the $\prec_T$-minimal covering vertex for $R_{i_j}(\bar y_{i_j})$}{\label{line:minimalc-repairs}
			\ForEach{$B \in \block{\dep}{R_{i_j},D}$}{ \label{line:being-choice-repairs}
				\lIf{$B = \{\beta\}$}{$\alpha := \beta;$}
				\lElseIf{$R_{i_j}(\bar c_j) \in B$}{$\alpha := R_{i_j}(\bar c_j)$;}
				\lElse{Guess $\alpha \in B \cup \{\bot\}$;}
				
				{Label with $\alpha$;\\} \label{line:end-choice-repairs}
			}
		}
	}
	
	\If{$v$ is not a leaf of $T$}{\label{line:end-repairs}
		Assuming $u_1,u_2$ are the (only) children of $v$ in $T$, universally guess $i \in \{1,2\}$;\\
		{$v := u_i$; $A := A'$;\\}
		{\textbf{goto} line \ref{line:begin-repairs};\\}
	}\lElse{\textbf{accept};}
	\caption{The alternating procedure $\mathsf{Rep}[k]$}\label{alg:repairs}
\end{algorithm}

\medskip
\noindent \paragraph{The Procedure $\mathsf{Rep}[k]$.} Roughly, $\mathsf{Rep}[k]$ describes the computation of a well-behaved ATO $M_{R}^{k}$ that, given $D,\dep,Q,H,\bar c$, with $(D,Q,H)$ being in normal form, non-deterministically outputs a labeled tree which, if accepted, encodes in the way discussed above 
a repair of $\{D' \in \opr{D}{\dep} \mid \bar c \in Q(D') \}$.

Let us first clarify some of the conventions that we use in the pseudocode of $\mathsf{Rep}[k]$. As it is standard when describing alternating Turing machines (without output) using pseudocode, all operations/instructions, unless performed ``universally'' or by a ``guess'', are assumed to be deterministic and implemented by means of subsequent non-labeling existential configurations that write nothing on the labeling tape, i.e., each configuration $C$ visited by the underlying machine $M_{R}^{k}$ is existential and non-labeling, and $C$ has exactly one subsequent configuration $C'$, which is also existential and non-labeling.
%
When we write ``Guess $o$'', for some object $o$, we mean that $M_{R}^{k}$ non-deterministically writes on the working tape the binary encoding of the object $o$; the machine writes nothing on the labeling tape, and only visits non-labeling existential configurations while doing so.
When we write ``Universally guess $o$'', for some object $o$, we mean that $M_{R}^{k}$ first enters a non-labeling universal configuration, and then non-deterministically writes the binary encoding of $o$ on the working tape; the machine writes nothing on the labeling tape, only visits non-labeling universal configurations while doing so, and when the encoding of $o$ is completely written, it moves to a non-labeling existential configuration.

Finally, we adopt the following convention regarding labeling configurations and the labeling tape. 
We assume that at the very beginning of $\mathsf{Rep}[k]$, the machine $M_{R}^{k}$ immediately and deterministically moves from the initial configuration, which is always labeling, to a non-labeling existential configuration, without writing anything on the labeling tape, and then proceeds with its computation. Recall that having the initial configuration as labeling, the root of every output tree of $M_{R}^{k}$ is labeled with the empty string.
Finally, when we write ``Label with $o$'', for some object $o$, we mean that $M_{R}^{k}$ deterministically writes the binary encoding of $o$ on the labeling tape (always via non-labeling existential configurations), then moves to a labeling existential configuration, and then moves to a non-labeling existential configuration. This implies that after the execution of ``Label with $o$'', a new node of the output tree is created, and its label is the encoding of $o$. We can now discuss $\mathsf{Rep}[k]$.

Starting from the root $v$ of $H$, $\mathsf{Rep}[k]$ guesses a set $A'$ of tuple mappings from the tuples of terms in $\lambda(v)$ to constants. This set witnesses that the tree being constructed encodes an operational repair $D'$ such that $\bar c \in Q(D')$. For this to be the case, $A'$ must be coherent with $\{\bar x \mapsto \bar c\}$ and the set $A$ of tuple mappings guessed at the previous step.
%
%
%
Having $A'$ in place, lines~\ref{line:being-choice-repairs}-\ref{line:end-choice-repairs} non-deterministically construct a path of the output tree that encodes, as discussed above,
the choices needed to obtain the operational repair $D'$ being considered. In the statement ``Label with $\alpha$'', we choose as the binary encoding of $\alpha$ a \emph{pointer} to $\alpha$ in the input tape; writing a pointer to $\alpha$ allows to use at most logarithmic space on the labeling tape.
Note that due to the ``if'' statement in line~\ref{line:minimalc-repairs}, $\mathsf{Rep}[k]$ does not consider the same block more than once when visiting different nodes of $H$. Moreover, since $H$ is strongly complete, there is a guarantee that, for every visited node $v$ of $H$, $\mathsf{Rep}[k]$ will always output a node of the tree encoding $D'$ (this guarantees that $M_{R}^{k}$ is well-behaved).
The procedure then proceeds in a parallel branch for each child of $v$.
Since $H$ enjoys the connectedness condition, only the current set $A'$ of tuple mappings needs to be kept in memory when moving to a new node of $H$; this set can be stored in logarithmic space using pointers. Moreover, since $H$ is 2-uniform, only a constant number of non-labeling universal configurations are needed to universally move to each child of $v$, and thus $M_{R}^{k}$ is well-behaved.
%
When all branches accept, the output of $\mathsf{Rep}[k]$ is a tree that encodes an operational repair $D'$ of $D$ w.r.t.\ $\dep$ such that $\bar c \in Q(D')$ since each relation of $D$ occurs in $Q$, and thus no blocks of $D$ are left unrepaired.

Let us stress that by guessing $\alpha \in B$ in line 8 of $\mathsf{Rep}[k]$, without considering the symbol $\bot$, we can show that the problem $\sharp\mathsf{SRepairs}[k]$, defined as $\sharp\mathsf{Repairs}[k]$ with the difference that we focus on the classical subset repairs from~\cite{ArBC99}, is in $\mathsf{SpanTL}$, and thus it admits an FPRAS. To the best of our knowledge, this is the first result concerning the combined complexity of the problem of counting classical repairs entailing a query.

\OMIT{
We proceed to show item (1) of Theorem~\ref{the:in-spantl}. This is done via the procedure $\mathsf{Rep}[k]$, depicted in Algorithm~\ref{alg:repairs}, for which we can show the following technical lemma:
\OMIT{
\def\thmrepairsspantl{
	For every $k > 0$, $\sharp\mathsf{Repairs}[k] \in \spantl$.
}

\begin{theorem}\label{thm:repairs-spantl}
\thmrepairsspantl
\end{theorem}

The procedure that 

To prove the above theorem, we introduce a procedure (i.e., Algorithm~\ref{alg:repairs}) describing the computation of an ATO $M$ with input a database $D$, a set $\dep$ of primary keys, a query $Q(\bar x)$ from $\sjf$, a generalized hypertree decomposition $H$ of $Q$ of width $k$, and a tuple $\bar c \in \adom{D}^{|\bar x|}$, such that $(D,Q,H)$ is in normal form. With Algorithm~\ref{alg:repairs} in place, our goal is to prove the following lemma.

\def\lemrepairsato{
	For every $k > 0$, the following hold:
	\begin{itemize}
		\item Algorithm~\ref{alg:repairs} can be implemented as a well-behaved ATO $M$, and
		\item for every database $D$, set $\dep$ of primary keys, query $Q(\bar x)$, generalized hypertree decomposition $H$ of $Q$ of width $k$, and $\bar c \in \adom{D}^{|\bar x|}$ with $(D,Q,H)$ in normal form, $\mathsf{span}_M(D,\dep,Q,H,\bar c) = \sharp\mathsf{Repairs}[k](D,\dep,Q,H,\bar c)$.
	\end{itemize}
}
}

\def\lemrepairsato{
	For every $k > 0$, the following hold:
	\begin{enumerate}
		\item $\mathsf{Rep}[k]$ can be implemented as a well-behaved ATO $M_R^k$.
		\item For a database $D$, a set $\dep$ of primary keys, a CQ $Q(\bar x)$ from $\sjf$, a generalized hypertree decomposition $H$ of $Q$ of width $k$, and a tuple $\bar c \in \adom{D}^{|\bar x|}$, where $(D,Q,H)$ is in normal form, $\mathsf{span}_{M_R^k}(D,\dep,Q,H,\bar c) = |\{D' \in \opr{D}{\dep} \mid \bar c \in Q(D')\}|$.
	\end{enumerate}
}
\begin{lemma}\label{lem:repairs-ato}
\lemrepairsato
\end{lemma}

Item (1) of Theorem~\ref{the:in-spantl} readily follows from Lemma~\ref{lem:repairs-ato}. The rest of this section is devoted to discussing the procedure $\mathsf{Rep}[k]$ and why Lemma~\ref{lem:repairs-ato} holds. But first we need some auxiliary notions.

%

\medskip
\noindent \paragraph{Auxiliary Notions.}
Consider a database $D$, a set $\dep$ of primary keys, and a fact $\alpha = R(c_1,\ldots,c_n) \in D$. The \emph{key value} of $\alpha$ w.r.t. $\dep$ is
\begin{eqnarray*}
	\keyval{\dep}{\alpha}\
	= \left\{
	\begin{array}{ll}
		\langle R, \langle c_{i_1},\ldots,c_{i_m} \rangle \rangle & \text{if }\key{R} = \{i_1,\ldots,i_m\} \in \dep,\\
		&\\
		\langle R, \langle c_1,\ldots,c_n \rangle \rangle & \text{otherwise.}
	\end{array} \right.
\end{eqnarray*}
Moreover, we define
\[
\block{\dep}{\alpha,D}\ =\ \{\beta \in D \mid \keyval{\dep}{\beta} = \keyval{\dep}{\alpha}\},
\]
and, for a relation name $R$, we define
\[
\block{\dep}{R,D}\ =\ \{ \block{\dep}{R(\bar c),D} \mid R(\bar c) \in D\}.
\]
A \emph{tuple mapping} is an expression of the form $\bar x \mapsto \bar c$, where $\bar x$ is a tuple of variables and $\bar c$ a tuple of constants with $|\bar x| = |\bar c|$. A set $A$ of tuple mappings is \emph{coherent} if for any two tuple mappings $\bar x \mapsto \bar t, \bar y \mapsto \bar u \in A$, $\bar x[i] = \bar y[j]$ implies $\bar t[i] = \bar u[j]$, for all $i \in [|\bar x|]$, $j \in [|\bar y|]$.
%
We are now ready to discuss the procedure $\mathsf{Rep}[k]$.

\begin{algorithm}[t]
	\SetInd{0.7em}{0.7em}
	\DontPrintSemicolon
	\SetArgSty{textnormal}
	\KwIn{A database $D$, a set $\dep$ of primary keys, a CQ $Q(\bar x)$ from $\sjf$, a generalized hypertree decomposition $H=(T,\chi,\lambda)$ of $Q$ of width $k$, and $\bar c \in \adom{D}^{|\bar x|}$, where $(D,Q,H)$ is in normal form}
	\vspace{2mm}
	
	{$v := \mathsf{root}(T)$; $A := \emptyset$;\\}
	\vspace{1mm}
	
	{Assuming $\lambda(v) = \{ R_{i_1}(\bar y_{i_1}), \ldots, R_{i_\ell}(\bar y_{i_\ell})\}$, guess a set $A' =$ $\{ \bar y_{i_1} \mapsto \bar c_1,\ldots,\bar y_{i_\ell} \mapsto \bar c_\ell\}$, with $R_{i_j}(\bar c_j) \in D$ for $j \in [\ell]$, and verify $A \cup A' \cup \{\bar x \mapsto \bar c\}$ is coherent; if not, \textbf{reject};\\}\label{line:begin-repairs}
	
	\vspace{1mm}
	\For{$j=1, \ldots, \ell$}{
		\If{$v$ is the $\prec_T$-minimal covering vertex for $R_{i_j}(\bar y_{i_j})$}{
			\ForEach{$B \in \block{\dep}{R_{i_j},D}$}{
				\lIf{$B = \{\beta\}$}{$\alpha := \beta;$}
				\lElseIf{$R_{i_j}(\bar c_j) \in B$}{$\alpha := R_{i_j}(\bar c_j)$;}
				\lElse{Guess $\alpha \in B \cup \{\bot\}$;}
				
				{Label with $\alpha$;\\}
			}
		}
	}
	
	\If{$v$ is not a leaf of $T$}{\label{line:end-repairs}
		{Universally guess a child $u$ of $v$ in $T$;\\}
		{$v := u$; $A := A'$;\\}
		{\textbf{goto} line \ref{line:begin-repairs};\\}
	}\lElse{\textbf{accept};}
	\caption{The alternating procedure $\mathsf{Rep}[k]$}\label{alg:repairs}
\end{algorithm}

\medskip
\noindent \paragraph{The Procedure $\mathsf{Rep}[k]$.} Let us first say that we assume that $\mathsf{Rep}[k]$ immediately moves from the initial (necessarily labeling) state to a non-labeling state, and then proceeds with its computation. Moreover, we use the statement ``Label with $o$'', for some object $o$, to say that the underlying ATO writes the encoding of $o$ in the labeling tape, then moves to a labeling state, and finally moves to a non-labeling state.
At a high level, $\mathsf{Rep}[k]$ traverses the generalized hypertree decomposition $(T,\chi,\lambda)$ of the input CQ $Q$. Starting from the root, for each node $v$ of $T$, the procedure guesses a mapping from the variables of the atoms in $\lambda(v)$ to constants (the set of tuple mappings $A'$) that witnesses a homomorphism from $\lambda(v)$ to $D$. Of course, such a homomorphism must map $\bar x$ to $\bar c$, and must be also coherent with the homomorphism used to map the atoms of the previously visited node of $T$ (at the beginning, no previous node has been visited, and thus, $A = \emptyset$). 
Note that the set $A'$ contains a constant number of tuple mappings (i.e., at most $k$) of the form $\bar y_{i_j} \mapsto \bar c_j$, which can be stored in logarithmic space using two pointers to $\bar y_{i_j}$ and $\bar c_j$, respectively, since both appear in the input.

Having $A'$ in place, the procedure now non-deterministically outputs a path, whose vertices are labeled with all the facts over the relation names in $\lambda(v)$ of some operational repair $D'$. That is, for each relation name $R_{i_j}$ in $\lambda(v)$, in some fixed order, the procedure chooses, for each block $B$ of $R_{i_j}$, again in some fixed order, which fact of $B$ (if any) should be part of $D'$, and outputs this choice as a node labeled with this fact. If $B$ is a singleton, or the mappings of $A'$ force a certain fact of $B$ to be part of $D'$, then the choice is deterministic; otherwise, a fact from $B$ (or none at all) is chosen non-deterministically. To use no more than logarithmic space in the labeling tape, a \emph{pointer} to the fact is written instead of the fact.

Since multiple nodes of $T$ might share the same atom, the choices for the blocks  relative to the atom with relation $R_{i_j}$ are made only if $v$ is the $\prec_T$-minimal covering vertex. If not, it means that the procedure will make such a choice in a vertex $v'$ of $T$ with $v' \prec_T v$. Note that thanks to the fact that $Q$ is self-join free, no conflicting choices can be made on two different nodes for the same relation. Moreover, since $(T,\chi,\lambda)$ is strongly complete, $v$ is the $\prec_T$-minimal covering vertex of at least one atom in $\lambda(v)$, and thus, the procedure always outputs some facts, for each vertex of $T$ (this is needed to guarantee that the underlying ATO is well-behaved).

Once the path relative to $v$ is computed, if the current vertex $v$ is not a leaf, then the procedure proceeds in parallel to all children of $v$ in $T$, and updates the last set of mappings $A$ to be the current one (that is, $A'$). Note that the procedure does not need to remember $A \cup A'$ due to the connectedness condition satisfied by $H$. Moreover, since $(T,\chi,\lambda)$ is 2-uniform, the underlying ATO does not require an arbitrarily long sequence of universal, non-labeling configurations, which ensures that it is well-behaved.
Now, if $v$ is a leaf, then there are no other vertices to process, and the procedure accepts on this branch of the computation. If all branches of the computation are accepting, the output is a tree, encoding an operational repair $D'$ of $D$ w.r.t.\ $\dep$ such that $\bar c \in Q(D')$; in fact, since every relation of $D$ occurs in $Q$, no blocks of $D$ are left unrepaired.

We proceed to give an example that illustrates the above discussion concerning the procedure $\mathsf{Rep}[k]$.

\begin{example}\label{ex:repairs-algo}
Consider the Boolean CQ $Q$ 
\[
\text{Ans}() \text{ :- } P(x,y),S(y,z),T(z,x),U(y,w),
\]
which has generalized hypertreewidth $2$. A generalized hypertree decomposition $H=(T,\chi,\lambda)$ of $Q$ of width $2$ is the following:

\vspace*{2mm}
\centerline{\includegraphics[width=.25\textwidth]{example-htd}}
\vspace*{2mm}

\noindent where each box denotes a vertex $v$ of $T$, and the content of the box is the set of variables $\chi(v)$ while on its side is the set $\lambda(v)$.
Consider now the set of primary keys $\dep = \{\key{R} = \{1\} \mid R \in \{P,S,T,U\}\}$, and the database $D$
\begin{multline*}
\{ P(a_1,b), P(a_1,c), P(a_2,b), P(a_2,c), P(a_2,d),  S(c,d), S(c,e), \\
T(d,a_1), U(c,f), U(c,g), U(h,i), U(h,j), U(h,k)\}.
\end{multline*}
Note that $(D,Q,H)$ is in normal form. A possible accepting computation of the ATO described by $\mathsf{Rep}[2]$ with input $D$, $\dep$, $Q$, $H$, and the empty tuple $()$, is the one where, for the root of $H$, $A' = \{ (x,y) \mapsto (a_1,c), (y,z) \mapsto (c,d)\}$ and the computed path contains three nodes labeled with $P(a_1,c)$, $\bot$, and $S(c,d)$, respectively. Then, for the first child of the root, $A' = \{(z,x) \mapsto (d,a_1)\}$ and the path is made of one node with label $T(d,a_1)$, while for the second child, $A' = \{(y,w) \mapsto (c,f)\}$ and the path contains two nodes labeled with $U(c,f)$ and $U(h,i)$ respectively. The overal output of this computation is the tree\footnote{For the sake of presentation, we omit the root node that is always labeled with the empty string (recall that the initial state of an ATO is always labeling).} on the left below

\vspace*{2mm}
\centerline{\includegraphics[width=0.21\textwidth]{example-tree1} \qquad \includegraphics[width=0.21\textwidth]{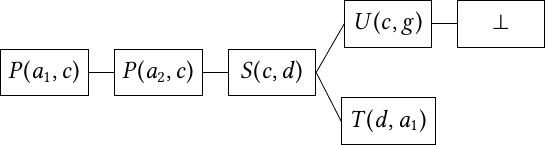}}
\vspace*{2mm}

\noindent that encodes the operational repair $D'$ where $P(a_1,c)$, $S(c,d)$, $T(d,a_1)$, $U(c,f)$, and $U(h,i)$ are all kept from their corresponding block, while no fact from $P(a_2,b), P(a_2,c), P(a_2,d)$ is kept. It is easy to verify that $D' \models Q$. Another valid output is the tree on the right in the figure above,
where the procedure chooses, at the root, one fact for each block of $P$, at the first child, one fact for each block of $T$, and at the second child, the fact $U(c,g)$ is chosen, while no fact from $U(h,i),U(h,j),U(h,k)$ is chosen.\hfill\markfull
\end{example}
}

\section{Concluding Remark}\label{sec:conclusion}
%


Considering all the operational repairs to be equally important is a self-justified choice. However, as discussed in~\cite{CLPS22}, this does not take into account the support in terms of the repairing process. In other words, an operational repair obtained by very few repairing sequences is equally important as an operational repair obtained by several repairing sequences. These considerations lead to the notion of uniform sequences, where all complete repairing sequences are considered equally important. We can then define the problem $\sharp\mathsf{Sequences}[k]$, similarly to $\sharp\mathsf{Repairs}[k]$, that counts the complete repairing sequences leading to an operational repair that entails the tuple in question. We can show that this problem can be placed in $\spantl$, and thus, $\ocqa$ admits an FPRAS in combined complexity also when considering uniform sequences.

\OMIT{
The take-home message of our work is that uniform operational CQA is flexible enough to lead to approximability results that go beyond the simple case of primary keys, which seems to be the limit of the classical approach to CQA.

Although we understand well uniform operational CQA, there are still interesting open problems  concerning approximability:
\begin{enumerate}
	\item the case of keys and uniform repairs (we only have a negative result for the problem of counting repairs), 
	\item the case of keys/FDs and uniform sequences, and 
	\item the case of FDs and uniform operations (we only have a positive result assuming singleton operations).
\end{enumerate}

Another interesting direction for future research is to consider conceptually relevant distributions that deviate from the uniform ones considered in this work, and perform the same complexity analysis of exact and approximate operational CQA.
}

\begin{acks}
	We thank the anonymous referees for their feedback. This
	work was funded by the \grantsponsor{eu-pnrr}{European Union}{http://dx.doi.org/10.13039/501100000780} - Next Generation
	EU under the MUR PRIN-PNRR grant \grantnum{eu-pnrr}{P2022KHTX7} ``DISTORT'',
	and by the \grantsponsor{epsrc}{EPSRC}{http://dx.doi.org/10.13039/501100000266} grant \grantnum{epsrc}{EP/S003800/1} ``EQUID''.
\end{acks}



\newpage
\appendix
\section{Proof of Theorem~\ref{the:ocqa-exact}}

We prove the following result:

\begin{manualtheorem}{\ref{the:ocqa-exact}}
\theocqaexact
\end{manualtheorem}

For brevity, for a fact $f$, we write $-f$ instead of $-\{f\}$ to denote the operation that removes $f$.
%

%

Fix an arbitrary $k>0$. We are going to show that $\ocqa^\ur[\sjf \cap \ghw_k]$ is $\sharp ${\rm P}-hard. Let us first introduce the $\sharp ${\rm P}-hard problem that we are going to reduce to $\ocqa^\ur[\sjf \cap \ghw_k]$. 
Consider the undirected bipartite graph $H = (V_H,E_H)$, depicted in Figure~\ref{fig:hard-graph}, where $V_{H,L} = \{1_L,0_L,?_L\}$ and $V_{H,R} = \{1_R,0_R,?_R\}$, with $\{V_{H,L},V_{H,R}\}$ being a partition of $V_H$, and $E_H = \{\{u,v\} \mid (u,v) \in (V_{H,L} \times V_{H,R}) \setminus \{(1_L,1_R)\}\}$.

\medskip
\begin{figure}[ht]
	\centering
	\begin{tikzpicture}[thick, main/.style = {draw, circle}]
	\node[main] (1L) {$1_L$}; 
	\node[main] (0L) [below=15mm of 1L] {$0_L$}; 
	\node[main] (?L) [below=15mm of 0L] {$?_L$}; 
	\node[main] (1R) [right=50mm of 1L] {$1_R$}; 
	\node[main] (0R) [below=15mm of 1R] {$0_R$}; 
	\node[main] (?R) [below=15mm of 0R] {$?_R$};
	\draw (1L) -- (0R);
	\draw (1L) -- (?R);
	\draw (0L) -- (1R);
	\draw (0L) -- (0R);
	\draw (0L) -- (?R);
	\draw (?L) -- (1R);
	\draw (?L) -- (0R);
	\draw (?L) -- (?R);
	\end{tikzpicture}
	\caption{The graph $H$.}
	\label{fig:hard-graph}
\end{figure}
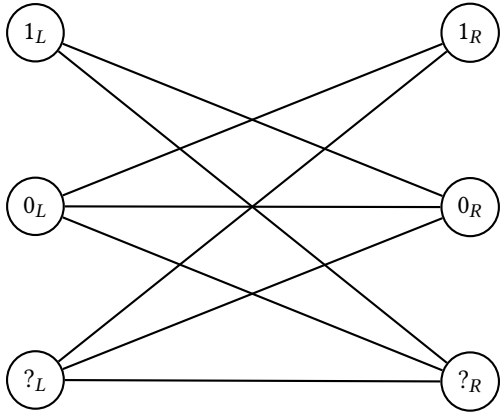

\medskip

\noindent Given an undirected graph $G = (V_G,E_G)$, a homomorphism from $G$ to $H$ is a mapping $h : V_G \rightarrow V_H$ such that $\{u,v\} \in E_G$ implies $\{h(u),h(v)\} \in E_H$. We write $\mathsf{hom}(G,H)$ for the set of homomorphisms from $G$ to $H$.
The problem $\sharp H\text{-}\mathsf{Coloring}$ is defined as follows:

\medskip

\begin{center}
	\fbox{\begin{tabular}{ll}
			{\small PROBLEM} : & $\sharp H\text{-}\mathsf{Coloring}$\\
			{\small INPUT} : & A connected undirected graph $G$.\\
			{\small OUTPUT} : &  The number $|\mathsf{hom}(G,H)|$.
	\end{tabular}}
\end{center}

\medskip

\noindent It is implicit in~\cite{Dyer00} that $\sharp H\text{-}\mathsf{Coloring}$ is $\sharp ${\rm P}-hard. In fact,~\cite{Dyer00} establishes the following dichotomy result: $\sharp \hat{H}\text{-}\mathsf{Coloring}$ is $\sharp ${\rm P}-hard if $\hat{H}$ has a connected component which is neither an isolated node without a loop, nor a complete graph with all loops present, nor a complete bipartite graph without loops; otherwise, it is solvable in polynomial time. Since our fixed graph $H$ above consists of a single connected component which is neither a single node, nor a complete graph with all loops present, nor a complete bipartite graph without loops (the edge $(1_L,1_R)$ is missing), we conclude that $\sharp H\text{-}\mathsf{Coloring}$ is indeed $\sharp ${\rm P}-hard.
Note that we can assume the input graph $G$ to be connected since the total number of homomorphisms from a (non-connected) graph $G'$ to $H$ is the product of the number of homomorphisms from each connected component of $G'$ to $H$.
We proceed to show via a polynomial-time Turing reduction from $\sharp H\text{-}\mathsf{Coloring}$ that $\ocqa^\ur[\sjf \cap \ghw_k]$ is $\sharp ${\rm P}-hard. 
In other words, we need to show that, given a connected undirected graph $G$, the number $|\mathsf{hom}(G,H)|$ can be computed in polynomial time in the size of $G$ assuming that we have access to an oracle for the problem $\ocqa^\ur[\sjf \cap \ghw_k]$.

Fix a connected undirected bipartite graph $G$ with vertex partition $\{V_{G,L},V_{G,R}\}$. Let $\ins{S}_k$ be the schema 
\[
\{V_L/2, V_R/2, E/2, T/1, T'/1\}\cup \{C_{i,j}/2 \mid i,j\in [k+1],\ i<j\},
\]
with $(A,B)$ being the tuple of attributes of both $V_L$ and $V_R$.
We define the database $D_G^k$ over $\ins{S}_k$ encoding $G$ as follows:
\begin{eqnarray*}
&& \{V_L(u,0),V(u,1) \mid u \in V_{G,L}\}\\ 
& \cup& \{V_R(u,0),V(u,1) \mid u \in V_{G,R}\}\\
&\cup& \{E(u,v) \mid (u,v) \in E_G\}\ \cup\ \{T(1)\}\ \cup\ \{T'(1)\}\\
&\cup& \{C_{i,j}(i,j) \mid i,j\in [k+1],\ i<j\}.    
\end{eqnarray*}
We also define the set $\dep$ of keys consisting of
\[
\key{V_L} = \{1\} \quad \textrm{and} \quad \key{V_R} = \{1\}
\]
and the (constant-free) self-join-free Boolean CQ $Q_k$ 
\[
\textrm{Ans}()\ \text{:-}\ E(x,y), V_L(x,z), V_R(y,z'), T(z), T'(z'), \bigwedge_{\substack{i,j\in [k+1]\\i<j}} C_{i,j}(w_i,w_j).
\]
It is easy to verify that $Q_k\in \ghw_k$ since the big conjunction over $C_{i,j}$-atoms encodes a clique of size $k+1$. 
%
Let us stress that the role of this conjunction is only to force the CQ $Q_k$ to be of generalized hypertreewidth $k$ as this part of the query will be entailed in every repair (since there are no inconsistencies involving a $C_{i,j}$-atom).
%

\OMIT{
\noindent
Given an undirected bipartite graph $G = (V_G,E_G)$, with vertex partition $V_{G,L}\cup V_{G,R}$, we define the following database over $\ins{S}$ encoding $G$:
\begin{align*}
D_G^k\ =\ &\{V_L(u,0),V(u,1) \mid u \in V_{G,L}\}\ \cup\ \{V_R(u,0),V(u,1) \mid u \in V_{G,R}\}\\
&\cup\ \{E(u,v) \mid (u,v) \in E_G\}\ \cup\ \{T(1)\}\ \cup\ \{T'(1)\}\\
&\cup\ \{C_{i,j}(i,j) \mid i,j\in [k+1],\ i<j\}.    
\end{align*}

\noindent
Further, we stress that the role of relations $C_{i,j}$ is only to provide the query with the required generalized hypertree-width, as this part of the query will be entailed in every repair (since there are no inconsistencies involving a $C_{i,j}$-atom).
}

Consider now the algorithm $\mathsf{HOM}$, which accepts as input a connected undirected graph $G = (V_G,E_G)$ and has access to an oracle for $\ocqa^\ur[\sjf \cap \ghw_k]$, 
and performs as follows:
\begin{enumerate}
	\item If $|V_G|=1$ and $E_G=\emptyset$, then output the number $6$; otherwise, continue.
	\item If $G$ is not bipartite, then output the number $0$; otherwise, continue.
	\item Let $r = \mathsf{RF}(D_G^k,\dep,Q_k,())$.
	\item Output the number $2 \cdot 3^{|V_G|} \cdot (1- r)$.
\end{enumerate}
It is clear that $\mathsf{HOM}(G)$ runs in polynomial time in $||G||$ assuming access to an oracle for the problem $\ocqa^\ur[\sjf \cap \ghw_k]$. It remains to show that $|\mathsf{hom}(G,H)| = \mathsf{HOM}(G)$.

If $G$ consists of only one isolated node, then there are six homomorphism from $G$ to $H$ (to each of the six vertices of $H$). Further, if $G$ is not bipartite, then clearly $|\mathsf{hom}(G,H)|=0$.
Let us now consider the case where $G$ is a connected undirected bipartite graph with vertex partition $\{V_{G,L},V_{G,R}\}$.
%
Recall that
\[\mathsf{RF}(D_G^k,\dep,Q_k,()) = \frac{|\{D \in \opr{D_G^k}{\dep} \mid D \models Q_k\}|}{|\opr{D_G^k}{\dep}|}.\]
%
Observe that there are $3^{|V_G|}$ operational repairs of $D_G^k$ w.r.t.~$\dep$. In particular, in each such a repair $D$, for each node $u \in V_{G,L}$ of $G$, either $V_L(u,0) \in D$ and $V_L(u,1) \not \in D$, or $V_L(u,0) \not \in D$ and $V_L(u,1) \in D$, or $V_L(u,0),V_L(u,1) \not \in D$. Further, for each node $u \in V_{G,R}$ of $G$, either $V_R(u,0) \in D$ and $V_R(u,1) \not \in D$, or $V_R(u,0) \not \in D$ and $V_R(u,1) \in D$, or $V_R(u,0),V_L(u,1) \not \in D$. Thus, there are $3^{|V_{G,L}|}\cdot 3^{|V_{G,R}|}=3^{|V_G|}$ many operational repairs of $D_G^k$ w.r.t.~$\dep$.
%
%
Therefore,
\[
\mathsf{RF}(D_G^k,\dep,Q_k,())\ =\ \dfrac{|\{D \in \opr{D_G^k}{\dep} \mid D \models Q_k\}|}{3^{|V_G|}}.
\]
Thus, $\mathsf{HOM}(G)$ coincides with
\begin{align*}
&2 \cdot 3^{|V_G|} \cdot \left(1 - \dfrac{|\{D \in \opr{D_G^k}{\dep} \mid D \models Q_k\}|}{3^{|V_G|}}\right)\\
=\; &2 \cdot \left(3^{|V_G|} - |\{D \in \opr{D_G^k}{\dep} \mid D \models Q_k\}|\right).
\end{align*}
Since $D_G^k$ has $3^{|V_G|}$ operational repairs w.r.t.~$\dep$, we can conclude that $3^{|V_G|} - |\{D \in \opr{D_G^k}{\dep} \mid D \models Q_k\}|$ is precisely the cardinality of the set $\{D \in \opr{D_G^k}{\dep} \mid D \not\models Q_k\}$.

Since $G$ and $H$ are both connected and bipartite, we know that any homomorphism from $G$ to $H$ must preserve the partitions, i.e., every vertex in one partition of $G$ must be mapped to the same partition in $H$ (and the other partition of $G$ must be mapped to the other partition of $H$). Let $\mathsf{hom'}(G,H)$ be the set of homomorphisms from $G$ to $H$ such that $V_{G,L}$ is mapped to $V_{H,L}$ and $V_{G,R}$ is mapped to $V_{H,R}$. Since $H$ is symmetric, we know that for every homomorphism in $h\in \mathsf{hom'}(G,H)$, there are two homomorphisms in $\mathsf{hom}(G,H)$. In particular, $h\in \mathsf{hom}(G,H)$ and $h'\in \mathsf{hom}(G,H)$, where $h'$ is the homomorphism that is essentially $h$ but maps the vertices of $V_{G,L}$ to $V_{H,R}$ and the vertices of $V_{G,R}$ to $V_{H,L}$. Thus, $|\mathsf{hom}(G,H)|=2\cdot|\mathsf{hom'}(G,H)|$. 
We proceed to show that $|\{D \in \opr{D_G^k}{\dep} \mid D \not\models Q_k\}|$ coincides with $|\mathsf{hom'}(G,H)|$.

\begin{lemma}\label{lem:keys-aux}
	$|\mathsf{hom'}(G,H)|\ =\ |\{D \in \opr{D_G^k}{\dep} \mid D \not\models Q_k\}|$.
\end{lemma}

\begin{proof}
	It suffices to show that there exists a bijection from the set $\mathsf{hom'}(G,H)$ to the set $\{D \in \opr{D_G^k}{\dep} \mid D \not\models Q_k\}$. To this end, we define the mapping $\mu : \mathsf{hom'}(G,H) \ra \PS(D_G^k)$, where $\PS(D_G^k)$ is the power set of $D_G^k$, as follows: for each $h \in \mathsf{hom'}(G,H)$,
	\begin{align*}
	\mu(h)\ =\ &\{V_L(u,0) \mid u \in V_{G,L} \text{ and } h(u) = 0_L\}\\
	&\cup\ \{V_L(u,1) \mid u \in V_{G,L} \text{ and } h(u) = 1_L\}\\
	&\cup\ \{V_R(u,0) \mid u \in V_{G,R} \text{ and } h(u) = 0_R\}\\
	&\cup\ \{V_R(u,1) \mid u \in V_{G,R} \text{ and } h(u) = 1_R\}\\
	&\cup\ \{E(u,v) \mid \{u,v\} \in E_G\}\ \cup\ \{T(1)\}\ \cup\ \{T'(1)\}\\
	&\cup\ \{C_{i,j}(i,j) \mid i,j\in [k+1],\ i<j\}.
	\end{align*}
	We proceed to show the following three statements:
	\begin{enumerate}
		\item $\mu$ is correct, that is, it is indeed a function from $\mathsf{hom'}(G,H)$ to $\{D \in \opr{D_G^k}{\dep} \mid D \not\models Q_k\}$.
		\item $\mu$ is injective.
		\item $\mu$ is surjective.
	\end{enumerate}
	
	\medskip 
	
	\noindent
	\paragraph{The mapping $\mu$ is correct.} Consider an arbitrary homomorphism $h \in \mathsf{hom'}(G,H)$. We need to show that there exists a $(D_G^k,\dep)$-repairing sequence $s_h$ such that $\mu(h) = s_h(D_G^k)$, $s_h(D_G^k) \models \dep$ (i.e., $s_h$ is complete), and $Q_k(s_h(D_G^k)) = \emptyset$. Let $V_G = \{u_1,\ldots,u_n\}$. Consider the sequence $s_h = \op_1,\ldots,\op_n$ such that, for every $i \in [n]$:
	\[
	\op_{i}\ =\
	\begin{cases}
	-V_L(u_i,1) & \text{if } h(u_i) = 0_L \\
	-V_L(u_i,0) & \text{if } h(u_i) = 1_L \\
	-\{V_L(u_i,0),V_L(u_i,1)\} & \text{if } h(u_i) =\ ?_L \\
	-V_R(u_i,1) & \text{if } h(u_i) = 0_R \\
	-V_R(u_i,0) & \text{if } h(u_i) = 1_R \\
	-\{V_R(u_i,0),V_R(u_i,1)\} & \text{if } h(u_i) =\ ?_R \\
	\end{cases}
	\]
	In simple words, the homomorphism $h$ guides the repairing process, i.e., $h(u_i) = 0_L$ (resp., $h(u_i)=1_L$) implies $V_L(u_i,0)$ (resp., $V_L(u_i,1)$) should be kept, while $h(u_i) =\ ?_L$ implies none of the atoms $V_L(u_i,0),V_L(u_i,1)$ should be kept. The analogous is true for the vertices that are mapped to one of $\{0_R, 1_R, ?_R\}$. It is easy to verify that $s_h$ is indeed a $(D_G^k,\dep)$-repairing sequence $s_h$ such that $\mu(h) = s_h(D_G^k)$ and $s_h(D_G^k) \models \dep$.
	The fact that $Q_k(s_h(D_G^k)) = \emptyset$ follows from the fact that, for every edge $\{u,v\} \in E_G$, $\{h(u),h(v)\} \in E_H$ cannot be the edge $(1_L,1_R)$, since it is not in $H$. This implies that for every $\{u,v\} \in E_G$, it is not possible that the atoms $V_L(u,1)$ and $V_R(v,1)$ (or $V_R(u,1)$ and $V_L(v,1)$) coexist in $s_h(D_G^k)$, which in turn implies that $Q_k(s_h(D_G^k)) = \emptyset$, as needed.

	\medskip

	\noindent
	\paragraph{The mapping $\mu$ is injective.} Assume that there are two distinct homomorphisms $h,h' \in \mathsf{hom'}(G,H)$ such that $\mu(h) = \mu(h')$. By the definition of $\mu$, we get that $h(u) = h'(u)$, for every node $u \in V_G$. But this contradicts the fact that $h$ and $h'$ are different homomorphisms of $\mathsf{hom'}(G,H)$. Therefore, for every two distinct homomorphisms $h,h' \in \mathsf{hom'}(G,H)$, $\mu(h) \neq \mu(h')$, as needed.

	\medskip

	\noindent
	\paragraph{The mapping $\mu$ is surjective.} Consider an arbitrary operational repair $D \in \{D \in \opr{D_G^k}{\dep} \mid D \not\models Q_k\}$. We need to show that there exists $h \in \mathsf{hom'}(G,H)$ such that $\mu(h) = D$. We define the mapping $h_D : V_G \ra V_H$ as follows: for every $u \in V_G$:
	\[
	h_D(u)\ =\
	\begin{cases}
	1_L & \text{if } u\in V_{G,L} \text{ and } V_L(u,1) \in D \text{ and } V_L(u,0) \not\in D \\
	0_L & \text{if } u\in V_{G,L} \text{ and } V_L(u,1) \not\in D \text{ and } V_L(u,0) \in D \\
	?_L & \text{if } u\in V_{G,L} \text{ and } V_L(u,1) \not\in D \text{ and } V_L(u,0) \not\in D \\
	1_R & \text{if } u\in V_{G,R} \text{ and } V_R(u,1) \in D \text{ and } V_R(u,0) \not\in D \\
	0_R & \text{if } u\in V_{G,R} \text{ and } V_R(u,1) \not\in D \text{ and } V_R(u,0) \in D \\
	?_R & \text{if } u\in V_{G,R} \text{ and } V_R(u,1) \not\in D \text{ and } V_R(u,0) \not\in D \\
	\end{cases}
	\]
	It is clear that $h_D$ is well-defined: for every $u \in V_G$, $h_D(u) = x$ and $h_D(u) = y$ implies $x=y$. It is also clear that $\mu(h_D) = D$. It remains to show that $h_D \in \mathsf{hom'}(G,H)$. Consider an arbitrary edge $\{u,v\} \in E_G$. W.l.o.g.~let $u\in V_{G,L}$ and $v\in V_{G,R}$. By contradiction, assume that $\{h_D(u),h_D(v)\} \not\in E_H$. This implies that $h_D(u) = 1_L$ and $h_D(v) = 1_R$, since $u$ is mapped to $V_{H,L}$ and $v$ is mapped to $V_{H,R}$. Therefore, $D$ contains both atoms $V_L(u,1)$ and $V_R(v,1)$, which in turn implies that $Q_k(D) \neq \emptyset$, which contradicts the fact that $D \in \{D \in \opr{D_G^k}{\dep} \mid D \not\models Q_k\}$.
\end{proof}

Summing up, since $\mathsf{HOM}(G) = 2\cdot|\{D \in \opr{D_G^k}{\dep} \mid D \not\models Q_k\}|$, Lemma~\ref{lem:keys-aux} implies that
\[
\mathsf{HOM}(G)\ =\ 2\cdot|\mathsf{hom'}(G,H)|\ =\ |\mathsf{hom}(G,H)|.
\]
This shows that indeed $\mathsf{HOM}$ is a polynomial-time Turing reduction from $\sharp H\text{-}\mathsf{Coloring}$ to $\ocqa^\ur[\sjf \cap \ghw_k]$. Thus, $\ocqa^\ur[\sjf \cap \ghw_k]$ is $\sharp ${\rm P}-hard, as needed.


\section{Proof of Theorem~\ref{the:ocqa-apx}}

We proceed to prove the following result:

\begin{manualtheorem}{\ref{the:ocqa-apx}}
\theocqaapx
\end{manualtheorem}

For brevity, for a fact $f$, we write $-f$ instead of $-\{f\}$ to denote the operation that removes $f$.
We prove each item of the theorem separately. 

\subsection{Item (1)}

\noindent We proceed to show that, unless ${\rm RP} = {\rm NP}$, $\ocqa^\ur[\sjf]$ does not admit an FPRAS.
Consider the decision version of $\ocqa^\ur[\sjf]$, dubbed $\mathsf{Pos}\ocqa^\ur[\sjf]$, defined as follows:

\medskip

\begin{center}
	\fbox{\begin{tabular}{ll}
			{\small PROBLEM} : & $\mathsf{Pos}\ocqa^\ur[\sjf]$
			\\
			{\small INPUT} : & A database $D$, a set $\dep$ of primary keys,\\
			& a query $Q(\bar x)$ from $\mathsf{SJF}$, a tuple $\bar c \in \adom{D}^{|\bar x|}$.
			\\
			{\small QUESTION} : & $\mathsf{RF}(D,\dep,Q,\bar c)>0$.
	\end{tabular}}
\end{center}

\medskip

\noindent To establish our claim it suffices to show that $\mathsf{Pos}\ocqa^\ur[\sjf]$ is NP-hard.
Indeed, assuming that an FPRAS for $\ocqa^\ur[\sjf]$ exists, we can place the problem $\mathsf{Pos}\ocqa^\ur[\sjf]$ in BPP, which implies that $\textrm{NP} \subseteq \textrm{BPP}$. This in turn implies that ${\rm RP} = {\rm NP}$~\cite{Jerrum03}, which contradicts our assumption that ${\rm RP}$ and ${\rm NP}$ are different.

To show that $\mathsf{Pos}\ocqa^\ur[\sjf]$, we reduce from the problem of deciding whether an undirected graph is 3-colorable, which is a well-known NP-hard problem.
Recall that, for an undirected graph $G = (V,E)$, we say that $G$ is {\em 3-colorable} if there is a mapping $\mu: V \ra \{1,2,3\}$ such that, for every edge $\{u,v\}\in E$, we have that $\mu(u)\neq \mu(v)$.
The $3{\rm -}\mathsf{Colorability}$ problem is defined as follows:

\medskip

\begin{center}
	\fbox{\begin{tabular}{ll}
			{\small PROBLEM} : & $3{\rm-}\mathsf{Colorability}$\\
			{\small INPUT} : & An undirected graph $G$.\\
			{\small QUESTION} : &  Is $G$ 3-colorable?
	\end{tabular}}
\end{center}

\medskip


Our reduction is actually an adaptation of the standard reduction from $3{\rm-}\mathsf{Colorability}$ for showing that CQ evaluation NP-hard.
Consider an undirected graph $G=(V,E)$. Let $\ins{S}_G$ be the schema
\[
\left\{C^{u,v}/2, C^{v,u}/2 \mid \{u,v\}\in E\right\}.
\]
We define the database $D_G$ over $\ins{S}_G$ as follows:
\[
\left\{C^{u,v}(i,j), C^{v,u}(i,j) \mid \{u,v\}\in E;\ i,j\in \{1,2,3\};\ i\neq j\right\}.
\]
We further consider the set of primary 
keys $\dep$ to be empty and define the (constant-free) Boolean CQ $Q_G$ as
\[
\textrm{Ans}()\ \text{:-}\  \bigwedge_{\{u,v\}\in E} C^{u,v}(x_u,x_v)\wedge C^{v,u}(x_v,x_u).
\]
Intuitively, the database stores all possible ways to colour the vertices of an edge in different colours. Note that we store every edge twice since we do not know in which order $u$ and $v$ occur in the edge. The role of the query is to check whether there is a coherent mapping from vertices to colours across all vertices.

Since $\dep$ is empty, the only repair of $D_G$ is $D_G$ itself. Therefore, $\mathsf{RF}(D_G,\dep,Q_G,())=1>0$ iff $D_G \models Q_G$. Hence, it remains to show that $G$ is 3-colorable iff there exists a homomorphism $h: \var{Q_G} \ra \adom{D_G}$. We proceed to show the latter equivalence:

\medskip
\noindent $(\Rightarrow)$
Assume that $G$ is 3-colorable. Thus, there is a mapping $\mu: V \ra \{1,2,3\}$ such that for every edge $\{u,v\}\in E$ we have that $\mu(u)\neq \mu(v)$. Now, consider the mapping $h: \var{Q_G} \ra \{1,2,3\}$ such that, for every $x_v\in \var{Q_G}$, $h(x_v)=\mu(v)$. We need to show that $h$ is indeed a homomorphism from $Q_G$ to $D_G$. 
Consider an arbitrary edge $\{u,v\}\in E$. We need to show that $\{C^{u,v}(h(x_u),h(x_v)), C^{v,u}(h(x_v),h(x_u))\}\subseteq D_G$, which will imply that $h$ is indeed a homomorphism from $Q_G$ to $D_G$. Since $\{u,v\}$ is an edge in $G$, we know that $\mu(u)\neq \mu(v)$; thus, $h(x_u)\neq h(x_v)$. Hence, $\{C^{u,v}(h(x_u),h(x_v)), C^{v,u}(h(x_v),h(x_u))\}\subseteq D_G$, as all possible combinations $h(x_u),h(x_v)\in\{1,2,3\}$ with $h(x_u)\neq h(x_v)$ are stored in the relations $C^{u,v}$ and $C^{v,u}$.

\medskip

\noindent $(\Leftarrow)$ Conversely, assume there is a homomorphism $h: \var{Q_G} \ra \adom{D_G}$. Consider the mapping $\mu: V \ra \{1,2,3\}$ such that, for every $v\in V$, $\mu(v)=h(x_v)$. We show that $G$ is 3-colorable with $\mu$ being the witnessing mapping. First, we note that $\mu$ is well-defined since there exists exactly one variable $x_v$ for every vertex $v\in V$ that is reused across the atoms representing different edges (and thus, making sure the mapping is consistent across different edges). Consider an edge $\{u,v\}\in E$. Since the variables $x_u$ and $x_v$ appear together in the query in the atom $C^{u,v}(x_u,x_v)$, we know that $h$ is such that $C^{u,v}(h(x_u),h(x_v))\in D_G$. By the construction of $D_G$, this implies that $h(x_u)\neq h(x_v)$, and thus, $\mu(u)\neq \mu(v)$, as needed.

\subsection{Item (2)}

\noindent A Boolean formula $\varphi$ in Pos2CNF is of the form $\bigwedge_{i\in[m]}C_i$, where $C_i=v_1^i\vee v_2^i$ with $v_1^i,v_2^i$, for each $i\in[m]$, being Boolean variables.
We write $\var{\varphi}$ for the set of variables occurring in $\varphi$ and $\sharp \varphi$ for the number of satisfying assignments of $\varphi$. We consider the following decision problem: 
\medskip

\begin{center}
	\fbox{\begin{tabular}{ll}
			{\small PROBLEM} : & $\sharp \mathsf{MON2SAT}$\\
			{\small INPUT} : & A formula $\varphi$ in Pos2CNF.\\
			{\small OUTPUT} : &  $\sharp \varphi$.
	\end{tabular}}
\end{center}

\medskip

\noindent It is known that $\sharp \mathsf{MON2SAT}$ does not admit an FPRAS, unless ${\rm NP}={\rm RP}$~\cite{+2001}. Fix an arbitrary $k>0$. We show via an approximation preserving reduction from $\sharp \mathsf{MON2SAT}$ that also $\ocqa^\ur[\ghw_k]$ does not admit an FPRAS, unless ${\rm NP}={\rm RP}$. 
Given a Pos2CNF formula $\varphi=\bigwedge_{i\in[m]}C_i$, let $\ins{S}_\varphi$ be the schema 
\[
\{C_i/2 \mid i\in[m]\}\cup\{\mathsf{Var}_v/1 \mid v\in\var{\varphi}\}\cup\{V/2, E/2\}
\]
with $(A,B)$ being the tuple of attributes of relation $V$.
We define the database $D_\varphi^k$ as
\begin{eqnarray*}
&&\{C_i(v_1^i, 1), C_i(v_2^i, 1) \mid i\in[m]\}\\
&\cup& \{\mathsf{Var}_v(v) \mid v\in\var{\varphi}\}\\
&\cup& \{V(v,0), V(v,1) \mid v\in\var{\varphi}\}\\
&\cup& \{E(i,j) \mid i,j\in [k+1],\ i<j\}.
\end{eqnarray*}
Let $\dep$ be the singleton set of keys consisting of
\[
\key{V} = \{1\}.
\]
We finally define the Boolean CQ $Q_\varphi^k$ as $\textrm{Ans}()\ \text{:-}\  \psi_1 \wedge \psi_2 \wedge \psi_3$, where 
\begin{align*}
\psi_1\ &=\ \bigwedge_{i\in[m]}C_i(x_i,y_i)\wedge V(x_i,y_i),\\
\psi_2\ &=\ \bigwedge_{v\in\var{\varphi}}\mathsf{Var}_v(z_v)\wedge V(z_v,\_),\\
\psi_3\ &=\ \bigwedge_{\substack{i,j\in [k+1]\\i<j}} E(w_i,w_j).
\end{align*}
We use $\_$ to denote a variable that occurs only once. Note that $\psi_1$ and $\psi_2$ are acyclic subqueries, whereas $\psi_3$ encodes a clique of size $k+1$ and has generalized hypertreewidth $k$.
Since $\psi_1$, $\psi_2$, and $\psi_3$ do not share variable, we can conclude that  $Q_\varphi^k$ has generalized hypertreewidth $k$.
Note that the sole purpose of the relation $E$ and the subquery $\psi_3$ is to force the generalized hypertreewidth of $Q_\varphi^k$ to be $k$. It is clear that since this part of $Q_\varphi^k$ will be satisfied in every repair of $D_\varphi^k$. 
Intuitively, $\psi_1$ ensures that the repair satisfying the query entails an assignment to the variables such that $\varphi$ is satisfied, $\psi_2$ ensures that a repair entailing the query encodes some variable assignment for every variable $v$ occurring in $\varphi$, and $\psi_3$, as said above, ensures that $Q_\varphi^k$ has generalized hypertreewidth $k$.

Let $n=|\var{\varphi}|$. It is clear that there are $3^n$ operational repairs of $D_\varphi^k$ since for every $v\in\var{\varphi}$ we can either keep $V(v,0)$ and remove $V(v,1)$ in the repair, or we keep $V(v,1)$ and remove $V(v,0)$, or we remove both facts in the repair. Hence, $|\opr{D_\varphi^k}{\dep}|=3^n$.

We proceed to show that $|\{D \in \opr{D_\varphi^k}{\dep} \mid D \models Q_\varphi^k\}|=\sharp\varphi$. 
%
%
To this end, let $\mathsf{sat}(\varphi)$ be the set of satisfying assignments from the variables of $\varphi$ to $\{0,1\}$. We will show that there exists a bijection from the set $\mathsf{sat}(\varphi)$ to the set $\{D \in \opr{D_\varphi^k}{\dep} \mid D \models Q_\varphi^k\}$. We define the mapping $\mu$ as follows: for each $h \in \mathsf{sat}(\varphi)$,
\begin{align*}
\mu(h)\ =\ &\{V(v,0) \mid v \in \var{\varphi} \text{ and } h(v) = 0\}\\
&\cup\ \{V(v,1) \mid v \in \var{\varphi} \text{ and } h(v) = 1\}\\
&\cup\ \{C_i(v_1^i, 1), C_i(v_2^i, 1) \mid i\in[m]\}\\
&\cup\ \{\mathsf{Var}_v(v) \mid v\in\var{\varphi}\}\\
&\cup\ \{E(i,j) \mid i,j\in [k+1],\ i<j\}.
\end{align*}
We proceed to show the following three statements:
\begin{enumerate}
	\item $\mu$ is correct, that is, it is indeed a function from $\mathsf{sat}(\varphi)$ to $\{D \in \opr{D_\varphi^k}{\dep} \mid D \models Q_\varphi^k\}$.
	\item $\mu$ is injective.
	\item $\mu$ is surjective.
\end{enumerate}

\medskip 

\noindent
\paragraph{The mapping $\mu$ is correct.} Consider an arbitrary satisfying assignment $h\in\mathsf{sat}(\varphi)$. We need to show that there exists a $(D_\varphi^k,\dep)$-repairing sequence $s_h$ such that $\mu(h) = s_h(D_\varphi^k)$, $s_h(D_\varphi^k) \models \dep$ (i.e., $s_h$ is complete), and $s_h(D_\varphi^k) \models Q_\varphi^k$ (i.e., the resulting repair entails the query). Let $\var{\varphi} = \{v_1,\ldots,v_n\}$. Consider the sequence $s_h = \op_1,\ldots,\op_n$ such that, for every $i \in [n]$:
\[
\op_{i}\ =\
\begin{cases}
-V(v_i,1) & \text{if } h(v_i) = 0 \\
-V(v_i,0) & \text{if } h(v_i) = 1.
\end{cases}
\]
In simple words, the homomorphism $h$ guides the repairing process, i.e., $h(v_i) = 0$ (resp., $h(v_i)=1$) implies $V(v_i,0)$ (resp., $V(v_i,1)$) should be kept in the repair. It is easy to verify that $s_h$ is indeed a $(D_\varphi^k,\dep)$-repairing sequence, such that $\mu(h) = s_h(D_\varphi^k)$ and $s_h(D_\varphi^k) \models \dep$.
%
To show that $s_h(D_\varphi^k) \models Q_\varphi^k$ it suffices to see that each of the subqueries $\psi_1$, $\psi_2$, and $\psi_3$ is satisfied in $s_h(D_\varphi^k)$. Since $h$ is a satisfying assignment of the variables of $\varphi$, we know that for every clause $C_i$, where $i\in [m]$, one of its variables $v_1^i$ or $v_2^i$ evaluates to $1$ under $h$, i.e.~$h(v_l^i) = 1$ for some $l\in \{1, 2\}$. Hence, for every $i\in [m]$ and (at least) one $l\in \{1, 2\}$, we have that $C_i(v_l^i, 1)$ and $V(v_l^i, 1)$ remains in the repair $s_h(D_\varphi^k)$ and thus $\psi_1$ is satisfied. To see that $\psi_2$ is satisfied as well it suffices to see that for every variable $v\in\var{\varphi}$ either $V(v,0)$ or $V(v,1)$ is kept in the repair and the relation $\mathsf{Var}_v$ remains untouched by the repairing sequence. Since the relation $E$ remains also unchanged in every repair, it is clear that $\psi_3$ is also satisfied. We can therefore conclude that $s_h(D_\varphi^k) \models Q_\varphi^k$.

\medskip 

\noindent
\paragraph{The mapping $\mu$ is injective.} Assume that there are two distinct satisfying assignments $h,h'\in\mathsf{sat}(\varphi)$ such that $\mu(h) = \mu(h')$. By the definition of $\mu$, we get that $h(v) = h'(v)$, for every $v \in \var{\varphi}$. But this contradicts the fact that $h$ and $h'$ are different satisfying assignments of the variables in $\varphi$. Therefore, for every two distinct satisfying assignments $h,h' \in \mathsf{sat}(\varphi)$, $\mu(h) \neq \mu(h')$, as needed.

\medskip

\noindent
\paragraph{The mapping $\mu$ is surjective.} Consider an arbitrary operational repair $D \in \opr{D_\varphi^k}{\dep}$ such that $D \models Q_\varphi^k\}$. We need to show that there exists some $h \in \mathsf{sat}(\varphi)$ such that $\mu(h) = D$. Note that, since $D \models Q_\varphi^k$, necessarily the subquery $\psi_2$ must also be satisfied in $D$. In other words, for every variable $v\in\var{\varphi}$ either $V(v,0)\in D$ or $V(v,1)\in D$ (and never both since $D \models \dep$). We define the mapping $h_D : \var{\varphi} \ra \{0,1\}$ as follows: for every $v \in \var{\varphi}$:
\[
h_D(v)\ =\
\begin{cases}
1 & \text{if } V(v,1) \in D \\
0 & \text{if } V(v,0) \in D.
\end{cases}
\]
It is clear that $h_D$ is well-defined: for every $v \in \var{\varphi}$, $h_D(v) = x$ and $h_D(v) = y$ implies $x=y$. It is also clear that $\mu(h_D) = D$. The fact that $h_D \in \mathsf{sat}(\varphi)$ follows directly from the fact that $D$ entails the subquery $\psi_1$. Since, for every $i\in [m]$ and some $l\in\{1,2\}$, $V(v_l^i, 1)\in D$ and thus $h_D(v_l^i) = 1$, for every $C_i$, where $i\in [m]$, one of its literals evaluates to true under $h_D$ and thus $h_D\in\mathsf{sat}(\varphi)$. 

\medskip

Summing up, recall that the repair relative frequency is the ratio
\[
\mathsf{RF}(D_\varphi^k,\dep,Q_\varphi^k,()) = \frac{|\{D \in \opr{D_\varphi^k}{\dep} \mid D \models Q_\varphi^k\}|}{|\opr{D_\varphi^k}{\dep}|} = \frac{\sharp\varphi}{3^n}.
\]
Hence, assuming an FPRAS $A$ for the problem $\ocqa^\ur[\ghw_k]$, we can build an FPRAS $A'$ for $\sharp \mathsf{MON2SAT}$ by simply running $A$ on the input $(D_\varphi^k, \dep, Q_\varphi^k, \varepsilon, \delta)$ and returning its result multiplied by $3^n$. We note that $A'$ is an FPRAS for $\sharp \mathsf{MON2SAT}$ with the same probabilistic guarantees $\varepsilon$ and $\delta$, contradicting the assumption ${\rm NP}\neq{\rm RP}$. Thus, $\ocqa^\ur[\ghw_k]$ does not admit an FPRAS, unless ${\rm NP}={\rm RP}$.

\section{Proof of Proposition~\ref{pro:notinspanl}}

We show that:
\begin{manualproposition}{\ref{pro:notinspanl}}
	\pronotinspanl
\end{manualproposition}

	We show that $\sharp\mathsf{Repairs}[k] \in \spanl$, for some $k > 0$, implies $\logcfl \subseteq \nlogspace$; the inclusion $\nlogspace \subseteq \logcfl$ is well-known. For this, we need some auxiliary notation.
	
	For a function $f : \Lambda^* \ra \mathbb{N}$, over some alphabet $\Lambda$, we use $f_{>0}$ to denote the decision problem associated to $f$, i.e., the language of all strings $w \in \Lambda^*$ such that $f(w) > 0$. In other words, $f_{>0}$ is the problem of checking whether the count $f(w)$ is non-zero. Hence, for $k > 0$, $\sharp\mathsf{Repairs}[k]_{>0}$ is the decision problem that, given a database $D$, a set $\dep$ of primary keys, a query $Q(\bar x)$ from $\sjf$, a generalized hypertree decomposition of $Q$ of width $k$, and a tuple $\bar c \in \adom{D}^{|\bar x|}$, checks whether there exists at least one operational repair $D' \in \opr{D}{\dep}$ such that $\bar c \in Q(D')$.
	
	We can now proceed with the proof.
	Assume there is $k > 0$ such that $\sharp\mathsf{Repairs}[k] \in \spanl$.
	By definition of $\spanl$, for every $f \in \spanl$, $f_{>0}$ is in $\nlogspace$. Thus, $\sharp\mathsf{Repairs}[k]_{>0} \in \nlogspace$. However, $\sharp\mathsf{Repairs}[k]_{>0}$ is $\logcfl\hard$. The latter follows from the fact that when the input database $D$ is consistent w.r.t.\ $\dep$, then $\sharp\mathsf{Repairs}[k]_{>0}$ coincides with the problem of evaluating a conjunctive query, when a generalized hypertree of $Q$ of width $k$ is given, which is known to be $\logcfl\hard$ \cite{GoLS02}. The fact that $\sharp\mathsf{Repairs}[k]_{>0}$ is $\logcfl\hard$ means that every language in $\logcfl$ reduces in logspace to $\sharp\mathsf{Repairs}[k]_{>0}$. This, together with the fact that $\sharp\mathsf{Repairs}[k]_{>0} \in \nlogspace$, and that $\nlogspace$ is closed under logspace reductions, implies that every language in $\logcfl$ belongs to $\nlogspace$.

\section{Proof of Proposition~\ref{pro:spanl-in-spantl}}

We proceed to show that:
\begin{manualproposition}{\ref{pro:spanl-in-spantl}}
	\prospanlinspantl
\end{manualproposition}

The inclusion $\spanl \subseteq \spantl$ is an easy consequence of how $\spantl$ is defined. That is, consider a non-deterministic, logspace Turing machine $M$ with output. We can convert $M$ to an ATO $M'$ as follows. Whenever $M$ writes a symbol $\alpha$ in the output tape, it instead writes $\alpha$ to the labeling tape, then moves to a labeling state, and then moves to a non-labeling configuration. Hence, each accepted output $\alpha_1,\ldots,\alpha_n$ of $M$ with input some string $w$ corresponds to a node-labeled tree of the form $v_0 \ra v_1 \ra \cdots \ra v_n$, with $v_0$ labeled with the empty string $\epsilon$, and $v_i$ labeled with $\alpha_i$, for $i \in [n]$. Hence, for every input string $w$, the number of accepted outputs of $M$ on input $w$ coincides with the number of accepted outputs of $M'$ on input $w$. The fact that $M'$ is well-behaved follows from the fact that $M$ is a logspace machine, and from the fact that $M'$ uses only existential states, and at most one symbol at the time is stored on the labeling tape.

For the second part of the claim, we show that $\spanl = \spantl$ implies $\logcfl \subseteq \nlogspace$; the inclusion $\nlogspace \subseteq \logcfl$ is well-known. As done in the proof of Proposition~\ref{pro:notinspanl}, for a function $f : \Lambda^* \ra \mathbb{N}$, over some alphabet $\Lambda$, we use $f_{>0}$ to denote the decision problem associated to $f$, i.e., the language of all strings $w \in \Lambda^*$ such that $f(w) > 0$.
Assume $\spanl = \spantl$. From Proposition~\ref{pro:other-prob-spantl} (which we prove at the end of the appendix), the problem $\sharp \mathsf{GHWCQ}[k]$ is in $\spantl$, for every $k > 0$. By definition of $\spanl$, for every $f \in \spanl$, $f_{>0}$ is in $\nlogspace$. Thus, $\sharp \mathsf{GHWCQ}[k]_{>0} \in \nlogspace$, for every $k > 0$. However, for every $k > 0$, $\sharp \mathsf{GHWCQ}[k]_{>0}$ is the decision problem checking, given a database $D$, a CQ $Q(\bar x)$, and  generalized hypertree decomposition of $Q$ of width $k$, whether there exists a tuple $\bar c \in \adom{D}^{|\bar x|}$ such that $\bar c \in Q(D)$, which is known to be $\logcfl\hard$~\cite{GoLS02}. The fact that $\sharp \mathsf{GHWCQ}[k]_{>0}$, for every $k > 0$, is $\logcfl\hard$ means that every language in $\logcfl$ reduces in logspace to $\sharp \mathsf{GHWCQ}[k]_{>0}$, for every $k > 0$. This, together with the fact that $\sharp \mathsf{GHWCQ}[k]_{>0} \in \nlogspace$, for every $k > 0$, and that $\nlogspace$ is closed under logspace reductions, implies that every language in $\logcfl$ belongs to $\nlogspace$.

\section{Proof of Proposition~\ref{pro:logspace-closure}}

We proceed to show that:

\begin{manualproposition}{\ref{pro:logspace-closure}}
	\prologspaceclosure
\end{manualproposition}

Consider two functions $f : \Lambda_1^* \rightarrow \mathbb{N}$ and $g : \Lambda_2^* \rightarrow \mathbb{N}$, for some alphabets $\Lambda_1$ and $\Lambda_2$, with $g \in \spantl$, and assume there is a logspace computable function $h : \Lambda_1^* \rightarrow \Lambda_2^*$ with $f(w) = g(h(w))$ for all $w \in \Lambda_1^*$. We prove that $f \in \spantl$. The proof employees a standard argument used to prove that logspace computable functions can be composed.
It is well known that $h$ can be converted to a logspace computable function $h' : \Lambda_1^* \times \mathbb{N} \rightarrow \Lambda_2$ that, given a string $w \in \Lambda_1^*$ and an integer $i \in [|h(w)|]$, where $|h(w)|$ denotes the length of the string $h(w)$, outputs the $i$-th symbol of $h(w)$, and $h'$ is computable in logspace (e.g., see~\cite{ArBa09}). This can be achieved by modifying $h$ so that it first initializes a counter $k$ to $1$, and then, whenever it needs to write a symbol on the output tape, it first checks whether $k = i$; if not, then it does not write anything on the output tape, and increases $k$, otherwise, it writes the required symbol, and then halts. The counter $k$ can be stored in logarithmic space since $h$ works in logarithmic space, and thus $h(w)$ contains no more than polynomially many symbols, i.e., $|h(w)| \in O(|w|^c)$, for some constant $c$. 

With the above in place, proving that $f \in \spantl$ is straightforward. Let $M$ be the well-behaved ATO such that $g = \mathsf{span}_M$. We modify $M$ to another ATO $M'$ such that $f = \mathsf{span}_{M'}$; again, this construction is rather standard, but we give it here for the sake of completeness. In particular, we modify $M$ so that, with input a string $w$, it keeps a counter $k$ initialized to $1$, and whenever $M$ needs to read the current symbol from the input tape, it executes $h'$ with input $(w,k)$, where its (single) output symbol $\alpha$ is written on the working tape; then $M'$ uses $\alpha$ as the input symbol. Then, when the input tape head of $M$ moves to the right (resp., left, stays), $M'$ increases (resp., decreases, leaves untouched) the counter $k$. The fact that $\mathsf{span}_{M'} = g(h(w))$ follows by construction of $M'$ and the fact that when executing $h'$, no symbols are written on the labeling tape, and no labeling states of $M'$ are visited. Hence, $\mathsf{span}_{M'} = f(w)$. To show that $f \in \spantl$, it remains to to prove that $M'$ is well-behaved.

Obviously, since $h'$ is computable in logspace, $M'$ can execute $h'$ using no more than logarithmic space w.r.t.~$|w|$ on its working tape. Moreover, while executing $h'$, $M'$ does not need to write anything on the labeling tape, and thus, $M'$ uses no more space than the one that $M$ uses on the labeling tape. Furthermore, storing $k$ requires logarithmic space since $k \in [|h(w)|]$ can be encoded in binary. 
Since $h'$ is computable in (deterministic) logspace, an execution of $h'$ corresponds to a sequence of polynomially many (non-labeling) existential configurations in a computation of $M'$ with input $w$. Hence, each computation $T'$ of $M'$ corresponds to a computation $T$ of $M$, where, in the worst case, each node of $T$ becomes a polynomially long path in $T'$; hence, $T'$ is of polynomial size w.r.t.\ $|w|$. Finally, since, as already discussed, each execution of $h'$ coincides with a sequence of non-labeling \emph{existential} configurations, the maximum number of non-labeling universal configurations on a labeled-free path of any computation of $M'$ is the same as the one for $M$, and thus remains a constant. Consequently, $M'$ is well-behaved. 
\section{Proof of Theorem~\ref{the:spantl-fpras}}

We proceed to prove the following result:

\begin{manualtheorem}{\ref{the:spantl-fpras}}
\thespantlfpras
\end{manualtheorem}

As discussed in the main body, to establish the above result we exploit a recent result showing that the problem of counting the number of trees of a certain size accepted by a non-deterministic finite tree automaton admits an FPRAS~\cite{ACJR21}.
In particular, we reduce in polynomial time each function in $\mathsf{SpanTL}$ to the above counting problem for tree automata. In what follows, we recall the basics about non-deterministic finite tree automata and the associated counting problem, and then give our reduction.

\medskip

\noindent \paragraph{Ordered Trees and Tree Automata.} For an integer $k \geq 1$, a {\em finite ordered $k$-tree} (or simply {\em $k$-tree}) is a prefix-closed non-empty finite subset $T$ of $[k]^*$, that is, if $w \cdot i \in T$ with $w \in [k]^*$ and $i \in [k]$, then $w \cdot j \in T$ for every $w \in [i]$.
The root of $T$ is the empty string, and every maximal element of $T$ (under prefix ordered) is a leaf. For every $u,v \in T$, we say that $u$ is a child of $v$, or $v$ is a parent of $u$, if $u = v \cdot i$ for some $i \in [k]$. The size of $T$ is $|T|$. Given a finite alphabet $\Lambda$, let $\trees{k}{\Lambda}$ be the set of all $k$-trees in which each node is labeled with a symbol from $\Lambda$. By abuse of notation, for $T \in \trees{k}{\Lambda}$ and $u \in T$, we write $T(u)$ for the label of $u$ in $T$.

A {\em (top-down) non-deterministic finite tree automation} (NFTA) over $\trees{k}{\Lambda}$ is a tuple $A = (S,\Lambda,s_\text{\rm init},\delta)$, where $S$ is the finite set of states of $A$, $\Lambda$ is a finite set of symbols (the alphabet of $A$), $s_\text{\rm init} \in S$ is the initial state of $A$, and $\delta \subseteq S \times \Lambda \times \left(\bigcup_{i = 0}^{k} S^k\right)$ is the transition relation of $A$.
A {\em run} of $A$ over a tree $T \in \trees{k}{\Lambda}$ is a function $\rho : T \ra S$ such that, for every $u \in T$, if $u \cdot 1,\ldots,u \cdot n$ are the children of $u$ in $T$, then $(\rho(u),T(u),(\rho(u \cdot 1),\ldots,\rho(u \cdot n))) \in \delta$. In particular, if $u$ is a leaf, then $(\rho(u),T(u),()) \in \delta$. We say that $A$ {\em accepts} $T$ if there is a run $\rho$ of $A$ over $T$ with $\rho(\epsilon) = s_\text{\rm init}$, i.e., $\rho$ assigns to the root the initial state. We write $L(A) \subseteq \trees{k}{\Lambda}$ for the set of all trees accepted by $A$, i.e., the language of $A$. We further write $L_n(A)$ for the set of trees $\{T \in L(A) \mid |T| = n\}$, i.e., the set of trees of size $n$ accepted by $A$. The relevant counting problem for NFTA follows:

\medskip

\begin{center}
    \fbox{\begin{tabular}{ll}
            {\small PROBLEM} : & $\sharp\mathsf{NFTA}$
            \\
            {\small INPUT} : & An NFTA $A$  and a string $0^n$ for some $n \geq 0$.
            \\
            {\small OUTPUT} : &  $\left|\bigcup_{i = 0}^n L_i(A)\right|$.
    \end{tabular}}
\end{center}

\medskip

The notion of FPRAS for $\sharp\mathsf{NFTA}$ is defined in the obvious way. 
We know from~\cite{ACJR21} that $\sharp\mathsf{NFTA}_=$, defined as $\sharp\mathsf{NFTA}$ with the difference that it asks for $|L_n(A)|$, i.e., the number trees of size $n$ accepted by $A$, admits an FPRAS. By using this result, we can easily show that:

\def\thenfta{
    $\sharp\mathsf{NFTA}$ admits an FPRAS.
}

\begin{theorem}\label{the:nfta}
    \thenfta
\end{theorem}
\begin{proof}
    The goal is to devise a randomized algorithm $\mathsf{A}$ that takes as input an NFTA $A$, a string $0^n$, with $n \geq 0$, $\epsilon > 0$ and $0 < \delta < 1$, runs in polynomial time in $||A||$, $n$, $1/\epsilon$ and $\log(1/\delta)$, and produces a random variable $\mathsf{A}(A,0^n,\epsilon,\delta)$ such that
\[
\text{\rm Pr}\left(\left|\mathsf{A}(A,0^n,\epsilon,\delta) - L_{0}^{n} \right|\ \leq\ \epsilon \cdot L_{0}^{n}\right)\ \geq\
1-\delta,
\]
where  $L_{0}^{n} = \left|\bigcup_{i = 0}^n L_i(A)\right|$, or equivalently
\[
\text{\rm Pr}\left((1-\epsilon) \cdot L_{0}^{n}\ \leq  \mathsf{A}(A,0^n,\epsilon,\delta)\ \leq\ (1 + \epsilon) \cdot L_{0}^{n}\right)\ \geq\
1-\delta.
\]
We know from~\cite{ACJR21} that the problem $\sharp\mathsf{NFTA}_=$ admits an FPRAS. This means that there is a randomized algorithm $B$  that takes as input an NFTA $A$, a string $0^n$, with $n \ge 0$, $\epsilon > 0$ and $0 < \delta < 1$, runs in polynomial time in $||A||$, $n$, $1/\epsilon$ and $\log(1/\delta)$, and produces a random variable $\mathsf{B}(A,0^n,\epsilon,\delta)$ such that
\[
\text{\rm Pr}\left((1-\epsilon) \cdot |L_n(A)|\ \leq  \mathsf{B}(A,0^n,\epsilon,\delta)\ \leq\ (1 + \epsilon) \cdot |L_n(A)|\right)\ \geq\
1-\delta.
\]
The desired randomized algorithm $\mathsf{A}$ is defined in such a way that, on input $A$, $0^n$, $\epsilon$ and $\delta$, returns the number
\[
\sum_{i = 0}^{n} \mathsf{B}\left(A,0^i,\epsilon,\frac{\delta}{2(n+1)}\right).
\]
Note that the $n+1$ runs of $\mathsf{B}$ in the expression above are all independent,
and thus, the random variables $\mathsf{B}\left(A,0^i,\epsilon,\frac{\delta}{2(n+1)}\right)$ and $\mathsf{B}\left(A,0^j,\epsilon,\frac{\delta}{2(n+1)}\right)$ are independent, for $i,j \in \{0,\ldots,n\}$, with $i \neq j$;
it is also clear that 
\[
\sum_{i=0}^{n} \left|L_i(A)\right|\ =\ \left|\bigcup_{i=0}^{n} L_i(A)\right|
\]
since $L_i(A) \cap L_j(A) = \emptyset$. Therefore, we get that
$$
\text{\rm Pr}\bigg(\left(1-\epsilon\right) \cdot L_{0}^{n}\ \leq  \mathsf{A}(A,0^n,\epsilon,\delta)\ \leq\ \left(1 + \epsilon\right) \cdot L_{0}^{n}\bigg)\ \geq 
\left(1-\frac{\delta}{2(n+1)}\right)^{n+1}.
$$
We know that the following inequality hold (see, e.g.,~\cite{DBLP:journals/tcs/JerrumVV86}) for $0 \leq x \leq 1$ and $m \geq 1$,
\[
1-x\ \leq\ \left(1 - \frac{x}{2m}\right)^{m}.
\]
Consequently,
\[
\text{\rm Pr}\left((1-\epsilon) \cdot L_{0}^{n}\ \leq  \mathsf{A}(A,0^n,\epsilon,\delta)\ \leq\ (1 + \epsilon) \cdot L_{0}^{n}\right)\ \geq\
1-\delta,
\]
as needed. It is also clear that $\mathsf{A}$ runs in polynomial time in
$||A||$, $n$, $1/\epsilon$ and $\log(1/\delta)$, since $\mathsf{B}$ runs in polynomial time in $||A||$, $i$, $1/\epsilon$ and $\log(1/\delta)$, for each $i \in \{0,\ldots,n\}$, and the claim follows.
\end{proof}

\noindent\paragraph{The Reduction.} Recall that our goal is reduce in polynomial time every function of $\mathsf{SpanTL}$ to $\sharp\mathsf{NFTA}$, which, together with Theorem~\ref{the:nfta}, will immediately imply Theorem~\ref{the:spantl-fpras}. In particular, we need to establish the following technical result:

\def\proreductiontonfta{
Fix a function $f : \Lambda^* \ra \mathbb{N}$ of $\mathsf{SpanTL}$. For every $w \in \Lambda^*$, we can construct in polynomial time in $|w|$ an NFTA $A$ and a string $0^n$, for some $n \geq 0$, such that $f(w) = \left|\bigcup_{i = 0}^n L_i(A)\right|$.
}

\begin{proposition}\label{pro:reduction}
    \proreductiontonfta
\end{proposition}

We proceed to show the above proposition. Since, by hypothesis, $f : \Lambda^* \ra \mathbb{N}$ belongs to $\mathsf{SpanTL}$, there exists a well-behaved ATO $M = (S,\Lambda',s_\text{\rm init},s_\text{\rm accept},s_\text{\rm reject},S_{\exists},S_{\forall},S_{L},\delta)$, where $\Lambda' = \Lambda \cup \{\bot,\triangleright\}$ and $\bot,\triangleright \not\in \Lambda$, such that $f$ is the function $\mathsf{span}_M$. 
The goal is, for an arbitrary string $w \in \Lambda^*$, to construct in polynomial time in $|w|$ an NFTA $A = (S^A,\Lambda^A,s_{\text{\rm init}}^{A}\delta^A)$ and a string $0^n$, for some $n \geq 0$, such that $\mathsf{span}_M(w) = |\bigcup_{i = 0}^{n} L_i(A)|$. This is done in two steps:
\begin{enumerate}
    \item We first construct an NFTA $A$ in polynomial time in $|w|$ such that $\mathsf{span}_M (w) = |L(A)|$, i.e., $A$ accepts $\mathsf{span}_M(w)$ trees.
    
    \item We show that $L(A) = \bigcup_{i = 0}^{\mathsf{pol}(|w|)} L_i(A)$ for some polynomial function $\mathsf{pol} : \mathbb{N} \ra \mathbb{N}$. 
\end{enumerate}
After completing the above two steps, it is clear that Proposition~\ref{pro:reduction} follows with $n = \mathsf{pol}(|w|)$. Let us now discuss the above two steps.

\medskip

\noindent\paragraph{\underline{Step 1: The NFTA}}

\smallskip

\noindent For the construction of the desired NFTA, we first need to introduce the auxiliary notion of the computation directed acyclic graph (DAG) of the ATO $M$ on input $w \in \Lambda^*$, which compactly represents all the computations of $M$ on $w$. Recall that a DAG $G$ is rooted if it has exactly one node, the root, with no incoming edges. We also say that a node of $G$ is a leaf if it has no outgoing edges.

\begin{definition}[\textbf{Computation DAG}]\label{def:computation_dag}
    The {\em computation DAG} of $M$ on input $w \in \Lambda^*$ is the DAG $G = (\mathcal{C},\mathcal{E})$, where $\mathcal{C}$ is the set of configurations of $M$ on $w$, defined as follows:
    \begin{enumerate}
        \item[-] if $C \in \mathcal{C}$ is the root node of $G$, then $C$ is the initial configuration of $M$ on input $w$,
        
        \item[-] if $C \in \mathcal{C}$ is a leaf node of $G$, then $C$ is either an accepting or a rejecting configuration of $M$, and
        
        \item[-] if $C \in \mathcal{C}$ is a non-leaf node of $G$ with $C \ra_M \{C_1,\ldots,C_n\}$ for $n > 0$, then (i) for each $i \in [n]$, $(C,C_i) \in \mathcal{E}$, and (ii) for every configuration $C' \not\in \{C_1,\ldots,C_n\}$ of $M$ on $w$, $(C,C') \not\in \mathcal{E}$. \hfill\markfull 
    \end{enumerate}
\end{definition}

As said above, the computation DAG $G$ of $M$ on input $w$ compactly represents all the computations of $M$ on $w$. In particular, a computation $(V,E,\lambda)$ of $M$ on $w$ can be constructed from $G$ by traversing $G$ from the root to the leaves and (i) for every universal configuration $C$ with outgoing edges $(C,C_1),\dots,(C,C_n)$, add a node $v$ with $\lambda(v)=C$ and children $u_1,\dots,u_n$ with $\lambda(u_i)=C_i$ for all $i\in[n]$, and (ii) for every existential configuration $C$ with outgoing edges $(C,C_1),\dots,(C,C_n)$, add a node $v$ with $\lambda(v)=C$ and a single child $u$ with $\lambda(u)=C_i$ for some $i\in [n]$.

\OMIT{
proof of Proposition~\ref{pro:reduction}. Since, by hypothesis, $f : \Lambda^* \ra \mathbb{N}$ belongs to $\mathsf{SpanTL}$, there exists a well-behaved ATO $M = (S,\Lambda',s_\text{\rm init},s_\text{\rm accept},s_\text{\rm reject},S_{\exists},S_{\forall},S_{L},\delta)$, where $\Lambda' = \Lambda \cup \{\bot,\triangleright\}$ and $\bot,\triangleright \not\in \Lambda$, such that $f$ is the function $\mathsf{span}_M$. 
The goal is, for an arbitrary string $w \in \Lambda^*$, to construct in polynomial time in $|w|$ an NFTA $A = (S^A,\Lambda^A,s_{\text{\rm init}}^{A}\delta^A)$ and a string $0^n$ from some $n \geq 0$ such that $\mathsf{span}_M(w) = |\bigcup_{i = 0}^{n} L_i(A)|$. This is done in two steps:
\begin{enumerate}
    \item We first construct an NFTA $A$ in polynomial time in $|w|$ such that $\mathsf{span}_M (w) = |L(A)|$, i.e., $A$ accepts $\mathsf{span}_M(w)$ trees.
    
    \item We define a polynomial function $\mathsf{pol} : \mathbb{N} \ra \mathbb{N}$ such that $L(A) = |\bigcup_{i = 0}^{\mathsf{pol}(|w|)} L_i(A)|$. 
\end{enumerate}
After completing the above two steps, it is clear that Proposition~\ref{pro:reduction} follows with $n = \mathsf{pol}(|w|)$. We proceed to discuss those two steps.

\medskip

\noindent\paragraph{\underline{Step 1: The NFTA}}

\smallskip

%
}

The construction of the NFTA $A$ is performed by $\mathsf{BuildNFTA}$, depicted in Algorithm~\ref{alg:dagtonfta}, which takes as input a string $w \in \Lambda^*$. It first constructs the computation DAG $G = (\mathcal{C},\mathcal{E})$ of $M$ on $w$, which will guide the construction of $A$. It then initializes the sets $S^A,\Lambda^A,\delta^A$, as well as the auxiliary set $Q$, which will collect pairs of the form $(C,U)$, where $C$ is a configuration of $\mathcal{C}$ and $U$ a set of tuples of states of $A$, to empty.
Then, it calls the recursive procedure $\mathsf{Process}$, depicted in Algorithm~\ref{alg:process}, which constructs the set of states $S^A$, the alphabet $\Lambda^A$, and the transition relation $\delta^A$, while traversing the computation DAG $G$ from the root to the leaves. Here, $S^A$, $\Lambda^A$, $\delta^A$, $G$, and $Q$ should be seen as global structures that can be used and updated inside the procedure $\mathsf{Process}$. Eventually, $\mathsf{Process}(\rt{G})$ returns a state $s \in S^A$, which acts as the initial state of $A$, and $\mathsf{BuildNFTA}(w)$ returns the NFTA $(S^A,\Lambda^A,s,\delta^A)$.

Concerning the procedure $\mathsf{Process}$, when we process a labeling configuration $C \in \mathcal{C}$, we add to $S^A$ a state $s_C$ representing $C$. 
Then, for every computation $T=(V,E,\lambda)$ of $M$ on $w$, if $V$ has a node $v$ with $u_1,\dots,u_n$ being the nodes reachable from $v$ via a labeled-free path, and $\lambda(u_i)$ is a labeling configuration for every $i\in[n]$, we add the transition $(s_C,z,(s_{C_1},\dots,s_{C_n}))$ to $\delta^A$, where $\lambda(v)=C$, $C$ is of the form $(\cdot,\cdot,\cdot,z,\cdot,\cdot)$, and $\lambda(u_i)=C_i$ for every $i\in[n]$, for some arbitrary order $C_1,\dots,C_n$ over those configurations.
Since, by definition, the output of $T$ has a node corresponding to $v$ with children corresponding to $u_1,\dots,u_n$, using these transitions we ensure that there is a one-to-one correspondence between the outputs of $M$ on $w$ and the trees accepted by $A$.

Now, when processing non-labeling configurations, we accumulate all the information needed to add all these transitions to $\delta^A$. In particular, the procedure $\process$ always returns a set of tuples of states of $S^A$. For labeling configurations $C$, it returns a single tuple $(s_C)$, but for non-labeling configurations $C$, the returned set depends on the outgoing edges $(C,C_1),\dots,(C,C_n)$ of $C$. In particular, if $C$ is an existential configuration, it simply takes the union $P$ of the sets $P_i$ returned by $\process(C_i)$, for $i\in[n]$, because in every computation of $M$ on $w$ we choose only one child of each node associated with an existential configuration; each $P_i$ represents one such choice. If, on the other hand, $C$ is a universal configuration, then $P$ is the set $\bigotimes_{i \in [n]} P_i$ obtained by first computing the cartesian product $P' = \bigtimes_{i \in [n]} P_i$ and then merging each tuple of $P'$ into a single tuple of states of $S^A$. For example, with $P_1 = \{(),(s_1,s_2),(s_3)\}$ and $P_2 = \{(s_5),(s_6,s_7)\}$, $P_1 \otimes P_2$ is the set $\{(s_5),(s_6,s_7),(s_1,s_2,s_5),(s_1,s_2,s_6,s_7),(s_3,s_5),(s_3,s_6,s_7)\}$. We define $\bigotimes_{i \in [n]} P_i=\emptyset$ if $P_i=\emptyset$ for some $i\in[n]$.
We use the $\otimes$ operator since in every computation of $M$ on $w$ we choose all the children of each node associated with a universal configuration. When we reach a labeling configuration $C \in \mathcal{C}$, we add transitions to $\delta^A$ based on this accumulated information. In particular, we add a transition from $s_C$ to every tuple $(s_1,\dots,s_\ell)$ of states in $P$.

Concerning the running time of $\mathsf{BuildNFTA}(w)$, we first observe that the size (number of nodes) of the computation DAG $G$ of $M$ on input $w$ is polynomial in $|w|$. This holds since for each configuration $(\cdot,\cdot,y,z,\cdot,\cdot,\cdot)$ of $M$ on $w$, we have that $|y|,|z|\in O(\log(|w|))$. Moreover, we can construct $G$ in polynomial time in $|w|$ by first adding a node for the initial configuration of $M$ on $w$, and then following the transition function to add the remaining configurations of $M$ on $w$ and the
outgoing edges from each configuration.
Now, in the procedure $\process$ we use the auxiliary set $Q$ to ensure that we process each node of $G$ only once; thus, the number of calls to the $\process$ procedure is polynomial in the size of $|w|$. Moreover, in $\process(C)$, where $C$ is a non-labeling universal configuration that has $n$ outgoing edges $(C,C_1),\dots,(C,C_n)$, the size of the set $P$ is $|P_1| \times \dots \times |P_n|$, where $P_i=\process(C_i)$ for $i\in[n]$. When we process a non-labeling existential configuration $C$ that has $n$ outgoing edges $(C,C_1),\dots,(C,C_n)$, the size of the set $P$ is $|P_1|+\dots+ |P_n|$. We also have that $|P|=1$ for every labeling configuration $C$. Hence, in principle, many universal states along a labeled-free path could cause an exponential blow-up of the size of $P$. However, since $M$ is a well-behaved ATO, there exists $k\ge 0$ such that every labeled-free path of every computation $(V,E,\lambda)$ of $M$ on $w$ has at most $k$ nodes $v$ for which $\lambda(v)$ is a universal configuration. It is rather straightforward to see that every labeled-free path of $G$ also enjoys this property since this path occurs in some computation of $M$ on $w$. Hence, the size of the set $P$ is bounded by a polynomial in the size of $|w|$. From the above discussion, we get that $\mathsf{BuildNFTA}(w)$ runs in polynomial time in $|w|$, and the next lemma can be shown:

\def\lemmabuildnfta{
For a string $w\in \Lambda^*$, $\mathsf{BuildNFTA}(w)$ runs in polynomial time in $|w|$ and returns an NFTA $A$ such that $\spanm_M(w)=|L(A)|$.
}

\begin{lemma}\label{lem:buildnfta}
\lemmabuildnfta
\end{lemma}

Before we give the full proof of Lemma~\ref{lem:buildnfta}, let us discuss the rather easy second step of the reduction.


\begin{algorithm}[t]
    \KwIn{A string $w \in \Lambda^*$}
    \KwOut{An NFTA $A$}
    \vspace{2mm}
    
    {Construct the computation DAG $G = (\mathcal{C},\mathcal{E})$ of $M$ on $w$;}
 {\\ $S^A := \emptyset$; $\Lambda^A := \emptyset$; $\delta^A := \emptyset$; $Q :=\emptyset$;}
    {\\ $P :=\mathsf{Process}(\rt{G})$;}
    {\\ Assuming $P = \{(s)\}$, $A := \left(S^A,\Lambda^A, s, \delta^A\right)$;}
    {\\\Return{$A$;}}
    \caption{The algorithm $\mathsf{BuildNFTA}$}\label{alg:dagtonfta}
\end{algorithm}

\begin{algorithm}[t]
    \KwIn{A configuration $C=(s,x,y,z,h_x,h_y)$ of $\mathcal{C}$}
    \vspace{2mm}
    \If{$(C,U) \in Q$}
    {\Return $U$;}
    \If{$C$ is a leaf of $G$}
    {\If{$s\in S_\labeling$}
        {{$S^A := S^A \cup \{s_C\}$;\\}
            {$\Lambda^A := \Lambda \cup \{z\}$;\\}
            \If{$s=s_\accept$}
            {$\delta^A := \delta^A \cup \{(s_C,z,())\}$;}
            {$Q :=Q \cup \{(C,\{(s_C)\})\}$;}
        }
        \lElseIf{$s=s_\accept$}{$Q := Q \cup \{(C,\{()\})\}$}
            \lElse{$Q : = Q \cup \{(C,\emptyset)\}$}}
    \Else{
        {let $C_1,\dots,C_n$ be an arbitrary order over the nodes of $G$ with an incoming edge from $C$;\\}
        \lForEach{$i \in [n]$}{
            {$P_i := \process(C_i)$}}
        \lIf{$s\in S_\exists$}
        {$P := \bigcup_{i \in [n]} P_i$}
        \lElse{    
            $P := \bigotimes_{i \in [n]} P_i$
            %
        }
        \If{$s\in S_\labeling$}
        {{$S^A := S^A \cup \{s_C\}$;\\} 
            {$\Lambda^A := \Lambda^A \cup \{z\}$;\\}
            \ForEach{$(s_1,\dots,s_\ell)\in P$}{$\delta^A := \delta^A \cup \{(s_C,z,(s_1,\dots,s_\ell))\}$;}
            {$Q := Q \cup \{(C,\{(s_C)\})\}$;}}
        \lElse{$Q := Q \cup \{(C,P)\}$}}
    {\Return $U$ with $(C,U) \in Q$;}
    \caption{The recursive procedure $\mathsf{Process}$}\label{alg:process}
\end{algorithm}

\medskip

\noindent\paragraph{\underline{Step 2: The Polynomial Function}}

\smallskip

\noindent 
%
Since $M$ is a well-behaved ATO, there exists a polynomial function $\mathsf{pol}:\mathbb{N}\rightarrow \mathbb{N}$ such that the size of every computation of $M$ on $w$ is bounded by $\mathsf{pol}(|w|)$. Clearly, $\mathsf{pol}(|w|)$ is also a bound on the size of the valid outputs of $M$ on $w$. From the proof of Lemma~\ref{lem:buildnfta}, which is given below, it will be clear that every tree accepted by $A$ has the same structure as some valid output of $M$ on $w$. Hence, we have that $\mathsf{pol}(|w|)$ is also a bound on the size of the trees accepted by $A$, and the next lemma follows:

\def\lempolynomial{
    It holds that $L(A) = \bigcup_{i = 0}^{\mathsf{pol}(|w|)} L_i(A)$.
}

\begin{lemma}\label{lem:polynomial}
\lempolynomial
\end{lemma}

Proposition~\ref{pro:reduction} readily follows from Lemmas~\ref{lem:buildnfta} and~\ref{lem:polynomial}.


\subsection{Proof of Lemma~\ref{lem:buildnfta}}
We now prove the following lemma.

\begin{manuallemma}{\ref{lem:buildnfta}}
    \lemmabuildnfta
\end{manuallemma}

We have already discussed the running time of $\mathsf{BuildNFTA(w)}$ and shown that it runs in polynomial time in $|w|$. We proceed to show that there is one-to-one correspondence between the valid outputs of $M$ on $w$ and the $k$-trees accepted by $A$.

\medskip

\noindent\paragraph{Valid Output to Accepted Tree.} Let $O=(V',E',\lambda')$ be a valid output of $M$ on $w$. Then, $O$ is the output of some accepting computation $T=(V,E,\lambda)$ of $M$ on $w$. For every node $v$ of $T$ with $\lambda(v)$ being a labeling configuration, if $u_1,\dots,u_n$ are the nodes associated with a labeling configuration that are reachable from $v$ in $T$ via a labeled-free path, we denote by $\succ_v$ the order defined over $u_1,\dots,u_n$ in the following way. For $i\neq j\in[n]$, let $v'$ be the lowest common ancestor of $u_i,u_j$. Since $v$ is a common ancestor, clearly, $v'$ lies somewhere along the path from $v$ to $u_i$ (similarly, to $u_j$). Now, let $w_1$ be the child of $v'$ that is the ancestor of $u_i$ (or $u_i$ itself if $v'$ is the parent of $u_i$) and let $w_2$ be the child of $v'$ that is the ancestor of $u_j$ (or $u_j$ itself if $v'$ is the parent of $u_j$); clearly, $w_1\neq w_2$. We define $u_j\succ_v u_i$ iff $\lambda(w_1)$ occurs before $\lambda(w_2)$ in the order defined over the nodes with an incoming edge from $\lambda(v')$ in line~13 of $\process(\lambda(v'))$ (inside $\mathsf{BuildNFTA}(w)$). Clearly, this induces a total order $\succ_v$ over $u_1,\dots,u_n$.
We will show that the $k$-tree $\tau$ inductively defined as follows is accepted by $A$:
\begin{enumerate}
    \item The root $\epsilon$ of $\tau$ corresponds to the root $v$ of $T$, with $\tau(\epsilon)=z$ assuming that $\lambda(v)$ is of the form $(\cdot,\cdot,\cdot,z,\cdot,\cdot)$.
    \item For every node $u$ of $\tau$ corresponding to a node $v\in V$, with 
     $v_1,\dots,v_n\in V$ being the nodes associated with a labeling configuration that are reachable from $v$ via a labeled-free path, if $v_n\succ_v,\dots,\succ_v v_1$, then the children $u\cdot 1,\dots,u\cdot n$ of $u$ are such that $\tau(u\cdot i)=z_i$ assuming that $\lambda(v_i)$ is of the form $(\cdot,\cdot,\cdot,z_i,\cdot,\cdot)$ for all $i\in [n]$.
\end{enumerate}
Note that every node $u$ of $\tau$ corresponds to some node $v\in V$ such that $\lambda(v)$ is a labeling configuration, that, in turn, corresponds to some node $v'$ of $O$.
Clearly, in this way, two different valid outputs of $M$ on $w$ give rise to two different $k$-trees. We  prove the following:

\begin{lemma}\label{lemma:recursive}
The following holds for $\mathsf{BuildNFTA}(w)$.
\begin{enumerate}
\item For every node $v\in V$ with $\lambda(v)=(s,x,y,z,h_x,h_y)$ being a labeling configuration, $\mathsf{Process}(\lambda(v))$ adds to $\delta^A$ the transition $(s_{\lambda(v)},z,(s_{\lambda(v_1)},\dots,s_{\lambda(v_n)}))$, 
where $v_1,\dots,v_n$ are the nodes associated with a labeling configuration that are reachable from $v$ in $T$ via a labeled-free path, and $v_n\succ_v,\dots,\succ_v v_1$, and returns the set $\{(s_{\lambda(v)})\}$.
\item For every node $v\in V$ with $\lambda(v)$ being a non-labeling configuration, $\mathsf{Process}(\lambda(v))$ returns a set containing the tuple $(s_{\lambda(v_1)},\dots, s_{\lambda(v_n)})$, 
where $v_1,\dots,v_n$ are the nodes associated with a labeling configuration that are reachable from $v$ in $T$ via a labeled-free path, and $v_n\succ_v,\dots,\succ_v v_1$.
\end{enumerate}
\end{lemma}
\begin{proof}
We prove the claim by induction on $h(v)$, that is, the height of the node $v$ in $T$.
The base case, $h(v)=0$, is when $v$ is a leaf of $T$. Clearly, in this case, there are no nodes associated with a labeling configuration reachable from $v$ via a labeled-free path. Since $T$ is an accepting computation, it must be the case that $\lambda(v)$ is of the form $(s_\accept,\cdot,\cdot,z,\cdot,\cdot)$. If $\lambda(v)$ is a labeling configuration, in line~8 of $\process(\lambda(v))$ the transition $(s_{\lambda(v)},z,())$ is added to $\delta^A$. Then, in line~9 we add $(\lambda(v),\{(s_{\lambda(v)})\})$ to $Q$, and in line~24 we return the set $\{(s_{\lambda(v)})\}$. If $\lambda(v)$ is a non-labeling configuration, in line~10 we add $(\lambda(v),\{()\})$ to $Q$, and in line~24 we return the set $\{()\}$. This concludes our proof for the base case.

Next, assume that the claim holds for $h(v)\le \ell$, and we prove that it holds for $h(v)= \ell+1$. We start with the case where $\lambda(v)$ is a labeling configuration. Let $w_1,\dots,w_m$ be the children of $v$ in $T$, and let $v_1,\dots,v_n$ be the nodes associated with a labeling configuration that are reachable from $v$ via a labeled-free path. Clearly, every node $v_i$ is reachable via a (possibly empty) labeled-free path from some child $w_j$ of $v$. Moreover, for each $w_j$ we have that $h(w_j)\le\ell$; hence, the inductive assumption implies that the claim holds for each node $w_j$. Each node $w_j$ that is associated with a labeling configuration is, in fact, one of the nodes $v_1,\dots,v_n$; that is, $w_j=v_i$ from some $i\in[n]$. The inductive assumption implies that $\process(\lambda(w_j))$ returns the set $\{(s_{\lambda(v_i)})\}$ that contains a single tuple. We denote this tuple by $p_j$. If, on the other hand, $\lambda(w_j)$ is a non-labeling configuration, let $v_{i_1},\dots,v_{i_{n_j}}$ be the nodes of $v_1,\dots,v_n$ that are reachable from $w_j$ via a labeled-free path. Clearly, these are the only nodes associated with a labeling configuration that are reachable from $w_j$ via a labeled-free path, as each such node is also reachable from $v$ via a labeled-free path. In this case, the inductive assumption implies that $\process(\lambda(w_j))$ returns a set containing the tuple $(s_{\lambda(v_{i_1})},\dots,s_{\lambda(v_{i_{n_j}})})$. We again denote this tuple by $p_j$. Note that if $\lambda(w_j)$ is a non-labeling configuration, and there are no nodes associated with a labeling configuration that are reachable from $w_j$ via a labeled-free path, we have that $p_j=()$.

If $\lambda(v)$ is a universal configuration,
then in line~16 we define $P := \bigotimes_{i \in [m]} P_i$. In particular, $P$ will contain the tuple obtained by merging the tuples
$p_1,\dots, p_m$, assuming that $\lambda(w_1),\dots,\lambda(w_m)$ is the order defined over the nodes with an incoming edge from $\lambda(v)$ in line~13 of $\process(\lambda(v))$. Note that for two nodes $v_i\neq v_k$ that are reachable from the same child $w_j$ of $v$, if $v_i \succ_{w_j} v_k$, then $v_i \succ_{v} v_k$, because the lowest common ancestor of $v_i, v_k$ occurs under $w_j$ (or it is $w_j$ itself). If $v_i$ is reachable from $w_{j_1}$ while $v_k$ is reachable from $w_{j_2}$ with $j_1<j_2$, then $v$ is the lowest common ancestor of $v_i$ and $v_j$, and since $\lambda(w_{j_1})$ occurs before $\lambda(w_{j_2})$ in the order defined over the nodes with an incoming edge from $\lambda(v)$ in line~13 of $\process(\lambda(v))$, we have that $v_j\succ_v v_i$. We conclude that the tuple obtained by merging the tuples
$p_1,\dots, p_m$ is precisely the tuple $(s_{\lambda(v_1)},\dots,s_{\lambda(v_n)})$, assuming that $v_n\succ_v\dots\succ_v v_1$. Then, in line~21 we add the transition $(s_{\lambda(v)},z,(s_{\lambda(v_1)},\dots,s_{\lambda(v_n)}))$ to $\delta^A$ and in line~24 we return the set $\{(s_{\lambda(v)})\}$.

Now, if $\lambda(v)$ is an existential configuration, $v$ has a single child $w_1$ in $T$. If $\lambda(w_1)$ is a labeling configuration, then $w_1$ is the only node associated with a labeling configuration that is reachable from $v$ via a labeled-free path, and the inductive assumption implies that $\process(\lambda(w_1))$ returns the set $\{(s_{\lambda(w_1)})\}$ that contains a single tuple, which we denote by $p$.
Otherwise, all the nodes $v_1,\dots,v_n$ are reachable from $w_1$ via a labeled-free path, and $\process(\lambda(w_1))$ returns a set that contains the tuple $(s_{\lambda(v_1)},\dots, s_{\lambda(v_n)})$, that we again denote by $p$. (As aforementioned, it cannot be the case that there is an additional node associated with a labeling configuration that is reachable from $w_1$ via a labeled-free path, as each such node is also reachable from $v$ via such a path.)
It is rather straightforward that $\succ_v$ is the same as $\succ_{w_1}$; hence, since $v_n\succ_{w_1}\dots\succ_{w_1} v_1$ based on the inductive assumption, it also holds that $v_n\succ_{v}\dots\succ_{v} v_1$. In line~15 of $\process(\lambda(v))$ we simply take the union of the sets returned by all the nodes of $G$ with an incoming edge from $\lambda(v)$, one of which is $\lambda(w_1)$; hence, the tuple $p$ will appear as it is in $P$. Finally, in line~21 we add the transition $(s_{\lambda(v)},z,p)$ to $\delta^A$ and in line~24 we return the set $\{(s_{\lambda(v)})\}$. 

Finally, if $\lambda(v)$ is a non-labeling configuration, instead of adding the transition, we add $(\lambda(v),P)$ to $Q$ in line~23, and return in line~24 the set $P$ that, as aforementioned, contains the tuple $(s_{\lambda(v_1)},\dots, s_{\lambda(v_n)})$, which concludes our proof.
\end{proof}

With Lemma~\ref{lemma:recursive} in place, it is not hard to show that there is a run $\rho$ of $A$ over $\tau$. In particular, for every node $u$ of $\tau$ corresponding to a node $v\in V$ (with $\lambda(v)$ being a labeling configuration), we define $\rho(u)=s_{\lambda(v)}$.
Lemma~\ref{lemma:recursive} implies that the transition
$(s_{\lambda(v)},z,(s_{\lambda(v_1)},\dots,s_{\lambda(v_n)}))$ occurs in $\delta^A$, where $v_1,\dots,v_n$ are the nodes associated with a labeling configuration that are reachable from $v$ via a labeled-free path and $v_n\succ_v\dots\succ_v v_1$. The children of $u$ in $\tau$ correspond, by construction, to the nodes $v_1,\dots,v_n$ following the same order $\succ_v$. Note that if no nodes associated with a labeling configuration are reachable from $v$ via a labeled-free path, we have the transition $(s_{\lambda(v)},z,())$ in $\delta^A$, and $u$ is a leaf of $\tau$.
It is now easy to verify that the mapping $\rho$ defines a run of $A$ over $\tau$.

\medskip

\noindent\paragraph{Accepted Tree to Valid Output.} 
Let $\tau$ be a $k$-tree accepted by $A$. We will show that the labeled tree $O=(V',E',\lambda')$ inductively defined as follows is a valid output of $M$ on $w$:

\begin{enumerate}
    \item The root $v'$ of $O$ corresponds to the root $\epsilon$ of $\tau$, with $\lambda(v')=\tau(\epsilon)$.
    \item For every node $v'$ of $O$ corresponding to a non-leaf node $u$ of $\tau$, if
$u\cdot 1,\dots,u\cdot n$ are the children of $u$ in $\tau$, then $v_1',\dots,v_n'$ are the children of $v'$ in $O$, with $\lambda'(v_i')=\tau(u\cdot i)$ for all $i\in[n]$.
\end{enumerate}
Clearly, with this definition, two different $k$-trees give rise to two different valid outputs of $M$ on $w$. We will show that there is an accepting computation $T$ of $M$ on $w$ such that $O$ is output of $T$.

Our proof is based on the following simple observation regarding $\mathsf{BuildNFTA}(w)$ and the labeling configurations of $M$ on $w$. Recall that $G$ is the computation DAG of $M$ on $w$.

\begin{observation}\label{obs:labeling}
Let $C$ be a labeling configuration of $M$ on $w$. Let $C_1,\dots,C_n$ be the nodes of $G$ with an incoming edge from $C$, and let $(s_C,z,(s_1,\dots,s_\ell))$ be one of the transitions added to $\delta^A$ in line~21 of $\process(C)$. The following hold:
\begin{itemize}
\item If $C$ is a universal configuration, then there are $p_1,\dots,p_n$ such that $p_i$ is one of the tuples returned by $\process(C_i)$ for all $i\in[n]$, and $(s_1,\dots,s_\ell)$ is the tuple obtained by merging $p_1,\dots,p_n$ in line~16, and
\item if $C$ is an existential configuration, then $(s_1,\dots,s_\ell)$ is one of the tuples returned by $\process(C_i)$ for some $i\in[n]$.
\end{itemize}
\end{observation}

We also need the following observation regarding $\mathsf{BuildNFTA}(w)$ and non-labeling configurations of $M$ on $w$. 

\begin{observation}\label{obs:non-labeling}
Let $C$ be a non-labeling configuration of $M$ on $w$. Let $C_1,\dots,C_n$ be the nodes of $G$ with an incoming edge from $C$, and let $(s_1,\dots,s_\ell)$ be one of the tuples returned by $\process(C)$. Then:
\begin{itemize}
\item If $C$ is a universal configuration, then there are $p_1,\dots,p_n$ such that $p_i$ is one of the tuples returned by $\process(C_i)$ for all $i\in[n]$, and $(s_1,\dots,s_\ell)$ is obtained by merging the tuples $p_1,\dots,p_n$ in line~16, and
\item if $C$ is an existential configuration, then $(s_1,\dots,s_\ell)$ is one of the tuples returned by $\process(C_i)$ for some $i\in[n]$.
\end{itemize}
\end{observation}

Since $\tau$ is accepted by $A$, there is a run $\rho$ of $A$ over $\tau$.  Based on this and the two observations above, we can now define an inductive procedure to
 construct $T=(V,E,\lambda)$. Intuitively, we construct $T$ while traversing the $k$-tree $\tau$ from the root to the leaves. We start by adding the root of $T$ associated with the initial configuration of $M$ on $w$. For each node $v$ that we add to $V$, we store two pieces of information: \emph{(1)} the current node $\node(v)$ of $\tau$ that we handle (starting from the root of $\tau$), and \emph{(2)} a tuple $\states(v)$ of states of $S^A$ that this node is responsible for. Intuitively, if $\states(v)=(s_1,\dots,s_\ell)$, then there should be labeled-free paths from the node $v$ to nodes $v_1,\dots,v_\ell$ such that $\lambda(v_j)$ is the configuration of $s_j$, for every $j\in[\ell]$. Note that every state in $S^A$ corresponds to a \emph{labeling} configuration of $M$ on $w$, because in $\mathsf{BuildNFTA}(w)$ we add a state $s_C$ to $S^A$ if only if $C$ is a labeling configuration; hence, $\lambda(v_j)$ is a labeling configuration, for every $j\in[\ell]$.

 In addition, for every node $v$ with $\lambda(v)$ being a labeling configuration, we store an additional tuple $\assign(v)$ that is determined by a node of $\tau$. In particular, when we reach a labeling configuration, we continue our traversal of $\tau$, and consider one of the children $u$ of $\node(v)$. Then, if the children of $u$ in $\tau$ are $u\cdot 1,\dots,u\cdot \ell$ with $\rho(u)=s_{\lambda(v)}$ and $\rho(u_i)=s_i$ for $i\in[\ell]$, and we have the transition $(s_{\lambda(v)},\tau(u),(s_1,\dots,s_\ell))$ in $\delta^A$ (such a transition must exist because $\rho$ is a run of $A$ over $\tau$), the tuple of states that we should assign responsibility for is $\assign(v)=(s_1,\dots,s_\ell)$. If $\lambda(v)$ is a universal configuration, this set is split between the children of $v$ (determined by the transition $\lambda(v) \ra_M \{C'_1,\ldots,C'_n\}$ of $M$), so each one of them is responsible for a (possibly empty) subtuple of $\assign(v)$. If, on the other hand, $\lambda(v)$ is an existential configuration, we add one child under $v$ and pass to him the responsibility for the entire tuple. Formally, we define the following procedure.
\begin{enumerate}
\item We add a root node $v$ to $V$ with $\lambda(v)$ being the initial configuration of $M$ on $w$. We define 
\[
\states(v) = (s_{\lambda(v)}).
\]

\item For every node $v\in V$ with $\lambda(v)$ being a labeling configuration, assuming $\lambda(v) \ra_M \{C'_1,\ldots,C'_n\}$ for some $n>0$, let $u$ be a child of $\node(v)$ (or the root of $\tau$ if $\node(v)$ is undefined), such that:

\emph{(1)} $\rho(u)=s_{\lambda(v)}$, 

\emph{(2)} $u\cdot 1,\dots,u\cdot \ell$ are the children of $u$ in $\tau$, 

\emph{(3)} $\rho(u\cdot j)=s_j$ for all $j\in[\ell]$. Consider the transition $(s_{\lambda(v)},\tau(u),(s_1,\dots,s_\ell))$ in $\delta^A$. We define 
\[\assign(v)=(s_1,\dots,s_\ell).\]
Then, if $\lambda(v)$ is an existential configuration,
we add a single child $v_i$ under $v$ with $\lambda(v_i)=C_i'$ for some $i\in[n]$, such that $\process(C'_i)$ returns the tuple $(s_1,\dots,s_\ell)$ based on Observation~\ref{obs:labeling}, and define 
\[
\node(v_i)=u \quad \text{and} \quad \states(v_i)=(s_1,\dots,s_\ell).
\]
 If $\lambda(v)$ is a universal configuration, we add the children $v_1,\dots,v_n$ under $v$ with $\lambda(v_i)=C_i'$ for all $i\in[n]$, and for each child $v_i$ we define
\[
\node(v_i)=u \quad \text{and} \quad \confs(v_i)=p_i
\]
where $p_i$ is the tuple returned by $\process(C_i')$, for every $i\in[n]$, such that $(s_1,\dots,s_\ell)$ is obtained by merging the tuples $p_1,\dots,p_n$ based on Observation~\ref{obs:labeling}.

\item For a node $v\in V$ with $\lambda(v)$ being a non-labeling configuration, assuming $\lambda(v) \ra_M \{C_1',\ldots,C_n'\}$ for some $n>0$ and $\states(v)=(s_1,\dots,s_\ell)$, if $\lambda(v)$ is an existential configuration, we add a single child $v_i$ under $v$ with $\lambda(c_i)=C_i'$ for some $i\in[n]$, such that the tuple $(s_1,\dots,s_\ell)$ is returned by $\process(C_i')$ based on Observation~\ref{obs:non-labeling}, and define 
\[
\node(v_i)=\node(v) \quad \text{and} \quad \states(v_i)=(s_1,\dots,s_\ell).
\]
If $\lambda(v)$ is a universal configuration, we add the children $v_1,\dots,v_n$ under $v$ with $\lambda(v_i)=C_i'$ for all $i\in[n]$, and for each child $v_i$ we define
\[\node(v_i)=\node(v) \quad \text{and} \quad
\states(v_i)=p_i\]
where $p_i$ is the tuple returned by $\process(C_i')$, for every $i\in[n]$, such that $(s_1,\dots,s_\ell)$ is obtained by merging the tuples $p_1,\dots,p_n$ based on Observation~\ref{obs:non-labeling}.

\item For a leaf node $v\in V$ with $\lambda(v)$ being a labeling configuration, we define
 \[
 \assign(v)=().
 \]
\end{enumerate}

We now show that the above procedure is well-defined. In particular, we show that we can always find a node of $\tau$ that satisfies the desired properties when considering labeling configurations, and that we can indeed apply Observations~\ref{obs:labeling} and~\ref{obs:non-labeling}.
\begin{lemma}\label{lemma:valid_procedure}
The following hold:
\begin{itemize}
\item For every added non-root node $v$, it holds that $\states(v)$ is a subtuple of $\assign(r)$, where $r$ is the lowest ancestor of $v$ with $\lambda(r)$ being a labeling configuration.
\item For every added node $v$, $\states(v)$ is a tuple returned by $\process(\lambda(v))$.
\item For every added node $v$ with $\lambda(v)$ being a labeling configuration, $\states(v)=(s_{\lambda(v)})$.
\item For every added node $v$ with $\lambda(v)$ being a labeling configuration, there exists a child $u$ of $\node(v)$ (or the root of $\tau$ if $\node(v)$ is undefined) with $\rho(u)=s_{\lambda(v)}$ such that, if $u\cdot 1,\dots,u\cdot \ell$ are the children of $u$ in $\tau$ with $\rho(u\cdot j)=s_j$ for every $j\in [\ell]$, then the transition $(s_{\lambda(v)},\tau(u),(s_1,\dots,s_\ell))$ occurs in $\delta^A$ and $\assign(v)=(s_1,\dots,s_\ell)$.
\item For every added leaf $v$, $\lambda(v)$ is an accepting configuration.
\end{itemize}
\end{lemma}

\begin{proof}
    We prove the claim by induction on $d(v)$, that is, the depth of the node $v$ in the constructed tree. If $d(v)=0$, then $v$ is the root of the tree. By construction, $\lambda(v)$ is the initial configuration of $M$ on $w$ and it is a labeling configuration by definition. Moreover, in the procedure we define $\states(v)=(s_{\lambda(v)})$, and $\process(C)$ returns the tuple $(s_C)$ for every labeling configuration $C$; in particular, $\process(\lambda(v))$ returns $(s_{\lambda(v)})$. Since $\node(v)$ is undefined, we consider the root $u$ of $\tau$. It must be the case that $\rho(u)=s_{\lambda(v)}$, since the state $s_{\lambda(v)}$ that corresponds to the initial configuration of $M$ on $w$, and so to the root of $G$, is the state $s^A_\init$ (as defined in line~4 of $\mathsf{BuildNFTA}(w)$). If $u\cdot 1,\dots,u\cdot\ell$ are the children of $u$ in $\tau$, and $\rho(u\cdot j)=s_j$ for every $j\in[\ell]$,  since $\rho$ is a run of $A$ over $\tau$, there must be a transition $(s_{\lambda(v)},\tau(u),(s_1,\dots,s_\ell))$ in $\delta^A$. In the procedure, assuming $\lambda(v) \ra_M \{C'_1,\ldots,C'_n\}$ for some $n>0$ (i.e.~$v$ is a non-leaf node), we define $\assign(v)=(s_1,\dots,s_\ell)$.

    As for the first point of the lemma, the non-root nodes $v_1,\dots,v_n$ of depth $1$ are the direct children of the root, and, in the procedure, we define $\states(v_i)=\assign(v)$ if $\lambda(v)$ is an existential configuration and $v_i$ is the only child of $v$ in the constructed tree (in this case, there is only one node of depth $1$), or $\states(v_i)\subseteq\assign(v)$ for every $i\in[n]$ if $\lambda(v)$ is a universal configuration. In both cases, $v$ is the lowest ancestor of $v_1,\dots,v_n$ with $\lambda(v)$ being a labeling configuration, and we have that $\states(v_i)\subseteq\assign(v)$. This concludes the base case.

    Next, assume that the claim holds for all added nodes $v$ with $d(v)\le m$. We prove that it also holds for all added nodes $v$ with $d(v)=m+1$. Let $v$ be such a node, and let $p$ be its parent in the constructed tree. Since $d(p)\le m$, the inductive assumption holds for $p$. 
    If $\lambda(p)$ is a non-labeling configuration, the inductive assumption implies that  $\states(v)=(s_1,\dots,s_\ell)$ is a tuple returned by $\process(\lambda(p))$. If, on the other hand, $\lambda(p)$ is a labeling configuration, the inductive assumption implies that there exists a child $u$ of $\node(p)$ (or the root of $\tau$ if $\node(p)$ is undefined) with $\rho(u)=s_{\lambda(p)}$, such that if $u\cdot 1,\dots,u\cdot \ell$ are the children of $u$ in $\tau$ with $\rho(u\cdot j)=s_j$ for every $j\in [\ell]$, the transition $(s_{\lambda(p)},\tau(u),(s_1,\dots,s_\ell))$ occurs in $\delta^A$ and, since $p$ is a non-leaf node, $\assign(v)=(s_1,\dots,s_\ell)$. Hence, if $\lambda(p)$ is a non-labeling configuration, we now consider the tuple $\states(p)=(s_1,\dots,s_\ell)$, and if $\lambda(p)$ is a labeling configuration, we consider the tuple 
    $\assign(p)=(s_1,\dots,s_\ell)$.
    
    Then, if $\lambda(p)$ is an existential configuration, Observation~\ref{obs:labeling} or~\ref{obs:non-labeling} (depending on whether $\lambda(p)$ is a labeling or non-labeling configuration) implies that, assuming $\lambda(p) \ra_M \{C_1',\ldots,C_n'\}$ for $n>0$, $(s_1,\dots,s_\ell)$ is one of the tuples returned by $\process(C_i')$ for $i\in[n]$. In this case, in the procedure, we add a single child $v$ under $p$ with $\lambda(v)=C_i'$, and define $\states(v)=(s_1,\dots,s_\ell)$.
    Clearly, this child is the node $v$; hence, $\process(\lambda(v))$ returns the tuple $(s_1,\dots,s_\ell)$.
    
    If $\lambda(p)$ is a universal configuration, Observation~\ref{obs:labeling} or Observation~\ref{obs:non-labeling} (depending on whether $\lambda(p)$ is a labeling or non-labeling configuration) implies that, again assuming $\lambda(p) \ra_M \{C_1',\ldots,C_n'\}$ for some $n>0$, there are $p_1,\dots,p_n$, such that for every $i\in[n]$, $p_i$ is one of the tuples returned by $\process(C_i')$, and $(s_1,\dots,s_\ell)$ is the tuple obtained by merging $p_1,\dots,p_n$. In the procedure, we add $n$ children $v_1,\dots,v_n$ under $p$, with $\lambda(v_i)=C_i'$ for $i\in[n]$, and, clearly, $v=v_i$ for some $i\in [n]$; hence, we define $\states(v)=p_i$.
    
    Note that in both cases ($\lambda(p)$ being universal or existential), we have shown that $\process(\lambda(v))$ returns the tuple $\states(v)$.
    Moreover,
    if $\lambda(p)$ is a labeling configuration, $p$ is the lowest ancestor of $v$ with $\lambda(p)$ being a labeling configuration, and we have shown that 
    $\states(v)$ is a subtuple of $\assign(p)$. If $\lambda(p)$ is a non-labeling configuration, the inductive assumption implies that $\states(p)$ is a subtuple of $\assign(r)$, where $r$ is the lowest ancestor of $p$ with $\lambda(r)$ being a labeling configuration. Clearly, $r$ is also the lowest ancestor of $v$ with $\lambda(r)$ being a labeling configuration, and we conclude that $\states(v)$ is a subtuple of $\assign(r)$.

    Now, if $\lambda(v)$ is a non-labeling configuration, then we are done, as we have proved that the first and second points of the lemma hold. However, if $\lambda(v)$ is a labeling configuration, we need to prove that the third and fourth points hold. In particular, we need to establish that there exists a child $u$ of $\node(v)$ with $\rho(u)=s_{\lambda(v)}$, such that if $u\cdot 1,\dots,u\cdot \ell$ are the children of $u$ in $\tau$ with $\rho(u\cdot j)=s_j$ for every $j\in[\ell]$, the transition $(s_{\lambda(v)},\tau(u),(s_1,\dots,s_\ell))$ occurs in $\delta^A$. Consider the lowest ancestor $r$ of $v$ with $\lambda(r)$ being a labeling configuration. Since $d(r)\le m$, the inductive assumption implies that there exists a child $u'$ of $\node(r)$ with $\rho(u')=s_{\lambda(r)}$, such that if $u'\cdot 1,\dots,u'\cdot m$ are the children of $u'$ in $\tau$ with $\rho(u'\cdot j)=S_j'$ for every $j\in[m]$, the transition $(s_{\lambda(r)},\tau(u'),(s_1',\dots,s_m'))$ occurs in $\delta^A$. Note that $\assign(r)=(s'_1,\dots,s_m')$.
   Clearly, we have that $\node(v)=u'$, as only labeling configurations modify this information, and all the nodes along the path from $r$ to $v$ are associated with non-labeling configurations.

As we have already shown, it holds that $\states(v)$ is a subtuple of $\assign(r)$, that is, a subtuple of $(s'_1,\dots,s_m')$. Since $\lambda(v)$ is a labeling configuration, the only tuple returned by $\process(\lambda(v))$ is $(s_{\lambda(v)})$. We have also shown that $\process(\lambda(v))$ returns the tuple $\states(v)$; thus, we conclude that $\states(v)=(s_{\lambda(v)})$ (which prove the fourth point of the lemma) and so $s_{\lambda(v)}=s_j'$ for some $j\in[m]$. Therefore, there is indeed a child $u'\cdot j$ of $\node(v)=u'$ with 
$\rho(u'\cdot j)=s_j'=s_{\lambda(v)}$. If the children of $u'\cdot j$ in $\tau$ are $u'\cdot j\cdot 1,\dots,u'\cdot j\cdot t$ with $\rho(u'\cdot j\cdot i)=s''_i$ for all $i\in[t]$, since $\rho$ is a run of $A$ over $\tau$, clearly there is a transition $(s_{\lambda(v)},\tau(u'\cdot j),(s_1'',\dots,s_t''))$ in $\delta^A$, and in the procedure we define $\assign(v)=(s_1'',\dots,s_t'')$. Note that if $u'\cdot j$ is a leaf of $\tau$, since $\rho$ is a run of $A$ over $\tau$, we must have a transition $(s_{\lambda(v)},\tau(u'\cdot j),())$ in $\delta^A$.

Finally, if $v$ is a leaf and $\lambda(v)$ is a non-labeling configuration, we have already established that $\states(v)$ is a tuple returned by $\process(\lambda(v))$, and it is a subtuple of $\assign(r)$, where $r$ is the lowest ancestor of $v$ with $\lambda(r)$ being a labeling configuration. Since $v$ is a leaf of $T$, there is no transition of the form $\lambda(v) \ra_M \{C_1',\ldots,C_n'\}$, and so $\lambda(v)$ is a leaf of $G$. Therefore, in $\process(\lambda(v))$ we return the set $\{()\}$ in line~24 if $\lambda(v)$ is an accepting configuration or the empty set if $\lambda(v)$ is a rejecting configuration (this set is added to $Q$ in line~10 or~11). Hence, the only possible case is that $\states(v)=()$ and $\lambda(v)$ is an accepting configuration.

If $v$ is a leaf and $\lambda(v)$ is a labeling configuration, we have already shown that the transition $(s_{\lambda(v)},\tau(u),(s_1\dots,s_\ell))$ is in $\delta^A$, where $u$ is the child of $\node(v)$ in $\tau$ that we consider when adding $v$, $u\cdot 1,\dots,u\cdot\ell$ are the children of $u$ in $\tau$, $\rho(u\cdot j)=s_j$ for all $j\in\ell$. Since $v$ is a leaf, in $\process(\lambda(v))$ we return a set containing the single tuple $(s_\lambda(v))$ in line~24 (after adding this set to $Q$ in line~9), 
and the only transition added to $\delta^A$ with $s_\lambda(v)$ on its left-hand side is $(s_{\lambda(v)},\tau(u),())$ if $\lambda(v)$ is an accepting configuration. If $\lambda(v)$ is a rejecting configuration, we do not add any transition with $\lambda(v)$ on its left-hand side to $\delta^A$. 
Thus, the only possible case is that $\lambda(v)$ is an accepting configuration, and $u$ is a leaf of $\tau$.
\end{proof}

 In the described procedure, we start with a node associated with the initial configuration of $M$ on $W$, and we systematically add a single child for each node $v$ with $\lambda(v)$ being an existential configuration and all children for each node $v$ with $\lambda(v)$ being a universal configuration, while following the transitions $\lambda(v) \ra_M \{C'_1,\ldots,C'_n\}$ of $M$. Therefore, in combination with Lemma~\ref{lemma:valid_procedure}, it is easy to verify that the constructed $T$ is indeed a computation of $M$ on $w$. Moreover, since for every leaf $v$ of $T$ we have shown that $\lambda(v)$ is an accepting configuration, $T$ is an accepting computation. It is only left to show that $O=(V',E',\lambda')$ is the output of $T$. The proof consists of three steps. We start by showing the following.

 \begin{lemma}\label{lemma:path-to-labeled-free}
     Let $v$ be a node of $T$. Assume that $\assign(v)=(s_{C_1},\dots,s_{C_\ell})$ if $v$ is a node of $T$ with $\lambda(v)$ being a labeling configuration, or $\states(v)=(s_{C_1},\dots,s_{C_\ell})$ otherwise. Then, there is a labeled-free path from $v$ to a node $v'$ with $\lambda(v')$ being a labeling configuration if and only if $\lambda(v')=C_i$ for some $i\in[\ell]$.
 \end{lemma}
 \begin{proof}
     We prove the claim by induction on $h(v)$, the height of the node $v$ in $T$. If $h(v)=0$, then $v$ is a leaf of $T$. Lemma~\ref{lemma:valid_procedure} implies that $\lambda(v)$ is an accepting configuration. If $\lambda(v)$ is a non-labeling configuration, $\process(\lambda(v))$ returns the set $\{()\}$.
     Then, Lemma~\ref{lemma:valid_procedure} implies that $\states(v)=()$, and the claim follows. If, on the other hand, $\lambda(v)$ is a labeling configuration, in the procedure we define $\assign(v)=()$, and the claim again follows.
     
     Next, assume that the claim holds for all nodes $v$ with $h(v)\le m$. We prove that it holds for all nodes $v$ with $h(v)=m$. If $\lambda(v)$ is a labeling configuration, assuming $\lambda(v) \ra_M \{C'_1,\ldots,C'_n\}$ for some $n>0$, in the procedure we either add one node $w_1$ under $v$ with $\states(w_1)=\assign(v)$ (if $\lambda(v)$ is existential), or we add $n$ nodes $w_1,\dots,w_n$ under $v$ with $\states(w_i)$ being a (possibly empty) subtuple of $\assign(v)$ for every $i\in[n]$, and such that $\assign(v)$ is obtained by merging the tuples $\states(w_1),\dots,\states(w_n)$.
     Similarly, if $\lambda(v)$ is a non-labeling configuration, assuming $\lambda(v) \ra_M \{C'_1,\ldots,C'_n\}$ for some $n>0$, in the procedure we either add one node $w_1$ under $v$ with $\states(w_1)=\states(v)$ (if $\lambda(v)$ is existential), or we add $n$ nodes $w_1,\dots,w_n$ under $v$ with $\states(w_i)$ being a (possibly empty) subtuple of $\states(v)$ for every $i\in[n]$, and such that $\states(v)$ is obtained by merging the tuples $\states(w_1),\dots,\states(w_n)$.
     
     If $\lambda(w_i)$ is a labeling configuration for some $i\in[n]$, Lemma~\ref{lemma:valid_procedure} implies that $\states(w_i)=(s_{\lambda(w_i)})$; hence, $s_{\lambda(w_i)}$ is one of the states in $(s_{C_1},\dots,s_{C_\ell})$, and clearly there is an (empty) labeled-free path from $v$ to $w_i$. If $\lambda(w_i)$ is a non-labeling configuration for some $i\in[n]$, since $h(w_i)\le m$, the inductive assumption implies that there is a labeled-free path from $w_i$ to nodes $v_{i_1},\dots,v_{i_t}$ such that $\states(w_i)=(s_{C_{i_1}},\dots,s_{C_{i_t}})$ and $\lambda(v_{i_j})=C_{i_j}$ for all $j\in[t]$. Clearly, this means that there is a labeled-free path from $v$ to $v_{i_j}$ for all $j\in[t]$. Since, as aforementioned, $(s_{C_1},\dots,s_{C_\ell})$ is obtained by merging the tuples $\states(w_1),\dots,\states(w_n)$ (or it is the single tuple $\states(w_1)$ if $\lambda(v)$ is existential), and every path from $v$ passes through one of its children, we conclude that there is a labeled-free path from $v$ to a node $v'$ with $\lambda(v')$ being a labeling configuration if and only if $\lambda(v')=C_i$ for some $i\in\ell$.
     \end{proof}

 Next, we show that every node of $\tau$ is associated with some node $v$ of $T$ with $\lambda(v)$ being a labeling configuration.

 \begin{lemma}\label{lemma:label-cover-tau}
     For every node $u$ of $\tau$, there is a node $v$ of $T$ with $\lambda(v)$ being a labeling configuration, such that $u$ is a child of $\node(v)$ (or the root of $\tau$ is $\node(v)$ is undefined) with $\rho(u)=s_{\lambda(v)}$.
 \end{lemma}
 \begin{proof}
     We prove the claim by induction on $d(u)$, the depth of the node $u$ in $\tau$. If $d(u)=0$, then $u$ is the root of $\tau$. By construction, the root $v$ of $T$ is such that $\lambda(v)$ is the initial configuration of $M$ on $w$, which is a labeling configuration. Moreover, in the procedure, we consider the node $u$ when adding $v$ to $T$, and we have that $\rho(u)=s_{\lambda(v)}$ because $s_{\lambda(v)}$ is the initial state $s^A_\init$ of the NFTA.

     Next, assume that the claim holds for all nodes $u$ of $\tau$ with $d(u)\le m$, and we prove that it holds for nodes $u$ with $d(u)=m+1$. Let $u'$ be the parent of $u$ in $\tau$ (hence, $u$ is of the form $u'\cdot j$ for some $j$). The inductive assumption implies that there is a node $v'$ of $T$ with $\lambda(v')$ being a labeling configuration, such that $u'$ is a child of $\node(v')$ with $\rho(u')=s_{\lambda(v')}$. Since $\rho$ is a run of $A$ over $\tau$, there is a transition $(s_{\lambda(v')},\tau(u'),(s_{C_1},\dots,s_{C_\ell}))$ in $\delta^A$ such that $\rho(u)=s_{C_i}$ for some $i\in[\ell]$. In the procedure, we define $\assign(v')=(s_{C_1},\dots,s_{C_\ell}))$. Lemma~\ref{lemma:path-to-labeled-free} implies that there is a labeled-free path from $v'$ to some node $v$ with $\lambda(v)=C_i$. Note that $\node(v)=u'$, as this is the case for every node with a labeled-free path from $v'$ (only labeling configurations alter this information). When we consider the node $v$ in the procedure, we choose a child of $\node(v)$ (hence, a child of $u'$) such that $\rho(u)=s_{\lambda(v)}$, and this is precisely the node $u$. Hence, $u$ is a child of $\node(v)$ with $\rho(u)=s_{\lambda(v)}$, and this concludes our proof.
 \end{proof}

Finally, we show that $O$ is indeed the output of $T$.

\begin{lemma}
The following hold.
\begin{enumerate}
    \item $V' = \{v \in V \mid \lambda(v) \text{ is a labeling configuration of } M\}$,
    \item $E' = \{(u,v) \mid u \text{ reaches } v \text{ in } T \text{ via a } \text {labeled-free path}\}$,
    
    \item for every $v' \in V'$, $\lambda'(v') = z$ assuming that $\lambda(v)$ is of the form $(\cdot,\cdot,\cdot,z,\cdot,\cdot)$.
\end{enumerate}
\end{lemma}
\begin{proof}
    The first point is rather straightforward. Every node $v$ of $T$ with $\lambda(v)$ being a labeling configuration is associated with a (unique) node $u$ of $\tau$; this is the child of $\node(v)$ that we choose when adding the node $v$ to $T$. Moreover, Lemma~\ref{lemma:label-cover-tau} implies that every node $u$ of $\tau$ is associated with a node $v$ of $T$ with $\lambda(v)$ being a labeling configuration.
    Hence, there is one-to-one correspondence between the nodes $v$ of $T$ with $\lambda(v)$ being a labeling configuration, and the nodes $u$ of $\tau$. Since $O$ contains a node for every node of $\tau$ by construction, and no other nodes, item (1) follows.

    The second item follows from Lemma~\ref{lemma:path-to-labeled-free} showing that there is a labeled-free path from a node $v$ to a node $v'$ with both $\lambda(v)$ and $\lambda(v')$ being labeling configurations, if and only if $\assign(v)=(s_{C_1},\dots,s_{C_\ell})$ and $\lambda(v')=C_i$ for some $i\in[\ell]$. In the procedure, we define $\assign(v)=(s_{C_1},\dots,s_{C_\ell})$ if there is a node $u$ of $\node(v)$ such that $\rho(u)=s_{\lambda(v)}$, with children $u\cdot 1,\dots,u\cdot\ell$ such that $\rho(u\cdot i)=s_{C_i}$ for all $i\in[\ell]$. Clearly, we have that $\node(v')=u$, as only nodes associated with labeling configurations alter this information.
    When we consider the node $v'$ in the procedure, we select a child $u'$ of $u$ such that $\rho(u')=s_{\lambda(v')}$; thus $s_{\lambda(v')}=s_{C_i}$ for some $i\in[\ell]$. Hence, the node $u'$ associated with $v'$ is a child of the node $u$ associated with $v$, and there is an edge $(u,u')$ in $\tau$, and so there is an edge between the corresponding nodes in $O$. 

    Finally, for every node $v\in V$ with $\lambda(v)$ being a labeling configuration, the node $u$ of $\tau$ that we consider when adding $v$ to $T$ is such that $\rho(u)=s_{\lambda(v)}$, where $\rho$ is a run of $A$ over $\tau$, and if $u\cdot 1,\dots,u\cdot \ell$ are the children of $u$ in $\tau$ with $\rho(u\cdot j)=s_{C_j}$ for every $j\in [\ell]$, the transition $(s_{\lambda(v)},\tau(u),(s_{C_1},\dots,s_{C_\ell}))$ occurs in $\delta^A$. When we add a transition to $\delta^A$ in $\mathsf{BuildNFTA(w)}$, we always add transitions of the form $(s_C,z,\cdots)$, where $C$ is a configuration of the form $(\cdot,\cdot,\cdot,z,\cdot,\cdot)$. Hence, if $\lambda(v)$ is a configuration of the form $(\cdot,\cdot,\cdot,z,\cdot,\cdot)$, we have that $\tau(u)=z$ and $\lambda'(v')=z$ for the corresponding node $v'$ of $O$. This concludes our proof.
\end{proof}

This concludes our proof of Lemma~\ref{lem:buildnfta}.

\OMIT{
\paragraph{Finalize the proof.} We have shown that there is one-to-one correspondence between the trees accepted by $A$ and the valid outputs of $M$ on $w$. However, the FPRAS for NFTAs is for trees of a given size. Since $M$ is a well-behaved ATO, there exists some polynomial $\mathsf{pol}:\mathbb{N}\rightarrow \mathbb{N}$ such that the size of every computation is bounded by $\mathsf{pol}(|w|)$. Clearly, $\mathsf{pol}(|w|)$ is also a bound on the size of the valid outputs of $M$ on $w$. We have already shown that every tree accepted by $A$ is obtained from a valid output by adding a single node under each leaf. Hence, if we define a polynomial $\mathsf{pol}'$ such that $\mathsf{pol}'(x)=2\times \mathsf{pol}(x)$ for every $x\in\mathbb{N}$, we have that $\mathsf{pol}'(|w|)$ is a bound on the size of the trees accepted by $A$.

Now, for every $i\in[\mathsf{pol}'(|w|)]$, we denote by $\mathcal{A}_i(A,\epsilon,\delta)$ the randomized algorithm that, given $A$, $\epsilon$ and $\delta$, satisfies:
\[
\text{\rm Pr}\left(|\mathcal{A}_i(A,\epsilon,\delta) - |L_i(A)||\ \leq\ \epsilon \cdot |L_i(A|\right)\ \geq\ 1-\delta.
\]
(Recall that ${L}_i(A)$ is the set of trees of size $i$ accepted by $A$.)
We claim that the randomized algorithm $\mathcal{A}(A,\epsilon,\delta)$ that executes $\mathcal{A}_i(A,\epsilon,\delta')$ for every $i\in[\mathsf{pol}'(|w|)]$, where $\delta'=2\times\mathsf{pol}'(|w|)$, and takes the sum of the results is an FPRAS for the problem of computing $|L(A)|$. Note that $L(A)=\bigcup_{i=1}^{\mathsf{pol}'(|w|)} L_i(A)$; hence, $|L(A)|=\sum_{i=1}^{\mathsf{pol}'(|w|)}|L_i(A)|$ (clearly, $L_i(A)\cap L_j(A)=\emptyset$ for every $i\neq j$).

Since each $\mathcal{A}_i$ runs in time polynomial in $|A|$, $\frac{1}{\epsilon}$, and $\log(\frac{1}{\delta'})$, clearly the algorithm $\mathcal{A}$ runs in time polynomial in $|A|$, $\frac{1}{\epsilon}$, and $\log(\frac{1}{\delta})$. Moreover, it holds that:
\begin{align*}
  &\text{\rm Pr}\left(|\mathcal{A}(A,\epsilon,\delta) - |L(A)||\ \leq\ \epsilon \cdot |L(A)|\right) 
  \\=&\text{\rm Pr}\left(\frac{1}{1+\epsilon}\cdot|L(A)|\ \leq \ \mathsf{A}(S,\epsilon,\delta)\ \leq \ (1+\epsilon) |L(A)|\right) 
  \\=&\text{\rm Pr}\left(\frac{1}{1+\epsilon}\cdot|L(A)|\ \leq \ \sum_{i=1}^{\mathsf{pol}'(|w|)}\mathcal{A}_i(A,\epsilon,\delta)\ \leq \ (1+\epsilon) |L(A)|\right) 
    \\&\text{\rm Pr}\left(\sum_{i=1}^{\mathsf{pol}'(|w|)}\frac{1}{1+\epsilon}\cdot|L_{i}(A)|\ \leq \ \sum_{i=1}^{\mathsf{pol}'(|w|)}\mathcal{A}_i(A,\epsilon,\delta)\ \leq \ \sum_{i=1}^{\mathsf{pol}'(|w|)}(1+\epsilon) |L_i(A)|\right) 
\end{align*}
Since for every $i\in[\mathsf{pol}'(|w|)]$ we have that:
\[\text{\rm Pr}\left(\frac{1}{1+\epsilon}\cdot|L_{i}(A)|\ \leq \ \mathcal{A}_i(A,\epsilon,\delta)\ \leq \ (1+\epsilon) |L_{i}(A)|\right) \ \geq \ 1-\delta'\]
and the events for each $i$ are independent, we conclude that:
\[\text{\rm Pr}\left(|\mathcal{A}(A,\epsilon,\delta) - |L(A)||\ \leq\ \epsilon \cdot |L(A)|\right)  \geq \ (1-\delta')^{\mathsf{pol}'(|w|)}.\]
Finally, we know (see, e.g.,~\cite{DBLP:journals/tcs/JerrumVV86}) that the following inequality holds:
\[(1-\frac{x}{2n})^n\ge 1-x\]
for any $0\le x\le 1$ and $n\ge 1$;
hence, we obtain the required FPRAS.
Therefore, $\mathcal{A}$ is also a randomized algorithm that, given $w$, $\epsilon$ and $\delta$ satisfies:
\[\text{\rm Pr}\left(|\mathcal{A}(A,\epsilon,\delta) - \mathsf{span}_M(w)|\ \leq\ \epsilon \cdot \mathsf{span}_M(w)\right) \ \geq \ 1-\delta.\]
where $A$ is the NFTA that we construct from $M$ using $\mathsf{BuildNFTA}$, and since we construct $A$ is polynomial time in $|w|$, $\mathcal{A}$ runs in time polynomial in $|w|$, $1/\epsilon$, and $\log(1/\delta)$.

\OMIT{With this claim in place, it is not hard to show that there is a run of $\mathcal{T}$ over $\tau$. In particular, we define the mapping $\Lambda$ as follows. For every $v\in V''$ that also occurs in $T$ we define:
\[\Lambda(v)=s_{v''}.\]
For every $v\in V''$ that does not occur in $T$ (these are the leaves of $\tau$ we define:
\[\Lambda(v)=s_{\accept}.\]
Every node $v$ of $\tau$ corresponding to a non-leaf of $T$ is also associated with some  node $v'$ of $T'$, and the claim implies that the transition
\[s_{v''}(\sem{O}(x_1,\dots,x_m))\rightarrow \sem{O}(s_{u_1''}(x_1),\dots, s_{u_m''}(x_m)\] occurs in $\delta$, where $u_1,\dots,u_m$ are the children of $v$ in $T$ (thus, also in $\tau$).
Every node $v$ of $\tau$ corresponding to a leaf of $T$ is also associated with some leaf $v'$ of $T'$, and the claim implies that the transition
\[s_{v''}(\sem{O}(x))\rightarrow \sem{O}(s_\accept(x))\] occurs in $\delta$.
Moreover, we always add the transition
$s_\accept(\epsilon)\rightarrow\epsilon$ to $\delta^A$ in line~3 of $\mathsf{DAGtoNFTA(G)}$. It is now easy to verify, that the mapping $\Lambda$ indeed defined a run of $\mathcal{T}$ over $\tau$.}

\OMIT{Let $T=(V,E,\mu)$ be a valid output of $M$ w.r.t.~$w$. Then, $T$ is the output of some accepting induced tree $T'=(V',E',\lambda)$ of $M$ w.r.t.~$w$.
For every node $v\in V$, we denote by $v^{T'}$ its corresponding node in $T'$, and by $v^G$ the node of $G$ representing the configuration $\lambda(v^{T'})$. Moreover, if the children of $v$ are $u_1,\dots,u_m$, we denote by $w_i^{T'}$ the  first node in the labeled-free path from $v_{T'}$ to $u_i^{T'}$ for every $i\in[m]$. If $u_i^{T'}$ is a direct child of $v^{T'}$, then we denote $w_i^{T'}=u_i^{T'}$.
Hence, the node representing the configuration of $w_i^{T'}$ in $G$ is $w_i^{G}$.

We will show that the ordered labeled tree $\tau=(V'',E'',\varphi,\succ)$ obtained from $T$ via the following procedure is accepted by $\mathcal{T}$:
\begin{enumerate}
    \item We start with $V''=V$, $\varphi(v)=\mu(v)$ for every $v\in V$, and $E''=E$.
    \item For every node $v\in V$ with children $u_1,\dots,u_m$, we use the order defined in line~18 when processing $v^G$ to define $\succ$; that is, if the order defined over the children of $v^G$ is $w_1^{G},\dots,w_m^{G}$, then the successor relation $\succ$ of $\tau$ is such that $u_{i+1}\succ u_i$, for all $i\in[m-1]$.
    \item We add a node $u$ under each leaf of $V''$ with $\varphi(u)=\epsilon$.
\end{enumerate}
Clearly, in this way, two different valid outputs $T_1, T_2$ of $M$ w.r.t.~$w$ give rise to two different ordered labeled trees $\tau_1,\tau_2$.

We prove, by induction on the size (i.e.,~number of vertices) of $\tau$, that there is a mapping $\Lambda:V''\rightarrow S$ such that:
\begin{itemize}
    \item for every $u\in V''$ with an ordered sequence $u_1,\dots,u_m$ of children, if $\Lambda(u)=s$ and $\Lambda(u_i)=s_i$ for $i\in[n]$, then
there exists $s(\varphi(u)(x_1,\dots,x_m))\rightarrow \varphi(u)(s_1(x_1),\dots,s_m(x_m))$ in $\delta$, and
\item for every leaf $u\in V''$, there is $s(\varphi(u))\rightarrow \varphi(u)$ in $\delta$.
\end{itemize}

The base case, $|V''|=2$, is simple. (Note that we always add an additional node when constructing $V''$; hence, the base case is when $V''$ contains two nodes.) In this case, the only node $r$ of $T'$ (the induced tree) is the root of the computation tree of $M$ w.r.t.~$w$. Since $T'$ is an accepting induced tree, we have that $r$ is of the form $(q_\accept,I,W,O,h_I,h_W,h_O)$. Therefore, $q_\init=q_\accept$ and $q_\accept\in Q_\mathcal{L}$. Moreover, since the computation tree contains a single node, the corresponding DAG $G$ is the computation tree itself. Hence, in lines $4,5,6$ of $\mathsf{Process}(\mathcal{T}, G, v)$ we will add a state $s_r$ to $S$ (which, in line~5 of $\mathsf{BuildNFTA}(w)$ will also be added to $S_0$) and a transition $s_r(\sem{O}(x))\rightarrow \sem{O}(s_\accept(x))$ to $\delta^A$. Together with the transition $s_\accept(\epsilon)\rightarrow\epsilon$ that is added to $\delta^A$ in line~3 of $\mathsf{BuildNFTA}(w)$, this is sufficient to allow a mapping $\Lambda$ satisfying the desired properties. In particular, we have that $\Lambda(v)=s_r$ for the node $v$ of $\tau$ corresponding to the node $r$. It is easy to verify that $\varphi(v)=\sem{O}$. Moreover, we define $\Lambda(u)=s_\accept$ for the only child $u$ of $v$ which is a leaf of the tree with $\varphi(u)=\epsilon$.

Next, we assume that the claim holds for $|V''|\in[c-1]$ and prove that it holds for $|V''|=c$. Let $z$ be the root of $\tau$, and let $y_1,\dots y_m$ be its children. For every $i\in[m]$ we denote by $\tau[u_i]$ the subtree of $\tau$ rooted at $u_i$. Each one of these subtrees contains at most $c-1$ nodes; thus, by the inductive assumption, for each such subtree there is a mapping $\Lambda_i$ that satisfies the desired properties. We take the union of these mapping (note that they are disjoint) and denote it by $\Lambda$. We will now show how to extend this mapping to the entire tree. 

Let $v$ be the node of $T$ corresponding to $z$ and let $u_1,\dots,u_m$ be the nodes of $T$ corresponding to $y_1,\dots,y_m$, respectively. For each $i\in[m]$, the node $u^{T'}_i$ (i.e.,~the node of $T'$ corresponding to $u_i$) must be such that $\lambda(u^{T'}_i)$ is a configuration associated with a labeling state of $M$. Therefore, whenever we process the corresponding node $u^G_i$ of $G$ in $\mathsf{Process}(\mathcal{T}, G, u^G_i)$, we return $\{(s_{u^G_i})\}$ in line~30.

Now, for everylet $w^i_1,\dots,w^i_t$ be the labeled-free path in $T'$ from }
}
\section{The Normal Form}

We prove the following result; note that there is no corresponding statement in the main body of the paper since the discussion on the normal form was kept informal.

\def\pronormalform{
	Consider a database $D$, a set $\dep$ of primary keys $\dep$, a CQ $Q(\bar x)$ from $\sjf$, a generalized hypertree decomposition $H$ of $Q$ of width $k$, and $\bar c \in \adom{D}^{|\bar x|}$. There exists a database $\hat{D}$, a CQ $\hat{Q}(\bar x)$ from $\sjf$, and a generalized hypertree decomposition $\hat{H}$ of $\hat{Q}$ of width $k+1$ such that:
	\begin{itemize}
		\item $(\hat{D},\hat{Q},\hat{H})$ is in normal form,
		\item $|\{D' \in \opr{D}{\dep} \mid \bar c \in Q(D')\}| = |\{D' \in \opr{\hat{D}}{\dep} \mid \bar c \in \hat{Q}(D')\}|$, and
		\item $(\hat{D},\hat{Q},\hat{H})$ is logspace computable in the combined size of $D,\dep,Q,H,\bar c$.
	\end{itemize}
}

\begin{proposition}\label{pro:normal-form}
	\pronormalform
\end{proposition}
With the above proposition in place, together with the fact that $\spantl$ is closed under logspace reductions, to prove that $\sharp\mathsf{Repairs}[k]$ and $\sharp\mathsf{Sequences}[k]$ are in $\spantl$, for each $k > 0$, it is enough to focus on databases $D$, queries $Q$, and generalized hypertree decompositions $H$ such that $(D,Q,H)$ is in normal form.

We first prove Proposition~\ref{pro:normal-form} in the special case that $H$ is already complete. Then, the more general result will follow from the fact that logspace reductions can be composed (see, e.g.,~\cite{ArBa09}), and from the following well-known result:

\begin{lemma}[\cite{GoLS02}]\label{lem:complete-logspace}
	Given a CQ $Q(\bar x)$ and a generalized hypertree decomposition $H$ for $Q$ of width $k$, a complete generalized hypertree decomposition $H'$ of $Q$ of width $k$ always exists, and can be computed in logarithmic space  in the combined size of $Q$ and $H$.
\end{lemma}

In the rest of this section, we focus on proving Proposition~\ref{pro:normal-form}, assuming that $H$ is complete. We start by discussing how $\hat{D}$, $\hat{Q}(\bar x)$ and $\hat{H}$ are constructed. In what follows, fix a database $D$, a set $\dep$ of primary keys, a CQ $Q(\bar x)$ from $\sjf$, a \emph{complete} generalized hypertree decomposition $H$ of $Q$ of width $k$, and a tuple $\bar c \in \adom{D}^{\bar x}$.

\medskip
\noindent\paragraph{\underline{Construction of $(\hat{D},\hat{Q},\hat{H})$}}

\smallskip
\noindent
Assume $H = (T,\chi,\lambda)$, with $T=(V,E)$, and assume $Q(\bar x)$ is a CQ of the form $\text{Ans}(\bar x) \text{ :- } R_1(\bar y_1),\ldots,R_n(\bar y_n)$.

\medskip
\noindent
\textbf{The Query $\hat{Q}$.} The query $\hat{Q}$ is obtained by modifying $Q$ as follows. If $P_1/n_1,\ldots,P_m/n_m$ are all the relations occurring in $D$ but not in $Q$, add two atoms of the form $P_i(\bar z_i)$ and $P'_i(z_i')$ to $Q$, for each $i \in [m]$, where $P_i'/1$ is a fresh unary relation, $\bar z_i$ is a tuple of $n_i$ distinct variables not occurring in $Q$, and $z_i' \not \in \bar z_i$ is a variable not occuring in $Q$. Moreover, if $v_1,\ldots,v_\ell$ are all the vertices in $V$, with $v_i$ having $h_i \ge 0$ children in $T$, create $h_i + 1$ fresh unary relations $S^{(1)}_{v_i}/1,\ldots,S^{(h_i + 1)}_{v_i}/1$, for each $i \in [\ell]$, not occurring anywhere in $Q$, and add to $Q$ the atom $S^{(j)}_{v_i}(w^{(j)}_{v_i})$, for each $i \in [\ell]$, and $j \in [h_i + 1]$, where $w^{(j)}_{v_i}$ is a variable. Hence, $\hat{Q}$ is of the form:
\begin{multline*}
\text{Ans}(\bar x) \text{ :- } R_1(\bar y_1),\ldots,R_n(\bar y_n), P_1(\bar z_1),P'_1(z_1') \ldots, P_m(\bar z_m),P'_m(z_m'), \\
S^{(1)}_{v_1}(w^{(1)}_{v_1}),\ldots,S^{(h_1 + 1)}_{v_1}(w^{(h_1 + 1)}_{v_1}),\ldots,S^{(1)}_{v_\ell}(w^{(1)}_{v_\ell}),\ldots,S^{(h_\ell + 1)}_{v_\ell}(w^{(h_\ell + 1)}_{v_\ell}).
\end{multline*}
The query $\hat{Q}$ is clearly self-join free.

\medskip
\noindent
\textbf{The Database $\hat{D}$.} The database $\hat{D}$ is obtained from $D$ by adding an atom of the form $P'_i(c)$, for each $i \in [m]$, and an atom of the form $S^{(j)}_v(c)$, for each $v \in V$, and $j \in [h+1]$, where $h \ge 0$ is the number of children of $v$ in $T$, and $c$ is some constant.

\medskip
\noindent
\textbf{The Decomposition $\hat{H}$.} We now discuss how $\hat{H} = (\hat{T},\hat{\chi},\hat{\lambda})$, with $\hat{T} = (\hat{V},\hat{E})$, is constructed:
\begin{itemize}
	\item Concerning the vertex set $\hat{V}$, for each vertex $v \in V$ in $T$, having $h \ge 0$ children $u_1,\ldots,u_h$, $\hat{V}$ contains $h+1$ vertices $v^{(1)},\ldots,v^{(h+1)}$, and for each $i \in [m]$, $\hat{V}$ contains the vertices $v_{P_i}$ and $v_{P'_i}$.
	
	\item The edge set $\hat{E}$ contains the edges $v_{P_i} \rightarrow v_{P'_i}$, $v_{P_i} \rightarrow v_{P_{i+1}}$, for each $i \in [m-1]$, and the edges $v_{P_m} \rightarrow v_{P'_m}$, $v_{P_m} \rightarrow v^{(1)}$, where $v \in V$ is the root of $T$. Moreover, for each \emph{non-leaf} vertex $v \in V$ with $h > 0$ children $u_1,\ldots,u_h$, $\hat{E}$ contains the edges $v^{(i)} \rightarrow u^{(1)}_i$ and $v^{(i)} \rightarrow v^{(i+1)}$, for each $i \in [h]$.
	
	\item Concerning $\hat{\chi}$ and $\hat{\lambda}$, for each $i \in [m]$, $\hat{\chi}(v_{P_i}) = \bar z_i$, $\hat{\chi}(v_{P'_i}) = \{z_i'\}$, and $\hat{\lambda}(v_{P_i}) = \{P_i(\bar z_i)\}$, $\hat{\lambda}(v_{P'_i}) = \{P_i'(z_i')\}$, and for each vertex $v \in V$ with $h \ge 0$ children, $\hat{\chi}(v^{(i)}) = \chi(v) \cup \{w^{(i)}_v\}$, and $\hat{\lambda}(v^{(i)}) = \lambda(v) \cup \{S^{(i)}_v(w^{(i)}_v)\}$, for each $i \in [h+1]$.
\end{itemize}

\medskip
\noindent\paragraph{\underline{Proving Proposition~\ref{pro:normal-form} when $H$ is complete}}

\smallskip
\noindent
We now proceed to show that $(\hat{D},\hat{Q},\hat{H})$ enjoys all the properties stated in Proposition~\ref{pro:normal-form}, assuming that $H$ is complete.

\medskip
\noindent
\textbf{Complexity of construction.}
We start by discussing the complexity of constructing $(\hat{D},\hat{Q},\hat{H})$. Constructing the query $\hat{Q}$ requires producing, in the worst case, two atoms for each relation in $D$, and $|V|$ atoms, for each vertex $v \in V$, and so the number of atoms in $\hat{Q}$ is $O(|Q| + |D| + |V|^2)$, where $|Q|$ is the number of atoms in $Q$. Each atom can be constructed individually using logarithmic space, and by reusing the space for each atom. We can show the database $\hat{D}$ is constructible in logspace with a similar argument. Regarding $\hat{H}$, the set $\hat{V}$ contains at most $|V|$ vertices for each node of $V$, plus two nodes for each relation in $D$. Hence, $|\hat{V}| \in O(|V|^2 + |D|)$, and each vertex of $\hat{V}$ can be constructed by a simple scan of $E$, $D$ and $Q$. Moreover, $\hat{E}$ contains $2 \times m \le 2 \times |D|$ edges, i.e. the edges of the binary tree involving the vertices of the form $v_{P_i},v_{P'_i}$, and the vertex $v^{(1)}$, where $v \in V$ is the root of $T$, plus at most $2 \times |V|$ edges for each vertex of $V$, so overal $|\hat{E}| \in O(|V|^2 + |D|)$. Again, each edge is easy to construct. Finally, constructing $\hat{\chi}$ and $\hat{\lambda}$ requires a simple iteration over $D$ and $\hat{V}$, and copying the contents of $\chi$, and $\lambda$.

\medskip
\noindent
\textbf{$\hat{H}$ is a proper decomposition of width $k+1$.}
We now argue that $\hat{H}$ is a generalized hypertree decomposition of $\hat{Q}$ of width $k+1$. We need to prove that \emph{(1)} for each atom $R(\bar y)$ in $\hat{Q}$, there is a vertex $v \in \hat{V}$ with $\bar y \setminus \bar x \subseteq \hat{\chi}(v)$, \emph{(2)} for each $x \in \var{\hat{Q}} \setminus \bar x$, the set $\{v \in \hat{V} \mid x \in \hat{\chi}(v)\}$ induces a connected subtree of $\hat{T}$, \emph{(3)} for each $v \in \hat{V}$, $\hat{\chi}(v) \subseteq \var{\hat{\lambda}(v)}$, and \emph{4)} $\max_{v \in \hat{V}}|\lambda(v)| = k+1$.

We start with \emph{(1)}. For each atom $R_i(\bar y_i)$, for $i \in [n]$ in $\hat{Q}$, since there is a vertex $v \in V$ with $\bar y_i \setminus \bar x \subseteq \chi(v)$, then, e.g., the vertex $v^{(1)} \in \hat{V}$ is such that $\bar y_i \setminus \bar x \subseteq \hat{\chi}(v^{(1)})$, by construction of $\hat{\chi}$. Furthermore, for each atom $P_i(\bar z_i)$ (resp., $P_i'(z_i')$) in $\hat{Q}$, with $i \in [m]$, by construction of $\hat{\chi}$, $\bar z_i \setminus \bar x = \bar z_i \subseteq \hat{\chi}(v_{P_i})$ (resp., $z_i' \not \in \bar x$, and $z_i' \in \hat{\chi}(v_{P'_i})$), and for each atom $S^{(i)}_v(w^{(i)}_v)$ in $\hat{Q}$, with $v \in V$, we have that $w^{(i)}_v \not \in \bar x$, and $w^{(i)}_v \in \hat{\chi}(v^{(i)})$, again by construction of $\hat{\chi}$.

We now show item \emph{(2)}. Consider a variable $x \in \var{\hat{Q}} \setminus \bar x$. We distinguish two cases:
\begin{itemize} 
\item Assume $x \not \in \var{Q} \setminus \bar x$, i.e., $x \in \var{\hat{Q}} \setminus \var{Q}$. Hence, either $x$ occurs in some atom of the form $P_i(\bar z_i)$ (resp., $P_i'(z_i')$) in $\hat{Q}$, which means $x$ occurs only in $\hat{\chi}(v_{P_i})$ (resp., $\hat{\chi}(v_{P_i'})$), and thus $v_{P_i}$ (resp., $v_{P_i'}$) trivially forms a connected subtree, or $x = w^{(i)}_v$, for some atom of the form $S^{(i)}_v(w^{(i)}_v)$ in $\hat{Q}$, with $v \in V$, which means $x$ occurs only in $\hat{\chi}(v^{(i)})$, and thus $v^{(i)}$ trivially forms a connected subtree. 

\item Assume now that $x \in \var{Q} \setminus \bar x$. By contruction of $\hat{H}$, for each $v \in V$ having  $x \in \chi(v)$, $x$ occurs in $\hat{\chi}(v^{(i)})$, for each $i \in [h+1]$, where $h \ge 0$ is the number of children of $v$ in $T$, and $x$ does not occur anywhere else in $\hat{T}$. Since the set of vertices $\{v \mid x \in \chi(v)\}$ forms a subtree of $T$, by assumption, and since for each $v \in V$, with $h \ge 0$ children $u_1,\ldots,u_h$ in $T$, the vertices $v^{(1)},\ldots,v^{(h+1)}$ form a path in $\hat{T}$, and since when $h > 0$, each $v^{(i)}$, with $i \in [h]$, is connected to $u_i^{(1)}$ in $\hat{T}$, the set $\{v \in \hat{V} \mid x \in \hat{\chi}(v)\}$ forms a subtree of $\hat{T}$, as needed.
\end{itemize}

We finally show item \emph{(3)}. Consider a vertex $u \in \hat{V}$. If $v = v_{P_i}$ (resp., $v = v_{P_i'}$), for some $i \in [m]$, then $\hat{\chi}(v) = \bar z_i$ (resp., $z_i'$), and $\hat{\lambda}(v) = \{P_i(\bar z_i)\}$ (resp., $\{P_i'(z_i')\}$), and the claim follows trivially. If $u = v^{(i)}$, for some $v \in V$, and $i \in [h+1]$, where $h \ge 0$ is the number of children of $v$ in $T$, then $\hat{\chi}(v^{(i)}) = \chi(v) \cup \{w^{(i)}_c\}$ and $\hat{\lambda(v)} = \lambda(v) \cup \{S^{(i)}_v(w^{(i)}_v)\}$. Since $\chi(v) \subseteq \var{\lambda(v)}$, the claim follows. Finally, to prove \emph{4)} it is enough to observe that for any vertex $u \in \hat{V}$, $\hat{\lambda}(u)$ contains either one atom, or one atom more than $\lambda(v)$, where $v \in V$ is the node of $T$ from which $u$ has been contructed. Hence, $\hat{H}$ is a generalized hypertree decomposition of $\hat{Q}$ of width $k+1$.

\smallskip
\noindent
\textbf{$(\hat{D},\hat{Q},\hat{H})$ is in normal form.}
We now prove that $(\hat{D},\hat{Q},\hat{H})$ is in normal form, i.e., \emph{(1)} every relation in $\hat{D}$ also occurs in $\hat{Q}$, \emph{(2)} $\hat{H}$ is 2-uniform, and \emph{(3)} strongly complete. 
The fact that \emph{(1)} holds follows by construction of $\hat{D}$ and $\hat{Q}$.
To see why \emph{(2)} holds, note that every non-leaf node in $\hat{T}$  has always exactly two children. 
Finally, we show item \emph{(3)}. Consider an atom $\alpha$ in $\hat{Q}$. By construction of $\hat{\chi}$ and $\hat{\lambda}$, if $\alpha$ is of the form $P_i(\bar z_i)$ (resp., $P_i'(z_i)$), then $v_{P_i}$ (resp., $v_{P_i'}$) is a covering vertex for $\alpha$. Moreover, if $\alpha$ is of the form $S^{(i)}_v(w^{(i)}_v)$, for some $v \in V$, then $v^{(i)}$ is a covering vertex for $S^{(i)}_v(w^{(i)}_v)$. Finally, if $\alpha$ is of the form $R_i(\bar y_i)$, with $i \in [n]$, then if $v \in V$ is the covering vertex of $\alpha$ in $T$ (it always exists, since $H$ is complete, by assumption), then any vertex of the form $v^{(j)}$ in $\hat{V}$ is a covering vertex for $R_i(\bar y_i)$ in $\hat{T}$. Hence, $\hat{H}$ is complete. To see that $\hat{H}$ is \emph{strongly} complete it is enough to observe that every vertex $v$ of $\hat{T}$ is such that $\hat{\lambda}(v)$ contains an atom $\alpha$, whose variables all occur in $\hat{\chi}(v)$, such that that $\alpha \not \in \hat{\lambda}(u)$, for any other $u \in \hat{V}$, different from $v$.
So, $(\hat{D},\hat{Q},\hat{H})$ is in normal form.

\smallskip
\noindent
\textbf{$(\hat{D},\hat{Q},\hat{H})$ preserves the counts.} In this last part we need to show
\[
|\{D' \in \opr{D}{\dep} \mid \bar c \in Q(D')\}| =\\
|\{D' \in \opr{\hat{D}}{\dep} \mid \bar c \in \hat{Q}(D')\}|.
\]
To prove the claim, we rely on the following auxiliary notion:

\begin{definition}[\textbf{Key Violation}]\label{def:violation}
	Consider a schema $\ins{S}$, a database $D'$ over $\ins{S}$, and a key $\kappa = \key{R} = A$, with $R \in \ins{S}$. A {\em $D'$-violation} of $\kappa$ is a set $\{f,g\} \subseteq D'$ of facts such that $\{f,g\} \not\models \kappa$.	
	We denote the set of $D'$-violations of $\kappa$ by $\viol{D'}{\kappa}$. Furthermore, for a set $\dep'$ of keys, we denote by $\viol{D'}{\dep'}$ the set $\{(\kappa,v) \mid \kappa \in \dep' \textrm{~~and~~} v \in \viol{D'}{\kappa}\}$. \hfill\markfull
\end{definition}

Note that any $(D',\dep')$-justified operation of the form $-F$, for some database $D'$ and set $\dep'$ of keys, is such that $F \subseteq \{f,g\}$, where $\{f,g\}$ is some $D'$-violation in $\viol{D'}{\dep'}$.
We now reason about the sets $\viol{D}{\dep}$ and $\viol{\hat{D}}{\dep}$ of $D$-violations and $\hat{D}$-violations, respectively, of all keys in $\dep$. In particular, note that $\hat{D} = D \cup C$, where $D \cap C = \emptyset$, and each fact $\alpha \in C$ is over a relation that no other fact mentions in $\hat{D}$. Hence, we conclude that $\viol{D}{\dep} = \viol{\hat{D}}{\dep}$. Moreover, we observe that $\hat{Q}$ contains all atoms of $Q$, and in addition, for each fact $R(\bar t) \in C$, a single atom $R(\bar y)$, where each variable in $\bar y$ occurs exactly ones in $\hat{Q}$. The above observations imply that:
\begin{enumerate}
	\item for every database $D'$, $D' \in \opr{D}{\dep}$ iff $D' \cup C \in \opr{\hat{D}}{\dep}$;
	\item for each $D' \in \opr{D}{\dep}$, $\bar c \in Q(D)$ iff $\bar c \in \hat{Q}(D' \cup C)$, where $D' \cup C \in \opr{\hat{D}}{\dep}$.
\end{enumerate}
Hence, items~(1) and item~(2) together imply $|\{D' \in \opr{D}{\dep} \mid \bar c \in Q(D')\}| = |\{D' \in \opr{\hat{D}}{\dep} \mid \bar c \in \hat{Q}(D')\}|$. This concludes our proof.

\section{Proof of Lemma~\ref{lem:repairs-ato}}
	
We proceed to show that:

\begin{manuallemma}{\ref{lem:repairs-ato}}
	\lemrepairsato
\end{manuallemma}

We start by proving that the algorithm can be implemented as a well-behaved ATO $M^k_R$, and then we show that the number of valid outputs of $M^k_R$ on input $D$, $\dep$, $Q(\bar x)$, $H=(T,\chi,\lambda)$, and $\bar c \in \adom{D}^{\bar x}$, with $(D,Q,H)$ in normal form, is precisely the number of operational repairs $D' \in \opr{D}{\dep}$ such that $\bar c\in Q(D')$.

\subsection{Item (1) of Lemma~\ref{lem:repairs-ato}}
We are going to give a high level description of the ATO $M^k_R$ underlying the procedure $\sharp\mathsf{Rep}[k]$ in Algorithm~\ref{alg:repairs}.
We start by discussing the space used in the working tape. Every object that $\sharp\mathsf{Rep}[k]$ uses from its input and stores in memory, such as vertices of $T$, atoms/tuples from $D$ and $Q$, etc., can be encoded using logarithmically many symbols via a pointer encoded in binary to the part of the input where the object occurs, and thus can be stored in the working tape of $M^k_R$ using logarithmic space. Hence, the node $v$ of $T$, the sets $A$ and $A'$, and the fact $\alpha$, all require logarithmic space in the working tape (recall that $A$ and $A'$ contain a fixed number, i.e., $k$, of tuple mappings). It remains to discuss how $M^k_R$ can check that \emph{1)} $A \cup A' \cup \{\bar x \mapsto \bar c\}$ is coherent, \emph{2)} a vertex $v$ of $T$ is the $\prec_T$-minimal covering vertex of some atom, and \emph{3)} iterate over each element of $\block{\dep}{R_{i_j},D}$, for some relation $R_{i_j}$; checking whether $v$ is a leaf, or guessing a child of $u$ are all trivial tasks. To perform \emph{1)}, it is enough to try every variable $x$ in $A' \cup \{\bar x \mapsto \bar c\}$, and for each such $x$, try every pair of tuple mappings, and check if they map $x$ differently; $x$ and the two tuple mappings can be stored as pointers. To perform \emph{2)}, $M^k_R$ must be able to compute the depth of $v$ in $T$, and iterate over all vertices at the same depth of $v$, in lexicographical order, using logarithmic space. The depth can be easily computed by iteratively following the (unique) path from the root of $T$ to $v$, in reverse, and increasing a counter at each step. By encoding the counter in binary, the depth of $v$ requires logarithmically many bits and can be stored in the working tape. Iterating all vertices with the same depth of $v$ in lexicographical order can be done iteratively, by running one pass over each node of $T$, computing its depth $d$, and in case $d$ is different from the depth of $v$, the node is skipped. During this scan, we keep in memory a pointer to a node that at the end of the scan will point to the lexicographically smaller node, e.g., $u$, among the ones with the same depth of $v$. Then, when the subsequent node must be found, a new scan, as the above, is performed again, by looking for the lexicographically smallest vertex that is still larger than the previous one, i.e., $u$. This is the next smallest node, and $u$ is updated accordingly. The process continues, until the last node is considered. Storing a pointer to $u$ is feasible in logspace, and thus, together with the computation of the depth $v$ in logspace, it is possible to search for all nodes $u$ of $T$ such that $u \prec_T v$ in logspace. Then, checking for each of these nodes whether $\lambda(u)$ contains $R_{i_j}(\bar c_j)$, i.e., $v$ is not the $\prec_T$-minimal covering vertex for $R_{i_j}(\bar c_j)$ is trivial. 
Regarding \emph{3)}, one can follow a similar approach to the one for iterating the nodes of the same depth of $v$ lexicographically, where all the key values of the form $\keyval{\dep}{R_{i_j}(\bar t)}$ are instead traversed lexicographically.

We now consider the space used in the labeling tape of $M^k_R$. For this, according to Algorithm~\ref{alg:repairs}, $M^k_R$ needs to write in the labeling tape facts $\alpha$ of the database $D$, which, as already discussed, can be uniquely encoded as a pointer to the input tape where $\alpha$ occurs.

Regarding the size of a computation of $M^k_R$, note that each vertex $v$ of $T$ makes the machine $M^k_R$ execute all the lines from line~\ref{line:begin-repairs} to line~\ref{line:end-repairs} in Algorithm~\ref{alg:repairs}, before it  universally branches to one configuration for each child of $v$. Since all such steps can be performed in logspace, they induce at most a polynomial number of existential configurations, all connected in a path of the underlying computation of $M^k_R$. Since the number of vertices of $T$ is linear w.r.t.\ the size of the input, and $M^k_R$ never processes the same node of $T$ twice, the total size of any computation of $M^k_R$ be polynomial.

It remains to show that for any computation $T'$ of $M^k_R$, all labeled-free paths in $T'$ contain a bounded number of universal computations. This follows from the fact that before line~\ref{line:end-repairs}, $M^k_R$ has necessarily moved to a labeling state (recall that $H$ is strongly complete, and thus every vertex of $T$ is the $\prec_T$-minimal covering vertex of some atom). Hence, when reaching line~\ref{line:end-repairs}, either $v$ is a leaf, and thus $M^k_R$ halts, and no universal configuration have been visited from the last labeling configuration, or $M^k_R$ universally branches to the (only two) children of $v$, for each of which, a labeling configuration will necessarily be visited (again, because $H$ is strongly complete). Branching to two children requires visiting only one non-labeling configuration on each path. Hence, $M^k_R$ is a well-behaved ATO.

\subsection{Item (2) of Lemma~\ref{lem:repairs-ato}}
Fix a database $D$, a set $\dep$ of primary keys, a CQ $Q(\bar x)$, a generalized hypertree decomposition $H=(T,\chi,\lambda)$ of $Q$ of width $k$, and a tuple $\bar c \in \adom{D}^{\bar x}$, where $(D,Q,H)$ is in normal form.
We show that there is a bijection from the valid outputs of $M^k_R$ on input $(D,\dep,Q,H,\bar c)$ to the set of operational repairs $\{D' \in \opr{D}{\dep} \mid \bar c\in Q(D')\}$. We define a mapping $\mu$ in the following way. Each valid output $O = (V,E,\lambda')$ is mapped to the set $\{\lambda'(v)\mid v\in V\}\setminus\{\bot\}$. Intuitively, we collect all the labels of the nodes of the valid output which are not labeled with $\bot$. To see that $\mu$ is indeed a bijection as needed, we now show the following three statements:
\begin{enumerate}
	\item $\mu$ is correct, that is, it is indeed a function from the set
	$$\{O \mid O \text{ is a valid output of } M \text{ on } (D,\dep,Q,H,\bar c)\}$$ to 
	$\{D' \in \opr{D}{\dep} \mid \bar c\in Q(D')\}$.
	\item $\mu$ is injective.
	\item $\mu$ is surjective.
\end{enumerate}

\noindent
\paragraph{The mapping $\mu$ is correct.} Consider an arbitrary valid output $O = (V,E,\lambda')$ of $M^k_R$ on $(D,\dep,Q,H,\bar c)$. Let $D' = \mu(O)$. First note that all labels assigned in Algorithm~\ref{alg:repairs} are database atoms or the symbol $\bot$. Thus, $D'$ is indeed a database. Further, we observe in Algorithm~\ref{alg:repairs} that for every relation $R$ and every block $B\in\block{\dep}{R,D}$ exactly one vertex is labeled in the output $O$ (with either one of the atoms of the block or $\bot$) and, hence, we have that $D'\models \dep$, i.e., $D' \in \opr{D}{\dep}$ since every consistent subset of $D$ is an operational repair of $D$ w.r.t.~$\dep$. What remains to be argued is that $\bar c\in Q(D')$.

Since we consider a complete hypertree decomposition, every atom $R_i(\bar y_i)$ of $Q$ has a $\prec_{T}$-minimal covering vertex $v$. By definition of covering vertex, we have that $\bar y_i\subseteq \chi(v)$ and $R_i(y_i)\in\lambda(v)$. In the computation, for each vertex $v$ of $T$, we guess a set $A' = \{ \bar y_{i_1} \mapsto \bar c_1,\ldots,\bar y_{i_\ell} \mapsto \bar c_\ell\}$ assuming $\lambda(v) = \{ R_{i_1}(\bar y_{i_1}), \ldots, R_{i_\ell}(\bar y_{i_\ell})\}$, such that $R_{i_j}(\bar c_j) \in D$ for $j \in [\ell]$ and all the mappings are coherent with $\bar x\mapsto \bar c$. In particular, when we consider the $\prec_{T}$-minimal covering vertex of the atom $R_i(y_i)$ of $Q$, we choose some mapping $\bar y_i\mapsto \bar c_i$ such that $R_i(\bar c_i)\in D$. When we consider the block $B$ of $R_i(\bar c_i)$ in line~5, we choose $\alpha=R_i(\bar c_i)$ in line~7, which ensures that $D'\cap B=R_i(\bar c)$; that is, this fact will be in the repair. Since $Q$ is self-join-free, the choice of which fact to keep from the other blocks of $\block{\dep}{R_i,D}$ can be arbitrary (we can also remove all the facts of these blocks).  When we transition from a node $v$ of $T$ to its children, we let $A:=A'$, and then, for each child, we choose a new set $A'$ in line~2 that is consistent with $A$ and with $\bar x\mapsto \bar c$.

This ensures that two atoms $R_i(\bar y_i)$ and $R_j(\bar y_j)$ of $Q$ that share at least one variable are mapped to facts $R_i(\bar c_i)$ and $R_j(\bar c_j)$ of $D$ in a consistent way; that is, a variable that occurs in both $\bar y_i$ and $\bar y_j$ is mapped to the same constant in $\bar c_i$ and $\bar c_j$. Let $v$ and $u$ be the $\prec_{T}$-minimal covering vertices of $R_i(\bar y_i)$ and $R_j(\bar y_j)$, respectively. Let $z$ be a variable that occurs in both $\bar y_i$ and $\bar y_j$. By definition of covering vertex, we have that $\bar y_i\subseteq \chi(v)$ and $\bar y_j\subseteq \chi(u)$. By definition of hypertree decomposition, the set of vertices $s\in T$ for which $z\in\chi(s)$ induces a (connected) subtree of $T$. Therefore, there exists a subtree of $T$ that contains both vertices $v$ and $u$, and such that $z\in\chi(s)$ for every vertex $s$ of this subtree. When we process the root $r$ of this subtree, we map the variable $z$ in the set $A'$ to some constant (since there must exist an atom $R_k(\bar y_k)$ in $\lambda(r)$ such that $z$ occurs in $\bar y_k$ by definition of a hypertree decomposition). Then, when we transition to the children of $r$ in $T$, we ensure that the variable $z$ is mapped to the same constant, by choosing a new set $A'$ of mappings that is consistent with the set of mappings chosen for $r$, and so forth. Hence, the variable $z$ will be consistently mapped to the same constant in this entire subtree, and for every atom $R_i(\bar y_i)$ of $Q$ that mentions the variable $z$ it must be the case that the $\prec_{T}$-minimal covering vertex $v$ of $R_i(\bar y_i)$ occurs in this subtree because $z\in\chi(v)$. Hence, the sets $A$ of mappings that we choose in an accepting computation map the atoms of the query to facts of the database in a consistent way, and this induces a homomorphism $h$ from $Q$ to $D$. since we choose mappings that are coherent with the mapping $\bar x\mapsto\bar c$ of the answer variables of $Q$, we get that $\bar c\in Q(D')$.

\medskip
\noindent
\paragraph{The mapping $\mu$ is injective.} Assume that there are two valid outputs $O = (V,E,\lambda')$ and $O' = (V',E',\lambda'')$ of $M^k_R$ on $(D,\dep,Q,H,\bar c)$ such that $\mu(O) = \mu(O')$. Note that the order in which the algorithm goes through relations $R_{i_j}$ in $\lambda(v)$, for some vertex $v$ in $T$, and blocks $B$ of $R_{i_j}$ is arbitrary, but fixed. Further, we only produce one label per block, since the choices for each block of some atom $R_{i_j}(\bar y_i)$ are only made in the $\prec_{T}$-minimal covering vertex of $R_i(\bar y_i)$. Now, since $Q$ is self-join-free, no conflicting choices can be made on two different nodes for the same relation. Since every block corresponds to exactly one vertex in the output, and the order in which the blocks are dealt with in the algorithm is fixed, all outputs are the same up to the labeling function, i.e., $V=V'$ and $E=E'$ (up to variable renaming) and we only need to argue that $\lambda'=\lambda''$ to produce a contradiction to the assumption that $O\neq O'$. To this end, see again that every block of $D$ corresponds to one vertex of $V$, which is labeled with either an atom of that block or the symbol $\bot$ (to represent that no atom should be kept in the repair). Fix some arbitrary vertex $v\in V$. Let $B_v$ be the block that corresponds to the vertex $v$. If $\lambda'(v)=\bot$, we know that $\mu(O)\cap B_v = \emptyset$. Since $\mu(O) = \mu(O')$, we know that also $\mu(O')\cap B_v = \emptyset$ and thus $\lambda''(v)=\bot$ (if $\lambda''(v)=\alpha$ for some $\alpha\in B_v$, then $\mu(O')\cap B_v = \alpha$). On the other hand, if $\lambda'(v)=\alpha$ for some $\alpha\in B_v$, with an analogous argument, we have that also $\lambda''(v)=\alpha$ as needed.

\medskip
\noindent
\paragraph{The mapping $\mu$ is surjective.} Consider an arbitrary operational repair $D'\in \{D' \in \opr{D}{\dep} \mid \bar c\in Q(D')\}$. We need to show that there exists a valid output $O = (V,E,\lambda')$ of $M^k_R$ on $(D,\dep,Q,H,\bar c)$ such that $\mu(O)=D'$. Since $D' \in \opr{D}{\dep}$ such that $\bar c\in Q(D')$, we know that there exists a homomorphism $h$ from $Q(\bar x)$ to $D'$ such that $h(\bar x) = \bar c$. This means that in line~2 of the algorithm, we can always guess the set $A'$ according to $h$, since this will guarantee that $A' \cup A \cup \{\bar x \mapsto \bar c\}$ is coherent. Then, when traversing the hypertree decomposition (and the relations in each node and the blocks in each relation), we can always choose the labels according to $D'$. Note also that since every relation of $D$ occurs in $Q$, every block of $D$ is considered in the algorithm. Fix some arbitrary block $B$ in $D$. If $|B|=1$, then $B$ does not participate in a violation and it will always be present in every operational repair $D'$. Then, when $B$ is considered in Algorithm~\ref{alg:repairs}, we set $\alpha$ to the only fact in $B$ in line~6 and label the corresponding node in the output with $\alpha$. If $|B|\geq 2$ and some fact $R_{i_j}(\bar{c_j})$ of $B$ is in the image of the homomorphism $h$, since we chose $A'$ in such a way that it is coherent with $h$, we guess $\alpha$ as the fact $R_{i_j}(\bar{c_j})$ in line~7 of the algorithm. Otherwise, we guess $\alpha$ as the fact that is left in $B$ in $D'$ (if it exists), i.e., $\alpha:=B\cap D'$, if $B\cap D'\neq\emptyset$, and $\alpha:=\bot$ otherwise. It should now be clear that for an output $O$ generated in such a way, we have that $\mu(O)=D'$.

\section{The Case of Uniform Sequences}

Considering all the operational repairs to be equally important is a self-justified choice. However, as discussed in~\cite{CLPS22}, this does not take into account the support in terms of the repairing process. In other words, an operational repair obtained by very few repairing sequences is equally important as an operational repair obtained by several repairing sequences. This is rectified by the notion of sequence relative frequency, which is the percentage of repairing sequences that lead to an operational repair that entails a candidate answer.
For a database $D$, a set $\dep$ of keys, a CQ $Q(\bar x)$, and $\bar c \in \adom{D}^{|\bar x|}$, the {\em sequence relative frequency} of $\bar c$ w.r.t.~$D$, $\dep$, and $Q$ is
\begin{eqnarray*}
	\mathsf{RF}^{\us}(D,\dep,Q,\bar c) &=& \frac{|\{s \in \crs{D}{\dep} \mid \bar c \in Q(s(D))\}|}{|\crs{D}{\dep}|}.
\end{eqnarray*}

In this case, the problem of interest in the context of uniform operational CQA, focusing on {\em primary keys}, is defined as follows: for 
a class $\mathsf{Q}$ of CQs (e.g., $\sjf$ or  $\ghw_k$ for $k > 0$),

\medskip

\begin{center}
	\fbox{\begin{tabular}{ll}
			{\small PROBLEM} : & $\ocqa^\us[\mathsf{Q}]$
			\\
			{\small INPUT} : & A database $D$, a set $\dep$ of primary keys,\\
			& a query $Q(\bar x)$ from $\mathsf{Q}$, a tuple $\bar c \in \adom{D}^{|\bar x|}$.
			\\
			{\small OUTPUT} : &  $\mathsf{RF}^\us(D,\dep,Q,\bar c)$.
	\end{tabular}}
\end{center}

\medskip

\noindent As in the case of uniform repairs, we can show that the above problem is hard, even for self-join-free CQs of bounded generalized hypertreewidth. In particular, with $\sjf \cap \ghw_k$ for $k > 0$ being the class of self-join-free CQs of generalized hypertreewidth $k$, we show that:

\def\theocqasequencesexact{
	For every $k > 0$, $\ocqa^\us[\sjf \cap \ghw_k]$ is $\sharp ${\rm P}-hard.
}

\begin{theorem}\label{the:ocqa-us-exact}
\theocqasequencesexact
\end{theorem}
\begin{proof}
    The fact that $\ocqa^\us[\sjf \cap \ghw_k]$ is $\sharp ${\rm P}-hard is proven similarly to the uniform repairs case. In particular, we use the exact same construction and simply show that the two relative frequencies, i.e., the repair relative frequency and the sequence relative frequency coincide. Let us give more details.

Fix an arbitrary $k>0$. We show that $\ocqa^\us[\sjf \cap \ghw_k]$ is $\sharp ${\rm P}-hard via a polynomial-time Turing reduction from $\sharp H\text{-}\mathsf{Coloring}$, where $H$ is the graph employed in the first part of the proof.
%
%
For a given graph $G$, let $S_k$, $D_G^k$, $\dep$, and $Q_k$ be the schema, database, set of primary keys, and Boolean self-join-free CQ of generalized hypertreewidth $k$, respectively, from the construction in the proof for the uniform repairs case.
We show that
\[
\mathsf{RF}(D_G^k,\dep,Q_k,())\ =\ \mathsf{RF}^\us(D_G^k,\dep,Q_k,()),
\] 
which implies that the polynomial-time Turing reduction from $\sharp H\text{-}\mathsf{Coloring}$ to $\ocqa[\sjf \cap \ghw_k]$ is a polynomial-time Turing reduction from $\sharp H\text{-}\mathsf{Coloring}$ to $\ocqa^\us[\sjf \cap \ghw_k]$.
Recall that
\[
\mathsf{RF}^\us(D_G^k,\dep,Q_k,())\ =\ \frac{|\{s \in \crs{D_G^k}{\dep} \mid s(D) \models Q_k\}|}{|\crs{D_G^k}{\dep}|}.
\]
By construction of $D_G^k$, each operational repair $D \in \opr{D_G^k}{\dep}$ can be obtained via $|V_G|!$ complete sequences of $\crs{D_G^k}{\dep}$. Thus, 
\begin{eqnarray*}
\mathsf{RF}^\us(D_G^k,\dep,Q_k,()) &=& \frac{|\{D \in \opr{D_G^k}{\dep} \mid D \models Q_k\}| \cdot |V_G|!}{|\opr{D_G^k}{\dep}| \cdot |V_G|!}\\
&=& \frac{|\{D \in \opr{D_G^k}{\dep} \mid D \models Q_k\}|}{|\opr{D_G^k}{\dep}|}.
\end{eqnarray*}
The latter expression is precisely $\mathsf{RF}(D_G^k,\dep,Q_k,())$, as needed.
\end{proof}

With the above intractability result in place, the question is again whether the problem of interest is approximable, that is, whether it admits an FPRAS. An FPRAS for $\ocqa^\us[\mathsf{Q}]$ is a randomized algorithm $\mathsf{A}$ that takes as input a database $D$, a set $\dep$ of primary keys, a query $Q(\bar x)$ from $\mathsf{Q}$, a tuple $\bar c \in \adom{D}^{|\bar x|}$, $\epsilon > 0$ and $0 < \delta < 1$, runs in polynomial time in $||D||$, $||\dep||$, $||Q||$, $||\bar c||$, $1/\epsilon$ and $\log(1/\delta)$, and produces a random variable $\mathsf{A}(D,\dep,Q,\bar c,\epsilon,\delta)$ such that
$
\text{\rm Pr}(|\mathsf{A}(D,\dep,Q,\bar c,\epsilon,\delta) - \mathsf{RF}^\us(D,\dep,Q,\bar c)|\ \leq\ \epsilon \cdot \mathsf{RF}^\us(D,\dep,Q,\bar c))\ \geq\
1-\delta.
$
As in the case of uniform repairs, the answer to the above question is negative, even if we focus on self-join-free CQs or CQs of bounded generalized hypertreewidth.

\def\theocqausapx{
	Unless  ${\rm RP} = {\rm NP}$,
\begin{enumerate}
	\item There is no FPRAS for $\ocqa^\us[\sjf]$.
	
	\item For every $k > 0$, there is no FPRAS for $\ocqa^\us[\ghw_k]$.
\end{enumerate}
}

\begin{theorem}\label{the:ocqa-us-apx}
\theocqausapx
\end{theorem}
\begin{proof}
We prove each item of the theorem separately.

\noindent\paragraph{Item (1).}
 We use the same proof strategy and construction as in the uniform repairs case. Consider the decision version of $\ocqa^\us[\sjf]$, dubbed $\mathsf{Pos}\ocqa^\us[\sjf]$, defined as expected:

\medskip

\begin{center}
	\fbox{\begin{tabular}{ll}
			{\small PROBLEM} : & $\mathsf{Pos}\ocqa^\us[\sjf]$
			\\
			{\small INPUT} : & A database $D$, a set $\dep$ of primary keys,\\
			& a query $Q(\bar x)$ from $\mathsf{SJF}$, a tuple $\bar c \in \adom{D}^{|\bar x|}$.
			\\
			{\small QUESTION} : & $\mathsf{RF}^\us(D,\dep,Q,\bar c)>0$.
	\end{tabular}}
\end{center}

\medskip

\noindent We show that $\mathsf{Pos}\ocqa^\us[\sjf]$ is NP-hard, which in turn implies that $\ocqa^\us[\sjf]$ does not admit an FPRAS, under the assumption that ${\rm RP}$ and ${\rm NP}$ are different.
For an undirected graph $G$, let $\ins{S}_G$, $D_G$, and $Q_G$ be the schema, database, and Boolean CQ from the construction in the proof of the uniform repairs case. Again, the set $\dep$ of keys is considered to be empty. Hence, we only have the empty repairing sequence, denoted $\varepsilon$, with $\varepsilon(D_G) = D_G$. It is easy to see that $\mathsf{RF}(D_G,\dep,Q_G,\bar c)=\mathsf{RF}^\us(D_G,\dep,Q_G,\bar c)$. This implies that $\mathsf{RF}^\us(D_G,\dep,Q_G,\bar c)=1>0$ iff $G$ is 3-colorable, as needed.

\noindent\paragraph{Item (2).}
Fix $k>0$. For proving the desired claim, it suffices to show that
\[
\mathsf{RF}(D_\varphi^k,\dep,Q_\varphi^k,())\ =\ \mathsf{RF}^\us(D_\varphi^k,\dep,Q_\varphi^k,()).
\]
We can then easily use the same construction from the proof for the uniform repairs case to show that there cannot exist an FPRAS for $\ocqa^\us[\ghw_k]$, as this would imply the existence of an FPRAS for $\sharp \mathsf{MON2SAT}$, which contradicts the assumption ${\rm NP}\neq{\rm RP}$.

For a Pos2CNF formula $\varphi$, let $\ins{S}_\varphi$, $D_\varphi^k$, $\dep$, and $Q_\varphi^k$ be the schema, database, set of keys, and CQ, respectively, from the construction in the proof for the uniform repairs case. Further, let $n=\var{\varphi}$. Since every complete repairing sequence of $D_\varphi^k$ consists of $n$ operations, one to resolve each violation of the form $\{V(v,0),V(v,1)\}$ for every $v\in\var{\varphi}$, and the fact that these operations can be applied in any order, it is easy to verify that 
\[
|\crs{D_\varphi^k}{\dep}|\ =\ |\opr{D_\varphi^k}{\dep}| \cdot n!
\]
and  
\begin{eqnarray*}
&&|\{s \in \crs{D_\varphi^k}{\dep} \mid s(D_\varphi^k) \models Q_\varphi^k\}|\\
&=&|\{D \in \opr{D_\varphi^k}{\dep} \mid D \models Q_\varphi^k\}| \cdot n!.
\end{eqnarray*}
Hence,
\begin{eqnarray*}
\mathsf{RF}^\us(D_\varphi^k,\dep,Q_\varphi^k,()) &=& \frac{|\{D \in \opr{D_\varphi^k}{\dep} \mid D \models Q_\varphi^k\}| \cdot n!}{|\opr{D_\varphi^k}{\dep}| \cdot n!}\\
&=& \frac{|\{D \in \opr{D_\varphi^k}{\dep} \mid D \models Q_\varphi^k\}|}{|\opr{D_\varphi^k}{\dep}|}\\
&=& \frac{\sharp\varphi}{3^n}\\
&=& \mathsf{RF}(D_\varphi^k,\dep,Q_\varphi^k,()),
\end{eqnarray*}
and the claim follows.
\end{proof}

Given Theorem~\ref{the:ocqa-us-apx}, the key question is whether we can provide an FPRAS if we focus on self-join-free
CQs of bounded generalized hypertreewidth. In the rest of this section, we provide an affirmative answer to this question.

\def\themainusfpras{
	For every $k > 0$, $\ocqa^\us[\sjf \cap \ghw_k]$ admits an FPRAS.
}

\begin{theorem}\label{the:main-us-fpras}
\themainusfpras
\end{theorem}

\smallskip

As in the case of uniform repairs, the denominator of the ratio in question, can be computed in polynomial time~\cite{CLPS22}; hence, to establish Theorem~\ref{the:main-us-fpras}, it suffices to show that the numerator of the ratio can be efficiently approximated. In fact, it suffices to show that the following auxiliary counting problem admit an FPRAS.
For $k>0$,

\begin{center}
	\fbox{\begin{tabular}{ll}
			{\small PROBLEM} : & $\sharp\mathsf{Sequences}[k]$
			\\
			{\small INPUT} : & A database $D$, a set $\dep$ of primary keys,\\
			& a query $Q(\bar x)$ from $\sjf$, a generalized hypertree\\ & decomposition of $Q$ of width $k$, $\bar c \in \adom{D}^{|\bar x|}$.
			\\
			{\small OUTPUT} : &  $|\{s \in \crs{D}{\dep} \mid \bar c \in Q(s(D))\}|$.
	\end{tabular}}
\end{center}

We prove the following.

\begin{theorem}\label{the:in-us-spantl}
	For every $k>0$, $\sharp\mathsf{Sequences}[k]$ is in $\spantl$.
\end{theorem}

To prove the above result, for $k>0$, we need to devise a procedure, which can be implemented as a well-behaved ATO $M_S^k$,  that take as input a database $D$, a set $\dep$ of primary keys, a CQ $Q(\bar x)$ from $\sjf$, a generalized hypertree decomposition $H$ of $Q$ of width $k>0$, and a tuple $\bar c \in \adom{D}^{|\bar x|}$, such that
\begin{eqnarray*}
\mathsf{span}_{M_{S}^{k}}(D,\dep,Q,H,\bar c) &=& |\{s \in \crs{D}{\dep} \mid \bar c \in Q(s(D))\}|.
\end{eqnarray*}

As in the case of uniform repairs, we are going to assume, w.l.o.g., that the input to this procedure is in normal form. We can show that, given a database $D$, a set $\dep$ of primary keys, a CQ $Q(\bar x)$ from $\sjf$, a generalized hypertree decomposition $H$ of $Q$ of width $k>0$, and a tuple $\bar c \in \adom{D}^{|\bar x|}$, we can convert in logarithmic space the triple $(D,Q,H)$ into a triple $(\hat{D},\hat{Q},\hat{H})$, where $\hat{H}$ is a generalized hypertree decomposition of $\hat{Q}$ of width $k+1$, that is in normal form while
$|\{s \in \crs{D}{\dep} \mid \bar c \in Q(s(D))\}| = |\{s \in \crs{\hat{D}}{\dep} \mid \bar c \in \hat{Q}(s(\hat{D}))\}|$. Formally,

\def\pronormalformus{
	Consider a database $D$, a set $\dep$ of primary keys $\dep$, a CQ $Q(\bar x)$ from $\sjf$, a generalized hypertree decomposition $H$ of $Q$ of width $k$, and $\bar c \in \adom{D}^{|\bar x|}$. There exists a database $\hat{D}$, a CQ $\hat{Q}(\bar x)$ from $\sjf$, and a generalized hypertree decomposition $\hat{H}$ of $\hat{Q}$ of width $k+1$ such that:
	\begin{itemize}
		\item $(\hat{D},\hat{Q},\hat{H})$ is in normal form,
		\item $|\{s \in \crs{D}{\dep} \mid \bar c \in Q(s(D))\}| = |\{s \in \crs{\hat{D}}{\dep} \mid \bar c \in \hat{Q}(s(\hat{D}))\}|$, and
		\item $(\hat{D},\hat{Q},\hat{H})$ is logspace computable in the combined size of $D,\dep,Q,H,\bar c$.
	\end{itemize}
}

\begin{proposition}\label{pro:normal-form-us}
	\pronormalformus
\end{proposition}
\begin{proof}
    The proof is very similar to the proof of Theorem~\ref{pro:normal-form} (in particular, we use the exact same construction), with the only difference being the last part, where we prove that $(\hat{D},\hat{Q},\hat{H})$ preserves the counts. Here, we need to show that \[|\{s \in \crs{D}{\dep} \mid \bar c \in Q(s(D))\}| = |\{s \in \crs{\hat{D}}{\dep} \mid \bar c \in \hat{Q}(s(\hat{D}))\}|.\]
    To prove the claim, we again rely on the notion of key violation. As discussed in the proof of Theorem~\ref{pro:normal-form}, we have that $\viol{D}{\dep} = \viol{\hat{D}}{\dep}$. Moreover, $\hat{Q}$ contains all atoms of $Q$, and in addition, for each fact $R(\bar t) \in C$, a single atom $R(\bar y)$, where each variable in $\bar y$ occurs exactly ones in $\hat{Q}$. The above observations imply that:
\begin{enumerate}
	\item $\crs{D}{\dep} = \crs{\hat{D}}{\dep}$;
	\item for every database $D'$, $D' \in \opr{D}{\dep}$ iff $D' \cup C \in \opr{\hat{D}}{\dep}$;
	\item for each $D' \in \opr{D}{\dep}$, $\bar c \in Q(D)$ iff $\bar c \in \hat{Q}(D' \cup C)$, where $D' \cup C \in \opr{\hat{D}}{\dep}$.
\end{enumerate}
All three items together imply $|\{s \in \crs{D}{\dep} \mid \bar c \in Q(s(D))\}| = |\{s \in \crs{\hat{D}}{\dep} \mid \bar c \in \hat{Q}(s(\hat{D}))\}|$. This concludes our proof.
\end{proof}

We now proceed to prove Theorem~\ref{the:in-us-spantl}. This is done via the procedure $\mathsf{Seq}[k]$, depicted in Algorithm~\ref{alg:sequences}. 
We show that:

\def\lemsequencesato{
	For every $k > 0$, the following hold:
	\begin{enumerate}
		\item $\mathsf{Seq}[k]$ can be implemented as a well-behaved ATO $M_S^k$.
		\item For a database $D$, a set $\dep$ of primary keys, a CQ $Q(\bar x)$ from $\sjf$, a generalized hypertree decomposition $H$ of $Q$ of width $k$, and a tuple $\bar c \in \adom{D}^{|\bar x|}$, where $(D,Q,H)$ is in normal form, $\mathsf{span}_{M_S^k}(D,\dep,Q,H,\bar c) = |\{s \in \crs{D}{\dep} \mid \bar c \in Q(s(D))\}|$.
	\end{enumerate}
}

\begin{lemma}\label{lem:sequences-ato}
	\lemsequencesato	
\end{lemma}

Theorem~\ref{the:in-us-spantl} readily follows from Lemma~\ref{lem:sequences-ato}. The rest of this section is devoted to discussing the procedure $\mathsf{Seq}[k]$ and proving that Lemma~\ref{lem:sequences-ato} holds. But first we need some auxiliary notions.

\OMIT{
\def\thmsequencesspantl{
	For every $k > 0$, $\sharp\mathsf{Sequences}[k] \in \spantl$.
}

\begin{theorem}\label{thm:sequences-spantl}
\thmsequencesspantl
\end{theorem}

To prove the above theorem, we introduce a procedure (i.e., Algorithm~\ref{alg:sequences}) describing the computation of an ATO $M$ with input a database $D$, a set $\dep$ of primary keys, a query $Q(\bar x)$ from $\sjf$, a generalized hypertree decomposition $H$ of $Q$ of width $k$, and a tuple $\bar c \in \adom{D}^{|\bar x|}$, such that $(D,Q,H)$ is in normal form. With Algorithm~\ref{alg:sequences} in place, our goal is to prove the following lemma.

\def\lemsequencesato{
	For every $k > 0$, the following hold:
	\begin{itemize}
		\item Algorithm~\ref{alg:sequences} can be implemented as a well-behaved ATO $M$, and
		\item for every database $D$, set $\dep$ of primary keys, query $Q(\bar x)$, generalized hypertree decomposition $H$ of $Q$ of width $k$, and $\bar c \in \adom{D}^{|\bar x|}$ with $(D,Q,H)$ in normal form, $\mathsf{span}_M(D,\dep,Q,H,\bar c) = \sharp\mathsf{Sequences}[k](D,\dep,Q,H,\bar c)$.
	\end{itemize}
}

\begin{lemma}\label{lem:sequences-ato}
\lemsequencesato	
\end{lemma}

Theorem~\ref{thm:sequences-spantl} immediately follows from Lemma~\ref{lem:sequences-ato}, Proposition~\ref{pro:normal-form}, and Proposition~\ref{pro:logspace-closure}.
The rest of this section is devoted to discuss Algorithm~\ref{alg:sequences}, and proving Lemma~\ref{lem:sequences-ato}. We first need some auxiliary notions.
}

\medskip
\noindent \paragraph{Auxiliary Notions.}
Consider a database $D$ and a set $\dep$ of primary keys over a schema $\ins{S}$. 
Moreover, for an integer $n>0$ and $\alpha \in  (\{R(\bar c) \mid R/m \in \ins{S} \text{ and } \bar c \in \adom{D}^m\} \cup \{\bot\})$, we define the set
\[
	\mathsf{shape}(n,\alpha)\ =\ \begin{cases}
		\{-1,-2\} & \text{if } n>1, \\
		\{-1\} & \text{if } n=1 \text{ and } \alpha \neq \bot,\\
		\emptyset & \text{otherwise.}
	\end{cases}
\]
The above set essentially collects the ``shapes'' of operations (i.e., remove a fact or a pair of facts) that one can apply on a block of the database, assuming that only $n$ facts of the block must be removed overall, and either the block must remain empty (i.e., $\alpha = \bot$), or the block must contain exactly one fact (i.e., $\alpha \neq \bot$).
Finally, for an integer $n>0$ and $i \in \{-1,-2\}$, let $\#\mathsf{opsFor}(n,i) = n$, if $i = -1$, and $\#\mathsf{opsFor}(n,i) = \frac{n(n-1)}{2}$, if $i = -2$. Roughly, assuming that $n$ facts of a block must be removed, $\#\mathsf{opsFor}(n,i)$ is the total number of operations of the shape determined by $i$ that can be applied.

\setlength{\textfloatsep}{2.5em}
\begin{algorithm}[t]
	\SetInd{0.7em}{0.7em}
	\DontPrintSemicolon
	\SetArgSty{textnormal}
	\KwIn{A database $D$, a set $\dep$ of primary keys, a CQ $Q(\bar x)$ from $\sjf$, a generalized hypertree decomposition $H=(T,\chi,\lambda)$ of $Q$ of width $k$, and $\bar c \in \adom{D}^{|\bar x|}$, where $(D,Q,H)$ is in normal form}
	\vspace{2mm}
	{$v := \mathsf{root}(T)$; $A := \emptyset$, $b := 0$;\\}
	{Guess $N \in [|D|]$;\\}
	\vspace{1mm}
	{Assuming $\lambda(v) = \{ R_{i_1}(\bar y_{i_1}), \ldots, R_{i_\ell}(\bar y_{i_\ell})\}$, guess a set $A' =$ $\{ \bar y_{i_1} \mapsto \bar c_1,\ldots,\bar y_{i_\ell} \mapsto \bar c_\ell\}$, with $R_{i_j}(\bar c_j) \in D$ for $j \in [\ell]$, and verify $A \cup A' \cup \{\bar x \mapsto \bar c\}$ is coherent; if not, \textbf{reject};\\}\label{line:begin-sequences}
	\vspace{1mm}
	\For{$j=1, \ldots, \ell$}{
		\If{$v$ is the $\prec_T$-minimal covering vertex for $R_{i_j}(\bar y_{i_j})$}{
			\ForEach{$B \in \block{\dep}{R_{i_j},D}$}{
				
				\lIf{$B = \{\beta\}$}{$\alpha := \beta$;}
				\lElseIf{$R_{i_j}(\bar c_j) \in B$}{$\alpha := R_{i_j}(\bar c_j)$;}
				\lElse{Guess $\alpha \in B \cup \{\bot\}$;}
				
				\lIfElse{$\alpha = \bot$}{$n := |B|$}{$n:= |B| - 1$;}\label{line:being-inner-seq}
				
				{$b' := b$;\\}
				
				\While{ $n>0$ }{
					\lIf{$\mathsf{shape}(n,\alpha) = \emptyset$}{\textbf{reject};}\lElse{Guess $-g \in \mathsf{shape}(n,\alpha)$;}
					{Guess $p \in [\#\mathsf{opsFor}(n,-g)]$;\\}
					{Label with $(-g,p)$;\\}
					{$b := b + 1$; $N := N -1$; $n := n - g$;\\}
				}
				{Guess $p \in \left[b \choose b'\right]$;\\}
				{Label with $(\alpha,p)$;\\}\label{line:end-inner-seq}
			}
		}
	}
	
	\If{$v$ is not a leaf of $T$}{
		{Guess $p \in \{0,\ldots,N\}$;\\}\label{line:being-outer-seq}
		Assuming $u_1,u_2$ are the (only) children of $v$ in $T$, universally guess $i \in \{1,2\}$;\\
		\lIf{$u_i$ is the first child of $v$}{$N := p$;}
		\lElse{$N := N - p$; $b := b + p$;}\label{line:end-outer-seq}
		{$v := u_i$; $A := A'$;\\}
		{\textbf{goto} line \ref{line:begin-sequences};\\}
	}\lElse{\textbf{if} $N = 0$ \textbf{then} \textbf{accept}; \textbf{else} \textbf{reject};}
	
	\caption{The alternating procedure $\mathsf{Seq}[k]$}\label{alg:sequences}
\end{algorithm}

\medskip
\noindent \paragraph{The Procedure $\mathsf{Seq}[k]$.}
Let $M_S^k$ be the ATO underlying $\mathsf{Seq}[k]$. As for $\mathsf{Rep}[k]$, we assume that $M_S^k$ immediately moves from the initial state to a non-labeling state.
$\mathsf{Seq}[k]$ proceeds similarly to $\mathsf{Rep}[k]$, i.e., it traverses the generalized hypertree decomposition $H$, where for each vertex $v$, it guesses a set of tuple mappings $A'$ coherent with $\{\bar x \mapsto \bar c\}$ and with the set of tuple mappings of the previous vertex, and non-deterministically chooses which atom (if any) to keep from each block relative to the relation names in $\lambda(v)$.

The key difference between $\mathsf{Seq}[k]$ and $\mathsf{Rep}[k]$ is that, for a certain operational repair $D'$ that $\mathsf{Seq}[k]$ internally guesses, $\mathsf{Seq}[k]$ must output a number of trees that coincides with the number of complete $(D,\dep)$-repairing sequences $s$ such that $s(D) = D'$. To achieve this, at the beginning of the computation, $\mathsf{Seq}[k]$ guesses the length $N$ of the repairing sequence being computed. Then, whenever it chooses some $\alpha \in B \cup \{\bot\}$ for a certain block $B$ (lines~7-9), then, in lines~10-17, it will non-deterministically output a \emph{path} of the form $L_1 \rightarrow \cdots \rightarrow L_n$, where $L_1,\ldots,L_n$ are special symbols describing the ``shape'' of operations to be performed in order to leave in $B$ only the fact $\alpha$, if $\alpha \neq \bot$, or no facts, if $\alpha = \bot$.
%
%
Each symbol $L_i$ is a pair of the form $(-g,p)$, where $-g \in \{-1,-2\}$ denotes whether the operation removes one or a pair of facts. For example, a path $(-1,n_1) \rightarrow (-2,n_2) \rightarrow (-2,n_3) \rightarrow (-1,n_4)$ is a template specifying that the first operation on the block should remove a single fact, then the next two operations should remove pairs of facts, and the last operation should remove a single fact.

However, a symbol in $\{-1,-2\}$ alone only denotes a \emph{family} of operations over the block $B$, i.e., there may be multiple ways for translating such a symbol into an actual operation, depending on the size of $B$ and the value of $\alpha$. 
%
%
Thus, in a pair $(-g,p)$, the integer $p$ acts as an \emph{identifier} of the concrete operation being applied having the shape of $-g$.\footnote{Identifiers are used in place of actual operations as the procedure would not be able to remember the whole ``history'' of applied operations in logarithmic space.}
The last important observation is that $\mathsf{Seq}[k]$ outputs a template of \emph{all} operations of a certain block, before moving to another block, and further, blocks are considered in a fixed order. However, a complete $(D,\dep)$-repairing sequence might interleave operations coming from different blocks in an arbitrary way as each block is independent. The procedure accounts for this in lines 18-19 as follows.
Let $k_1,\dots,k_m$ be the number of operations applied over each individual block during a run of $\mathsf{Seq}[k]$ (where $m$ is the total number of blocks of $D$). Any arbitrary interleaving of the above operations gives rise to a complete $(D,\dep)$-repairing sequence. Hence, the number of complete $(D,\dep)$-repairing sequences that uses those numbers of operations is the multinomial coefficient
\[{k_1+\dots+k_m\choose k_1,\dots,k_m}={k_1\choose 0}\times {k_1+k_2\choose k_1}\times\dots\times {k_1+\dots+k_m\choose k_1+\dots+k_{m-1}}.\]
Hence, when $\mathsf{Seq}[k]$ completely ``repairs'' the $i$-th block of $D$ (i.e., it reaches line~18), in lines~18-19 it ``amplifies'' the number of trees by guessing an integer $p \in \left[{b\choose b'}\right]$, with $b = k_1 + \cdots + k_i$ and $b' = k_1 + \cdots k_{i-1}$, and by running the statement ``Label with $(\alpha,p)$''.
%
%
We can show that $b \choose b'$, and actually any number $p \in \left[{b \choose b'}\right]$, can be computed in logspace using ideas from~\cite{Chiu2021}. However, the number of bits required to write $p \in \left[{b \choose b'}\right]$ in the labeling tape is polynomial. So, by ``Label $(\alpha,p)$'' we actually mean that $M_S^k$ outputs a \emph{path} of polynomial length of the form $(\alpha,b_1) \rightarrow (\alpha,b_2) \rightarrow \cdots \rightarrow (\alpha,b_n)$, where $b_i$ is the $i$-th bit of the binary representation of $p$.

With all blocks relative to a vertex $v$ completely processed, the procedure moves in parallel to the two children of $v$. To this end, it guesses how many operations are going to be applied in the sequences produced in the left and the right branch, respectively. It does this by guessing an integer $p \in \{0,\ldots,N\}$, which essentially partitions $N$ as the sum of two numbers, i.e., $p$ and $N-p$, and then adapts the values of $N$ (number of operations left) and $b$ (number of operations applied so far) in each branch, accordingly. A computation is accepting if in each of its leaves all required operations have been applied (i.e., $N = 0$).


\begin{example}\label{ex:sequences-algo}
	Consider the database $D$, the set $\dep$ of primary keys, the Boolean CQ $Q$, and the generalized hypertree decomposition $H$ of $Q$ used in Example~\ref{ex:repairs-algo}.
	A possible accepting computation of the ATO $M_S^2$ described by $\mathsf{Seq}[2]$ on input $D$, $\dep$, $Q$, $H$, and the empty tuple $()$, is the one corresponding to the operational repair $D' = \{P(a_1,c), S(c,d), T(d,a_1), U(c,f), U(h,i)\}$, with $D' \models Q$, that outputs a tree of one of the following forms:
	
	\medskip
	\noindent
	\includegraphics[width=\textwidth]{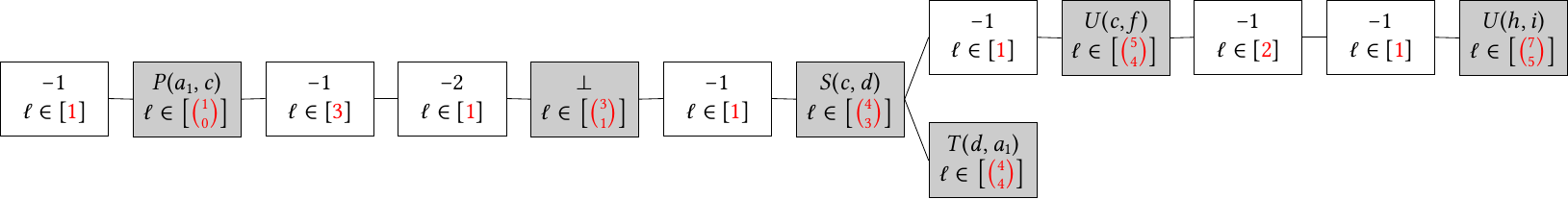}
	\ \\
	\noindent
	\includegraphics[width=1\textwidth]{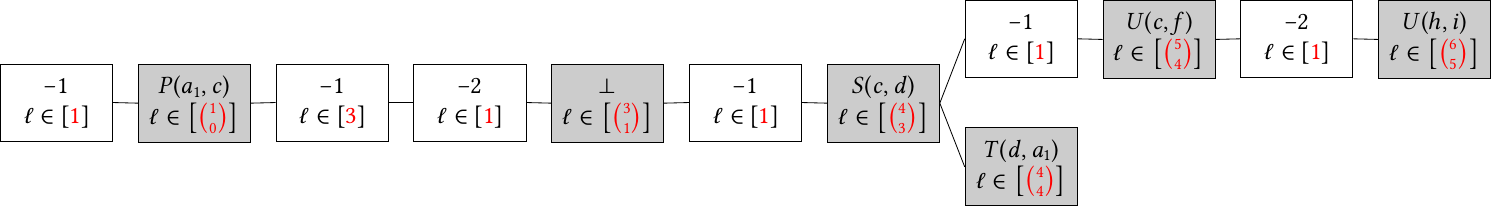}
	%
	
	\noindent Note that, for simplicity, we omit the root with label $\epsilon$, and each pair that labels a node is presented by showing each of its elements on a separate line inside the node. Moreover, the nodes with a gray background containing a pair $(\alpha,\ell)$ are actually a shorthand for paths of the form $(\alpha,b_1) \rightarrow (\alpha,b_2) \rightarrow \cdots \rightarrow (\alpha,b_n)$, where $b_i$ is the $i$-th bit of $\ell$.
	Different accepting computations of the procedure that internally guess the same operational repair $D'$ will output a tree obtained from one of the above two forms, having each $\ell$ equal to some value in the corresponding set. Hence, the total number of such trees is $s_1 + s_2$, where $s_1$ (resp., $s_2$) is the product of all ``amplifying factors'' written in red in each label of the first (resp., second) tree in the figure above, i.e., 
	\[
	\begin{array}{ll}
		s_1 = & 1 \times {1 \choose 0} \times 3 \times 1 \times {3 \choose 1} \times 1 \times {4 \choose 3} \times {4 \choose 4} \times 1 \times {5 \choose 4} \times 2 \times 1 \times {7 \choose 5}, \\\ \\
		s_2 = & 1 \times {1 \choose 0} \times 3 \times 1 \times {3 \choose 1} \times 1 \times {4 \choose 3} \times {4 \choose 4} \times 1 \times {5 \choose 4} \times 1 \times {6 \choose 5}.
	\end{array}
	\]
	One can verify that $s_1 + s_2$ is precisely the number of complete $(D,\dep)$-repairing sequences $s$ such that $s(D) = D'$, as needed. \hfill\markfull
\end{example}

\OMIT{
\medskip
\noindent \paragraph{The Procedure $\mathsf{Seq}[k]$.}
As for $\mathsf{Rep}[k]$, we assume that $\mathsf{Seq}[k]$ immediately moves from the initial state to a non-labeling state.
It proceeds similarly to $\mathsf{Rep}[k]$, i.e., it performs a traversal of the generalized hypertree decomposition $H$, and during the traversal it guesses a set of mappings $A'$  that map $\bar x$ to $\bar c$ and is coherent with the set of mappings of the previous vertex. It then makes a (possibly non-deterministic) choice of which atom (if any) to keep from each block of each relation name $R_{i_j}$ occurring in $\lambda(v)$, where $v$ is the currently processed vertex of $H$.

The key difference of $\mathsf{Seq}[k]$ compared to $\mathsf{Rep}[k]$ is that, for a certain operational repair $D'$ that the algorithm internally guesses by non-deterministically choosing the final state of each block of $D$, the algorithm needs to produce a number of trees that coincides with the number of complete $(D,\dep)$-repairing sequences $s$ such that $s(D) = D'$. To achieve this, at the beginning of the computation, $\mathsf{Seq}[k]$ guesses the length $N$ of the repairing sequence being computed. Then, whenever the algorithm chooses some $\alpha \in B \cup \{\bot\}$ for a certain block $B$ of the database $D$, it will non-deterministically output a path that provides a ``template'' describing the operations to be performed on that block in order to leave in $B$ only the fact $\alpha$, if $\alpha \neq \bot$, or no facts if $\alpha = \bot$.
This template is represented in the output of the procedure as a path of nodes labeled with $-1$ or $-2$, indicating the removal of a single fact or a pair of facts, respectively.
For example, consider the sequence $-1,-2,-2,-1$. Here, the first operation on the block should remove a single fact, then the next two operations should remove pairs of facts, and the last operation should remove a single fact. After this sequence of $-1$ and $-2$ symbols, the last symbol for the block will be $\alpha$, which indicates the final state of the block, i.e., either empty or contains a single fact.

A crucial observation is that each symbol $-1$ or $-2$ actually denotes a \emph{family} of operations over the block $B$, i.e., there may be multiple ways for translating it into an actual operation. For instance, in a block of five facts, either four or five potential operations could remove a single fact, depending on the block's final state. If the block is left empty, any of the five facts can be removed. On the other hand, if a specific fact is kept, only the remaining four facts can be removed. Thus, the actual label the algorithm produces for each operation is a \emph{pair} of a symbol $-i$, with $i \in \{1,2\}$, and an integer $p$ guessed from $[\#\mathsf{opsFor}(n,-i)]$, where $\#\mathsf{opsFor}(n,-i)$ is precisely the number of possible operations of the shape $-i$ that can be applied; this accurately accounts for all the possible translations of the symbol $-i$ into an actual operation.\footnote{The reason why the algorithm produces such templates, rather than actual sequences of operations is because, given it has chosen the first $\ell$ operations for a block $B$, in order to choose the next operation, it must remember al the previous $\ell$ operations, which cannot be done in logarithmic space.}
 
The last important observation is that $\mathsf{Seq}[k]$ produces \emph{all} operations of a certain block, before moving to another block, and blocks are considered in a fixed order. However, a complete $(D,\dep)$-repairing sequence might interleave operations coming from different blocks in an arbitrary way as each block is independent from the others. The procedure accounts for this as follows.
Let $k_1,\dots,k_m$ be the number of operations applied over each individual block during a run of $\mathsf{Seq}[k]$ (where $m$ is the total number of blocks of $D$). The number of ways to create a complete $(D,\dep)$-repairing sequence using those operations is provided by the multinomial coefficient
\[
{k_1+\dots+k_m\choose k_1,\dots,k_m}\ =\ \dfrac{(k_1 + \cdots + k_m)}{k_1! \times \cdots \times k_m!}.
\]
It is well-known that such a multinomial coefficient can be represented as a product of binomial coefficients as follows:
\[{k_1+\dots+k_m\choose k_1,\dots,k_m}={k_1\choose 0}\times {k_1+k_2\choose k_1}\times\dots\times {k_1+\dots+k_m\choose k_1+\dots+k_{m-1}}.\]
Hence, when $\mathsf{Seq}[k]$ completely ``repairs'' the $i$-th block of $D$, it uses the corresponding binomial coefficient ${k_1 + \cdots + k_i} \choose {k_1 + \cdots k_{i-1}}$ to ``amplify'' the number of trees produced; the binomial coefficient of the current block is kept updated using the variables $b$ and $b'$. For instance, when the first block $B$ of $D$ has been completely processed, the algorithm guesses an integer $p \in \left[{b\choose b'}\right]$, where $b' = 0$ and $b= k_1$, and outputs a node labeled with $(\alpha,p)$, where $\alpha$ is the atom to keep in the block (encoded as a pointer), or $\bot$ if no atom should be kept. When the algorithm completes the second block, the binomial coefficient will be $b \choose b'$, where $b' = k_1$ and $b = k_1 + k_2$, and so on until all blocks determined by the current vertex $v$ of $H$ have been processed.
Note that, although one can show that the binomial coefficient $b \choose b'$, and thus, any number $p \in \left[{b \choose b'}\right]$, can be computed in logspace (this is a consequence of~\cite{Chiu2021}), the number of bits required to store $p \in \left[{b \choose b'}\right]$ is polynomial. Hence, we cannot afford to label a single node with $(\alpha,p)$, as this would violate the well-behavedness of the underlying ATO. Instead, ``Label $(\alpha,p)$'' means that the machine outputs a \emph{path} of polynomial length the form 
$(\alpha,b_1) \rightarrow (\alpha,b_2) \rightarrow \cdots \rightarrow (\alpha,b_n)$, where $b_i$ is the $i$-th bit of the binary representation of $p$.

With all blocks relative to a vertex $v$ completely processed, the procedure moves in parallel to the two children of $v$. To this end, it guesses how many operations are going to be applied in the sequences produced in the left and the right branch, respectively. It does this by guessing an integer $p \in \{0,\ldots,N\}$, which essentially partitions $N$ as the sum of two numbers, i.e., $p$ and $N-p$, and then adapts the values of $N$ and $b$ in each branch, accordingly. A computation is accepting if in each of its leaves the required number $N$ of operations have all been applied (i.e., $N = 0$).

We proceed to give an example that illustrates the above discussion concerning the procedure $\mathsf{Seq}[k]$.

\begin{example}\label{ex:sequences-algo}
Consider the database $D$, the set $\dep$ of primary keys, the Boolean CQ $Q$, and the generalized hypertree decomposition $H$ of $Q$ given in Example~\ref{ex:repairs-algo}.
A possible accepting computation of the ATO described by $\mathsf{Seq}[2]$ with input $D$, $\dep$, $Q$, $H$, and the empty tuple $()$, is the one corresponding to the operational repair $D' = \{P(a_1,c), S(c,d), T(d,a_1), U(c,f), U(h,i)\}$, with $D' \models Q$, that outputs a tree of one of the following forms:

\vspace*{2mm}
\noindent
\includegraphics[width=0.48\textwidth]{example-tree1-seq}
\ \\
\noindent
\includegraphics[width=0.44\textwidth]{example-tree2-seq}
\vspace*{2mm}

\noindent where each pair that labels a node is presented by showing each of its elements on a separate line inside the node. Moreover, the nodes with a thick border containing a pair $(\alpha,\ell)$ are actually a shorthand for paths of the form $(\alpha,b_1) \rightarrow (\alpha,b_2) \rightarrow \cdots \rightarrow (\alpha,b_n)$, where $b_i$ is the $i$-th bit of $\ell$.
Different accepting computations of the procedure that internally guess the same operational repair $D'$ will output a tree obtained from one of the above two forms, having each $\ell$ equal to some value in the corresponding set. Hence, the total number of such trees is $s_1 + s_2$, where $s_1$ (resp., $s_2$) is the product of all ``amplifying factors'' written in red in each label of the first (resp., second) tree in the figure above, i.e., 
\[
\begin{array}{ll}
s_1 = & 1 \times {1 \choose 0} \times 3 \times 1 \times {3 \choose 1} \times 1 \times {4 \choose 3} \times {4 \choose 4} \times 1 \times {5 \choose 4} \times 2 \times 1 \times {7 \choose 5}, \\\ \\
s_2 = & 1 \times {1 \choose 0} \times 3 \times 1 \times {3 \choose 1} \times 1 \times {4 \choose 3} \times {4 \choose 4} \times 1 \times {5 \choose 4} \times 1 \times {6 \choose 5}.
\end{array}
\]
One can verify that $s_1 + s_2$ is precisely the number of complete $(D,\dep)$-repairing sequences $s$ such that $s(D) = D'$, as needed. \hfill\markfull
\end{example}
}

We now proceed to prove Lemma~\ref{lem:sequences-ato}.
We start by proving that the algorithm can be implemented as a well-behaved ATO $M^k_S$, and then we show that the number of valid outputs of $M^k_S$ is precisely the number of complete repairing sequences $s$ such that $\bar c\in Q(s(D))$.

\subsection{Item (1) of Lemma~\ref{lem:sequences-ato}}
The proof of this item proceeds similarly to the one for the first item of Lemma~\ref{lem:repairs-ato}. The additional part we need to argue about is how the  ATO $M^k_S$ underlying Algorithm~\ref{alg:sequences} implements lines~\ref{line:being-inner-seq}-\ref{line:end-inner-seq}, and lines~\ref{line:being-outer-seq}-\ref{line:end-outer-seq} using logarithmic space in the working and labeling tape. Regarding lines~\ref{line:being-inner-seq}-\ref{line:end-inner-seq}, the first two lines require storing a number no larger than $|D|$, which can be easily stored in logarithmic space. Then, computing $\mathsf{shape}(n,\alpha)$ and $\sharp\mathsf{opsFor}(n,-g)$ is trivially doable in logarithmic space. Similarly, guessing the number $p \in [\sharp\mathsf{opsFor}(n,-g)]$ requires logarithmic space. The operation "Label with $(-g,p)$" requires only logarithmic space in the labeling tape, since both $-g$ and $p$ can be encoded in logarithmic space. Similarly, additions and subtractions between two logspace-encoded numbers can be carried out in logarithmic space.

The challenging part is guessing $p \in \left[{b \choose b'}\right]$ and writing $(\alpha,p)$ in the labeling tape, without explicitly storing $p$ and ${b \choose b'}$ in the working tape, and not going beyond the space bounds of the labeling tape. This is because ${b \choose b'}$ and $p$ require, in principle, a polynomial number of bits, since they encode exponentially large numbers w.r.t.\ the size of the input. To avoid the above, $M^k_S$ does the following: it guesses one bit $d$ of $p$ at the time, starting from the most significant one. Then, $M^k_S$ writes $(\alpha,d)$ in the labeling tape, moves to a labeling state, and then moves to a non-labeling state, and continues with the next bit of $p$. We now need to show how each bit of $p$ can be guessed in logarithmic space, by also guaranteeing that the bits chosen so far place $p$ in the interval $\left[{b \choose b'}\right]$. Note that
$${b \choose b'} = \dfrac{b!}{b'! \times (b-b')!} = \dfrac{b \times (b-1) \times \cdot \times (b-b'+1)}{b'!},$$
and thus ${b \choose b'}$ can be computed by performing two iterated multiplications: $B_1 = b \times (b-1) \times \cdot \times (b-b'+1)$, and $B_2 = b'! = b' \times (b'-1) \times \cdots \times 1$, of polynomially many terms, and then by dividing $B_1 / B_2$. Iterated multiplication and integer division are all doable in deterministic logspace~\cite{Chiu2021}, and thus, by composition of logspace-computable functions~\cite{ArBa09}, $M^k_S$ can compute one bit of ${b \choose b'}$ at the time, using at most logspace, and use such a bit to guide the guess of the corresponding bit of $p$.

We conclude by discussing lines~\ref{line:being-outer-seq}-\ref{line:end-outer-seq}. Guessing $p \in \{0,\ldots,N\}$ is clearly feasible in logspace, since $N \le |D|$. The remaining lines can be implemented in logarithmic space because of arguments similar to the ones used for the proof of item~(1) of Lemma~\ref{lem:repairs-ato}.

\subsection{Item (2) of Lemma~\ref{lem:sequences-ato}}
We prove that there is a one-to-one correspondence between the valid outputs of $M^k_S$ on $D,\dep,Q,H,\bar c$ and the complete repairing sequences $s$ such that $\bar c\in Q(s(D))$. Assume that $H=(T^H,\chi^H,\lambda^H)$ with $T^H=(V^H,E^H)$. In what follows, for each set $P$ of $D$-operations, we assume an arbitrary fixed order over the operations in $P$ that we denote by $\prec_P$. In addition, let $B_1,\dots,B_m$ be the order defined over all the blocks of $D$ by first considering the order $\prec_{T^H}$ over the nodes of $T^H$, then for each node $v$ of $T^H$ considering the fixed order used in Algorithm~\ref{alg:sequences} over the relation names $R_i$ such that $v$ is the $\prec_{T^H}$-minimal covering vertex of the atom $R_i(\bar y_i)$ of $Q$, and finally, for each relation name $R_i$, considering the fixed order used in Algorithm~\ref{alg:sequences} over the blocks of $\block{\dep}{R_i,D}$. For each $i\in[m]$, consider the first $i$ blocks $B_1,\dots,B_i$, and let $s_{\le i-1}$ be a sequence of operations over the blocks $B_1,\dots,B_{i-1}$ and $s_i$ be a sequence of operations over $B_i$. We assume an arbitrary fixed order over all the permutations of the operations in $s_{\le i-1},s_i$ that are coherent with the order of the operations in $s_{\le i-1}$ and with the order of operations in $s_i$ (i.e.,~all the possible ways to interleave the sequences $s_{\le i-1},s_i$), and denote it by $\prec_{s_{\le i-1},s_i}$. For simplicity, we assume in the proof that in line~19 of Algorithm~\ref{alg:sequences} we produce a single label $(\alpha,p)$ rather than a sequence of labels representing the binary representation of $p$. This has no impact on the correctness of the algorithm.

\medskip
\noindent\paragraph{Sequences to Valid Outputs.} Consider a complete repairing sequence $s$ such that $\bar c \in Q(s(D))$. We construct a valid output $O=(V',E',\lambda')$ as follows. First, for each vertex $v\in V^H$, we add to $V'$ a corresponding vertex $v'$. If there exists and edge $(v,u)\in E^H$, then we add the corresponding edge $(v',u')$ to $E'$.

Next, for every vertex $v'\in V'$ corresponding to a vertex $v\in V^H$, we replace $v'$ by a path $v'_1 \rightarrow \dots \rightarrow v'_m$ of vertices. We add the edges $(u',v'_1)$ and $(v'_m,u'_\ell)$ for $\ell\in[2]$, assuming that $u'$ is the parent of $v'$ and $u'_1,u'_2$ are the children of $v'$ in $V'$. Note that initially each node of $V'$ has precisely two children since we assume that $(D,Q,H)$ is in normal form; hence, every node of $V^H$ has two children.
The path $v'_1 \rightarrow \dots \rightarrow v'_m$ is such that for every relation name $R_i$ with $v$ being the $\prec_{T^H}$-minimal covering vertex of the atom $R_i(\bar y_i)$ of $Q$ (note that for every $v$ there is at least one such relation name since $(D,Q,H)$ is in normal form), and for every block $B\in \block{\dep}{R_i,D}$, there is subsequence $v'_{i_1} \rightarrow \dots \rightarrow v'_{i_t}$ where:
\begin{itemize}
\item $t-1$ is the number of operations in $s$ over the block $B$,
\item $\lambda'(v'_{i_r})=(-1,p)$ if the $r$-th operation of $s$ over $B$ is the $p$-th operation according to $\prec_{P}$, where $P$ is the set of operations removing a single fact $f\not\in s(D)$ of $B'$, and $B'\subseteq B$ is obtained by applying the first $r-1$ operations of $s$ over $B$, and excluding the fact of $s(D)\cap B$ (if such a fact exists),
\item $\lambda'(v'_{i_r})=(-2,p)$ if the $r$-th operation of $s$ over $B$ is the $p$-th operation according to $\prec_{P}$, where $P$ is the set of operations removing a pair of facts $f,g\not\in s(D)$ of $B'$, and $B'\subseteq B$ is obtained by applying the first $r-1$ operations of $s$ over $B$, and excluding the fact of $s(D)\cap B$ (if such a fact exists),
\item $\lambda'(v'_{i_t})=(\alpha,p)$ if $B\cap s(D)=\{\alpha\}$, $B$ is the $i$-th block of $D$, and the permutation of the operations of $s$ over the blocks $B_1,\dots,B_i$ is the $p$-th in $\prec_{s_{\le i-1},s_i}$, where $s_{\le i-1}$ is the sequence of operations in $s$ over the blocks $B_1,\dots,B_{i-1}$ and $s_i$ is the sequence of operations in $s$ over the block $B_i$,
\item $\lambda'(v'_{i_t})=(\bot,p)$ if $B\cap s(D)=\emptyset$, $B$ is the $i$-th block of $D$, and then the permutation of the operations of $s$ over the blocks $B_1,\dots,B_i$ is the $p$-th in $\prec_{s_{\le i-1},s_i}$, where $s_{\le i-1}$ is the sequence of operations in $s$ over the blocks $B_1,\dots,B_{i-1}$ and $s_i$ is the sequence of operations in $s$ over the block $B_i$.
\end{itemize}
The order of such subsequences in $v'_1 \rightarrow \dots \rightarrow v'_m$ must be coherent with the order $B_1,\dots,B_n$: the first subsequence corresponds to the first block of the first relation name $R_i$ with $v'$ being the $\prec_{T^H}$-minimal covering vertex of the atom $R_i(\bar y_i)$ of $Q$, etc.

Since $H$ is complete, every atom $R_i(\bar y_i)$ of $Q$ has a $\prec_{T^H}$-minimal covering vertex, and since $(D,Q,H)$ is in normal form, every relation of $D$ occurs in $Q$. Hence, every block of $D$ has a representative path of nodes in $O$. Since $\bar c \in Q(s(D))$, there exists a homomorphism $h$ from $Q$ to $D$ with $h(\bar x)=\bar c$. Hence, there is an accepting computation of $M^k_S$ with $O$ being its output, that is obtained in the following way via Algorithm~\ref{alg:sequences}:
\begin{enumerate}
    \item We guess $N=|s|$ (i.e.,~the length of the sequence $s$) in line~2,
    \item for every node $v$, assuming $\lambda(v) = \{ R_{i_1}(\bar y_{i_1}), \ldots, R_{i_\ell}(\bar y_{i_\ell})\}$, we guess the set $A' =$ $\{ \bar y_{i_1} \mapsto h(\bar y_{i_1}),\ldots,\bar y_{i_\ell} \mapsto h(\bar y_{i_\ell})\}$ in line~3,
    \item for every block $B$, we guess $\alpha=s(D)\cap B$ if $s(D)\cap B\neq\emptyset$ or $\alpha=\bot$ otherwise in lines~7---9,
    \item for every block $B$ with $v'_{i_1} \rightarrow \dots \rightarrow v'_{i_t}$ being its representative path in $O$, in the $r$-th iteration of the while loop of line~12 over this block, for all $1\le r\le t-1$, we guess $-g$ in line~14 and $p$ in line~15 assuming that $\lambda'(v'_{i_r})=(-g,p)$,
    \item for every block $B$ with $v'_{i_1} \rightarrow \dots \rightarrow v'_{i_t}$ being its representative path in $O$, we guess $p$ in line~18 assuming $\lambda'(v'_{i_t})=(\alpha,p)$,
    \item we guess the total number $p$ of operations of $s$ over the blocks $B$ such that $B\in \block{\dep}{R_i,D}$ and the $\prec_{T^H}$-minimal covering vertex of the atom $R_i(\bar y_i)$ of $Q$ occurs in the subtree rooted at $u$ (the first child of $v$) in line~21.
\end{enumerate}

The number $N$ that we guess in line~2 clearly belongs to $[|D|]$ because $|s|\le |D|$ for every complete repairing sequence $s$. Moreover, since $h$ is a homomorphism from $Q$ to $D$ with $h(\bar x)=\bar c$, for every atom $R_i(\bar y_i)$ of $Q$ it holds that $R_i(h(\bar y_i))\in D$; hence, it is always possible to choose a set $A'$ in line~3 that is coherent with the homomorphism as well as with the previous choice and with the mapping $\bar x\mapsto \bar c$---for every atom $R_i(\bar y_i)$ that we consider, we include the mapping $\bar y_i \mapsto h(\bar y_i)$ in the set. In addition, for every block $B$, since $s$ is a complete repairing sequence, $s(D)\cap B$ either contains a single fact or no facts at all. Since $\bar c\in Q(s(D))$, as witnessed by $h$, if $R_i(h(\bar y_i))\in B$, then $s(D)$ contains the fact $R_i(h(\bar y_i))$ and this is the only fact of $s(D)\cap B$. In this case, we will define $\alpha=R_i(h(\bar y_i))$ in line~7 or line~8, when considering the block $B$ in line~6.

Next, if $v'_{i_1},\dots,v'_{i_t}$ is the path of nodes in $O$ representing the block $B$, and $\lambda'(v'_{i_r})=(-g,p)$ for some $r\in[t-1]$, then $p\in [\#\mathsf{opsFor}(n,-g)]$. This holds since we start with $n=|B|$ if $\alpha=\bot$ or $n=|B|-1$ otherwise, and decrease, in line~17, the value of $n$ by $1$ when guessing $-1$ in line~14 (i.e,~when applying an operation that removes a single fact) or by $2$ when guessing $-2$ in line~14 (i.e.,~when applying an operation that removes a pair of facts); hence, at the $r$-th iteration of the while loop of line~12 over $B$, it holds that $n$ is precisely the size of the subset $B'$ that is obtained by applying the first $r-1$ operations of $s$ over $B$, and excluding the fact of $s(D)\cap B$ if such a fact exists. This implies that $\#\mathsf{opsFor}(n,-g)$ is precisely the number of operations over $B'$ that remove a single fact if $-g=-1$ (there are $n$ such operations, one for each fact of $B'$) or the number of operations over $B'$ that remove a pair of facts if $-g=-2$ (there are $\frac{n (n-1)}{2}$ such operations, one for every pair of facts of $B'$). Note that the number of possible operations only depends on the size of a block and not its specific set of facts, as all the facts of a block are in conflict with each other. Since $\prec_P$ is a total order over precisely these operations, we conclude that $p\in [\#\mathsf{opsFor}(n,-g)]$.

Next, we argue that if $\lambda'(v'_{i_t})=(\alpha,p)$ (or $\lambda'(v'_{i_t})=(\bot,p)$), then $p$ is a value in $[{b \choose b'}]$.
An important observation here is that whenever we reach line~18, and assuming that $B$ is the block that is currently being handled and it is the $i$-th block of $D$ (w.r.t.~the order $B_1,\dots,B_m$ over the blocks), we have that  $b'$ is the number of operations in $s$ over the blocks $B_1,\dots,B_{i-1}$ and $b$ is the number of operations in $s$ over the blocks $B_1,\dots,B_i$.
This holds since we start with $b=b'=0$, and increase the value of $b$ by one in every iteration of the while loop of line~12 (hence, with every operation that we apply over a certain block). When we are done handling a block, we define $b'=b$ in line~11, before handling the next block. Moreover, for each relation name $R_j$, we handle the blocks of $\block{\dep}{R_j,D}$ in the fixed order defined over these blocks, and for each node $v$, we handle the relation names $R_j$ with $v$ being the $\prec_{T^H}$-minimal covering vertex of the atom $R_j(\bar y_j)$ of $Q$ in the fixed order defined over these relation names. 

In lines~20---25, when we consider the children of the current node $v$, we guess the number $p$ of operations for the first child of $v$ (hence, for all the blocks $B'\in \block{\dep}{R_j,D}$ such that the $\prec_{T^H}$-minimal covering vertex of the atom $R_j(\bar y_j)$ of $Q$ is in the subtree rooted at $v$, as each block is handled under its $\prec_{T^H}$-minimal covering vertex), and we define $b=b+p$ for the second child. By doing so, we take into account all the operations over the blocks that occur before the blocks handled under the second child of $v$ in the order, even before they are actually handled. Note that all the blocks handled under the first child appear before the blocks handled under the second child, as the first child appears before the second child in the order $\prec_{T^H}$ by definition. Hence, whenever we consider a certain block in line~6, $b$ always holds the total number of operations applied over the  previous blocks in the order, in line~11 we define $b'=b$, and when we finish handling this block in the last iteration of the while loop of line~12, $b$ holds the total number of operations applied over the previous blocks in the order and the current block.

Now, in our procedure for constructing $O$, the value $p$ that we choose for $\lambda'(v_{i_t}')$ is determined by the number of permutations of the sequences $s_{\le i-1}$ and $s_i$, where $s_{\le i-1}$ is the sequence of operations in $s$ over the blocks $B_1,\dots,B_{i-i}$ and $s_i$ is the sequence of operations over $B_i$. As said, the number of operations in $s_{\le i-1}$ is the value of $b'$ defined in line~11 when considering the block $B$ in line~6, and the number of operations in $s_{\le i}$ is precisely the value of $b$ after the last iteration of the while loop of line~12. The number of ways to interleave these sequences, which is the number of permutations of these operations that are coherent with the order of the operations in both $s_{\le i-1}$ and $s_i$ is precisely ${b\choose b'}$, and so $p\in[{b\choose b'}]$.

Finally, the number $p$ that we guess in line~21 is clearly in the range $0,\dots,N$, as the number of operations of $s$ over the blocks $B$ that are handled under some subtree of $T^H$ is bounded by the total number of operations in $s$. 
Note that for every leaf of the described computation we have that $N=0$ in line~27, as for each node $v$, the number $N$ of operations that we assign to it is precisely the number of operations over the blocks handled under its subtree. We conclude that the described computation is indeed an accepting computation, and since \emph{(1)} the described computation and $O$ consider the blocks of $D$ in the same order $B_1,\dots,B_n$, each block under its $\prec_{T^H}$-minimal covering vertex, and \emph{(2)} the only labeling configurations are those that produce the labels in line~16 and line~19, and we use the labels of $O$ to define the computation, it is now easy to verify that $O$ is the output of this computation.

It remains to show that two different sequences $s$ and $s'$ give rise to two different outputs $O$ and $O'$. Assume that $s=(o_1,\dots,o_m)$ and $s'=(o_1',\dots,o_\ell')$. If $m\neq \ell$, then $O\neq O'$ since the number of nodes in $O$ is $m+n$, where $n$ is the total number of blocks of $D$ (one node for each operation and an additional node for each block), and the number of nodes in $O'$ is $\ell+n$. If $m=\ell$, there are three possible cases. There first is when there is a block $B$ such that the operations of $s$ over $B$ and the operations of $s'$ over $B$ conform to different templates or $s(D)\cap B\neq s'(D)\cap B$. In this case, the path of nodes in $O$ representing the block $B$ and the path of nodes in $O'$ representing this block will have different labels, as at some point they will use a different $-g$, or a different $\alpha$ for the last node. 

The second case is when $s$ and $s'$ conform to the same templates over all the blocks and $s(D)\cap B=s'(D)\cap B$ for every block $B$, but they differ on the operations over some block $B$. Assume that they differ on the $r$-th operation over $B$, and this is the first operation over this block they differ on. In this case, the $r$-th node in the representative path of $B$ in $O$ and the $r$-th node in the representative path of $B$ in $O'$ will have different labels, as they will use a different $p$. This is the case since when determining the value of $p$ we consider the same subset $B'$ of $B$ in both cases (as the first $r-1$ operations over $B$ are the same for both sequences and $s(D)\cap B=s'(D)\cap B$), and the same $-g$ (as they conform to the same template), and so the same set $P$ of operations. The value of $p$ then depends on the $r$-th operation of $s$ over $B$ (respectively, the $r$-th operation of $s'$ over $B$) and the order $\prec_P$ over the operations of $P$, and since these are two different operations, the values will be different. 

The last case is when $s$ and $s'$ use the exact same operations in the same order for every block. In this case, the only difference between $s$ and $s'$ is that they interleave the operations over the different blocks differently. For each $i\in[m]$, let $s_{\le i-1}$ be the sequence of operations in $s$ over the blocks $B_1,\dots,B_{i-1}$, and let $s'_{\le i-1}$ be the sequence of operations in $s'$ over these blocks. Let $B_j$ be the first block such that $s_{\le j-1}=s'_{\le j-1}$, but $s_{\le j}\neq s'_{\le j}$. Note that $s_j$ (the sequence of operations of $s$ over $B_j$) and $s'_j$ (the sequence of operations of $s'$ over $B_j$ are the same). In this case, the last node in the representative path of $B_j$ in $O$ and the last node in the representative path of $B_j$ in $O'$ will have a different label. They will share the same $\alpha$ since they apply the same operations; hence keep the same fact of this block (or remove all its facts). However, the value of $p$ will be different as it is determined by the sequence $s_{\le j}$ (which is the same as $s'_{\le j}$, the sequence $s_j$ (which is the same as $s'_j$), and the order $\prec_{s_{\le j-1}, s_j}$ over all the ways to interleave these sequences. Since $s_{\le j}\neq s'_{\le j}$, these sequences are interleaved in different ways, and the values will be different.

\medskip
\noindent\paragraph{Valid Outputs to Sequences.} Next, let $O=(V',E',\lambda')$ be a valid output of $M^k_S$ on $D,\dep,Q,H,\bar c$. In each computation, we traverse the nodes of $T^H$ in the order defined by $\prec_{T^H}$, for each node $v$ we go over the relation names $R_i$ such that $v$ is the $\prec_{T^H}$-minimal covering vertex of the atom $R_i(\bar y_i)$ of $Q$ in some fixed order (there is at least one such relation name for every $v$ since $(D,Q,H)$ is in normal form), and for each relation name $R_i$ we go over the blocks of $\block{\dep}{R_i,D}$ in some fixed order. For every block, we produce a set of labels that must all occur along a single path of the output of the computation, as all the guesses we make along the way are existential. Since $O$ is the output of some computation of $M^k_S$ on $D,\dep,Q,H,\bar c$, it must be of the structure described in the proof of the previous direction. That is, for every node $v$ of $T^H$, we have a path of nodes in $O$ representing the blocks handled under $v$, and under the last node representing the last block of $v$, we have two children, each representing the first block handled by a child of $v$ in $T^H$. Note that since $H$ is complete, every relation name $R_i$ is such that the atom $R_i(\bar y_i)$ of $Q$ has a $\prec_{T^H}$-minimal covering vertex, and so every block of $D$ has a representative path in $O$.

It is now not hard to convert an output $O=(V',E',\lambda')$ of an accepting computation into a complete repairing sequence $s$ with $\bar c\in Q(s(D))$. In particular, for every node $v$ of $T^H$, every relation name $R_i$ such that $v$ is the $\prec_{T^H}$-minimal covering vertex of the atom $R_i(\bar y_i)$ of $Q$, and every block $B\in\block{\dep}{R_i,D}$, we consider the path $v'_1 \rightarrow \dots \rightarrow v'_m$ of vertices in $V'$ representing this block. Clearly, the first $m-1$ vertices are such that $\lambda'(v_i')$ is of the form $(-g,p)$ (these labels are produced in the while loop of line~12), and the last vertex is such that $\lambda'(v_m')$ is of the form $(\alpha,p)$ (this label is produced in line~19). Hence, the sequence $s$ will have $m-1$ operations over the block $B$, such that the $r$-th operation is a removal of a single fact if $\lambda'(v_r')=(-1,p)$ for some $p$, or it is the removal of a pair of facts if $\lambda'(v_r')=(-2,p)$ for some $p$. The specific operation is determined by the value $p$ in the following way. 

If after applying the first $r-1$ operations over $B$, and removing the fact $\alpha$ of $\lambda'(v_m')$ (if $\alpha\neq\bot$) from the result, we obtain the subset $B'$, then the operation that we choose is the $p$-th operation in $\prec_P$, where $P$ is the set of all operations removing a single fact of $B'$ if $\lambda'(v_r')=(-1,p)$ or it is the set of all operations removing a pair of facts of $B'$ if $\lambda'(v_r')=(-2,p)$. Note that such an operation exists since when we start handling the block $B$ in line~6, we define $n=|B|$ if $\alpha=\bot$ or $n=|B|-1$ otherwise, and we decrease the value of $n$ by one when we guess $-1$ in line~14 and by $2$ when we guess $-2$. Thus, in the $r$-th iteration of the while loop of line~12 over this block, $n$ is precisely the size of the subset $B'$ of $B$ as described above. Then, the $-g$ that we guess in this operation depends on the size of $B'$ and the chosen $\alpha$. If $n>1$, then the removal of any single fact of $B'$ and the removal of any pair of facts of $B'$ are all $D$-justified operations, as all the facts of $B'$ are in conflict with each other; in this case, $-g\in\{-1,-2\}$. If $n=1$, the only $D$-justified operation is the removal of the only fact of $B'$ in the case where $\alpha\neq\bot$, as these two facts are in conflict with each other; in this case, $-g=-1$. If $n=1$ and $\alpha=\bot$ or $n=0$, there are no $D$-justified operations over $B'$.
Hence, we only choose $-1$ (respectively, $-2$) if there is at least one $D$-justified operation that removes a single fact (respectively, a pair of facts), as determined by $\mathsf{shape}(n,\alpha)$. The value $p$ is from the range $[\#\mathsf{opsFor}(n,-g)]$, where $\#\mathsf{opsFor}(n,-g)$ is precisely the number of such operations available. If $-g=-1$, then the number of available operations is $n$---one for each fact of $B'$. If $-g=-2$, then the number of available operations is $\frac{n (n-1)}{2}$---one for each pair of facts of $B'$. Note that the number of operations only depends on $n$ and not on the specific block, as all the facts of a block are in conflict with each other.

As for the last node $v_m'$, $\lambda'(v_m')=(\alpha,p)$ determines the final state of the block $B$, as well as the positions of the operations over $B$ in the sequence established up to that point (as we will explain later). Since we always disregard the fact of $\alpha$ when choosing the next operation, clearly none of the operations of $s$ over $B$ removes this fact. Hence, we eventually have that $s(D)\cap B=\alpha$ if $\alpha\neq\bot$, or $s(D)\cap B=\emptyset$ if $\alpha=\bot$ (in this case, we allow removing any fact of $B'$ in the sequence). Note that if $\alpha=\bot$, $|B'|=2$, and we choose an operation that removes a single fact, this computation will be rejecting, as in the next iteration of the while loop we will have that $n=1$ and $\alpha=\bot$, and, in this case, $\mathsf{shape}(n,\alpha)=\emptyset$. Hence, if $\alpha=\bot$, the last operation of $s$ over $D$ will always be the removal of a pair of facts when precisely two facts of the block are left, and we will have that $s(D)\cap B=\emptyset$.

We have explained how we choose, for every block of $D$, the sequence of operations in $s$ over this block, and it is only left to show how to interleave these operations in the sequence. We do so by following the order $B_1,\dots,B_m$ over the blocks. We start with the block $B_1$ and denote by $s_{\le 1}$ the sequence of operations of $s$ over $B_1$. This sequence of operations has already been determined. Next, when we consider the node $B_i$ for some $i>1$, we denote by $s_i$ the sequence of operations over this block, and by $s_{\le i-1}$ the sequence of operations over the blocks $B_1,\dots,B_{i-1}$. If $s_{\le i}$ is of length $b$ and $s_{\le i-1}$ is of length $b'$, we have $b$ positions in the sequence $s_{\le i}$, and all we have to do is to choose, among them, the $b'$ positions for the sequence $s_{\le i-1}$, and the rest of the positions will be filled by the sequence $s_i$; the number of choices is ${b\choose b'}$. As argued in the proof of the previous direction, when we handle the block $B_i$, in line~18 the value of $b$ is the number of operations in $s_{\le i}$ and the value of $b'$ is the number of operations in $s_{\le i-1}$; hence, if the path of nodes representing the block $B_i$ is $v_1' \rightarrow \dots \rightarrow v_m'$ and $\lambda'(v_m')=(\alpha,p)$, $p$ is a number in $[{b\choose b'}]$, and it uniquely determines the sequence $s_{\le i}$; this is the $p$-th permutation in $\prec_{s_{\le i-1},s_i}$. Finally, we define $s=s_{\le n}$.

It is left to show that $s$ is a complete repairing sequence, and that $\bar c\in Q(s(D))$. Since for every block $B$ with a representative path $v_1' \rightarrow \dots \rightarrow v_m'$ we have that $\lambda'(v_m')=(\alpha,p)$, where $\alpha$ is either a single fact of $B$ or $\bot$, $s(D)$ contains, for every block of $D$, either a single fact or no facts. Since there are no conflicts among different blocks, we conclude that $s$ is indeed a complete repairing sequence. Moreover, since we consider a complete hypertree decomposition, every atom $R_i(\bar y_i)$ of $Q$ has a $\prec_{T^H}$-minimal covering vertex $v$. By definition of covering vertex, we have that $\bar y_i\subseteq \chi^H(v)$ and $R_i(y_i)\in\lambda^H(v)$. In the computation, for each vertex $v$ of $T^H$, we guess a set $A' =$ $\{ \bar y_{i_1} \mapsto \bar c_1,\ldots,\bar y_{i_\ell} \mapsto \bar c_\ell\}$ assuming $\lambda(v) = \{ R_{i_1}(\bar y_{i_1}), \ldots, R_{i_\ell}(\bar y_{i_\ell})\}$, such that $R_{i_j}(\bar c_j) \in D$ for $j \in [\ell]$ and all the mappings are coherent with $\bar x\mapsto \bar c$. In particular, when we consider the $\prec_{T^H}$-minimal covering vertex of the atom $R_i(y_i)$ of $Q$, we choose some mapping $\bar y_i\mapsto \bar c_i$ such that $R_i(\bar c_i)\in D$. When we consider the block $B$ of $R_i(\bar c_i)$ in line~6, we choose $\alpha=R_i(\bar c_i)$ in line~8, which ensures that eventually $s(D)\cap B=R_i(\bar c)$; that is, we keep this fact in the repair. Since $Q$ is self-join-free, the choice of which fact to keep from the other blocks of $\block{\dep}{R_i,D}$ can be arbitrary (we can also remove all the facts of these blocks).  When we transition from a node $v$ of $T^H$ to its children, we define $A=A'$, and then, for each child, we choose a set $A$ in line~3 that is consistent with $A'$ and with $\bar x\mapsto \bar c$.

This ensures that two atoms $R_i(\bar y_i)$ and $R_j(\bar y_j)$ of $Q$ that share at least one variable are mapped to facts $R_i(\bar c_i)$ and $R_j(\bar c_j)$ of $D$ in a consistent way; that is, a variable that occurs in both $\bar y_i$ and $\bar y_j$ is mapped to the same constant in $\bar c_i$ and $\bar c_j$. Let $v$ and $u$ be the 
$\prec_{T^H}$-minimal covering vertices of $R_i(\bar y_i)$ and $R_j(\bar y_j)$, respectively. Let $z$ be a variable that occurs in both $\bar y_i$ and $\bar y_j$. By definition of covering vertex, we have that $\bar y_i\subseteq \chi^H(v)$ and $\bar y_j\subseteq \chi^H(u)$. By definition of hypertree decomposition, the set of vertices $s\in T^H$ for which $z\in\chi^H(s)$ induces a (connected) subtree of $T^H$. Therefore, there exists a subtree of $T^H$ that contains both vertices $v$ and $u$, and such that $z\in\chi^H(s)$ for every vertex $s$ of this substree. When we process the root $r$ of this subtree, we map the variable $z$ in the set $A'$ to some constant (since there must exist an atom $R_k(\bar y_k)$ in $\lambda(r)$ such that $z$ occurs in $\bar y_k$ by definition of a hypertree decomposition). Then, when we transition to the children of $r$ in $T^H$, we ensure that the variable $z$ is mapped to the same constant, by choosing a new set $A'$ of mappings that is consistent with the set of mappings chosen for $r$, and so fourth. Hence, the variable $z$ will be consistently mapped to the same constant in this entire subtree, and for every atom $R_i(\bar y_i)$ of $Q$ that mentions the variable $z$ it must be the case that the 
$\prec_{T^H}$-minimal covering vertex $v$ of $R_i(\bar y_i)$ occurs in this subtree because $z\in\chi^H(v)$. Hence, the sets $A$ of mappings that we choose in an accepting computations map the atoms of the query to facts of the database in a consistent way, and this induces a homomorphism $h$ from $Q$ to $D$. And since we choose mappings that are coherent with the mapping $\bar x\mapsto\bar c$ of the answer variables of $Q$, we conclude that $\bar c\in Q(s(D))$.

Finally, we show that two different valid outputs $O=(V^O,E^O,\lambda^O)$ and $O'=(V^{O'},E^{O'},\lambda^{O'})$ give rise to two different complete repairing sequences $s$ and $s'$, respectively. If the values of $N$ guessed in line~2 in the computations of $O$ and $O'$ are different, then clearly $s\neq s'$, as the length of the sequence is determined by $N$. This holds since we decrease this value by one with every operation that we apply, and only accept if $N=0$ for every leaf of $T^H$; that is, if all our guesses for the number of operations in the computation are coherent with the sequence. Similarly,
if the length of the representative path of some block $B$ in $O$ is different from the length of the representative path of $B$ in $O'$, then $s\neq s'$, as they differ on the number of operations over $B$. If both sequences have the same length and the same number of operations for each block, then there are three possible cases. 
The first is when there exists a block $B$ such that the representative path $v_1 \rightarrow \dots \rightarrow v_m$ of $B$ in $O$ and the representative path $u_1 \rightarrow \dots \rightarrow u_m$ of $B$ in $O'$ are such that for some $i\in[m-1]$, if $\lambda^O(v_i)=(-g,p)$ and $\lambda^{O'}(u_i)=(-g',p')$ then $-g\neq -g'$, or $\lambda^O(v_m)=(\alpha,p)$ and $\lambda^{O'}(u_m)=(\alpha',p')$ with $\alpha\neq\alpha'$.
In this case, the operations of $s$ over $B$ and the operations of $s'$ over $B$ conform to different templates or keep a different fact of $B$ in the repair; hence, we again conclude that $s\neq s'$.

The second case is when for every block $B$, if $v_1 \rightarrow \dots \rightarrow v_m$ is the path in $O$ representing the block $B$ and $u_1 \rightarrow \dots \rightarrow u_m$ is the path in $O'$ representing the block $B$, then for every $i\in[m-1]$, if $\lambda^O(v_i)=(-g,p)$ and $\lambda^{O'}(u_i)=(-g',p')$, then $-g=-g'$, but for some $i\in[m-1]$, $p\neq p'$. Assume that this holds for the $r$-th node in the path, and this is the first node they differ on; that is, $\lambda^O(v_i)=\lambda^{O'}(u_i)$ for all $i\in[r-1]$, but $\lambda^O(v_r)\neq\lambda^{O'}(u_r)$.
In this case, the $r$-th operation of $s$ over $B$ and the $r$-th operation of $s'$ over $B$ will be different. This is the case since when determining the $r$-th operation of the sequence over $B$, we consider the same subset $B'$ of $B$ in both cases (as the first $r-1$ operations over $B$ are the same for both sequences) and the same $\alpha$, and so the same set $P$ of operations. The $r$-th operation in $s$ (respectively, $s'$) over $B$ is the $p$-th (respectively, $p'$-th) operation in the order $\prec_P$ over such operations, and since $p\neq p'$, these are different operations.

The last case is when for some block $B$, with $v_1 \rightarrow \dots \rightarrow v_m$ being its representative path in $O$ and $u_1 \rightarrow \dots \rightarrow u_m$ being its representative path in $O'$, we have that $\lambda'^O(u_i)=\lambda^O(v_i)$ for all $i\in[m-1]$, $\lambda^O(v_m)=(\alpha,p)$, $\lambda^{O'}(v_m)=(\alpha',p')$, $\alpha=\alpha'$, and $p\neq p'$. In this case, we consider the first block $B$ in the order $B_1,\dots,B_m$ for which this is the case. For each $i\in[m]$, let $s_{\le i-1}$ be the sequence of operations in $s$ over the blocks $B_1,\dots,B_{i-1}$, and let $s'_{\le i-1}$ be the sequence of operations in $s'$ over these blocks. Clearly, $s_{\le i-1}=s'_{\le i-1}$, as for all the previous blocks, $O$ and $O'$ agree on all the labels along the representative path of the block. However, the sequences $s_{\le i}$ and $s'_{\le i}$ will be different, and since they are subsequences of $s$ and $s'$, respectively, we will conclude that $s\neq s'$.
This is the case since the way we choose to interleave the sequences $s_{\le i-1}$ (resp., $s'_{\le i-1}$) and $s_i$ (resp., $s'_i$) depends on the values $p$ and $p'$; that is, we choose the $p$-th (resp., $p'$-th) permutation in the order 
 $\prec_{s_{\le i-1}, s_i}$ (note that $s_i=s'_i$). Since $p\neq p'$, these will be two different permutation, and $s_i\neq s'_i$. This concludes our proof.
\begin{algorithm}[t]
	\SetInd{0.7em}{0.7em}
	\DontPrintSemicolon
	\SetArgSty{textnormal}
		\KwIn{A database $D$, a CQ $Q(\bar x)$, and a generalized hypertree decomposition $H=(T,\chi,\lambda)$ of $Q$ of width $k$, where $(D,Q,H)$ is in normal form}
		\vspace{2mm}
		
		{$v := \mathsf{root}(T)$; $A := \emptyset$;\\}
		\vspace{1mm}
		
		{Assuming $\lambda(v) = \{ R_{i_1}(\bar y_{i_1}), \ldots, R_{i_\ell}(\bar y_{i_\ell})\}$, guess a set $A' =$ $\{ \bar y_{i_1} \mapsto \bar c_1,\ldots,\bar y_{i_\ell} \mapsto \bar c_\ell\}$, with $R_{i_j}(\bar c_j) \in D$ for $j \in [\ell]$, and verify $A \cup A'$ is coherent; if not, \textbf{reject};\\}
		
		\vspace{1mm}
		\ForEach{$x \in \chi(v) \cap \bar x$}{
			{Assuming $c$ is the constant to which $x$ is mapped in $A'$, Label with $(x,c)$;\\}
			
		}
		
		\If{$v$ is not a leaf of $T$}{
			Assuming $u_1,u_2$ are the (only) children of $v$ in $T$, universally guess $i \in \{1,2\}$;\\
			{$v := u_i$; $A := A'$;\\}
			{\textbf{goto} line 2;\\}
		}\lElse{\textbf{accept};}
		\caption{The alternating procedure $\mathsf{GHWCQ}[k]$}\label{alg:ghwcq}
\end{algorithm}

\section{Other Database Problems in $\spantl$}

In this section, we show the following:

\begin{proposition}\label{pro:other-prob-spantl}
	For every $k > 0$, $\sharp \mathsf{GHWCQ}[k],\sharp\mathsf{UR}[k] \in \spantl$, and unless $\nlogspace = \logcfl$, $\sharp \mathsf{GHWCQ}[k],\sharp\mathsf{UR}[k] \not \in \spanl$.
\end{proposition}

In what follows, fix an integer $k > 0$. To prove the second part of Proposition~\ref{pro:other-prob-spantl} we can use similar arguments used to prove Proposition~\ref{pro:notinspanl}. In particular, it is enough to observe that the decision problems $\sharp \mathsf{GHWCQ}[k]_{>0}$ and $\sharp\mathsf{UR}[k]_{>0}$, associated to $\sharp \mathsf{GHWCQ}[k]$ and $\sharp\mathsf{UR}[k]$, respectively, are both $\logcfl\hard$. This is true, since $\sharp \mathsf{GHWCQ}[k]_{>0}$ corresponds to the problem of checking whether there exists at least a tuple $\bar c$ that is an answer to a given CQ $Q$ over a database $D$, when a generalized hypertree decomposition of $Q$ of width $k$ is given~\cite{GoLS02}, while $\sharp\mathsf{UR}[k]_{>0}$ is the problem of checking whether a given tuple $\bar c$ is an answer to a CQ $Q$ over a database $D$, when a generalized hypertree decomposition of $Q$ of width $k$ is given~\cite{GoLS02}. Hence, using the same arguments used in the proof of Proposition~\ref{pro:notinspanl}, we conclude that unless $\nlogspace = \logcfl$,  $\sharp \mathsf{GHWCQ}[k]$ and $\sharp\mathsf{UR}[k]$ are not in $\spanl$.

For the membership of $\sharp \mathsf{GHWCQ}[k]$ and $\sharp\mathsf{UR}[k]$ in $\spantl$, we present alternating procedures with output for the problems $\sharp \mathsf{GHWCQ}[k]$ and $\sharp\mathsf{UR}[k]$. The procedures are similar in spirit to the procedure $\mathsf{Rep}[k]$ depicted in Algorithm~\ref{alg:repairs} used in the main body to show that the problem $\sharp \mathsf{Repairs}[k]$ is in $\spantl$. To prove that such procedures are well-behaved, and that the number of accepted outputs coincides with the count of the corresponding problem, one can use very similar arguments as the ones we use for the proof of Lemma~\ref{lem:repairs-ato}.

	\begin{algorithm}[t]
	\SetInd{0.7em}{0.7em}
	\DontPrintSemicolon
	\SetArgSty{textnormal}
		\KwIn{A database $D$, a CQ $Q(\bar x)$ from $\sjf$, a generalized hypertree decomposition $H=(T,\chi,\lambda)$ of $Q$ of width $k$, and $\bar c \in \adom{D}^{|\bar x|}$, where $(D,Q,H)$ is in normal form}
		\vspace{2mm}
		
		{$v := \mathsf{root}(T)$; $A := \emptyset$;\\}
		\vspace{1mm}
		
		{Assuming $\lambda(v) = \{ R_{i_1}(\bar y_{i_1}), \ldots, R_{i_\ell}(\bar y_{i_\ell})\}$, guess a set $A' =$ $\{ \bar y_{i_1} \mapsto \bar c_1,\ldots,\bar y_{i_\ell} \mapsto \bar c_\ell\}$, with $R_{i_j}(\bar c_j) \in D$ for $j \in [\ell]$, and verify $A \cup A' \cup \{\bar x \mapsto \bar c\}$ is coherent; if not, \textbf{reject};\\}
		
		\vspace{1mm}
		\For{$j=1, \ldots, \ell$}{
			\If{$v$ is the $\prec_T$-minimal covering vertex for $R_{i_j}(\bar y_{i_j})$}{
				\ForEach{atom $\beta \in D$ with predicate $R_{i_j}$}{
					\lIf{$\beta = R_{i_j}(\bar c_j)$}{$\alpha := \beta$}
					\lElse{Guess $\alpha \in \{\beta,\bot\}$;}
					
					{Label with $\alpha$;\\}
				}
			}
		}
		
		\If{$v$ is not a leaf of $T$}{
			Assuming $u_1,u_2$ are the (only) children of $v$ in $T$, universally guess $i \in \{1,2\}$;\\
			{$v := u_i$; $A := A'$;\\}
			{\textbf{goto} line 2;\\}
		}\lElse{\textbf{accept};}
		\caption{The alternating procedure $\mathsf{UR}[k]$}\label{alg:ur}
\end{algorithm}

\medskip
\noindent \textbf{Counting Answers to CQs.}
We start by focusing on the problem  $\sharp \mathsf{GHWCQ}[k]$, i.e., the probem that given as input a database $D$, a CQ $Q(\bar x)$, and a generalized hypertree decomposition $H$ of $Q$ of width $k$, asks for $|Q(D)|$. In particular, for this problem, a generalized hypertree decomposition $H=(T,\chi,\lambda)$ is understood to consider also the output variables $\bar x$ of $Q$, i.e., the tree decomposition $(T,\chi)$ considers all the variables of $Q$ and not only the variables $\var{Q} \setminus \bar x$. Note that this assumption is also needed in~\cite{ACJR21} for providing an FPRAS for the problem $\sharp \mathsf{GHWCQ}[k]$.
Furthermore, as already done for $\sharp \mathsf{Repairs}[k]$, we assume w.l.o.g. that the triple $(D,Q,H)$ is in normal form.

Algorithm~\ref{alg:ghwcq} depicts an alternating procedure with output, dubbed $\mathsf{GHWCQ}[k]$, with inputs $D$, $Q$, and $H$, where $(D,Q,H)$ is in normal form, and whose accepted outputs are node-labeled trees corresponding to the answers in $Q(D)$.
The procedure $\mathsf{GHWCQ}[k]$ follows the same structure as the procedure $\mathsf{Rep}[k]$ of Algorithm~\ref{alg:repairs} in the main body of the paper, i.e., it traverses the generalized hypertree decomposition $H$ of $Q$, and guessess a witness $A'$, at each node $v$ of $H$, of the fact that the atoms in $\lambda(v)$ can be homomorphically mapped to the database $D$. The key difference is that now, when such a witness $A'$ is found for the current node $v$ of $H$, the content of $A'$ is used to produce \emph{part of an answer of $Q$}. This part is the one relative to the atoms in $\lambda(v)$. Hence, a single accepted output of $\mathsf{GHWCQ}[k]$ collects in each of its nodes, a part of a tuple mapping of the form $\bar x \mapsto \bar c$, where $\bar c \in Q(D)$. Each part is encoded as a pair of the form $(x,c)$, with $x \in \bar x$ and $c \in \bar c$. When all such pairs are combined together, we obtain a complete description of the tuple mapping $\bar x \mapsto  \bar c$. Here, for the tuple $\bar c$ to have \emph{all} its parts occuring in some node of an output, it is crucial that the tree decomposition of $H$ considers all variables in $\var{Q}$, and not only the variables $\var{Q} \setminus \bar x$. The fact that $\mathsf{GHWCQ}[k]$ is well-behaved follows by the same arguments used to prove that $\mathsf{Rep}[k]$ is well-behaved.

\medskip
\noindent \textbf{Uniform Reliability.}
We now focus on the problem $\sharp\mathsf{UR}[k]$, i.e., the probem that given as input a database $D$, a CQ $Q(\bar x)$ from $\sjf$, a generalized hypertree decomposition $H$ of $Q$ of width $k$, and a tuple $\bar c \in \adom{D}^{|\bar x|}$, asks for the number of subsets $D'$ of $D$ such that $\bar c \in Q(D')$. In particular, here we assume $H$ is a generalized hypertree decomposition as defined in Section~\ref{sec:preliminaries}. Again, we assume w.l.o.g. that $(D,Q,H)$ is in normal form.

We observe that $\sharp\mathsf{UR}[k]$ is very similar to our problem $\sharp \mathsf{Repairs}[k]$. Indeed, an operational repair $D'$ is nothing else than a special kind of subset of $D$. Hence, it is enough to extend the procedure $\mathsf{Rep}[k]$ for the problem $\sharp \mathsf{Repairs}[k]$ so that it can output an arbitrary subset $D'$ of $D$ for which $\bar c \in Q(D')$. The alternating procedure with output depicted in Algorithm~\ref{alg:ur}, dubbed $\mathsf{UR}[k]$, does exactly this. In fact, one can verify that it is just a slight modification of the procedure $\mathsf{Rep}[k]$ used to show that $\sharp \mathsf{Repairs}[k]$ is in $\spantl$. The key difference is that now, when the set of tuple mappings $A'$ is guessed, the procedure has to choose, for each predicate $R_{i_j}$ in $\lambda(v)$, where $v$ is the node of $H$ for which $A'$ has been guessed, and for each atom in $D$ with predicate $R_{i_j}$, whether this atom is going to be part of the subset $D'$, with the condition that the atom with predicate $R_{i_j}$ whose tuple is already determined by $A'$ must necessarily be part of $D'$. 
The fact that $\mathsf{UR}[k]$ is well-behaved can be proved with the same arguments used to prove the well-behavedness of the procedure $\mathsf{Rep}[k]$.

\bibliographystyle{ACM-Reference-Format}

\bibliography{references}

\end{document}